\documentclass[11pt,a4paper]{article}

\usepackage{jheppub}
\usepackage{macros}

\preprint{CALT-TH 2022-021}

\title{Branes and DAHA Representations}

\author[1]{Sergei Gukov}
\author[2,3]{Peter Koroteev}
\author[4]{Satoshi Nawata}
\author[5]{Du Pei}
\author[6]{Ingmar Saberi}
\date{\today}
\affiliation[1]{Walter Burke Institute for Theoretical Physics, California Institute of Technology,\\ Pasadena, CA 91125, USA}
\affiliation[2]{Department of Mathematics, University of California, 970 Evans Hall,
Berkeley, CA 94720, U.S.A.}
\affiliation[3]{NHETC at Rutgers University, Piscataway, NJ, 08854, U.S.A.}
\affiliation[4]{Department of Physics and Center for Field Theory and Particle Physics, Fudan University, \\
220 Handan Road, 200433 Shanghai, China}
\affiliation[5]{Center of Mathematical Sciences and Applications, Harvard University, \\ 20 Garden Street, Cambridge, MA 02138, U.S.A.}
\affiliation[6]{Ludwig-Maximilians-Universität München,
Fakultät für Physik,
Theresienstraße 37,\\
80333 München,
Deutschland}

\emailAdd{gukov@theory.caltech.edu}
\emailAdd{pkoroteev@math.berkeley.edu}
\emailAdd{snawata@gmail.com}
\emailAdd{dpei@cmsa.fas.harvard.edu}
\emailAdd{i.saberi@physik.uni-muenchen.de}

\abstract{
Using brane quantization, we study the representation theory of the spherical double affine Hecke algebra of type $A_1$  in terms of the topological $A$-model on the moduli space of flat $\SL(2,\bC)$-connections on a once-punctured torus. In particular, we provide an explicit match between finite-dimensional representations and $A$-branes with compact support; one consequence is the discovery of new finite-dimensional indecomposable representations. We proceed to embed the $A$-model story in an M-theory brane construction, closely related to the one used in the 3d/3d correspondence; as a result, we identify modular tensor categories behind particular finite-dimensional representations with $\PSL(2,\Z)$ action. Using a further connection to the fivebrane system for the class $\cS$ construction, we go on to study the relationship of Coulomb branch geometry and algebras of line operators in 4d $\cN=2^*$ theories to the double affine Hecke algebra.
}

\begin{document}
\setcounter{tocdepth}{3}
\maketitle

\newpage
\section{Introduction}
\subsection{Background}

In string theory, the term ``brane'' is used for certain extended objects.
As is typical of string theory, there are many different ways of seeing or defining these objects, depending on one's preferred point of view.
For example, from the target space perspective, where string theory can be thought of as modeling the motion of strings in a target space \MS,
one can picture branes as particular distinguished submanifolds of~\MS{} (decorated with additional data) on which open strings can end.
Relatedly, from the point of view of the string worldsheet, branes are simply boundary conditions of the two-dimensional worldsheet theory.
But branes can also be viewed as sources for higher-form gauge symmetries in the effective field theory of the target space.
In the supergravity approximation, such extended sources produce interesting solutions,
called ``black branes'' by analogy with familiar ``black hole'' solutions in standard general relativity.
This perspective is especially useful in eleven-dimensional M-theory, where a first-quantization perspective (which would replace the string worldsheet by an appropriate ``membrane'' theory) is currently unavailable.

Branes, or at least models of certain special versions of branes, have also made numerous appearances in the mathematics literature, where they may go by different names. For example, topological string theory (which, from the physical point of view, comes from a twist of the worldsheet sigma-model discussed above) comes in two flavors, known as the $A$- and $B$-models. The category of branes in each of these can be identified with a fairly well-defined mathematical structure associated to a Calabi--Yau target space $\MS$. For the $B$-model, this is the derived category of coherent sheaves on~\MS, whereas the $A$-model is expected to be some appropriately defined version of---or generalization of---the Fukaya category $\Fuk(\MS,\omega_\MS)$, where $\omega_\MS$ is the symplectic form. Since this generalization may be nontrivial, we will write $\ABrane(\MS,\omega_\MS)$ for the category of $A$-branes, in which $\Fuk(\MS,\omega_\MS)$ is expected to be a full subcategory.
The homological mirror symmetry proposal of Kontsevich~\cite{KontsevichICM} identifies the category of $A$-branes on a Calabi--Yau threefold with the category of $B$-branes on its mirror, and is the subject of ongoing intense mathematical research.

While the category of $B$-branes belongs squarely to the realm of algebraic geometry, the category of $A$-branes is much more subtle, and has appeared in numerous different guises in mathematical physics. To give another example, the proposed framework of \emph{brane quantization}~\cite{Gukov:2008ve} suggests that the problem of quantizing a symplectic manifold $M$ can be approached by studying the topological $A$-model on a \emph{different} target space \MS, which is chosen to be a so-called ``complexification'' of~$M$. (When $M$ is the set of real points of an algebraic symplectic manifold, this complexification can be taken to be the obvious one.) This complexification should, in any case, be a complex manifold whose dimension is twice that of $M$; $M$ should map to~\MS{}, and \MS{} should be equipped with a holomorphic symplectic form $\Omega$, whose real part $\Re \Omega$ restricts to the symplectic form on~$M$, and imaginary part $\Im\Omega$ restricts to zero on~$M$.

One is then instructed to consider the $A$-model of the complexification with respect to the \emph{imaginary} part of the holomorphic symplectic form, $\omega_\MS = \Im \Omega$. This gives rise to a category $\ABrane(\MS,\omega_\MS)$ of $A$-branes, which includes not only Lagrangian objects but also much more unfamiliar branes supported on \emph{coisotropic} submanifolds of~\MS. Coisotropic branes were introduced in~\cite{Kapustin:2001ij};  \cite{aldi2005coisotropic} conjectured that spaces of morphisms between $A$-branes should be identified with deformation quantizations of the functions on their intersections.
While coisotropic branes remain mysterious in general, and do not occur at all on simply-connected Calabi--Yau three-folds, they are needed for mirror symmetry to work, even on flat target spaces.

In fact, since the dimension of~\MS{} is always zero modulo four, one can define a particularly useful exotic $A$-brane on~\MS, known as the \emph{canonical coisotropic brane}. This brane was introduced in~\cite{Kapustin:2006pk}, where it played an important role in connecting $A$-branes to $D$-modules. Its support is the entire space \MS, and it is furthermore expected to have a very interesting algebra of endomorphisms. In fact, in keeping with the proposal of~\cite{aldi2005coisotropic}, one expects that
\deq{
\End(\Bcc) = \OO^q(\MS),
}
where the object on the right-hand side is the \emph{deformation quantization} of the ring $\OO(\MS)$ of holomorphic functions (with appropriate polynomial growth conditions at infinity) on the complexification, taken with respect to its holomorphic symplectic form. In the case of an affine variety, $\OO(\MS)$  is just the coordinate ring. (Although the $A$-model depends only on the symplectic form $\omega_\MS = \Im \Omega$, the real part of~$\Omega$ enters the definition of the boundary condition~$\Bcc$, which is only canonically definable on a holomorphic symplectic manifold.)

As with any category, there is an action of this algebra by precomposition (physically speaking, by joining strings at boundary conditions) on the space of morphisms from~\Bcc{} to any other $A$-brane \brane.
In other words, brane quantization naturally proposes a functor
\deq[eq:pre-functor]{
\Hom(\Bcc,-): \ABrane(\MS,\omega_\MS) \to \Rep(\OO^q(\MS)),
}
which allows us to generate a representation of this algebra from an $A$-brane.
A category is said to be \emph{generated} by an object $A$ if $\Hom(A,-)$ is an equivalence of categories.
In fact, Kapustin~\cite{Kapustin:2005vs} proposed that $\Bcc$ is a generating object of the category of $A$-branes, and that $\Rep(\OO^q(\MS))$ can be taken as a definition of the category $\ABrane(\MS)$, when \MS{} is a \HK space.
We remark that there are some subtleties here. The Fukaya category as typically studied in homological mirror symmetry~\cite{KontsevichICM} requires each object to carry a choice of grading, so that there is at least a family of $A$-branes supported on the same Lagrangian which are shifts of one another, forming a torsor over~$\Z$. There is typically no canonical choice of a preferred grading datum on an  $A$-brane. One should more properly expect \deq[eq:functor]{
\RHom(\Bcc,-): D^b\ABrane(\MS,\omega_\MS) \to D^b \Rep(\OO^q(\MS))
}
to provide a derived equivalence between the category of $A$-branes and the derived category of $\OO^q(\MS)$-modules. (From the physical perspective, this corresponds to working with the notion of equivalence appropriate to the twist, treating $A$-branes as boundaries for the $A$-twisted theory rather than boundaries for the full theory that are compatible with the twist.)
The relevance of derived categories to boundary conditions in topological string theory has been understood for a long time;  see~\cite[for example]{Douglas:2000gi}.

Returning briefly to the perspective of brane quantization, the gist now consists in the fact that $M$ is a Lagrangian submanifold in $(\MS,\omega_\X)$, so that the original symplectic manifold itself can be used to define an $A$-brane $\brane_M$ in $(\MS,\omega_\X)$. In fact, it is shown in~\cite{Gukov:2008ve,Gukov:2010sw} that the morphism space $\Hom(\Bcc,\brane_M)$ can be identified in a precise fashion with the geometric quantization of~$M$, at least under the assumption that $M$ is a K\"ahler manifold.
As such, brane quantization provides a bridge between deformation quantization---which is guaranteed to formally produce the algebra of quantum observables $\OO^q(\MS)$, but gives no candidate for a natural module or Hilbert space on which it acts---and standard geometric quantization.
(For a recent study of issues in geometric quantization from this perspective, see~\cite{Gaiotto:2021kma}.)
However, as we have already argued, the functor $\Hom(\Bcc,-)$ is \emph{much more} than this: assuming that it is an equivalence, it provides a natural description of the category of~$\OO^q(\MS)$-modules in geometric terms. Indeed, the role of $M$ in the story is no longer distinguished: it is just one $A$-brane among (at least potentially) many, each of which corresponds naturally to an $\OO^q(\MS)$-module.
This broader perspective was already appreciated in~\cite{Gukov:2008ve}, where a particular space $\X=T^*\CP^1$ was used to generalize the orbit method and give geometric constructions for all representations of~$\SL(2,\R)$.
Therefore, the proposed equivalence \eqref{eq:functor} between $A$-branes and $\OO^q(\MS)$-modules is the natural way to think about a geometric approach to representation theory for algebras that deformation-quantize \HK manifolds \X{}.

As the definition of the $A$-brane category is not available yet, much of this discussion is not at a mathematical level of rigor.
Nonetheless, with an appropriate choice of $(\MS,\omega_\MS)$, we can provide concrete evidence for the equivalence \eqref{eq:functor} if we restrict ourselves to Lagrangian objects belonging to the Fukaya category   $\Fuk(\MS,\omega_\MS)$ of $\MS$, which forms a subcategory in $\ABrane(\MS,\omega_\MS)$.
We will take the target space \MS{} of the 2d sigma-model to be the moduli space of complex flat connections (or parabolic Higgs bundles) on a once-punctured torus $C_p$. Then, as proved in \cite{oblomkov2004double}, the algebra $\OO^q(\MS)$ will be the spherical subalgebra of double affine Hecke algebra (DAHA in short) \cite{Cherednik-book}. One of our goals in this paper is to explore the idea described above in this setup, presenting solid evidence for the equivalence \eqref{eq:functor}.\footnote{A related functor of a similar kind is constructed in \cite{ben2016quantum,ben2016betti,ben2018integrating}. The constructions there give a description of the factorization homology of a particular $E_2$ algebra valued in categories in terms of modules.
(One may equivalently think of such an algebra as a braided tensor category).
Taking the braided tensor category to be $\Rep_q GL_n$ and applying the general result to a once-punctured torus, one obtains a Morita equivalence between the spherical DAHA of type $\mathfrak{gl}(N,\C)$ and the endomorphisms of a generating object of the factorization homology.}

\paragraph{Remark:}
In the past few years, Kontsevich and Soibelman \cite{KSFloer} have been developing a new formalism within the framework of `holomorphic Floer theory,' which among other things, allows for a rigorous formulation of brane quantization.
According to the {\it generalized Riemann--Hilbert correspondence} of Kontsevich--Soibelman, there is an embedding of the Fukaya category $\Fuk(\MS)$ into the right-hand side of \eqref{eq:functor} as the category of so-called holonomic $D_q$-modules. Their approach provides a realization of the category of representations of $\OO^q(\X)$ in terms of sheaves on its Lagrangian skeleton. Some of our results in this paper about DAHA representations can thus be interpreted as a particular example of the generalized Riemann-Hilbert correspondence.

\subsection{Results}
We first study the representation theory of spherical double affine Hecke algebra $\SH$ of type $A_1$ from the viewpoint of brane quantization in great detail.
We explicitly identify a compact Lagrangian brane in $\X=\MF(C_p,\SL(2,\C))$, the moduli space of flat $\SL(2,\C)$-connections on~$C_p$, for each finite-dimensional irreducible representation of $\SH$.
In particular, we match objects including the parameter spaces, dimensions and shortening conditions on both sides.
We also study the spaces of derived morphisms of the two categories. As a by-product, we find new finite-dimensional representations of $\SH$ that do not appear in \cite{Cherednik-book}. We see examples in which two irreducible branes can form bound states in more than one way, corresponding to a higher-dimensional $\Ext^1$; these bound states are related to subtleties defining $A$-branes supported on singular submanifolds. Hence, the careful study in \S\ref{sec:2d} in terms of brane quantization provides solid evidence for the following:

\begin{claim}\label{Th:Claim1}
For $\X=\MF(C_p,\SL(2,\C))$, the functor \eqref{eq:functor} restricts to a derived equivalence of the full subcategory of compact Lagrangian $A$-branes of $\MS$ and the category of finite-dimensional $\SH$-modules.
\end{claim}

We also consider a particular example of a non-compact brane corresponding to the polynomial representation of~\SH{} studied by Cherednik. In fact, the brane perspective suggests straightforward generalizations of this representation.

\vskip.1in

While the brane quantization proposal---and thus the physics of the $A$-model---is our starting point,
many of the various other types of branes in string theory and M-theory, and the guises in which they appear, will have a role to play in this paper.
As was already emphasized, just for example, in the constructions of~\cite{Kapustin:2006pk}, the moduli space \MS{} plays an important role in higher-dimensional gauge theories, which allows for an embedding of the physics of $A$-branes into a richer system. We focus on one such construction: M5-branes on a once-punctured torus (or equivalently with $\Omega$-deformation orthogonal to M5-branes on a torus) in an appropriate setup of M-theory. This construction will provide many new angles to view the structure of the category of representations of (spherical) DAHA.

As such, branes will lead us to a geometric interpretation of previously known facts about $\SH$-modules,
as well as to new results, not previously known in the representation theory literature.
It is rather straightforward from the geometry of the target space $\MS$ to identify finite-dimensional $\SH$-modules that carry representations of~$\PSL(2,\Z)$. More interestingly, by connecting the M5-brane setup of the 3d/3d correspondence to the 2d $A$-model, we can naturally identify the corresponding $\PSL(2,\Z)$ representations.
Let us recall the fivebrane setup for the 3d/3d correspondence where M5-branes are located on $S^1\times D^2\times M_3$ with the $\Omega$-background. Then, a suitable compactification on $T^2\times T^2$ can relate this setup to the 2d $A$-model described above, where the center of $D^2$ is associated to $\Bcc$ and a boundary condition at the boundary of $D^2$ gives rise to $\brane_M$.

For various choices of boundary conditions $\brane_M$, the partition function of 3d $\cN=2$ theory $\cT[M_3]$ on $S^1 \times D^2$ computes the corresponding invariant of the 3-manifold $M_3$. In some cases, such topological invariants of 3-manifolds can be lifted to a 3d TQFT, \textit{i.e.}\ can be constructed via cutting-and-gluing. In turn, the algebraic structure underlying a 3d TQFT often can be encoded in a modular tensor category (that, in general, may be non-unitary or non-semisimple). In particular, in the present setup of 3d/3d correspondence, this algebraic structure itself can be viewed as a special case $\MTC [S^1 \times (S^2 \setminus \text{pt}), \brane_M]$ of a more general algebraic structure dubbed $\MTC [M_3]$ in \cite{Gukov:2016gkn} for its close resemblance to the structure of a modular tensor category.
We will explain how concrete instances of this algebraic structure can be realized via branes in the 2d $A$-model and the corresponding $\SH$-modules:
\be
\Hom(\Bcc,\brane_M)\cong K^0(\MTC)~.
\ee
In particular, one such boundary condition leads to a TQFT associated to a 4d Argyres-Douglas theory.
In general, branes supported on $M$ that are invariant under $\PSL(2,\Z)$ action (not pointwise) give rise to interesting $\PSL(2,\Z)$ representations $\Hom(\Bcc,\brane_M)$, and 3d/3d correspondence can help us to relate them to the modular data (and the Grothendieck group) of an $\MTC$-like structure.

\vskip.1in

Another relevant brane setting appears in the class $\cS$ construction \cite{Gaiotto:2009we,Gaiotto:2009hg} of a 4d $\cN=2^*$ theory $\cT[C_p]$ where M5-branes are placed on $S^1\times \bR^3\times C_p$. An algebra of line operators becomes the coordinate ring of the Coulomb branch of 4d $\cN=2^*$ theory on $S^1\times \bR^3$ \cite{Gaiotto:2010be}, and we can study it again in a rank-one case from the relation to $\SH$.
 As in \cite{Aharony:2013hda}, the spectrum of line operators in the 4d $\cN=2^*$ theory is sensitive to the global structure of the gauge group, which can be specified by imposing additional discrete data. In fact, the Coulomb branch of 4d $\cN=2^*$ theory of rank-one is given as the quotient of $\MF(C_p,\SL(2,\bC))$ by this additional discrete choice $\bZ_2\subset \bZ_2\oplus \bZ_2$, which can be interpreted as an automorphism group of $\SH$. Therefore, we can study the elliptic fibration of the Coulomb branch, and the algebra of line operators on the $\Omega$-background is a $\bZ_2$-invariant subalgebra of $\SH$. Furthermore, by introducing a surface operator of codimension two in the system, an algebra of line operators on the surface operator is related to the full (rather than spherical) DAHA. By compactifying the 4d theory to the 2d sigma-model, we propose a canonical coisotropic brane $\wh\frakB_{cc}$ of higher rank where the algebra of $(\wh\frakB_{cc},\wh\frakB_{cc})$-open strings realizes the full DAHA. In this way, the interplay among moduli spaces, algebras of line operators, and DAHA can be studied from the viewpoint of the compactification of fivebrane systems.

\subsection{Structure}

The structure of the paper follows a simple principle. We start in the world of two-dimensional physics, and we gradually proceed to higher-dimensional theories. One advantage of this approach is that lower-dimensional theories can be analyzed much more explicitly and often can be described in mathematically rigorous terms. For example, the two-dimensional sigma-model perspective is phrased in the language of the topological $A$-model, which is reasonably well understood in the mathematical literature. Likewise, many explicit calculations can be done easily and many questions can be answered more concretely in low-dimensional systems. The advantage of higher-dimensional systems, on the other hand, is that they reveal a much richer (higher categorical) structure, that helps to see a ``bigger picture,'' dualities and relations between various low-dimensional descriptions, which otherwise might seem worlds apart.

To give a concrete overview of what follows:
In \S \ref{sec:2d}, we provide a detailed study of the equivalence \eqref{eq:functor} between DAHA representations and $A$-branes. To this end, we study the 2d sigma-model on the moduli space of flat $\SL(2,\C)$-connections (or the Hitchin moduli space) on the punctured torus $C_p$ in this section. We begin by constructing the spherical DAHA $\SH$ in the 2d $A$-model. We review the relevant geometry of the target space in~\S\ref{sec:target}. Then we move on to study the algebraic side, reviewing the double affine Hecke algebra of type $A_1$ and its spherical subalgebra $\SH$ in~\S\ref{sec:DAHA-main}. We introduce the canonical coisotropic brane in \S\ref{sec:Bcc}, showing how the spherical DAHA $\SH$ arises as the algebra of $(\Bcc,\Bcc)$-strings. In the remainder of~\S\ref{sec:2d}, we discuss the match between representations of $\SH$ and open-string states between $A$-branes. To this end, \S\ref{sec:Lagrangian} reviews some details of the category \ABrane, explaining the correspondence between branes supported on Lagrangian submanifolds and modules of $\SH$. In particular, we will find branes for the polynomial representations in \S\ref{sec:poly-rep}.  \S\ref{sec:finite-rep} aims to show the match between branes with irreducible compact supports and finite-dimensional \SH{} representations. \S\ref{sec:bound-state} studies bound states of branes and the corresponding short exact sequences in representations, matching them between the two categories.

Some finite-dimensional $\SH$-modules carry $\PSL(2,\Z)$ representations. Taking the 3d/3d correspondence into account, we explore the geometric origin of these $\PSL(2,\Z)$ representations (and the conditions under which they are present) in \S\ref{sec:3d}. Moreover, the vantage point of three-dimensional physics reveals additional structure concealed behind these $\PSL(2,\Z)$ representations. We show in \S\ref{sec:modularity} that the fivebrane system of the 3d/3d correspondence connects the two-dimensional $A$-model to three-dimensional topological field theories on a 3-manifold $M_3$.
In particular, we show that the choice of an $\SH$-module with a $\PSL(2,\Z)$ action gives rise to a modular tensor category that describes such a 3d TQFT on $M_3$, whose Grothendieck group is identified with the chosen $\SH$-module. In \S\ref{sec:SL2Z-skein}, we propose that the categorification of the skein module of a closed oriented 3-manifold $M_3$ results in a modular tensor category so that there is a ``hidden'' $\SL(2,\Z)$ action on the skein module of $M_3$. We also explain the connection to $\SL(2,\C)$ Floer homology groups of $M_3$.

In \S\ref{sec:4d}, we move one more dimension up, and study our category of interest from the vantage point of four-dimensional physics, namely in the context of four-dimensional $\cN=2^*$ theories. $\cN=2^*$ theories can be constructed by wrapping a stack of M5 branes on the once-punctured torus $C_p$, labeled with some additional discrete data associated to $C_p$.
In~\S\ref{sec:4dCoulomb}, we study an elliptic fibration of the Coulomb branch of an $\cN=2^*$ theory of rank-one on $S^1\times \R^3$, based on the analysis of the Hitchin fibration performed in \S\ref{sec:target}. We also show that the algebra of line operators in the 4d $\cN=2^*$ theory in the $\Omega$-background is a subalgebra of $\SH$ specified by the discrete data. Here as well, the bird's-eye view provided by the fivebrane system connects 4d physics and 2d sigma-models.
We use this in \S\ref{sec:line} to sort out the relationships among line operators, Coulomb branches, and DAHA. Finally, we introduce a surface operator in the 4d $\cN=2^*$ theory and consider an algebra of line operators on the surface operator in \S\ref{sec:surface}. We also discuss a higher-rank bundle for the canonical coisotropic brane to realize the full DAHA and the Morita equivalence $\Rep(\HH)\cong\Rep(\SH)$.

In Appendix \ref{app:notation}, we list notations and symbols adopted in this paper.
A concise summary of some basics of DAHA is given in Appendix \ref{app:DAHA}. In Appendix \ref{app:qt}, we discuss the representation theory of the quantum torus algebra $\QT$ in terms of brane quantization. As a toy model, we show the match between representation theory of $\QT$ and $A$-branes on a flat space $\C^\times \times\C^\times$.  Then we consider an orbifold quotient by $\Z_2$, and we match $A$-branes to representations in this context in~Appendix \ref{app:SQT}.
Appendix \ref{app:3dN=4} is devoted to studying the relation between trigonometric and rational degenerations of the spherical DAHA and Coulomb branches of 3d $\cN=4$ theories.

\section{2d sigma-models and DAHA}
\label{sec:2d}
In this section, we study representation theory of DAHA, strictly speaking, the spherical subalgebra of DAHA of type $A_1$, in terms of \emph{brane quantization} in the 2d $A$-model \cite{Gukov:2008ve} on the moduli space of flat $\SL(2,\C)$-connections on a once-punctured torus. The brane quantization lends itself well to a geometric approach to representation theory of spherical DAHA, which provides novel viewpoints.
 The main goal of this section is to explicitly show the correspondence between $A$-branes with compact Lagrangian submanifolds and finite-dimensional representations of spherical DAHA with respect to dimensions, shortening conditions and morphisms. This matching enables us to find new finite-dimensional representations. The geometric picture also allows us to identify $\PSL(2,\Z)$ actions on some finite-dimensional modules. As another advantage, we generalize Cherednik's polynomial representation from a geometric viewpoint. These results play a crucial role in higher-dimensional physical theories and categorical structures in the subsequent sections.

\bigskip

DAHA associated to a root system $\sfR$ (or, equivalently, to a semisimple Lie algebra $\lie{g}$) can be constructed by beginning with the quantum torus algebra $\QT(\sfP\oplus \sfP^\vee,\omega)$ defined on the direct sum of the weight and coweight lattices of~\lie{g} with the symplectic pairing $\omega$ between~$\sfP$ and~$\sfP^\vee$.
More concretely, $\QT(\sfP\oplus \sfP^\vee,\omega)$ can be understood as the group algebra of the Heisenberg group with the relation
$$X^{\mu} Y^{\lambda}=q^{(\mu, \lambda)} Y^{\lambda} X^{\mu}, \quad \textrm{for} \quad \mu \in \sfP, \lambda \in \sfP^{\vee}~,$$
where $(\mu, \lambda)$ is the symplectic pairing.
 Note that this lattice is isomorphic to the standard pairing on~$\Z^{2\dim \sfP}\cong\Z^{2n}$, so that the algebra has outer automorphism group $\Out(\QT(\sfP\oplus\sfP^\vee,\omega))=\Sp(2n,\Z)$.

However, we have the additional data of the action of the Weyl group $W$ on~$\sfP$ and $\sfP^\vee$. This gives a distinguished embedding of~$\Weyl$ into $\Sp(2n,\Z)$, which therefore determines an extension
\deq[extension]{
0 \to \QT(\sfP\oplus \sfP^\vee,\omega) \to \HH_{t=1}(W) \to \C[\Weyl] \to 0
}
up to equivalence. The algebra $\HH_{t=1}(W)$ is known to be the group algebra of the \emph{double affine Weyl group} $\WW$: $\HH_{t=1}(W)\cong\C[\WW]$. Since the representation of~$\Weyl$ is just on~$\sfP$ (and contragredient on~$\sfP^\vee$), this extension leaves the ``diagonal'' $\Sp(2,\Z)$ subgroup unbroken as outer automorphisms of~$\HH_{t=1}(W)$. For the Cartan type $A_1$, this construction is equivalent to the algebra $\HH_{t=1}$ in Appendix \ref{app:SQT}.
Moreover, the algebra $\HH_{t=1}(W)$ can be further deformed by other formal parameters $t$, transforming the group algebra $\C[W]$ to the Hecke algebra. The result is DAHA $\HH(W)$. We will give a concrete description of the deformation in the Cartan type $A_1$ in this section. DAHAs of general Cartan types are explained in Appendix~\ref{app:DAHA}.
Through this construction, the quantum torus algebra and DAHA are closely related, and we can take the same approach to representation theory of the quantum torus algebra. Although the representation theory of the quantum torus algebra is well-known, it can be a useful guide for DAHA. Therefore, the reader can refer to Appendix \ref{app:qt} for the brane quantization of the quantum torus algebra and symmetrized quantum torus.

The algebra $\HH(W)$ is not commutative, even in the $q=1$ limit. Nonetheless, it contains the \emph{spherical subalgebra} $\SH(W)$, obtained by an idempotent projection, which is commutative as $q=1$. In the limit $t=1$, $\SH_{t=1}(W)$ is isomorphic to the Weyl-invariant subalgebra of~$\QT(\sfP\oplus \sfP^\vee,\omega)$ (after a lift of the Weyl group action is chosen). In the further specialization $q=1$, $\SH$ becomes precisely the algebra of Weyl-invariant functions on
$$
(\frakt_\bC /\sfQ^\vee) \times (\frakt^\vee_\bC/\sfQ)= T_\bC\times T_\bC~.
$$
Note that we take the coroot and root lattices $\sfQ^\vee \oplus\sfQ=   \Hom(\sfP,\Z)\oplus \Hom(\sfP^\vee,\Z)$ (namely the dual lattice) as the quotient lattice.
This space with group action is nothing other than the moduli space of flat connections on a two-torus $T^2$, valued in the corresponding complex Lie group $G_\C$:
\begin{equation}
\begin{aligned}
\MF(T^2,G_\C) &= \Hom(\pi_1(T^2),G_\C)/G_\C \\
&\cong \frac{ T_\bC\times T_\bC}{W}.
\end{aligned}
\end{equation}

We would like to consider an additional deformation of this moduli space to study the representation theory of spherical DAHA geometrically. Happily, for type $A$, this can be achieved just by adding a ``puncture'' on a two-torus $T^2$. Despite this rather simple ``addition'', the story becomes incredibly deeper and more interesting. This section focuses on DAHA of rank one to illustrate and highlight all the delicate features and interesting phenomena. In rank one, we can perform concrete computations as explicitly as possible. For that reason, we will first review some necessary background on the moduli space of
flat $\SL(2,\C)$-connections on a once-punctured torus, which will play the role of the target space $\X$ in the 2d sigma-model.
Then, we will carve out $A$-branes in $\X$ for salient modules of the spherical DAHA. This will give solid evidence of the functor \eqref{eq:functor} from the categories of $A$-branes in $\X$ to the representation category of the spherical DAHA.

\subsection{Higgs bundles and flat connections}
\label{sec:target}

Figuratively speaking, the target space of the 2d sigma-model is the stage where our main characters (branes) will make their appearance. Thus, let us begin by setting the stage.

The target space of our system will be the moduli space of $G=\SU(2)$ Higgs bundles on a genus-one curve $C_p$, ramified at one point $p$:
\be
\label{targetX}
\MS:=\MH(C_p,G).
\ee
Although the geometry of this space, also called the Hitchin moduli space, is a fairly familiar character in mathematical physics literature, we review those aspects that will be especially important for applications to DAHA representations.

Recall \cite{0887284,simpson1990harmonic}, that a ramified (or stable parabolic) Higgs bundle is a pair $(E,\varphi)$ of a holomorphic
$\SU(2)$-bundle $E$ over a curve $C$ and a holomorphic section $\varphi$, called the Higgs field, of the bundle $K_{C}\otimes {\rm ad}(E) \otimes \cO(p)$. Here, $K_{C}$ denotes the canonical bundle of ${C}$, and $\cO(p)$ is the line bundle whose holomorphic sections are functions holomorphic away from $p$ with a first-order pole at $p$. The ramification at $p$ --- more precisely called {\it tame ramification} since we are considering first-order pole --- is described by the following conditions on the connection $A$ on $E$ and the Higgs field
\bea\label{tame}
A &=  \alpha_p \,d\vartheta +\cdots \cr
\varphi& = {\frac12}(\beta_p+i\gamma_p){\frac{dz}{z}}+\cdots
\eea
Here, $z=re^{i\vartheta}$ is a local coordinate on a small disk centered at $p$, and the ramification data is a triple of continuous parameters, $(\alpha_p,\beta_p,\gamma_p)\in T\times \frakt \times \frakt$ where we denote the Cartan subgroup $T\subset G$ and the Cartan subalgebra $\frakt\subset\frakg$. With this prescribed behavior at $p$, the Hitchin moduli space is the space of solutions to the equations
\bea
\label{Hitchin-eq}
F-[\varphi, \overline \varphi]=&0\cr
\overline D_A\,\varphi=&0~,
\eea
modulo gauge transformations. We denote this moduli space $\MH(C_p,G)$, where $C_p$ is a Riemann surface $C$ with the tame ramification \eqref{tame} at $p\in C$. It is a \HK space and the corresponding \K forms are
\bea\label{kahler-form}
\omega_I & =
-{\frac{i}{2\pi}}\int_C|d^2z|\,\Tr\Bigl(\delta A_{\bar z}\wedge
\delta A_z-\delta \bar\varphi\wedge \delta\varphi\Bigr)~,\cr
\omega_J & ={\frac{1}{2\pi}}\int_C |d^2z|\,\Tr\Bigl(
\delta\bar\varphi\wedge \delta A_z+\delta\varphi\wedge
\delta A_{\bar z} \Bigr)~,\cr
\omega_K & =
{\frac{i}{2\pi}}\int_C|d^2z|\,\Tr\Bigl(\delta\bar\varphi
\wedge\delta A_z-\delta\varphi\wedge\delta A_{\bar z}\Bigr)~.
\eea
There is also a triplet of holomorphic symplectic forms $\Omega_I=\omega_J+i\omega_K$, $\Omega_J=\omega_K+i\omega_I$, and $\Omega_K=\omega_I+i\omega_J$, holomorphic in complex structures $I$, $J$, and $K$, respectively. In the absence of ramification, it is easy to check that $\omega_J$ and $\omega_K$ are cohomologically trivial \cite[\S4.1]{Kapustin:2006pk}, whereas $\omega_I$ is non-trivial and, if properly normalized, can be taken as a generator of $H^2 (\MS, \mathbb{Z})$. On the other hand, in the presence of ramification \eqref{tame}, the cohomology classes of $\omega_J$ and $\omega_K$ are proportional to $\beta_p$ and $\gamma_p$, respectively.

The description of $\MH(C_p,G)$ as the moduli space of Higgs bundles given above is in complex structure $I$. Another useful description, in complex structure $J$, comes from identifying a complex combination $A_\bC=A+i\phi$ with a $G_{\bC}$-valued connection, where $\phi=\varphi+\bar\varphi$. The Hitchin equations then become the flatness condition $F_\bC=dA_\bC+A_\bC\wedge A_\bC=0$ for this $G_{\bC}$-valued connection $A_\bC$. According to \eqref{tame}, it has a non-trivial monodromy around the point $p$:
\be \label{monodromy}
U=\exp(2\pi (\gamma_p+i \alpha_p))~.
\ee
which depends holomorphically on $\gamma_p+i\alpha_p$ and is independent of
$\beta_p$. Indeed, in complex structure $J$, $\beta_p$ is a \K parameter and $\gamma_p+i\alpha_p$ is a complex structure parameter. Another useful fact, also explained in \cite{Gukov:2006jk}, is that the cohomology class of the holomorphic symplectic form $\Omega_J=\omega_K+i\omega_I$ is proportional to $\gamma_p+i\alpha_p$ and independent of $\beta_p$.

Similarly, in complex structure $I$ the \K modulus is $\alpha_p$, while $\beta_p+i\gamma_p$ is a complex structure parameter. The cohomology class of the holomorphic symplectic form $\Omega_I=\omega_J+i\omega_K$ is $\beta_p+i\gamma_p$. There is a similar story for complex structure $K$ and all these statements are summarized in Table \ref{tab:complex-Kahler}.

\begin{table}[ht]\centering
	\begin{tabular}{c|c|c}
		Complex structure & Complex modulus & K\"ahler modulus\tabularnewline
		\hline
		$I$  & $\beta_p+i\gamma_p$ & $\alpha_p$\tabularnewline
		$J$  & $\gamma_p+i\alpha_p$ & $\beta_p$\tabularnewline
		$K$  & $\alpha_p+i\beta_p$ & $\gamma_p$\tabularnewline
\end{tabular}
\caption{Complex and \K moduli of the moduli space $\MH$ with one ramification point.}\label{tab:complex-Kahler}
\par\end{table}

In a supersymmetric sigma-model with target $\X$, the \K modulus of the target space is always complexified. This fact plays an important role in mirror symmetry. In the present setup, too, the \K moduli are all complexified by the periods of the 2-form field $B$. For example, in complex structure $I$, the complexified \K modulus is $\alpha_p+i\eta_p$, where $\eta_p\in T^\vee=\Hom(\Lcochar, \U(1))$ and $\Lcochar$ is the cocharacter lattice of $G$. Therefore, taking into account the ``quantum'' parameter $\eta_p$, the ramification data consists of the quadruple  of parameters $(\alpha_p, \beta_p,\gamma_p, \eta_p)$.

All of these structures can be made completely explicit in the case when~$C_p$ is a punctured torus. In complex structure $J$, where $\MS=\MH(C_p,G)$ is the moduli space of complex flat connections on $C_p$, we can then use an explicit presentation of the fundamental group
\begin{equation}
\pi_1(C_p) = \langle \frakm,\frakl,\frakc | \frakm \frakl \frakm^{-1}\frakl^{-1} = \frakc \rangle~.
\end{equation}
to describe flat connections concretely, in terms of holonomies along the $(1,0)$-cycle $\frakm$, the $(0,1)$-cycle $\frakl$, and the loop $\frakc$ around $p$:
\deq{
x = \Tr(\rho(\frakm)), ~  y = \Tr(\rho(\frakl)),\text{ and } z = \Tr(\rho(\frakm\frakl^{-1}))~.
}
In terms of these holonomy variables, the space of $\SL(2,\C)$-representations $\rho: \pi_1(C_p)  \to \SL(2,\C)$ is a cubic surface (see {\it e.g.} \cite{goldman2009trace,Gukov:2010sw}):
\begin{equation}
\label{toruscubic}
\MF (C_p, \SL(2,\C))=
\{ (x,y,z) \in \C^3 |x^2 + y^2 + z^2 - xyz - 2 = \Tr(\rho(\frakc)) = \tilde{t}^2+\tilde{t}^{-2}\}~.
\end{equation}
Here we used the fact that, according to~\eqref{monodromy}, the holonomy of the complex flat connection around $p$ is conjugate to
\be
\label{monodromy2}
\rho(\frakc)\sim  \begin{pmatrix}\tilde{t}^{-2}&0 \\ 0&\tilde{t}^{2}\end{pmatrix}~.
\ee
This section will be devoted to studying the deformation quantization $\OO^q(\X)$  of this coordinate ring holomorphic in complex structure $J$, which is generated by $x$, $y$, $z$, and its representations geometrically.

For a complex surface defined by the zero locus of a polynomial $f(x,y,z)$, the holomorphic symplectic form (a.k.a. Atiyah-Bott-Goldman symplectic form) can be written as
\be
\label{holomorphic-2-form}
\Omega_{J}=\frac{1}{2\pi i} \frac{dx\wedge dy}{\partial f/\partial z} = \frac{1}{2\pi i} \frac{dx\wedge dy}{2z-xy}~.
\ee
and the \K form is
\be
\label{K2form}
\omega_J =\frac{i}{4\pi} (dx \wedge d \bar x + dy \wedge d \bar y + dz \wedge d \bar z)~.
\ee

In the special case $\alpha_p=\b_p=\gamma_p=0$, the moduli space of flat $\SL(2,\C)$-connections on $C_p$ is simply a quotient space
\be
\label{Hitchin-torus}
(\C^\times\times\C^\times)/\Z_2
\ee
by the Weyl group $\Z_2$. It can be understood as a moduli space of flat $\SL(2,\C)$-connections on a torus (without ramification), such that holonomy eigenvalues along A- and B-cycles each parametrize a copy of $\C^\times$. The ``real slice'' $(S^1\times S^1)/\Z_2$ is the moduli space of flat $\SU(2)$-connections on the (punctured) torus, and it is sometimes called the ``pillowcase''. According to the theorem of~\cite{narasimhan1965stable} (resp.~\cite{mehta1980moduli}), it can be identified with the moduli space $\mathrm{Bun}(C_p,G)$ of stable (resp.~parabolic) $G$-bundles on $C_p$. It is easy to see that $\mathrm{Bun}(C_p,G)$ is a holomorphic submanifold of $\MH(C_p,G)$ in complex structure $I$. Furthermore, because $\delta \varphi=0$ on $\mathrm{Bun}(C_p,G)$, it follows from \eqref{kahler-form} that $\mathrm{Bun}(C_p,G)$ is a holomorphic Lagrangian submanifold with respect to $\Omega_I$ (in particular, Lagrangian with respect to $\omega_J$ and $\omega_K$). Following the notation in \S\ref{sec:Lagrangian}, we write it by $\bfV$ as a Lagrangian submanifold in the target $(\X,\omega_\X)$.

In addition to $\bfV$, other special submanifolds of $\MH(C_p,G)$ will play a role in what follows.
For example, in complex structure $I$, the Hitchin moduli space is a completely integrable Hamiltonian system \cite{0887284}, \textit{i.e.}\ a fibration
\deq[Hitchin-fibration]{
\pi:\MH(C_p,G)\to \BH
}
over an affine space, the ``Hitchin base'' \BH{},
whose generic fibers are abelian varieties (sometimes called ``Liouville tori''). For $G = \SU (2)$, the map $\pi$ takes a pair $(E,\varphi)$ to $\Tr\varphi^2$, which is holomorphic in complex structure $I$. Specializing further to the case where $C_p$ is a genus-one curve gives a particularly simple integrable system: its generic fiber $\F$ is a torus that, just like $\bfV$, is holomorphic in complex structure $I$ and Lagrangian with respect to $\omega_J$ and $\omega_K$. We also note that the only singular fiber of the Hitchin fibration $\pi:\MH(C_p,G)\to \BH$ is the pre-image $\NC =  \pi^{-1}(0)$ of $0\in \BH$ which, in the limit $\alpha_p=\b_p=\gamma_p=0$, is the ``pillowcase'' $\bfV \cong (S^1\times S^1)/\Z_2$ with four torsion (orbifold) points.

Now let us consider what happens when we go away from the limit $\alpha_p=\b_p=\gamma_p=0$ and consider generic values of the ramification parameters. From the viewpoint of the complex structure $J$, the equation \eqref{toruscubic} describes the \emph{deformation} of the four $A_1$ singularities of the singular surface \eqref{Hitchin-torus}, where $\tilde{t}^2$ (or, equivalently, $\g_p + i \alpha_p$) plays the role of the complex structure deformation. On the other hand, turning on $\beta_p \ne 0$ leads to a \emph{resolution} of the $A_1$-singularities. In other words, $\beta_p$ is the \K structure parameter in complex structure $J$, {\it cf.} Table~\ref{tab:complex-Kahler}.

\begin{figure}
	\centering
	\includegraphics[width=\textwidth]{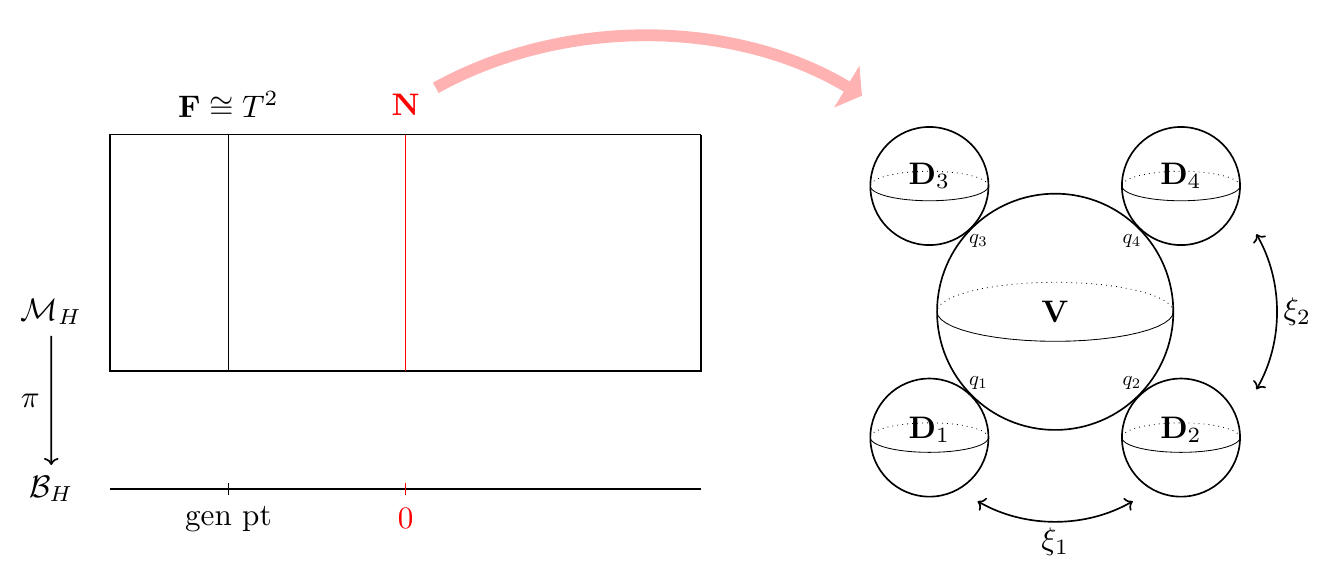}
	\caption{Schematic illustration of the Hitchin fibration $\MH(C_p,\SU(2))\to \cB_H$ and global nilpotent cone at $\beta_p=0=\gamma_p$ and a generic value of $\talpha_p$.}\label{fig:Hitchin-fibration}
\end{figure}

Recall that $\alpha_p$ is the \K structure parameter in complex structure $I$. If we turn on $\alpha_p$ while keeping $\beta_p=\gamma_p=0$, then the four torsion (orbifold) points are blown up in the Hitchin fibration. Consequently, the singular fiber in the Hitchin fibration, called the \emph{global nilpotent cone} $\bfN := \pi^{-1} (0)$, now contains five compact irreducible components (all rational) \cite{Hausel:1998qc,Gukov:2010sw}:
\be\label{nilconetoy}
\bfN  =  \bfV \cup \bigcup_{i=1}^{4} \bfD_i \,~.
\ee
In fact, it is a singular fiber of Kodaira type $I_0^*$ \cite{kodaira1964structure,kodaira1966structure}  in the elliptic fibration $\pi$. The irreducible components  $\bfV$ and $\bfD_i$ of the global nilpotent cone are holomorphic Lagrangians with respect to $\Omega_I$, sometimes called Lagrangians of type $(B,A,A)$. The homology classes of $\bfV$ and $\bfD_i$ provide a basis for the second homology groups $H_2 (\MH(C_p,G),\Z)$, and their intersection form is the affine Cartan matrix of type $\wh D_4$, as illustrated in Figure~\ref{fig:Hitchin-fibration}.
The intersection form has only one null vector, which
must be identified with the class of a generic fiber $\F$ of the Hitchin fibration, resulting in the relation
\be
\label{fiber-class-rel}
[\F]  =  2 [\bfV] + \sum_{i=1}^{4} [\bfD_i] \,.
\ee

Once we move away from $\beta_p=\gamma_p=0$, we are deforming the complex structure
modulus $\beta_p+i\gamma_p$ of complex structure~$I$, and so the structure of the Hitchin fibration drastically changes. For generic values of $(\beta_p,\gamma_p)$, the embeddings of the two-cycles $\bfV$ and $\bfD_i$ ($i=1,\ldots,4$) into $\MH(C_p,G)$ are no longer holomorphic with respect to complex structure $I$, and the singular fiber of type $I_0^*$ splits into three singular fibers of type $I_2$ ~\cite[\S3.4]{Frenkel:2007tx}. If we write the base genus-one curve $C_p$ of the Hitchin system by an algebraic equation $y^2=(x-e_1)(x-e_2)(x-e_3)$ with $e_1+e_2+e_3=0$ where the ramification point $p$ is located at infinity, then the singular fibers of type $I_2$ are preimages of points
\be\label{base-points}
\cB_H\ni b_i:=e_i \Tr\,(\beta_p+i\gamma_p)^2 \qquad  (i=1,2,3)~,
\ee
under the Hitchin fibration as depicted in Figure \ref{fig:3I_2}. In the singular fiber at $b_i\in \cB_H$, two irreducible components $\bfU_{2i-1}$ and $\bfU_{2i}$, which are topologically $\CP^1$, meet at two double points.

Hence, the two-cycles $\bfV$ and $\D_i$ ($i=1,\ldots,4$) are not projected to a point by the Hitchin fibration with a generic ramification, though they still give a basis of $H_2 (\MH(C_p,G),\Z)$ and satisfy the relation~\eqref{fiber-class-rel}. An analysis by the Mayer--Vietoris sequence tells us that the homology class of each irreducible component in a singular fiber $I_2$ can be expressed as
\begin{equation}\label{U-homology}
\begin{gathered}
	[\bfU_1]=[\bfV]+[\D_1]+[\D_2]~, \quad [\bfU_3]=[\bfV]+[\D_1]+[\D_3]~, \quad [\bfU_5]=[\bfV]+[\D_1]+[\D_4]~,\\
	[\bfU_2]=[\bfV]+[\D_3]+[\D_4]~, \quad [\bfU_4]=[\bfV]+[\D_2]+[\D_4]~, \quad[\bfU_6]=[\bfV]+[\D_2]+[\D_3]~,
\end{gathered}
\end{equation}
Additionally, there is another two-cycle $\bfW$, which projects to the blue curve as in Figure \ref{fig:3I_2} under the Hitchin vibration \eqref{Hitchin-fibration}, with homology class $[\bfW]=-[\D_1]$. Namely, the cycle $\bfW$ is suspended between $\pi^{-1}(b_1)$ and $\pi^{-1}(b_2)$ through $\pi^{-1}(b_3)$ where $\bfW$ intersects a Hitchin fiber by $S^1$ except the end points $b_1$ and $b_2$. With respect to the new basis
\deq{
[\bfU_1],[\bfU_2],[\bfU_3],[\bfU_5],[\bfW] \in H_2 (\MH(C_p,G),\Z)~,
}
the intersection form becomes
\begin{equation}
\begin{pmatrix}
	-2 & 2 & 0 & 0 & 1 \\
	2 & -2 & 0 & 0 & -1 \\
	0 & 0 & -2 & 0 & 1 \\
	0 & 0 & 0 & -2 & 1 \\
	1 & -1 & 1 & 1& -2
\end{pmatrix}.
\end{equation}
Note that the upper-left two-by-two matrix is the Cartan matrix of the affine type $\wh A_1$ up to minus sign as the intersection form of a singular fiber of type $I_2$.

\begin{figure}[t]
	\centering
	\includegraphics[width=\textwidth]{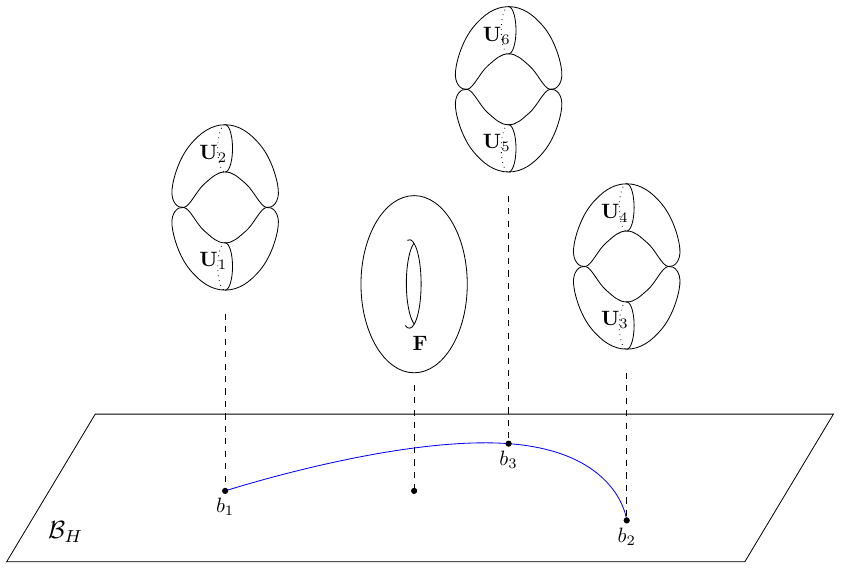}
	\caption{The Hitchin fibration with a generic ramification contains three singular fibers of Kodaira type $I_2$ at the base points $b_i$ ($i=1,2,3$).}
	\label{fig:3I_2}
\end{figure}

For our applications to branes and representations, we also need to know the type of the five compact two-cycles $\bfV$, $\D_i$ ($i=1,\ldots,4$) and periods of the \K forms over them. The integrals of $\Omega_J$ over $\bfV$ and over $\F$ were computed {\it e.g.} in \cite{Gukov:2010sw}. They can be expressed as the following set of relations:
\bea\label{integral-BunG}
\int_{\bfV}\frac{\omega_I}{2\pi}&=\frac12-|\talpha_p| ~,\quad & \textrm{diag}(\talpha_p,-\talpha_p)\sim \a_p\in T~,\cr
\int_{\bfV}\frac{\omega_J}{2\pi}&=-\mathrm{sign}(\talpha_p)\tbeta_p~,& \textrm{diag}(\tbeta_p,-\tbeta_p)\sim \b_p\in \frakt~, \cr
\int_{\bfV}\frac{\omega_K}{2\pi}&=-\mathrm{sign}(\talpha_p)\tgamma_p ~,& \textrm{diag}(\tgamma_p,-\tgamma_p)\sim \g_p\in \frakt~
\eea
and
\be
\label{generic-fiber-evaluation}
\int_{\F}\frac{\omega_I}{2\pi} =1~,\qquad \int_{\F} \frac{\omega_J}{2\pi}=0=\int_{\F}\frac{\omega_K}{2\pi}~,
\ee
where in the latter we used the fact that the Hitchin fiber $\F$ is holomorphic in complex structure $I$ and Lagrangian with respect to $\Omega_I$ for any $(\alpha_p,\beta_p,\gamma_p)$. We assume that $\talpha_p$ takes its value in the range $-\frac12<\talpha_p\le\frac12$.
Although we did not compute the periods of the 2-forms \eqref{holomorphic-2-form} and \eqref{K2form} over exceptional divisors $\D_i$ directly, we claim
\be
\label{integral-D}
\frac{|\talpha_p|}2 =\int_{\D_i}\frac{\omega_I}{2\pi}~,\qquad
\mathrm{sign}(\talpha_p)\frac{\tbeta_p}2 =\int_{\D_i}\frac{\omega_J}{2\pi}~,\qquad
\mathrm{sign}(\talpha_p)\frac{\tgamma_p}2 =\int_{\D_i}\frac{\omega_K}{2\pi}~,
\ee
independently of $i=1,2,3,4$.
One way to justify this claim is to compute the periods for small values of $\gamma_p + i \alpha_p \approx 0$, \textit{i.e.}\ for $\tilde{t} \approx 1$. Another way is to use \eqref{fiber-class-rel} together with the symmetries of $\MH(C_p,G)$ that we discuss next.
The formulae above are compatible with the fact that the Weyl group symmetry of the ramification parameters given by an overall sign change \be\label{ramification-Weyl}(\talpha_p, \tbeta_p, \tgamma_p) \to(-\talpha_p,-\tbeta_p,-\tgamma_p)\ee leaves the moduli space completely invariant.

Furthermore, the ``quantum'' parameter that complexifies a \K parameter can be understood as the period of the $B$-field in a 2d sigma-model over $\D_i$
\be\label{eta}
\mathrm{sign}(\talpha_p)\teta_p =\int_{\D_i}\frac{B}{2\pi}~~,\qquad\qquad \textrm{diag}(\teta_p,-\teta_p)\sim \eta_p\in T^\vee~.
\ee
In the following, we often use the parameters $(\talpha_p,\tbeta_p,\tgamma_p,\teta_p)\in S^1\times\bR\times\bR\times S^1$ and the quadruple $(\a_p,\b_p,\g_p,\eta_p)\in T\times\frakt\times\frakt\times T^\vee$ of the tame ramification \eqref{tame}  at $p\in C$ in the same meaning.

\paragraph{Symmetries}
The target space \eqref{targetX} of our sigma-model has the symmetry group\footnote{The symmetry of the $A$-model can be larger or smaller than the group of geometric symmetries. It can be larger due to quantum symmetries not directly visible from geometry, and it can be smaller if some geometric symmetries are $Q$-exact from the $A$-model viewpoint.}
\be\label{sym-group}
\Xi \times \MCG (C_p)  =  \mathbb{Z}_2 \times \mathbb{Z}_2 \times \SL(2,\mathbb{Z})
\ee
where $\Xi = \mathbb{Z}_2 \times \mathbb{Z}_2$ is the group of ``sign changes'' generated by twists of a Higgs bundle $E \to C_p$ by line bundles of order 2. Abusing notation, this group can be identified with $H^1(C,\Z_2)=\Z_2\oplus\Z_2$ where $\Z_2$ is the center of $\SU(2)$. Obviously, $\SL(2,\mathbb{Z})$ is the mapping class of the (punctured) torus:
\be
\MCG (C_p)  \cong  \SL(2,\mathbb{Z})~.
\ee
Both $\Xi$ and $\MCG (C_p)$ are symmetries in all complex and symplectic structures. In particular, in what follows, we will need their explicit presentations as holomorphic symplectomorphisms with respect to $\Omega_J$.

In complex structure $J$, the ``sign changes'' $\Xi= \mathbb{Z}_2 \times \mathbb{Z}_2$ are holomorphic involutions, and its generators $\xi_1$, $\xi_2$ and their combination $\xi_3:=\xi_1\circ \xi_2$ act as
\bea\label{sign-changes}
\xi_1 &: \quad (x,y,z) \mapsto (-x,y,-z)~,\cr
\xi_2 &: \quad (x,y,z) \mapsto (x,-y,-z)~,\cr
\xi_3&: \quad (x,y,z) \mapsto (-x,-y,z)~,
\eea
respectively. The ``sign changes'' symmetry plays a very important role to understand mirror symmetry \cite{Gukov:2010sw} and connections to 4d physics in \S\ref{sec:4d}.

The symmetry group $\Xi$ leaves $\bfV$ invariant (as a set, not pointwise) and acts on the exceptional divisors $\D_i$ as follows:
\bea\label{signchange-Di}
\xi_1 &: \D_1 \leftrightarrow \D_2 \quad \textrm{and} \quad  \D_3 \leftrightarrow \D_4~, \\
\xi_2 &: \D_1 \leftrightarrow \D_3 \quad \textrm{and} \quad  \D_2 \leftrightarrow  \D_4~, \\
\xi_3 &: \D_1 \leftrightarrow \D_4 \quad \textrm{and} \quad  \D_2 \leftrightarrow \D_3~.
\eea
This symmetry, illustrated in Figure~\ref{fig:Hitchin-fibration}, provides supporting evidence to our assumption in~\eqref{integral-D}.

In complex structure $I$, a point in the Hitchin base $\cB_H$ is invariant under $\Xi$ so that it acts on each fiber as translations of order two in the Hitchin fibration $\MH\to \cB_H$ \cite[\S3.5]{Frenkel:2007tx}. It acts freely on a generic fiber. On the other hand, $\xi_i$ acts on each irreducible component  of the singular fiber $\pi^{-1}(b_i)$, namely $\bfU_{2i-1}$ and $\bfU_{2i}$, respectively, where the fixed points are exactly the two double points. At the other singular fibers, it exchanges the two double points and swaps the two irreducible components
\be\label{signchange-Ui}
\xi_i: \bfU_{2i+1} \leftrightarrow  \bfU_{2i+2}\quad \textrm{and} \quad \bfU_{2i+3} \leftrightarrow  \bfU_{2i+4}~,
\ee
where the indices of $\bfU$ are counted modulo $6$. This is consistent with the homology classes \eqref{U-homology} and the actions \eqref{signchange-Di}.

The action of $\SL(2,\Z)$ on the eigenvalues of the holonomies $\rho(\frakm)$ and $\rho(\frakl)$ is indeed given in \eqref{SL2Z-coord}. In particular, the non-trivial central element $-1$ of $\SL(2,\mathbb{Z})$ indeed exchanges the eigenvalues of the holonomies $\rho(\frakm)$ and $\rho(\frakl)$ as well as the one around the puncture \eqref{monodromy2} to their inverses. Therefore, it acts as the Weyl group symmetry of $\SL(2,\C)$. Subsequently, the trace coordinates $x,y,z$ are invariant under the non-trivial central element $-1$ so that $\SL(2,\Z)$ acts projectively on the coordinate ring $\OO(\X)$ holomorphic in complex structure $J$.
However, the eigenvalues of the holonomy around the puncture are exchanged, which we denote
\be\label{iota-classical}
\iota:\tilde{t}\to\wt  t^{-1}~.
\ee
A quotient of $\MCG (C_p) \cong \SL(2,\mathbb{Z})$ by the center is $\PSL(2,\mathbb{Z})=\SL(2,\mathbb{Z}) / \pm 1$, which is the mapping class group of a 4-punctured sphere.
In order to find an explicit presentation of $\PSL(2,\mathbb{Z})$, it is convenient to note that $T^2 \to S^2$ is a double cover branched at 4 points, {\it cf.} \eqref{Hitchin-torus}
\be
\label{PSL2ZvsBr3}
\PSL(2,\mathbb{Z})  \cong  \Br_3 / \mathcal{Z}
\ee
where the second equality is a well-known relation to the Artin braid group $\Br_3$. In terms of standard generators $\tau_+$ and $\tau_-^{-1}$, which satisfy the braid relation $\tau_+ \tau_-^{-1} \tau_+ = \tau_-^{-1} \tau_+ \tau_-^{-1}$, the center $\mathcal{Z}$ of $\Br_3$ is generated by $(\tau_+ \tau_-^{-1})^3$. Under the surjective map onto $\PSL(2,\mathbb{Z})$, we have
\be\label{tau-pm}
\tau_+ \mapsto\begin{pmatrix} 1 & 0 \\ 1 & 1 \end{pmatrix},
\qquad
\tau_- \mapsto
\begin{pmatrix} 1 & 1 \\ 0 & 1 \end{pmatrix}
\ee
and
\be
\sigma:=\tau_+ \tau_-^{-1} \tau_+ = \tau_-^{-1} \tau_+ \tau_-^{-1} \mapsto
 \begin{pmatrix} 0 & -1 \\ 1 & 0 \end{pmatrix},
\qquad
\tau_+ \tau_-^{-1}  \mapsto
\begin{pmatrix} 1 & -1 \\ 1 & 0 \end{pmatrix}~.
\ee
In the quotient~\eqref{PSL2ZvsBr3}, the latter two elements have order 2 and 3, respectively.

Using \eqref{PSL2ZvsBr3}, we can relate our present problem to the mapping class group action on the character variety of the 4-punctured sphere\footnote{In the notations of \cite{Gukov:2007ck} we need to take $(x_1,x_2,x_3) = (-x,-y,-z)$, $\theta_1 = \theta_2 = \theta_3 = 0$, and $\theta_4 = -2-\tilde{t}^2 - \tilde{t}^{-2}$.} which is also a cubic surface of the form \eqref{toruscubic} and on various branes (submanifolds) on this surface \cite{Gukov:2007ck}:
\bea\label{MCG-classical}
\tau_+ &: \quad
(x,y,z) \mapsto (x,xy-z,y)~,\cr
\tau_- &: \quad
(x,y,z) \mapsto (xy-z,y, x)~,\cr
\sigma&:  \quad
(x,y,z) \mapsto (y,x,xy-z) ~.
\eea
It is easy to verify that these are indeed polynomial automorphisms of the cubic surface \eqref{toruscubic} and that they satisfy the braid relation.

Note, the action of $\PSL(2,\Z)$ leaves $\bfV$ invariant (as a set, not pointwise) and acts on the exceptional divisors $\D_i$ as on the set of $\Z_2$ torsion points on an elliptic curve, In other words, $\D_1$ is fixed by the $\PSL(2,\Z)$, also as a set, not pointwise, whereas $\D_2$, $\D_3$ and $\D_4$ transform as points $\frac12$, $\frac\tau2$, and $\frac12+\frac\tau2$, respectively. In terms of the generators of $\PSL(2,\Z)$, we have explicit transformation rules
\bea\label{PSL-Di}
\tau_+ &: \D_2 \leftrightarrow \D_4 \quad \textrm{and} \quad  \D_1,~ \D_3\quad \textrm{are fixed as a set}~, \\
\tau_- &: \D_3 \leftrightarrow \D_4 \quad \textrm{and} \quad  \D_1,~  \D_2\quad \textrm{are fixed as a set}~, \\
\sigma &: \D_2 \leftrightarrow \D_3 \quad \textrm{and} \quad  \D_1,~ \D_4\quad \textrm{are fixed as a set}~.
\eea
In addition, these generators permute the singular fibers of type $I_2$ in the Hitchin fibration as $\frakS_3$:
\bea\label{PSL-Ui}
\begin{tikzpicture}
  \node (A)  at (0,1.6) {$\pi^{-1}(b_2)$};
    \node (B) at (-1.5,0) {$\pi^{-1}(b_1)$};
  \node (C)  at (1.5,0)  {$\pi^{-1}(b_3)$};
  \draw[<->] (A)  to node [left,blue] {$\sigma$} (B);
  \draw[<->] (B)  to node [below,blue] {$\tau_+$} (C);
    \draw[<->] (C)  to node [right,blue] {$\tau_-$} (A) ;
\path[->] (A) edge  [loop above] node [blue]{$\tau_+$} (A);
\path[->] (B) edge  [out=175,in=185,loop] node [left,blue] {$\tau_-$} ();
\path[->] (C) edge  [out=5,in=-5,loop] node[right,blue] {$\sigma$} ();
\end{tikzpicture}
\eea

In the above, we pointed out that $\bfV$ is invariant under both symmetries $\Xi$ and $\PSL(2,\mathbb{Z})$ only as a set, not pointwise. Also, the same is true about $\PSL(2,\mathbb{Z})$ action on $\D_1$. While in the case of $\bfV$ the reason for both claims is fairly clear ({\it e.g.} it is manifest in the $\tilde{t} \to 1$ limit \eqref{Hitchin-torus}), the fact that $\PSL(2,\mathbb{Z})$ fixes $\D_1$ only as a set and not pointwise is less obvious.In order to explain it, let us consider the limit $\tilde{t} = 1 + \epsilon$, with $\epsilon \ll 1$, and take $(x,y,z) = (2 + a, 2 + b, 2 + c)$. Then, for small values of $(a, b, c)$, the surface \eqref{toruscubic} looks like a quadric
$$
a^2 + b^2 + c^2 - 2 (a b + b c + c a)  =  4 \epsilon^2~,
$$
on which the generators $\tau_\pm$ act as linear reparametrizations:
$$
\begin{aligned}
\tau_+ :
(a, b, c) &\mapsto (a , 2a + 2b - c, b)~,
\\
\tau_-:
(a,b,c) &\mapsto (2a + 2b - c, b, a)~.
\end{aligned}
$$

\subsection{DAHA of rank one and its spherical algebra}
\label{sec:DAHA-main}

Now let us review a few necessary details of DAHA of rank one here.
Much like the Hecke algebra sits, loosely speaking, between the Weyl group and the braid group---in the sense that the latter two can be obtained by either specialization or by omitting some of the relations---DAHA sits in between the double affine Weyl group and the double affine braid group. This perspective, reviewed in e.g.~\cite{Gukov:2014gja}, will be useful to us in what follows. In Cartan type $A_1$, the \emph{double affine braid group} (a.k.a. the  \emph{elliptic braid group}), denoted $\BB_{q=1}(\Z_2)$, is simply the orbifold fundamental group of the quotient space $(T^2\backslash p)/\Z_2$, the quotient of a once-punctured torus by $\Z_2$. It is generated by three generators $X$, $Y$, and $T$, illustrated in Figure \ref{fig:orb-fund}:
\be
\label{EBG}
\pi^\textrm{orb}_1\Bigl((T^2\backslash p)/\Z_2\Bigr) = \Bigl( T, X, Y | ~ TXT = X^{-1},\ \  TY^{-1}T = Y,  \ \  Y^{-1}X^{-1}YXT^2 = 1\Bigr)~.
\ee
Its central extension, denoted $\BB(\Z_2)$, is obtained by deforming the last relation to $Y^{-1}X^{-1}YXT^2=q^{-1}$.

\begin{figure}[ht]
\begin{center}
\subfloat[][]{
\begin{tikzpicture}[scale = 3]
\drawgrid{0}{0}{0.2}
\Ta{0}{0}{$T$}[black][black]
\Yb{-1}{0}{$Y$}[black][left]
\Xb{0}{-1}{$X$}[black][above]
\end{tikzpicture}
} \qquad
\subfloat[][]{
\begin{tikzpicture}[scale=3]
\drawgrid{0}{0}{0.2}
%%% TXT = X-
\Ta{1}{0}{$T$}[blue!60!black]
\Xb{0}{-1}{$X$}[blue!60!black]
\Tb{0}{0}{$T$}[blue!60!black]
\Xinva{0}{0}{$X^{-1}$}[blue!60!black]
%%% TY-T = Y
\Ta{0}{0}{$T$}[red!60!black]
\Yinvb{-1}{0}{$Y^{-1}$}[red!60!black]
\Tb{0}{1}{$T$}[red!60!black]
\Ya{0}{0}{$Y$}[red!60!black]
\end{tikzpicture}
} \qquad
\subfloat[][]{
\begin{tikzpicture}[scale = 3]
\drawgrid{0}{0}{0.2}
\Xb{0}{-1}{$X$}[green!60!black]
\Yb{0}{0}{$Y$}[green!60!black]
\Xinvb{0}{0}{$X^{-1}$}[green!60!black]
\Yinvb{-1}{0}{$Y^{-1}$}[green!60!black]
\Ta{0}{0}{$T$}[green!60!black]
\Tb{0}{0}{}[green!60!black]
\end{tikzpicture}
}
\end{center}
\caption[Relations in the orbifold fundamental group.]{Generators and relations in the orbifold fundamental group of the once-punctured torus. On the left, generators and relations are drawn on the double cover. The relations depicted are \textcolor{blue!60!black}{$ TXT = X^{-1}$}, \textcolor{red!60!black}{$TY^{-1}T = Y$}, and \textcolor{green!60!black}{$Y^{-1}X^{-1}YXT^2 = 1$}.}
\label{fig:orb-fund}
\end{figure}
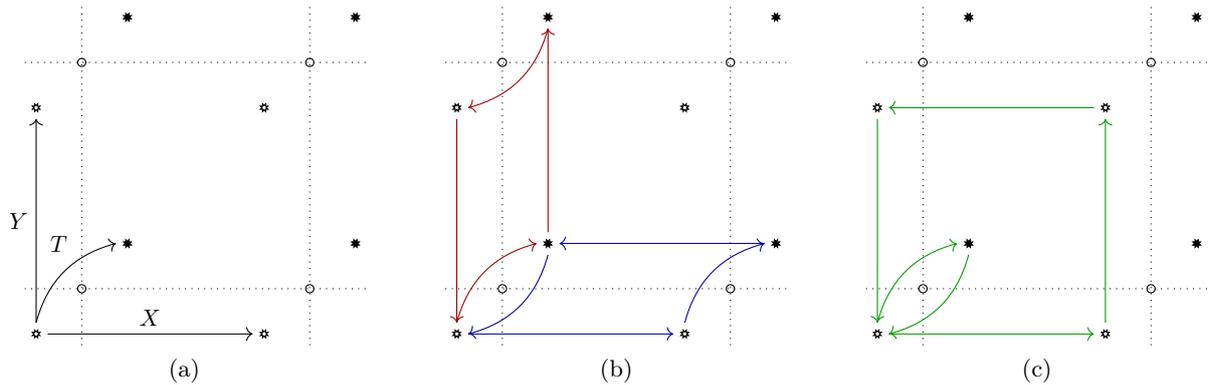

Then, rank-one DAHA $\HH(\Z_2)$ is obtained by imposing one more quadratic (``Hecke'') relation:
\be\label{DAHA1}
\HH(\Z_2)= \CR \bigl[ T^{\pm 1}, X^{\pm 1},Y^{\pm 1} \bigr] \bigg/
\left\{ \begin{array}{cc}TXT=X^{-1}~,
	& Y^{-1}X^{-1}YXT^2=q^{-1}~,\\
	TY^{-1}T=Y~,& (T-t)(T+t^{-1})=0\end{array}\right\} ~.
\ee
This involves the second deformation parameter $t$. Here $\CR$ is a ring of coefficients defined as follows.
Let $\C[q^{\pm\frac12},t^{\pm}]$ be the ring of Laurent polynomials in the formal parameters $q^{1/2}$ and~$t$, and consider a multiplicative system $M$ in $\C[q^{\pm\frac12},t^{\pm}]$ generated by elements of the form $(q^{\ell} t-q^{-\ell} t^{-1})$ for any non-negative integer $\ell \in \Z_{\ge0}$. We define the coefficient ring \CR{} to be the localization (or formal ``fraction'')\footnote{In other words, $\CR$ is the ring of rational functions in the formal parameters $q^{\frac12}$ and~$t$ where denominators are always elements in the multiplicative system $M$ such as
$$
\frac{f(X)}{( t- t^{-1})^{k_0}(q t-q^{-1} t^{-1})^{k_1}\cdots (q^{\ell} t-q^{-\ell} t^{-1})^{k_\ell}}~,\qquad f(X)\in \C[q^{\pm\frac12},t^{\pm},X^\pm]~.
$$} of the ring $\C[q^{\pm\frac12},t^{\pm}]$ at~$M$:
\deq[eq:coeffs]{
\CR =M^{-1}\C[q^{\pm\frac12},t^{\pm}]~.
}

This coefficient ring contains the two central generators of the algebra $\HH(\Z_2)$, $q$ and $t$, which can be thought of as continuous deformation parameters and start life (in any irreducible representation) as arbitrary complex numbers. Many remarkable things happen when these two parameters assume special values, as will be further discussed in the sequel. In a way, the behavior of the algebra and its representations under such specializations---and the match of this behavior to the $A$-brane category---is one of the most interesting aspects of the geometric/physical approach.

Another standard notation for the second deformation parameter (which is convenient for some of the specializations) is
\be
\label{central-charge}t=q^c~.
\ee
where $c$ is often called the ``central charge''. In what follows, we will use the shorthand notation $\HH=\HH(\Z_2)$ unless we wish to make a statement about DAHA of Cartan type other than~$A_1$.

For further details and properties of DAHA, we refer the reader to the fundamental book \cite{Cherednik-book}. The representation theory of DAHA there will be introduced throughout this section, as they emerge from physics and geometry.
Also, some basics of DAHA are assembled in Appendix \ref{app:DAHA}.

The construction of $\HH$ based on the punctured torus allows us to see the action of the symmetry group \eqref{sym-group}, and the symmetry plays a pivotal role in the geometric understanding of the representation theory of (spherical) DAHA in what follows. Under $\Xi$, the generators are transformed as
\bea\label{automorphisms1}
\begin{array}{cccccc}
\xi_1:&T\mapsto T,&X\mapsto -X,& Y\mapsto Y,&q\mapsto q,& t\mapsto t,\\
\xi_2:&T\mapsto T,&X\mapsto X,& Y\mapsto -Y,&q\mapsto q,& t\mapsto t~.\\
\end{array}
\eea
The mapping class group $\mathrm{SL}(2,\Z)$ acts on the generators of $\HH$ as follows\footnote{Although we follow the notation of \cite{Cherednik-book} for the transformations $\tau_\pm$ on the generators of DAHA here and in \eqref{DAHA-T}, we change matrix assignments to $\tau_\pm$ as in \eqref{tau-pm} and \eqref{tau-pm2} from \cite{Cherednik-book} since it is consistent with the projective action \eqref{PSL-Di}  of $\SL(2,\Z)$ on the exceptional divisors geometrically.}:
\bea\label{MCGloop}
\tau_+ &: \quad (X,Y,T)\mapsto (X , q^{-\frac12}XY , T) \cr
\tau_-  &: \quad (X,Y,T)\mapsto ( q^{\frac12}YX , Y , T)\cr
\sigma &: \quad  (X,Y,T)\mapsto (Y^{-1} , XT^2 , T)
\eea
Since $\sigma$ essentially exchanges the canonically conjugate variables $X$ and~$Y$, it is sometimes called the \emph{Fourier transform} of~$\HH$. Also,  $\HH$ enjoys the following (non-inner) involution,
\be\label{automorphisms2}
\tilde\iota:T\mapsto -T, \quad X\mapsto X,\quad  Y\mapsto Y,\quad q\mapsto q,\quad  t\mapsto t^{-1}~.\\
\ee

It is easy to check from the Hecke relation that $\mathbf{e}=({T+t^{-1}})/({t+t^{-1}})$ is  an idempotent element ($\mathbf{e}^2=\mathbf{e}$) of  $\HH$.
Then, the spherical subalgebra $\SH$ is defined by the idempotent projection
\deq{
\SH := \mathbf{e}\HH\mathbf{e}~.
}
The generators of $\SH$ can be identified with
\begin{align}\label{SH-gen}
x &=(1+t^2)\mathbf{e}X\mathbf{e}= (X+X^{-1})\mathbf{e}\\
y &=(1+t^{-2})\mathbf{e}Y\mathbf{e}= (Y+Y^{-1})\mathbf{e}\\
z &= (q^{-\frac12}Y^{-1}X + q^{\frac12}X^{-1}Y)\mathbf{e}=\frac{[x,y]_q}{(q^{-1}-q)}~,
\end{align}
and they satisfy relations
\bea\label{SH-rel}
[x,y]_q &= (q^{-1}-q)z\cr
{ } [y,z]_q &= (q^{-1}-q)x \cr
{ } [z,x]_q &= (q^{-1}-q)y \cr
q^{-1}x^2 + qy^2 + q^{-1}z^2-q^{-\frac12}xyz&=(q^{-\frac12}t-q^{\frac12}t^{-1})^2+(q^{\frac12} + q^{-\frac12})^2~,
\eea
where $q = e^{2\pi i\hbar}$ and the $q$-commutator is defined by
$$[a,b]_q:=q^{-\frac12}ab-q^{\frac12}ba~.$$
See {\it e.g.} \cite{terwilliger2012universal} for further details. The key point is that the spherical DAHA $\SH$ is commutative at the ``classical'' limit  $q=1$ while  the DAHA $\HH$ is not commutative even in the $q=1$ limit. Indeed, it is easy to see that in the ``classical'' limit $q \to 1$, the Casimir relation (the last one) in \eqref{SH-rel} becomes the equation for the cubic surface \eqref{toruscubic}:
\begin{equation}
\SH \xrightarrow[q\to 1]{} \OO(\MF(C_p,\SL(2,\C)))~.
\end{equation}
 Thus, $\SH$ is the deformation quantization $\OO^q(\X)$ of the coordinate ring~\eqref{toruscubic} of the moduli space of flat $\SL(2,\C)$-connections  $\X=\MF(C_p,\SL(2,\C))$ with respect to the Poisson bracket defined by~$\Omega_J$ \cite{Oblomkov:aa,oblomkov2004double}.

 Here, it is worth commenting on an important issue in the context of the deformation quantization of the coordinate ring on the affine cubic hypersurface of the form~\eqref{toruscubic}. It is clear that this equation is Weyl-group invariant, so that the monodromy parameter $\tilde{t}$ appears only through the symmetric combination $\tilde{t} + \tilde{t}^{-1}$, and that the same symmetry applies to the Poisson structure. Moreover, if we work with a specific value of $\tilde{t}$, we will obtain the deformation quantization at a specific value of the parameters, \textit{i.e.}\ for a specific choice of the central character (at least for the formal parameter $t$).

 Since the inputs to deformation quantization depend on~$\tilde{t}$ only in a $\Z_2$-invariant fashion, the output $\OO^q(\X_{\tilde{t}})$ will also have the corresponding symmetry. However, this clarifies that $\tilde{t} \neq t$, since the relations~\eqref{SH-rel} do \emph{not} depend symmetrically on $t$. The proper identification is
\begin{equation}\label{t-wtt}
   \tilde{t} = t q^{-\frac12}~,
\end{equation}
as will be made clear by the discussion of the formal outer automorphism $\iota$ below. There is no contradiction with the statement that $\SH$ is the deformation quantization of~$\OO(\X)$, since the classical limit of~$\SH$ still recovers the same commutative Poisson algebra.

 It is simple to check that the two involutions~\eqref{automorphisms1} straightforwardly reduce to the symmetry of $\SH$, which is the same as \eqref{sign-changes}. As in the classical case, the non-trivial central element $-1\in \SL(2,\Z)$ acts trivially on the generators of $\SH$, and the action of $\PSL(2,\Z)$ is quantized from \eqref{MCG-classical}
 \bea\label{MCG-quantum}
 \tau_+ &: \quad
 (x,y,z) \mapsto \Bigl(x,\frac{xy+yx}{q^{1/2}+q^{-1/2}}-z,y\Bigr)~,\cr
 \tau_- &: \quad
 (x,y,z) \mapsto \Bigl(\frac{xy+yx}{q^{1/2}+q^{-1/2}}-z,y, x\Bigr)~,\cr
 \sigma&:  \quad
 (x,y,z) \mapsto \Bigl(y,x,\frac{xy+yx}{q^{1/2}+q^{-1/2}}-z\Bigr) ~.
 \eea
 Thus, the symmetries $\Xi\times \PSL(2,\Z)$ can be seen in outer automorphisms of $\SH$. 
 The other outer automorphism
 $\tilde\iota$ in \eqref{automorphisms2} is somewhat more complicated; it does not preserve the idempotent, but it rather brings it into the other idempotent element
 \deq{
 \tilde\iota:\quad \mathbf{e} = \frac{T+t^{-1}}{t+t^{-1}} \mapsto \wt{\mathbf{e}}=\frac{-T+t}{t+t^{-1}}~.
 }
 Thus, $\tilde\iota$ maps $\SH$ to the other spherical subalgebra $\wt{\mathbf{e}}\HH\wt{\mathbf{e}}$ where the Casimir relations are different by $t\leftrightarrow t^{-1}$. However, the involution $\tilde\iota$ on $\HH$ does correspond in a sense to an outer automorphism of~\SH, which acts simply by
 \be\label{iota}
 \iota: t
  \mapsto qt^{-1}.
 \ee
 Indeed, it is easy to check that this map preserves the Casimir relation in~\eqref{SH-rel}. (Note that this automorphism only acts nontrivially when $q$ and $t$ are regarded as formal elements; it does not preserve the central character.)

In general, we are free to think of any commutative algebra as the coordinate ring of a certain affine space. In addition to the example above, we consider the case of $\MS=\C^\times\times\C^\times$ for the quantum torus algebra in Appendix \ref{app:qt}, and $\MS$ as 3d $\cN=4$ Coulomb branches in Appendix \ref{app:3dN=4} in this paper.  What is common between all of these examples are certain key properties of~\MS: First of all, it will always be a non-compact manifold, so that it has a large and interesting algebra $\OO(\MS)$ of holomorphic functions with polynomial growth at infinity. (In fact, in this paper, \MS{} will always be an affine variety over~$\C$.)
It will also be a \HK manifold, and an algebra is obtained by the deformation quantization of the coordinate ring of $\X$ with respect to a certain holomorphic symplectic form.
These conditions fit into the context of brane quantization \cite{Gukov:2008ve} in a  2d sigma-model. It is the central idea of this paper, and this will pave the way towards a geometric angle on the representation theory of $\SH$.

\subsection{\texorpdfstring{Canonical coisotropic branes in $A$-models}{Canonical coisotropic branes in A-models}}
\label{sec:Bcc}

Here, we will obtain the deformation quantization of the coordinate ring of $\X$ with respect to $\Omega_J$ by using the 2d $A$-model on a symplectic manifold $(\X,\omega_\X)$. The main character in our story is the \emph{canonical coisotropic brane}, denoted $\Bcc$. Eventually, we will investigate the representation theory of $\SH$ by the 2d $A$-model, but we begin by constructing the (presumably less familiar) canonical coisotropic brane $\Bcc$ here. Subsequently, we will discuss standard Lagrangian branes and some methods for computing spaces of morphisms in what follows. Our review is necessarily cursory;  for more details, we refer to the literature \cite{Gukov:2008ve,Gukov:2010sw}.

In  general, as was pointed out in~\cite{Kapustin:2001ij}, the $A$-model admits branes with support on coisotropic submanifolds which are equipped with a transverse holomorphic structure. The canonical coisotropic brane is supported on the target space $\X$ itself, which is a coisotropic submanifold of the target space $\X$ in a rather trivial way. More precisely, there is a family of such branes, labeled by a complex parameter
\be\label{hbar}
\hbar  =  |\hbar| e^{i \theta}~,
\ee
and we will identify it with the parameter of deformation quantizations by $q=e^{2\pi i\hbar}$. The fact that the support involves no additional choice is (at least part of) the reason for the term ``canonical.'' On a $2n$-dimensional target space, coisotropic branes can therefore be supported in dimension $n + 2j$ for integer $j$; when $n$ is even, there can be branes supported throughout the entire target. In our example, $n=2$,  so that no other coisotropic branes can occur just for dimension reasons.

In complex structure $\mathcal{I} =  I \cos \theta - K \sin \theta$, the data defining the brane $\Bcc$ is simply a holomorphic line bundle $\cL \to \X$, equipped with a connection whose curvature $F$ is of course equal to the first Chern class:
\be
\label{Bcc}\Bcc:\quad
\begin{tikzcd}
\cL \arrow[d ] \\
\X
\end{tikzcd}  \qquad\qquad c_1(\cL)=[F/2\pi]\in H^2(\X,\Z)~.
\ee
As usual, open strings ending on $\Bcc$ source the gauge-invariant combination $F + B$, where
\deq{
B \in H^2 (\X, \U(1))
}
is the 2-form $B$-field.
For our family of the canonical coisotropic branes $\Bcc$ parametrized by $\hbar$ on a symplectic manifold $(\X,\omega_\X)$, the values of $[B/2\pi]\in H^2(\X,\U(1))$ and the integral class $[F/2\pi]\in H^2(\X,\Z)$ are determined by the equation
\be\label{Bcc-Omega}
\Omega:=F+B+i\omega_\X=\frac{\Omega_J}{i\hbar}~,
\ee
so that at a generic value of $\hbar$ in \eqref{hbar} we can write
\begin{align}\label{generic-Bcc}
F + B & = \Re\,\Omega  =  \frac{1}{|\hbar|}( \omega_I\cos\theta - \omega_K\sin\theta)~,\cr
\omega_\X  &= \Im\,\Omega= - \frac{1}{|\hbar|}(\omega_I\sin\theta+\omega_K\cos\theta)~.
\end{align}
Since the \HK conditions ensure that $J=\omega_\X^{-1}(F+B)$, we have the condition for $\Bcc$ to be a coisotropic $A$-brane \cite{Kapustin:2001ij}
\be
\label{coisotropic-cond} \big( \omega_{\X}^{-1} (B + F) \big)^2  =  J^2  =  -1~.
\ee
In particular, when $\hbar$ is real, $\omega_\X = \omega_K$ and $\Bcc$ is a brane of type $(B,A,A)$, whereas for $\hbar$ purely imaginary, $\omega_\X = \omega_I$ and $\Bcc$ is an $(A,A,B)$-brane. \Bcc{} is also called ``canonical'' because its extra data corresponds in this fashion to the holomorphic symplectic structure.

Now comes the key point.
Under this circumstance, the space $\Hom(\Bcc,\Bcc)$ of open $(\Bcc,\Bcc)$ strings with both ends
on the canonical coisotropic brane $\Bcc$ is a non-commutative deformation of the Dolbeault cohomology $H_{\overline\partial}^{0,*}(\X)$ when $\X$ is regarded as a complex manifold with $J$, and we are interested in its zeroth degree, namely the space of holomorphic functions.
Moreover, for $\X$ an affine variety, a suitable condition at infinity for a ``good $A$-model'' is to allow only functions of polynomial growth.
In the presence of non-trivial background $F+B \ne 0$, $\Hom^0(\Bcc,\Bcc)$ is therefore the deformation quantization of the coordinate ring on $\X$,
holomorphic in complex structure $J$ \cite{aldi2005coisotropic,Gukov:2008ve}. \footnote{Since we are mainly interested in the zeroth degree of morphism spaces, we will usually omit the superscript $0$, meaning $\Hom=\Hom^0$ unless it is specified.}

In general, for any brane $\brane$, in either the $A$-model or the $B$-model, the space of open strings states $\End(\brane)$ forms an algebra. This can be easily understood by considering the process of joining open strings, illustrated in Figure~\ref{fig:Bcc-algebra-rep} (left).
However, generically, this algebra of $(\brane,\brane)$ strings is rather simple and not very interesting. Even if the brane $\brane$ is ``big enough,'' the algebra $\End(\brane)$ can be interesting, but may be hard to identify or relate to more familiar algebras. For example, various $(B,B,B)$ branes represented by hyper-holomorphic sheaves in \cite{Gukov:2010sw} lead to interesting endomorphism algebras, but apart from some special cases it is hard to recognize them in the world of more familiar algebras.
What makes the canonical coisotropic brane special is that the algebra~$\End(\Bcc)$ can be identified with the deformation quantization $\OO^q(\X)$ of the target manifold $\X$~\cite{Kapustin:2006pk}.

\begin{figure}[ht]
	 \centering
		 \includegraphics[width=4.5cm]{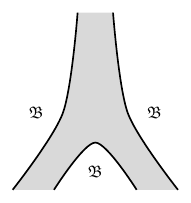} \hspace{3cm}
		 \includegraphics[width=4.5cm]{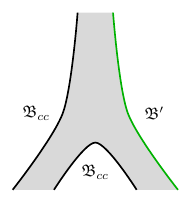}
 \caption{(Left) Open strings that start and end on the same brane $\brane$ form an algebra. \\ (Right) Joining a $(\Bcc,\Bcc)$-string with a $(\Bcc,\brane')$-string leads to another  $(\Bcc,\brane')$-string.}
 \label{fig:Bcc-algebra-rep}
\end{figure}

\subsubsection{\texorpdfstring{Spherical DAHA as the algebra of $(\Bcc,\Bcc)$-strings}{Spherical DAHA as the algebra of (Bcc,Bcc)-strings}}\label{sec:Bcc-SH}

In our example, the target space  $\X=\MF(C_p,\SL(2,\C))$  is the moduli space of flat $\SL(2,\C)$-connections over a punctured torus, which is a \HK manifold. Then, by construction, the algebra of $(\Bcc,\Bcc)$ strings is the deformation quantization $\OO^q(\X)$ of the coordinate ring on $\X$ with respect to $\Omega_J$, which is the spherical DAHA $\SH$.

The parameter $q$ of $\SH$ is identified with $\hbar$ in the data \eqref{Bcc-Omega} of $\Bcc$ via $q=\exp(2\pi i\hbar)$.
It is worth emphasizing that for a generic value of $q\in \bC^\times$, the $B$-field needs to be turned on in the sigma-model.  In fact, the target admits the Hitchin fibration \eqref{Hitchin-fibration} where a generic fiber is a two-torus $T^2$. Since a generic fiber $\bfF$ is Lagrangian with respect to $\omega_J$ and $\omega_K$ and it sees only $\omega_I$, the evaluation of $\Omega$ in \eqref{Bcc-Omega} over $\bfF$   yields
$$
\int_{\bfF}\frac{\Omega}{2\pi}=\frac1\hbar~,
$$
where $F+B$ is responsible for its real part.
Because $[F/2\pi]\in H^2(\X,\Z)$ is an element of the second integral cohomology class, the $B$-field needs to be switched on unless the real value  of $1/\hbar$ is an integer. Thus, a 2d $A$-model has to incorporate the $B$-field for a generic value of $\hbar$, and we will moreover witness that the $B$-field plays a more important role in the subsequent sections.

The parameter $t$ of $\SH$ is related to the ramification parameters of the target space. In fact, the monodromy parameter \eqref{monodromy2} around the puncture can be expressed by the ramification parameters \eqref{monodromy} as $$\tilde{t}=\exp(-\pi(\tgamma_p+i \talpha_p))~.$$ Furthermore, \eqref{t-wtt} compares the monodromy parameter $\tilde{t}$ with the central character $t$ of $\SH$. Then, it is easy to see from \eqref{integral-D} that the evaluation of \eqref{Bcc-Omega} on an exceptional divisor yields
\be \label{c-ramification}
\frac{1}{2\pi}\int_{\D_i}F+B+i\omega_\X=\int_{\D_i}\frac{\Omega_J}{2\pi i\hbar}=\frac{\tgamma_p+i \talpha_p}{2i\hbar}=-c+\frac12~.
\ee
where $c$ is the central charge in \eqref{central-charge}.

The canonical coisotropic brane enjoys the symmetries $\Xi\times \PSL(2,\Z)$ of the target space $\X$ analyzed in \S\ref{sec:target}, which become the outer automorphisms of $\SH$ given by \eqref{sign-changes} and \eqref{MCG-quantum}. The symmetry \eqref{iota} of $\SH$ is indeed the Weyl group symmetry $\tilde{t}\leftrightarrow \tilde{t}^{-1}$ of the monodromy matrix \eqref{monodromy2}. In fact, the Weyl group symmetry \eqref{ramification-Weyl} of the ramification parameters preserves the target space. Since the canonical coisotropic brane is sensitive only to $(\talpha_p,\tgamma_p)$ or $\tilde{t}$, the symmetry \eqref{iota} of $\SH$ is equivalent to the fact that the canonical coisotropic branes supported on $\X_{\tilde{t}}$ and $\X_{\tilde{t}^{-1}}$ related by the Weyl group symmetry give rise to the isomorphic algebra
\begin{equation}
  \End(\Bcc)\cong \SH\cong\End(\iota(\Bcc))~.
\end{equation}

\subsection{\texorpdfstring{Lagrangian $A$-branes and modules of $\OO^q(\X)$}{Lagrangian A-branes and modules of Oq(X)}}
\label{sec:Lagrangian}

Now we lay out the approach to the representation theory of $\OO^q(\X)$ by the 2d $A$-model on $(\X,\omega_\X)$.
This subsection also serves as a lightning review about the category of $A$-branes.

The approach to the representation theory of $\OO^q$ from the 2d $A$-model arises from a simple idea: given an open string boundary condition (or an $A$-brane) $\brane'$, the space of $(\Bcc,\brane')$ open strings forms a vector space. As in  the right of Figure~\ref{fig:Bcc-algebra-rep}, a joining of $(\Bcc,\Bcc)$ and $(\Bcc,\brane')$ string leads to another $(\Bcc,\brane')$ string. This implies that the space of $(\Bcc,\brane')$ strings receives an action of the algebra of $(\Bcc,\Bcc)$ strings \cite{Gukov:2008ve}. Namely, other $A$-branes $\brane'$ on~$\X$ give rise to modules for $\OO^q(\X)$:
\be\label{HHrep}
\begin{array}{ccc}
 \OO^q(\X) & = & \Hom(\Bcc,\Bcc) \\
 \rotatebox[origin=c]{180}{$\circlearrowright$} & & \rotatebox[origin=c]{180}{$\circlearrowright$}  \\
 \repB' & = & \Hom(\Bcc,\brane')
\end{array}
\ee
In our example, supports of other branes $\frakB'$ are always Lagrangian submanifolds so that we will review Lagrangian $A$-branes $\brL$ in the next subsection.
If the support of $\frakB'$ is a Lagrangian submanifold contained in the fixed point set of an antiholomorphic involution $\zeta:\X\to \X$ with $\zeta^{*}J=-J$, then the corresponding representation admits unitarity.

We now briefly recall a few standard facts about Lagrangian $A$-branes \cite{floer1988morse,floer1989witten} and their mathematical incarnation, the Fukaya category $\Fuk(\MS,\omega_\MS)$. For more detail, the reader is referred to the literature, which is  substantial; \cite{Auroux} is a good starting point, or~\cite{KontsevichICM} for the fundamentals of homological mirror symmetry.

The \emph{Lagrangian Grassmannian}, denoted \LGr, of a symplectic vector space parameterizes the collection of its Lagrangian subspaces. We can obtain a description of this space by thinking of the standard symplectic vector space $(\R^{2n},\omega)$, which can be equipped with a metric via a contractible choice. By the two-of-three property, the group preserving both the symplectic and orthogonal structures is $\U(n)$, which therefore acts on $\LGr(2n)$; the subgroup stabilizing a fixed Lagrangian subspace is $\Or(n)$, so that
\deq{
\LGr(2n) = \U(n) / \Or(n)~.
}
There is furthermore an obvious map
\deq{
{\det}^2: \LGr(2n) \to \U(1)
}
which can be shown to induce an isomorphism of fundamental groups. The \emph{Maslov index}~\cite{ArnoldMaslov} of a loop in~$\LGr(2n)$ is its image under this induced map in $\pi_1(\U(1)) \cong \Z$; it is responsible for both obstructions and gradings in the Fukaya category. The universal cover $\tLGr(2n)$ of $\LGr(2n)$ thus has deck group $\Z$, and the Maslov index of a loop is simply the element of $\Z$ that connects the endpoints of a lift to~$\tLGr(2n)$.

Let $(\X,\omega_\X)$ be a symplectic manifold with zero first Chern class (as is obviously  the case in our \HK examples). There is a bundle
\deq{
\LGr(\X) \to \X
}
whose fiber over $x\in\X$ is $\LGr(T_x\X)$. (We hope the reader will forgive the moderately abusive notation.)
We can furthermore define a bundle $\tLGr(\X)$, which is a covering space of the total space $\LGr(\X)$, such that the covering map is a bundle map and restricts over each fiber to the universal covering map.

A Lagrangian subspace $\L\subset \X$ comes with an obvious lift
\begin{equation}
\begin{tikzcd}
& \LGr(\X) \ar[d] \\
\L \arrow{r}{\subset} \arrow[hook]{ru}{} & \X
\end{tikzcd}
\end{equation}
defined by the Lagrangian subbundle $T\L \subset \left.T\X\right|_\L$. Lifting this canonical map to $\tLGr(\X)$ is obstructed by the image of $\pi_1(\L)$ under the Maslov map, which is an element of $H^1(\L,\Z)$ called the \emph{Maslov class}. Lagrangians with zero Maslov class admit so-called \emph{graded lifts}, which are maps
\begin{equation}
\begin{tikzcd}
\tLGr(\X) \arrow{r}{\cdot/\Z} & \LGr(\X) \ar[d] \\
\L \arrow[dashed]{u}{g} \arrow{r}{\subset} \arrow[hook]{ru}{} & \X
\end{tikzcd}
\end{equation}
making the square commute. The set of such maps is naturally a $\Z$-torsor under the action of deck transformations, but no canonical choice of graded lift exists. Given a Lagrangian object of $\ABrane(\X,\omega_\X)$, the set of graded lifts plays the role of its shifts.

  A (rank-one) Lagrangian object of $\ABrane(\X,\omega_\X)$ is supported on a Lagrangian submanifold $\L\subset \X$ of zero Maslov class, which is considered up to Hamiltonian isotopy. The additional data required to define a Lagrangian $A$-brane consists of a ``Chan-Paton'' bundle with unitary connection; a flat $\Spin^c$ structure on~$\L$; and a grade lift.
A Chan-Paton bundle for a Lagrangian $A$-brane is generally endowed with a flat $\Spin^c$ structure \cite{Witten:1998cd,Freed:1999vc,Katz:2002gh,Gukov:2008ve}.
A $\Spin^c$ structure arises if $\cL'$ does not exist as a line bundle, but is obstructed
by the same cocycle that obstructs the existence of the square root $K_\L^{-1/2}$ of the canonical bundle over the Lagrangian $\L$. Namely, putative transition functions $g_{i j}$ and $w_{i j}$ of $\cL'$ and $K_\L^{-1/2}$, respectively, obey  $g_{i j} g_{j k} g_{k i}=\phi_{i j k}=w_{i j} w_{j k} w_{k i}$ where $\phi_{i j k}=\pm1$. In this case, the cocycle cancels out in the transition functions $g_{i j} w_{i j}$ of an honest vector bundle $\cL'\otimes K_\L^{-1/2}\to \L$, called a $\Spin^c$ structure. The $K_\L^{-1/2}$ part in a $\Spin^c$ structure arises from boundary fermions of the open worldsheet \cite[\S5]{Hori:2000ic} \cite[\S3.2]{Herbst:2008jq}, which gives rise to a $\Spin^c$ structure of the normal bundle to the support of a brane. (This proposal is explicitly checked by Hemisphere partition functions in \cite{Kim:2013ola}.) Thus, the canonical coisotropic brane $\Bcc$ is endowed with an ordinary line bundle whereas a Lagrangian $A$-brane is equipped with a $\Spin^c$ structure.
Since most of the Lagrangian submanifolds in this paper are of real two dimensions, there always exists a spin bundle of $\L$, which is a square-root of the canonical bundle $K_\L^{\pm1/2}$ of $\L$, though it is not necessarily unique. Hence, both $\cL'$ and $K_\L^{-1/2}$ exist as genuine line bundles in most of the examples in this paper and we treat their tensor product $\cL'\otimes K_\L^{-1/2}$ as a $\Spin^c$ structure. However, a subtlety arises when an $A$-brane degenerates into a different spin structure, which will be considered in \S\ref{sec:bound-state}. Moreover, a Lagrangian $A$-brane is endowed with a flat $\Spin^c$ structure:
if $\cL'$ exists as a line bundle, the curvature $F'_\L$ of $\cL'$ must obey a gauge-invariant version of the flatness condition
\deq[deformed-flat]{
F'_\L + \left.B\right|_\L
=0~,
}
in the presence of a $B$-field. Even if $\cL'$ does not exist as a line bundle, its square $(\cL')^2$ does so that a half of the curvature of $(\cL')^2$ is subject to \eqref{deformed-flat}.
In sum, for a Lagrangian $A$-brane, we have a Chan-Paton bundle
\be
\label{BL}\brL:\quad
\begin{tikzcd}
\cL'\otimes K_\L^{-1/2} \arrow[d ] \\
\L
\end{tikzcd}
\ee
with a flat $\Spin^c$ structure \eqref{deformed-flat}.
We will sometimes denote a Chan-Paton bundle by $\brL\to \L$, abusing notation.
Morphisms between Lagrangian objects are defined in the usual way using the Floer--Fukaya complex generated by intersection points between the Lagrangians; see~\cite{Auroux} for details.

Defining the space of morphisms between Lagrangian and coisotropic objects is a bit more subtle, and is discussed in detail for flat targets in~\cite{aldi2005coisotropic}. The essential idea is that the morphism space should be thought of as related to the space of holomorphic functions on the intersection, with respect to the transverse holomorphic structure on coisotropic objects. For Lagrangian objects, this complex structure obviously plays no role, but instanton corrections can appear, in the guise of the contributions of holomorphic disks to the differential in the Floer--Fukaya complex. On the other hand, for \Bcc, the transverse holomorphic structure is just a standard complex structure and plays an essential role, but instanton corrections are forbidden. In the case of general coisotropic branes, both phenomena can be expected to be relevant. (For some further discussion of this fact from the worldsheet perspective, as well as generalizations to branes of higher rank, see~\cite{Herbst:2010js}.)

In a \HK manifold, we can make use of a $B$-model analysis as in~\cite{Gukov:2008ve,Gukov:2010sw} to compute the dimension of open strings. The dimension of the representation space $\scL:=\Hom(\Bcc,\brL)$ associated to a compact Lagrangian brane $\brL$ can be computed with the help of the Grothendieck--Riemann--Roch formula:
\bea\label{dimension2}
\dim \scL& = \dim H^0(\L,\Bcc\otimes \brL^{-1}) \\
&= \int_{\L} \text{ch} (\Bcc) \wedge \text{ch} (\brL^{-1}) \wedge \Td(T\bfL)~,
\eea
Here we denote, by $\frakB$, a bundle for the corresponding brane including an effect of the $B$-field, abusing notation.

If a Lagrangian $\L$ is of real two dimensions, then the Todd class $\Td(T\bfL)=\text{ch}(K_\L^{-1/2})\widehat A(T\bfL)$ is equivalent to $\text{ch}(K_\L^{-1/2})$. Consequently, the formula becomes a very simple form
\be \label{dimension}
\dim \scL=\int_{\L}\text{ch} (\Bcc)= \int_{\L}\frac{F+B}{2\pi}~,
\ee
for a real two-dimensional Lagrangian $\bfL$.

As explained in \cite{Gukov:2008ve}, for a Lagrangian brane $\brL$, the space of open strings $\Hom(\Bcc,\brL)$ can be understood as a geometric quantization of $\L$ with a curvature on a ``prequantum line bundle'' $\Bcc\otimes \brL^{-1}$. If \X{} is a complexification of~$\L$ in the sense of~\cite{Gukov:2008ve}, then the action of $\End(\Bcc)$ on the quantization $\Hom(\Bcc,\brL)$ plays the role of the quantized algebra of operators.

Finally, let us mention a brief word about coefficients. In general, the Fukaya category is defined with coefficients in the Novikov ring; this is necessary because the sums over instanton contributions that define the differential are formal and not necessarily guaranteed to converge. Similarly, deformation quantization of a Poisson manifold~\cite{Groenewold,Fedosov,KontsevichPoisson} is not guaranteed to produce convergent series, but only a formal deformation in general. We will restrict ourselves to target spaces \X{} for which a  ``good $A$-model'' is expected to exist, meaning that all the series involved should in fact converge. The existence of a complete \HK metric on~\X{} should be sufficient to ensure this; see~\cite{Gukov:2008ve} for further discussion of this issue.

We will proceed to compare the two categories $\ABrane(\X,\omega_\X)$ and $\Rep(\SH)$ via the brane quantization.\footnote{
Note, that spherical DAHA is Morita-equivalent to DAHA \eqref{DAHA1}, \textit{i.e} the category of representations of DAHA is equivalent to the category of representations of its spherical subalgebra \cite{oblomkov2004double}:
\be\label{Morita-equiv}
\Rep(\HH) \cong\Rep(\SH)~.
\ee
See also \eqref{Morita} for the explanation from the 2d sigma-model.}
For the comparison, the symmetries play a crucial role. In fact, the symmetries of the target space $\X$ become the group of auto-equivalences of the categories. More concretely, we will investigate the action of $\Xi\times \PSL(2,\Z)$ (\eqref{sign-changes} and \eqref{MCG-quantum}) and the Weyl group $\Z_2$ generated by $\iota$ \eqref{iota} on both categories.

Now we set up the framework so that we will start our expedition to ``see'' and ``touch'' representations of $\SH$ as if they were geometric objects in the target $\X$.

\subsection{\texorpdfstring{$(A,B,A)$-branes for polynomial representations}{(A,B,A)-branes for polynomial representations}}
\label{sec:poly-rep}
DAHA was introduced by Cherednik in the study of Macdonald polynomials from the viewpoint of representation theory \cite{Cherednik-daha} in which the distinguished infinite-dimensional representation on the ring $\PR := \CR[X^\pm]^{\Z_2}$ of symmetric Laurent polynomials, called \emph{polynomial representation}, plays an important role. (See also \cite{cherednik2017galois} for finite-dimensional modules.) Here, Laurent polynomials in a single variable $X$ over~$\CR$ are symmetrized under the inversion $\Z_2:X \mapsto X^{-1}$ so that the ring can also be expressed as $\PR = \CR[X+X^{-1}]$.
This polynomial representation of $\SH$ is defined by the following formulas:
\begin{equation}
\label{poly-rep-sym}
\begin{aligned}
x&\mapsto X+X^{-1}, \\
\pol:\SH \to \End(\PR),  \quad  y&\mapsto \frac{tX-t^{-1}X^{-1}}{X-X^{-1}}\varpi +\frac{t^{-1}X-tX^{-1}}{X-X^{-1}}\varpi^{-1},\\
z&\mapsto q^{\frac12}X\frac{tX-t^{-1}X^{-1}}{X-X^{-1}}\varpi +q^{\frac12}X^{-1}\frac{t^{-1}X-tX^{-1}}{X-X^{-1}}\varpi^{-1} ,
\end{aligned}
\end{equation}
where $\varpi^{\pm}(X)=q^{\pm}X$ is the exponentiated degree operator, often called the $q$-shift operator, that appeared in \eqref{shift-operator} for the quantum torus algebra. In particular, $\pol(y)$ is the so-called \emph{Macdonald difference operator}, whose eigenfunctions are \emph{symmetric Macdonald polynomials} \cite{macdonald1998symmetric,Cherednik-book}. The Macdonald functions of type  $A_1$ are labeled by spin-$\frac{j}{2}$ representations, and can be expressed in terms of the basic hypergeometric series
\be\label{sym-Macdonald-A1}
P_{j}(X;q,t):=X^j\, {}_2\phi_1(q^{-2j}, t^2;q^{-2j+2} t^{-2};q^2  ;q^2t^{-2}X^{-2})~.
\ee
They are acted on  diagonally by the Macdonald difference operator, with eigenvalues
\be\label{Macdonald-eigen}
\pol(y)\cdot P_{j}(X;q,t)=(q^j t+q^{-j} t^{-1})P_{j}(X;q,t)~.
\ee
Under this basis, the actions of the other generators are
\bea\label{x-z-action}
\pol(x)\cdot P_j(X; q, t)=&P_{j+1}(X; q, t)+\frac{\left(1-q^{2 j}\right)\left(1-q^{2 j-2} t^{4}\right)}{\left(1-q^{2j-2} t^{2}\right)\left(1-q^{2 j} t^{2}\right)} P_{j-1}(X; q, t)~,\cr
\pol(z)\cdot P_j(X; q, t)=&t q^{j+\frac{1}{2}} P_{j+1}(X; q, t)+t^{-1} q^{-j+\frac{1}{2}} \frac{\left(1-q^{2 j}\right)\left(1-q^{2 j-2} t^{4}\right)}{\left(1-q^{2 j-2} t^{2}\right)\left(1-q^{2 j} t^{2}\right)} P_{j-1}(X; q, t)~.
\eea

In fact, the Macdonald polynomials $P_j$ form a basis for the ring $\PR$ over~\CR, so that the polynomial representation can be studied with the help of raising and lowering operators~\cite{KN:1998}:
\bea\label{RL}
\sfR_j&:=x-q^{j-\frac12}tz=X(q^{j}t^{-1}Y-q^{2j}t^{2})+X^{-1}(q^{j}tY^{-1}-q^{2j}t^{2})~, \cr
\sfL_j&:=x-q^{-j-\frac12}t^{-1} z=X(q^{-j}t^{-3}Y-q^{-2j}t^{-2})+X^{-1}(q^{-j}t^{-1}Y^{-1}-q^{-2j}t^{-2})~.
\eea
These operators relate adjacent Macdonald polynomials, respectively increasing or decreasing the value of $j$:
\begin{align}
\pol(\sfR_j)\cdot P_{j}(X;q,t)&= (1-q^{2j}t^{2}) P_{j+1}(X;q,t)~,\label{raising} \\
\pol(\sfL_j)\cdot P_{j}(X;q,t)&= \frac{(1-q^{2j})(1-q^{2(j-1)}t^{4})}{q^{2j}t^{2}(q^{2(j-1)}t^{2}-1) }P_{j-1}(X;q,t)~.\label{lowering}
\end{align}
See Figure~\ref{fig:RLops} for a schematic diagram of this action.
\begin{figure}[ht]\centering
	\includegraphics[width=0.9\textwidth]{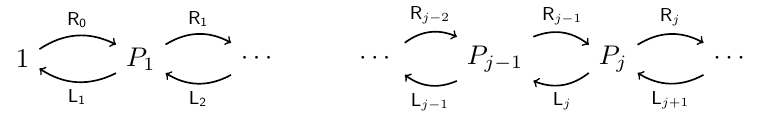}
	\caption{The action of raising and lowering operators on Macdonald polynomials}
	\label{fig:RLops}
\end{figure}
At $t=1$, this representation reduces to the pullback of the lift of $\PR^{y_1=1}$ in Proposition \ref{prop:liftPR} so that Cherednik's polynomial representation can be understood as its deformation from the symmetrized quantum torus to DAHA. Since the classical limit $(q=1)$ of the Macdonald eigenvalues \eqref{Macdonald-eigen} is always $t+t^{-1}$, the support of the corresponding $A$-brane $\brP$ is given by
\be\label{brane-poly}\bfP=\{y=\tilde{t}+\tilde{t}^{-1}~, \ z =\tilde{t}^{-1}x\}~.\ee
While the parameter $t$ in $\SH$ coincides with the monodromy parameter $\tilde{t}$ at the classical limit $(q=1)$ (see \eqref{t-wtt}), we use $\tilde{t}$ to specify the position of the brane because it is the geometric parameter of $\X$.
Since it is of type $(A,B,A)$, it is happily a Lagrangian submanifold with respect to $\omega_\X$ for any value of $\hbar$ or $q$.

To understand the brane $\brP$ for the polynomial representation $\repP$ of $\SH$ better, it is illuminating to consider its relation to the skein module.
The skein module of type $A_1$ \cite{turaev1990conway,przytycki2006skein}  of an oriented 3-manifold $M_3$ is defined as
\deq{
\Sk(M_3,\SU(2)) :=\Sk(M_3) = \frac{\bC[q^{\pm\frac12}] (\text{isotopy classes of  framed links in $M_3$})}%
{\Bigl( \cross = q^{-1/2} \resone\, + q^{1/2} \restwo\,  , ~ \unknot\, = -q-q^{-1} \Bigr) }
}
The skein algebra $\Sk(C)$ associated to an oriented closed surface $C$ is defined as
\deq{
\Sk(C):=\Sk(C\times [0,1],\SU(2))~~,
}
where the multiplication $\Sk(C)\times \Sk(C)\to \Sk(C)$ is given by stacking. As a result, $\Sk(C)$ is a $\bC[q^{\pm\frac12}]$-associative algebra \cite{turaev1991skein}.

At the $q=1$ specialization, the skein module $\Sk(M_3)$ becomes a commutative algebra.  Moreover, it was shown in \cite{bullock1997rings,przytycki1997skein} that by assigning a loop $\gamma:S^1\to M_3$ to $\Tr(\rho(\gamma))$ where $\rho:\pi_1(M_3) \to \SL(2,\C)$ is the holonomy homomorphism, the classical limit $q=1$ of $\Sk(M_3)$ is isomorphic to the coordinate ring of the character variety $\MF(M_3,\SL(2,\C))$.
Hence, the skein module $\Sk(M_3)$ can be understood as a BV quantization \cite{gunningham2019finiteness}
$$
\Sk(M_3)\cong \textrm{BV}^q(\MF(M_3,\SL(2,\C)))~.
$$
The skein module of a closed 3-manifold will be studied in \S\ref{sec:SL2Z-skein}.

If a 3-manifold has a boundary $\partial M_3 = C$, then we have a module  $\Sk(C)$ \rotatebox[origin=c]{-90}{$\circlearrowright$} $\Sk( M_3 )$ by pushing a framed  links in a thickened boundary $C\times [0,1]$ into the bulk $M_3$.
In fact, $\MF(M_3,\SL(2,\C))$ is a holomorphic Lagrangian submanifold of $\MF(C,\SL(2,\C))$ with respect to the holomorphic symplectic form $\Omega_J$. Therefore, it can be understood as an $(A,B,A)$-brane $\frakB_H$ on $\MF(C,\SL(2,\C))$, called a \emph{Heegaard brane}. From the viewpoint of brane quantization, the action of the skein algebra can be understood as
\be\label{skein-module}
\begin{array}{ccc}
	\Sk(C) & \cong & \Hom(\Bcc,\Bcc)\\
	\rotatebox[origin=c]{180}{$\circlearrowright$} & & \rotatebox[origin=c]{180}{$\circlearrowright$}  \\
	\Sk(M_3) &\cong & \Hom(\Bcc,\frakB_H) \\
\end{array}~.
\ee

Of our particular interest is the skein algebra $\Sk(T^2)$ of a torus, which is the $t=q$ specialization of $\SH$ \cite{bullock2000multiplicative}. Also, the skein module $\Sk(S^1\times D^2)$ of the solid torus is the Grothendieck ring of the category of finite-dimensional representations of $U_q(\fraksl(2))$
\deq{
\Sk(S^1\times D^2)\cong K^0 (\mathsf{Rep}\,U_q(\fraksl(2)))\otimes \bC[q^{\pm\frac12}]~~,
}
which is spanned by Chebyshev polynomials \(S_{j}(z)\) of the second kind \cite{frohman2000skein}. They are recursively defined by
\deq{
z S_{j}(z)=S_{j+1}(z)+S_{j-1}(z)
}
with the initial conditions \(S_{0}(z)=1, S_{1}(z)=z\), and they are actually the $t=q$ specialization of the Macdonald polynomials
\deq{
S_{j}(X+X^{-1})=P_{j}(X; q,t=q)=\frac{X^{j+1}-X^{-j-1}}{X-X^{-1}}~.
}
Consequently, the polynomial representation $\repP$  of $\SH_{t=q}$ is indeed the skein module $\Sk(T^2)$ \rotatebox[origin=c]{-90}{$\circlearrowright$} $\Sk(S^1\times D^2)$. In fact, the support of the Heegaard brane for the solid torus is given by $y=2$, which is the $A$-polynomial of the unknot complement in $S^3$. Indeed the eigenvalue of the $y$ operator on $S_{j}(X+X^{-1})$ under the polynomial representation $\repP$  at $t=q$ is $q^{j+1}+q^{-j-1}$ and its classical limit is $y=2$.
Thus, the polynomial representation $\repP$ of $\SH$ can be understood as the $t$-deformation of the skein module $\Sk(T^2)$ \rotatebox[origin=c]{-90}{$\circlearrowright$} $\Sk(S^1\times D^2)$ \cite{Hikami:2019jaw}.

Let us consider how the symmetries of $\SH$ act on the polynomial representation $\scP$. For instance, an action of $\PSL(2,\Z)$ on $\scP$ can be seen by using the maps \eqref{MCG-quantum} combined with \eqref{poly-rep-sym}. It is easy to see from \eqref{MCG-quantum} that the generators of $\PSL(2,\Z)$ maps $\frakB_\bfP$ to another $(A,B,A)$ brane
\bea\label{brane-poly-SL2Z}
\tau_{+}(\bfP)&=\{xy-z=\tilde{t}+\tilde{t}^{-1}~, ~ y=\tilde{t}^{-1} x\}~,\cr
\tau_{-}(\bfP)&=\{y=\tilde{t}+\tilde{t}^{-1}~, ~ z=\tilde{t}^{-1} x\}~,\cr
\sigma(\bfP)&=\{x=\tilde{t}+\tilde{t}^{-1}~, ~ y=\tilde{t}^{-1} z\}~.
\eea
Under the modular $T$-transformation $\tau_-$, the support does not change and the polynomial representation $\PR$ is invariant, $\tau_-(\PR)\cong\PR$ since the Macdonald polynomials are transformed under the modular $T$-transformation $\tau_-$ as
\be\label{Mac-T}
\tau_{- } ( P_{j} ) = q^{-\frac{j^{2 } }2}t^{-  j } P_{j } \quad \rightsquigarrow  \quad  T_{jj'}=q^{-\frac{j^{2 } }2}t^{-  j } \delta_{jj'}~.
\ee
The proof is given at the end \eqref{tau-p}  of Appendix~\ref{app:DAHA-poly}.
The image $\sigma(\scP)$ of the polynomial representation of $\SH$ under the $S$-transformation $\sigma$ is called the \emph{functional representation}, which is explained in Appendix~\ref{app:func-rep}. As for the group $\Xi$ of the sign changes, the image $\xi_1(\scP)$ is isomorphic to itself $\scP\cong \xi_1(\scP)$. On the other hand, the image under the involution $\xi_2$ can be obtained by multiplying the minus sign to $y$ and $z$ as in \eqref{sign-changes} and the support of the corresponding brane is
\be \label{zeta2-support}
\xi_2(\bfP)=\{y=-\tilde{t}-\tilde{t}^{-1}~, ~ z=-\tilde{t}^{-1} x\}~.
\ee
Finally, the outer automorphism \eqref{iota} changes the Chan-Paton bundle of $\brP$ as explained in \S\ref{sec:Bcc-SH} and the support becomes
\be \label{brane-poly-iota}
\iota(\bfP)=\{y=\tilde{t}+\tilde{t}^{-1}~, ~ z=\tilde{t} x\}~.
\ee
Note that the $\iota$ image $\iota(\PR)$ of the polynomial representation can be obtained by changing $t\to q/t$ in \eqref{poly-rep-sym}.

The perspective from the brane quantization also sheds new light on infinite-dimensional representations. We have seen that Cherednik's polynomial representation \eqref{poly-rep-sym} corresponds to the $A$-brane $\brP$ \eqref{brane-poly} at the particular value of $y$. It is natural to expect that it can be deformed in such a way that the corresponding brane is supported on a generic point of $y$.

This consideration leads us to the following.
Let us consider the multiplicative system $\widetilde M\subset \CR[X^{\pm}]$ generated by all elements of the form $(q^{\ell} X-q^{-\ell} X^{-1})$ for all integers $\ell \in \Z$. Then there is a family of representations of $\SH$ on the localization \footnote{In other words, $\PR^{y_1}$ is the ring of rational functions with coefficients in $\CR$ where denominators are always elements in the multiplicative system $\wt M$ such as
$$
\frac{f(X)}{(q^{-m} X-q^m X^{-1})^{k_{-m}}\cdots(X- X^{-1})^{k_0}\cdots (q^{\ell} X-q^{-\ell} X^{-1})^{k_\ell}}~,\qquad f(X)\in \CR[X^{\pm}]~.
$$} of the ring of Laurent polynomials by $\wt M$
\deq{
\PR^{y_1} = \widetilde M^{-1}  \CR[X^{\pm}]~,
}
labeled by a parameter $y_1 \in \C^\times$ where the representations are defined by
\begin{equation}
\label{poly-rep-y1}
\begin{aligned}
x&\mapsto X+X^{-1}, \\
\pol_{y_1}:\SH \to \End(\PR^{y_1}),  \quad  y&\mapsto y_1\frac{tX-t^{-1}X^{-1}}{X-X^{-1}}\varpi +y_1^{-1}\frac{t^{-1}X-tX^{-1}}{X-X^{-1}}\varpi^{-1},\\
z&\mapsto  q^{\frac12}y_1X\frac{tX-t^{-1}X^{-1}}{X-X^{-1}}\varpi +q^{\frac12}(y_1X)^{-1}\frac{t^{-1}X-tX^{-1}}{X-X^{-1}}\varpi^{-1} ~.
\end{aligned}
\end{equation}
Concretely, one is free to deform Cherednik's polynomial representation \eqref{poly-rep-sym} to this larger representation parametrized by $y_1$, as long as we allow denominators to be elements of the multiplicative system $\wt M$.
Only at $y_1=1$, it decomposes into two irreducible representations where one is Cherednik's polynomial representation, and the other irreducible representation is
$$
\widetilde M^{-1}  \CR[X^{\pm}]^{\Z_2}~.
$$
When $t=1$, the story reduces to the polynomial representations of the symmetrized quantum torus discussed in Appendix \ref{app:branes-SQT}. Thus, the support of the corresponding brane $\frakB_{\bfP}^{y_1}$ is expected to be
$$
\mathrm{supp}~\frakB_{\bfP}^{y_1}=\{y=y_1^{-1} \tilde{t}+y_1\tilde{t}^{-1}\}~.
$$

In fact, the eigenfunction  of $y$ under $\pol_{y_1}$ that generalizes the Macdonald polynomials is constructed in \cite{Bullimore:2014awa,Koroteev:2018azn,Koroteev:2018aa} \footnote{$\mathcal{Z}$ is the so-called uncapped vertex function in the quantum K-theory of $T^*\CP^1$.}
\begin{equation}
\mathcal{Z}(X,y_1,q,t)= {}_2\phi_1\left(y_1^2,t^2;q^2t^{-2}y_1^2;q^2;q^2t^{-2} X^{-2}\right)\,,
\label{eq:HolBlock}~
\end{equation}
where the eigenvalue is
\begin{equation}
\label{eq:MacdonaldEigen}
\pol_{y_1}(y)\cdot\mathcal{Z}= (y_1^{-1} t+y_1t^{-1})\mathcal{Z}\,.
\end{equation}
Thus, for a generic value of $y_1$, the eigenfunction is an infinite hypergeometric series \eqref{eq:HolBlock}. However, as easily seen from \eqref{sym-Macdonald-A1},
the series truncates to the symmetric Macdonald polynomial
\begin{equation}
\mathcal{Z}(X,y_1=q^{-j},q,t) = X^{-j} P_{j}(X;q,t)\,.
\end{equation}
at the specialization $y_1=q^{-j}$ $(j\in\mathbb{Z}_{\ge0})$.

A geometric interpretation of the multiplicative system $\widetilde{M}$ can be given by thinking about the $t=1$ limit, where we are interested in the quotient map $\C^\times \times \C^\times \rightarrow (\C^\times \times \C^\times) / \Z_2$. After deforming the target of this covering map, no natural ramified twofold cover by $\C^\times \times \C^\times$ exists. However, such a cover can be constructed once we extract the $\bZ_2$-invariant points $X=\pm 1$. In fact, $\OO(\C^\times\backslash \{X=\pm 1\} )$ admits the generator $\frac{1}{X-X^{-1}}$. A related story exists in the rational limit, where the relevant geometry is the deformation of the $A_1$ singularity $(\C \times \C)/\Z_2$ to the total space of $T^*\C P^1$; we discuss this in detail in Appendix~\ref{app:Rational-A1}.

\subsection{\texorpdfstring{Branes with compact supports and finite-dimensional representations:\\ object matching}{Branes with compact supports and finite-dimensional representations: object matching}}
\label{sec:finite-rep}

Cherednik's polynomial representation is of particular significance due to the theorems of Cherednik~\cite[\S2.8--9]{Cherednik-book}, which classify finite-dimensional representations of~\SH{} obtained as quotients of the polynomial representation paired with the action of outer automorphisms. Similar to the theory of Verma modules, the polynomial representation is generically irreducible. A raising operator \eqref{raising} never be null since the Macdonald polynomial $P_{2j}$ always has a factor $(1-q^{2j}t^{2})$ in the denominator. However, it can occur that a lowering operator $\sfL_j$ annihilates one of the Macdonald polynomials $P_{j}$ when  certain conditions on the central character are satisfied. If this occurs, $P_j$  generates a subrepresentation, and a finite-dimensional representation of the spherical DAHA appears as the quotient $\PR/(P_j)$.
We can therefore study finite-dimensional representations by asking that the condition $\pol(\sfL_j)\cdot P_{j}=0$ be satisfied for some $j$, \textit{i.e.}\ that the factor
\be
\label{factor-vanish}
\frac{(1-q^{2j})(1-q^{(j-1)}t^{2})(1+q^{(j-1)}t^{2})}{q^{2j}t^{2}(q^{2(j-1)}t^{2}-1) }
\ee
on the right hand side of \eqref{lowering} vanishes.

This amounts to the following three cases:
\begin{subequations}
\begin{align}
q^{2n}&=1~, \label{case1} \\
t^2&=-q^{-k}~, \label{case2} \\
t^2&=q^{-(2\ell-1)}~. \label{case3}
\end{align}
\end{subequations}
Here, the exponent in the right hand side of  \eqref{case3} must be an odd integer in order for  the denominators of Macdonald polynomials as well as \eqref{factor-vanish} to be non-zero; even exponents are excluded by the definition of the  coefficient ring~\CR{} in~\eqref{eq:coeffs}. We write this odd integer as $2\ell-1$.
Each of these separate shortening conditions will naturally appear as an existence condition of an $A$-brane with compact support in what follows; we will examine each of the resulting finite-dimensional representations and the corresponding compact Lagrangian branes in turn.

\subsubsection{Generic fibers of the Hitchin fibration}\label{sec:generic-fiber}

First we consider analogous $A$-branes in this setting; the ones supported on generic fibers in the Hitchin fibration. As explained in \S\ref{sec:target}, the Hitchin fibration \eqref{Hitchin-fibration} is completely integrable, and a generic Hitchin fiber  $\F$  is holomorphic in complex structure $I$ while it is a complex Lagrangian submanifold from the viewpoint of the holomorphic two-form $\Omega_I$  for a generic ramification data \eqref{tame}. Namely, it is a Lagrangian submanifold of type $(B,A,A)$ for any values of $(\a_p,\b_p,\g_p)$-triple. Therefore, a generic fiber $\F$ can be Lagrangian in a symplectic manifold $(\X,\omega_\X)$ only when the canonical coisotropic brane $\Bcc$ obeys the condition $\theta=0$ in \eqref{generic-Bcc} so that
\be\label{Bcc-BAA}\omega_\X=-\frac{\omega_K}{\hbar}~,\quad \textrm{and} \quad F+B=\frac{\omega_I}{\hbar}~.\ee
With $\theta\neq0$, there is no $A$-brane supported on $\F$ in the symplectic manifold $(\X,\omega_\X)$. Accordingly, $\hbar=|\hbar|$ is real (\textit{i.e.} $|q|=1$), and the canonical coisotropic brane $\Bcc$ is an $A$-brane of type $(B,A,A)$.

An analogous brane appears in the brane quantization of $\C^\times\times\C^\times$ for the quantum torus algebra. As in Appendix \ref{app:cyclic-rep-QT}, a brane is supported on a fiber $T^2$ of the elliptic fibration $T^*T^2 \cong\C^\times\times\C^\times$, which gives rise to a  finite-dimensional representation, called the cyclic representation. Therefore, we can study a brane supported on a generic fiber $\F$ of the Hitchin fibration, comparing with the case of the quantum torus algebra.

The branes are indexed by a position of the Hitchin base $\cB_H$ (see also Appendix \ref{app:cyclic-rep-QT}).  Also, the flatness condition \eqref{deformed-flat} of the line bundle $\cL'$ an $A$-brane supported $\brF$ is
$$
F'_{\bfF}+B\big|_{\bfF}=0~.$$
Since $\F$ is topologically a two-torus, the flat $\Spin^c$ structure $\cL'\otimes K_\L^{-1/2}$ of $\brF$ can have non-trivial $\U(1)^2$ holonomy with a choice of spin structure \cite{Gukov:2008ve}.
The branes $\brF^\lambda$ are parametrized by $\lambda=(x_m,y_m)\in\C^\times\times\C^\times$ where the absolute values $(|x_m|,|y_m|)$ describe its position and the angular phases illustrate the $\U(1)^2$ holonomy with a choice of spin structures. Namely, the angular phase $\U(1)$ encodes the holonomy $\U(1)$ and a choice of spin structures $\Z_2$ along a one-cycle of a Riemann surface via
$$
1\to \Z_2\to \U(1)\to \U(1)\to 1~.
$$
We assign the plus sign $+$ for $1\in \Z_2$ to the Ramond spin structure, and the minus sign $-$ for $-1\in \Z_2$ to the Neveu-Schwarz spin structure. The choice of spin structures appears in the representation of the symmetrized quantum torus discussed in Appendix \ref{app:branes-SQT}.

Consequently, the computation of the dimension \eqref{dimension} of the space $\Hom(\Bcc,\brF^\lambda)$ is reduced to the period integral \eqref{generic-fiber-evaluation}
\bea
\dim \Hom(\Bcc,\brF^\lambda) = \int_{\F} \frac{F+B}{2\pi} =\int_{\F}\frac{\omega_I}{2\pi \hbar}=\frac 1{\hbar}
\label{dim-generic-fiber}
\eea
for arbitrary $\lambda$. Hence, this leads to the Bohr-Sommerfeld quantization condition  $\hbar=1/m$, or equivalently that $q=e^{2\pi i/m}$ is a primitive $m$-th root of unity for $m\in \Z_{>0}$. In fact, since $[F/2\pi]$ is an integral cohomology class in $H^2(\X,\Z)$, the fiber class relation \eqref{fiber-class-rel} requires  $\int_\bfF F/2\pi$ to be an even integer. Thus, if $m$ is an odd positive integer, then we need non-trivial $B$-flux with
\be\label{odd-integral}
 \int_{\F} \frac{B}{2\pi} =-  \int_{\F} \frac{F'_\bfF}{2\pi}=1~,
\ee
up to an even integer shift. For instance, this can be achieved if the $B$-field flux over $\V$ is $1/2$ and those over the exceptional divisors $\D_i$ ($i=1,\ldots,4$) are zero.

In order for the $(\Bcc,\brF^\lambda)$-strings exist, $q$ has to be a primitive $m$-th root of unity whereas $t$ can be generic. Under this condition, the action of $\SH$ under the generalized polynomial representation $\pol_{y_1}$ in \eqref{poly-rep-y1} commute with $X^m-x_m$ for $x_m\in \C^\times$ because the shift operator $\varpi$ acts trivially on it. Consequently, the ideal $(X^m-x_m)$ is invariant under $\pol_{y_1}$ so that the quotient space
$$
\scF^{\lambda}_{m}=\scP^{y_1}/(X^m-x_m)~,
$$
is also a representation of $\SH$. Since the Taylor expansion of a denominator in the multiplicative system $\wt M$ always truncates under the condition $X^m=x_m$, this is indeed an $m$-dimensional representation parametrized by $\lambda=(x_m,y_m)$ where $y_1$ is any $m$-th root of $y_m$. Hence, we can identify $\Hom(\Bcc,\brF^\lambda)$ with $\scF^{\lambda}_{m}$ when $q$ is a primitive $m$-th root of unity where the parameter $\lambda \in \C^\times\times\C^\times$ exactly matches.

For generic values of $\lambda=(x_m,y_m)$, the support of a brane $\brF^\lambda$ is mapped to another Hitchin fiber up to Hamiltonian isotopy under the $\PSL(2,\Z)$ action, and the holonomy of the Chan-Paton bundle, which is a point in the dual torus $\Jac(\F)$, is also transformed appropriately. Namely, $\PSL(2,\Z)$ acts on $\lambda$. On the other hand, a generic fiber is invariant as a set under the group $\Xi$ of the sign changes as we have seen in \S\ref{sec:target}. Correspondingly, the representation $\scF^{\lambda}_{m}$ is invariant under $\Xi$ at a generic value of $\lambda$.

Setting $y_1=1$, we can symmetrize the story \cite[Thm 2.8.5 (iv)]{Cherednik-book}. Namely, since the ideal $(X^m+X^{-m}-x_m-x_m^{-1})$ is invariant under Cherednik's polynomial representation $\scP$ due to the same reason,  we have an $m$-dimensional representation
\be\label{Fxm}
\scF^{(x_m,+)}_{m}=\scP/(X^m+X^{-m}-x_m-x_m^{-1})~.
\ee
In this case, the corresponding brane $\brF^{(x_m,+)}$ supported on a Hitchin fiber intersects with the support $\bfP$ \eqref{brane-poly} of the polynomial representation. Also, the Chan-Paton bundle has the trivial holonomy and the Ramond spin structure around one generator, say the $(0,1)$-cycle, of $\pi_1(\bfF)\cong\Z\oplus\Z$. The parameter $x_m$ encodes its position in the $x$ coordinate and the holonomy around the other generator of $\pi_1(\bfF)$.

Therefore, the representations of this family are analogous to the finite-dimensional representations of both symmetrized and ordinary quantum torus in terms of $A$-branes on fibers of the elliptic fibration of the target in the 2d $A$-models as illustrated in  Appendix \ref{app:qt}.
As in the case of the symmetrized quantum torus Appendix \ref{app:SQT}, if a brane $\brF$ with trivial holonomies moves to a special position, we will see below that a special phenomenon occurs.

\subsubsection{\texorpdfstring{Irreducible components in singular fibers of type $I_2$}{Irreducible components in singular fibers of type I2}}\label{sec:Ui}

As in Figure~\ref{fig:3I_2}, the Hitchin fibration has three singular fibers of Kodaira type $I_2$ for generic ramification parameters of $(\a_p,\b_p,\g_p)$. Since they are still fibers in the Hitchin fibration, the irreducible components $\bfU_i$ ($i=1,\ldots,6$) in a singular fiber are also Lagrangian submanifolds of type $(B,A,A)$. Therefore, $\Bcc$ needs to satisfy \eqref{Bcc-BAA} in order for $\frakB_{\bfU_i}$ to be $A$-branes as in the previous subsection.

For instance, let us investigate a module that the brane $\frakB_{\bfU_1}$ gives rise to. The curvature of the line bundle $\cL'$ should obey the flatness condition \eqref{deformed-flat}
\be
F'_{\bfU_1}+\left.B\right|_{\bfU_1}=0~.
\ee
Since $\bfU_1$ is topologically $\bC \bfP^1$ and a position is fixed, there is no deformation parameter associated to the brane  $\frakB_{\bfU_1}$.
Subsequently, one can evaluate the dimension formula  \eqref{dimension}
\bea
\dim \Hom(\Bcc,\frakB_{\bfU_1}) = \int_{\bfU_1} \frac{F+B}{2\pi} =\int_{\bfU_1}\frac{\omega_I}{2\pi \hbar}=\frac 1{2\hbar}
\label{dim-U}
\eea
Consequently, the brane $\frakB_{\bfU_1}$ can exist only at $1/(2\hbar)=n\in \Z_{>0}$, or equivalently when $q$ is a primitive $2n$-th root of unity.

This is exactly one \eqref{case1} of the shortening conditions, and under this condition a lowering operator \eqref{lowering} annihilates the Macdonald polynomial
\be\label{PN}
\pol(\sfL_n)\cdot P_n(X;q,t)=0\qquad \textrm{where} \qquad P_n(X;q,t)=X^n+X^{-n}~~.
\ee
Therefore, the quotient space \be\label{Un} \scU_n^{(1)}:=\PR/(P_n)\ee by an ideal $(P_n)$ is an $n$-dimensional irreducible representation of spherical DAHA \cite[Thm 2.8.5 (ii)]{Cherednik-book} so that one can identify
$$
\scU_n^{(1)}=\Hom(\Bcc,\frakB_{\bfU_1})~.
$$

As seen in \S\ref{sec:target}, the irreducible component $\bfU_1$ is invariant under the sign change $\xi_1$ whereas it is mapped to $\bfU_2$ under $\xi_2$ \eqref{signchange-Ui}.
In fact, it follows from the form \eqref{PN} of $P_n(X)$ that the finite-dimensional module $\scU_n^{(1)}$ is invariant under the sign flip $\xi_1$.
On the other hand, the sign change $\xi_2$ leads to another non-isomorphic finite-dimensional module. Thus, the brane $\xi_2(\frakB_{\bfU_1})$ corresponds to a brane supported on the other irreducible component $\bfU_2$ in the same singular fiber from which the module comes from
$$\scU_n^{(2)}:=\xi_2(\scU_n^{(1)})=\Hom(\Bcc,\frakB_{\bfU_2})~.$$
In a similar fashion, a brane $\frakB_{\bfU_i}$ supported on another irreducible component in a singular fiber gives rise to an image of $\scU_n^{(1)}$ under $\PSL(2,\Z)$ and the sign changes $\xi_{1,2}$. The transformation rule can be read off from \eqref{PSL-Ui} so that the branes $\frakB_{\bfU_{1,2}}$ are invariant under $\tau_-$ whereas they are mapped as
\bea\label{SL2Z-singular}
\sigma(\frakB_{\bfU_{1}})=\frakB_{\bfU_{3}}~,&\qquad \sigma(\frakB_{\bfU_{2}})=\frakB_{\bfU_{4}}~,\cr
\tau_{+}(\frakB_{\bfU_{1}})=\frakB_{\bfU_{5}}~,&\qquad \tau_{+}(\frakB_{\bfU_{2}})=\frakB_{\bfU_{6}}~.
\eea
The corresponding modules $\scU_n^{(i)}$ are obtained from $\scU_n^{(1)}$ in the same way.

\subsubsection{\texorpdfstring{Moduli space of $G$-bundles}{Moduli space of G-bundles}}\label{sec:BunG}

  Next, we consider a brane $\frakB_{\bfV}$ supported on the moduli space $\bfV$ of $G$-bundles. For the sake of brevity, let us first see the case of $\b_p=0$. If $\hbar$ is real, only $\a_p$ can be turned on while $\g_p$ must vanish in order for $\bfV$ to be Lagrangian with respect to $\omega_K$. As $\hbar=|\hbar|e^{i\theta}$ is rotated $\theta\neq0$ in the complex plane, the symplectic form we are interested in is also rotated from $\omega_K$ to $\omega_\frakX$ according to \eqref{generic-Bcc}. However, this rotation can be actually compensated by switching on $\g_p$ so that $\bfV$ can stay Lagrangian with respect to $\omega_\frakX$.
 According to \eqref{integral-BunG} and \eqref{generic-Bcc}, the set $\bfV$ is Lagrangian with respect to $\omega_\frakX$ when the following condition holds:
\be
\label{Lag-BunG}
\Im \, \frac{\big(\frac{1}{2}-\talpha_p\big) +i\tgamma_p}{\hbar} \; = \; 0
\ee
As a simple check, one can easily see from \eqref{integral-BunG} and \eqref{c-ramification} that the integral of the symplectic form is zero
\bea\label{evaluation-BunG}
\int_{\bfV}\frac{\Im\,\Omega}{2\pi}&=\int_{\bfV}\frac{\omega_\frakX}{2\pi}=0~,\cr
\eea
In addition, if $\beta_p=0$, the submanifold $\V$ is also Lagrangian with respect to $\omega_J$. Namely, it is a complex Lagrangian submanifold with respect to a holomorphic two-form $\omega_\frakX+i\omega_J$.  When $\tbeta_p$ is varied, $\V$ stays as a Lagrangian submanifold with respect to $\omega_\frakX$ while they are no longer Lagrangian with respect to $\omega_J$. In fact, the variation of $\b_p$ does not change the holomorphic symplectic form $\Omega_J=\omega_K + i\omega_I$, and therefore keeps $\omega_\frakX$ fixed. In conclusion, $\V$ can be Lagrangian with respect to $\omega_\frakX$ only when \eqref{Lag-BunG} holds. Since our concern is the $A$-model in the symplectic manifold $(\frakX,\omega_\frakX)$, the value of $\b_p$ can be arbitrary.
For generic $(\beta_p,\gamma_p)$, $\bfV$ is no longer a Lagrangian of type $(B,A,A)$, and it is therefore not contained in a fiber of the Hitchin fibration.
Nonetheless, unlike a Hitchin fiber, we can consider the $A$-model in a generic symplectic form $\omega_\frakX$ in \eqref{generic-Bcc} where $\hbar$ can take any complex value.

Under the condition \eqref{Lag-BunG} with a generic value of $\hbar$, $\V$ is a unique compact Lagrangian submanifold, which is topologically $\bC\bfP^1$. Hence, there is no deformation parameter for $\brV$.
Consequently, we obtain the dimension of the space of $(\Bcc,\brV)$-strings from \eqref{c-ramification}
\be
\dim \Hom(\Bcc,\brV)  = \int_{\bfV} \frac{F+B}{2\pi}=\frac{1}{2\hbar}-\frac{\tgamma_p+i\talpha_p}{i\hbar}=\frac{1}{2\hbar}+2c-1~.
\label{dim-BunG}
\ee
The Bohr-Sommerfeld quantization condition imposes its dimension as a positive integer $1/2\hbar+2c-1=k+1 \in \Z_{>0}$, or equivalently $t^2=-q^{k+2}$.

One can observe that this quantization condition is equivalent to the image of the shortening condition \eqref{case2} under the involution $\iota$.
In fact, under the shortening condition $t^2=-q^{k+2}$, the lowering operator in the $\iota$-image of the polynomial representation becomes an annihilation operator
$$\pol(\sfL_{k+1})\cdot P_{k+1}(X;q,t)\Big|_{t\to\frac{q}{t}}=0~.$$
Consequently, the quotient space by an ideal $(P_{k+1})$
\be\label{quotient-V}
\iota(\scV_{k+1}):=\iota(\PR)/(P_{k+1}(X;q,\tfrac{q}{t}))~
\ee
is a $(k+1)$-dimensional irreducible representation of $\SH$ \cite{oblomkov2009finite}. This representation is called the \emph{additional series} in \cite[\S2.8.2]{Cherednik-book}, and we identify
$$
\iota(\scV_{k+1})=\Hom(\Bcc,\brV)~.
$$
In fact, the support \eqref{brane-poly-iota} of the brane $\iota(\brP)$ intersects with $\V$ at $t^2=-q^{k+2}$ so that $\Hom(\iota(\brP),\brV)\cong\bC$ becomes non-trivial. Hence, $\iota(\scV_{k+1})$ can be obtained as the quotient of $\iota(\PR)$ as in  \eqref{quotient-V}.

As we have seen at the end of \S\ref{sec:target}, the submanifold $\V$ is geometrically invariant under the sign changes $\xi_{1,2}$ so that we expect that the corresponding module $\scV_{k+1}$ is also endowed with the same property.
When $t^2=-q^{k+2}$, the Macdonald polynomials obey $$P_{k+1}(-X;q,\tfrac{q}{t})=(-1)^kP_{k+1}(X;q,\tfrac{q}{t})~,$$
which implies that $\iota(\scV_{k+1})$ is indeed invariant under $\xi_1$. In addition, it is easy to check that the full set of $y$-eigenvalues (the $\iota$-image of \eqref{Macdonald-eigen}) of $\iota(\scV_{k+1})$ is also invariant under $\xi_2$.

What makes the space of $(\Bcc,\brV)$-strings even more interesting is that it also carries a $\mathrm{PSL}(2,\Z)$ action. Indeed, as also explained in \S\ref{sec:target}, the submanifold $\bfV$ is invariant under $\mathrm{PSL}(2,\Z)$ symmetry and, as a result, the module $\iota(\scV_{k+1})$ is a $\mathrm{PSL}(2,\Z)$ representation.

Of course, it is then natural to ask \emph{which} representation it is, and in particular, what the corresponding $S$ and $T$ matrices are. To this end, it is more convenient to consider the space of $(\Bcc,\brV)$-strings in the target $\X_{\tilde{t}^{-1}}$ under \eqref{iota-classical} or \eqref{ramification-Weyl}. Then, the  corresponding representation  is given by
\be
\label{rCS-Hilbert-space} \scV_{k+1}:=\PR/(P_{k+1})
\ee
under the shortening condition  \eqref{case2}. Since the basis fo $\scV_{k+1}$ is spanned by the Macdonald polynomials $P_j(X)$ $(j=0,\ldots,k)$, the modular $T$-transformation $\tau_-$ acts diagonally in this basis due to \eqref{Mac-T}. Under the modular $S$ transformation, this basis is transformed to $P_j(Y)$ and the submanifold $\V$ intersects with both the support \eqref{brane-poly} of the branes $\brP$ and  that \eqref{brane-poly-SL2Z} of $\sigma(\brP)$.
Hence, the modular $S$-matrix can be written as
\be \label{S-poly}
S_{jj'}=\pol(P_j(Y^{-1}))\cdot P_{j'}(X)\big|_{X=t^{-1}}=P_j (tq^{j'};q,t) \, P_{j'} (t^{-1};q,t)~.
\ee
This is first introduced by Cherednik \cite{cherednik1995macdonald} as a symmetric bilinear pairing of Macdonald polynomials, which we also denote by $\textbf[P_j , P_{j'} \textbf]$ as in \eqref{Macdonald-pairing2}. Moreover, it becomes of rank $(k+1)$ when $t^2=-q^{-k}$, and it acts on $\scV_{k+1}$. Therefore, we find explicit forms of the $S$ and $T$ matrices as follows, and we will also find a 3d interpretation of this $\PSL(2,\Z)$ representation in \S\ref{sec:SU2}.

\begin{conjecture}\label{conj:rcs}
    The space $\repV_{k+1}$ is a $(k+1)$-dimensional $\mathrm{PSL}(2,\Z)$ representation, with modular $S$ and $T$ matrices given by
\begin{equation}
\begin{aligned}
T_{jj'}\big|_{\repV_{k+1}} &= e^{\frac{\pi i k}{12}}q^{-\frac{k(k-1)}{12}}~ i^{-j}q^{\frac{j(k-j)}2} \delta_{jj'}
\qquad 0 \le j,j' \le k \label{STBunG} \\
S_{jj'}\big|_{\repV_{k+1}} &=a_k^{-1}g_j(q,t=iq^{-k/2})^{-1} P_j (i q^{j'-k/2};q,t=iq^{-k/2}) \, P_{j'} (i q^{k/2};q,t=iq^{-k/2})~.
\end{aligned}
\end{equation}
	These matrices provide the $\mathrm{PSL}(2,\Z)$ representation for ``refined Chern-Simons theory''.
\end{conjecture}

Here we normalize the modular $S$-transformation \eqref{S-poly} by the Macdonald norm of type $A_1$ (See \eqref{Macdonald-norm} for the definition)
\be\label{Mac-norm-A1}
g_j(q,t):=\frac{(q^{2j};q^{-2})_j(t^4;q^2)_j}{(q^{2j-2};q^{-2})_j(t^2q^2;q^2)_j}
\ee
and
$$
a_k=\begin{cases} \sqrt{2}\prod_{i=0}^{\frac{k-3}2}(q^{\frac14+\frac i2}+q^{-\frac14-\frac i2}) & \quad k:~ \textrm{odd}\\ 2 \prod_{i=0}^{\frac{k-4}2}(q^{\frac12+\frac i2}+q^{-\frac12-\frac i2})& \quad  k :~\textrm{even} \end{cases}
$$
so that $S^2=1$. We also normalize the $T$-transformation \eqref{Mac-T} by $e^{\pi i k/12}q^{-k(k-1)/12}$ so that $(ST)^3= 1$.
For example, the first non-trivial case occurs at $k=1$
$$
T\big|_{\repV_2}=e^{\pi i/12}
\begin{pmatrix}
1 & 0 \\
0 &  -i
\end{pmatrix}~,\qquad S\big|_{\repV_2}=\frac1{\sqrt{2}}
\begin{pmatrix}
    1 & -i (q^{\frac12}-q^{-\frac12}) \\
    i(q^{\frac12}-q^{-\frac12})^{-1} & -1
\end{pmatrix}~.
$$
Next, we turn to less familiar and more interesting modular representation that arises from another Lagrangian $A$-brane in a similar fashion.

\subsubsection{Exceptional divisors}
\label{sec:divisorTQFT}
Now let us consider an interesting $A$-brane $\brDi$ supported on an exceptional divisor $\bfD_i$, $i=1,\ldots,4$. As we reviewed in the earlier part of this section, the ramification parameters $(\talpha_p, \tbeta_p, \tgamma_p)$ play the role of resolution/deformation parameters for $\bfD_i$. In particular, when $\b_p=0$ and $\hbar$ is real, only $\a_p$ can be turned on while $\g_p$ must vanish in order for $\bfD_i$ to be Lagrangian with respect to $\omega_K$. As $\hbar=|\hbar|e^{i\theta}$ is rotated $\theta\neq0$ in the complex plane, the exceptional divisors $\bfD_i$ stay Lagrangian with respect to $\omega_\X$ if the deformation parameter $\tgamma_p+i\talpha_p\in \bC$ in complex structure $J$ is proportional to $i\hbar$, namely,
\be\label{Lag-D}\Im\frac{\tgamma_p+i\talpha_p}{2i\hbar}=0~.\ee
Here the value of $\tbeta_p$ can be arbitrary as in the previous case.
It is easy to verify from \eqref{integral-D} and \eqref{c-ramification} that
$$
\int_{\bfD_i}\frac{\Im\,\Omega}{2\pi}=\int_{\bfD_i}\frac{\omega_\frakX}{2\pi}=0~.
$$

The story goes as before.
The flatness condition  \eqref{deformed-flat} of the Chan-Paton bundle for the brane $\brDi$ is
$$
F'_{\bfD_i}+B\big|_{\bfD_i}=0~,
$$
Since it is topologically $\CP^1$, there is no holonomy and no deformation parameter for $\brDi$. Subsequently, the dimension can be computed as
\bea
\dim \Hom(\Bcc,\brDi) =\int_{\bfD_i}\frac{F+B}{2 \pi}=-c+\frac12\label{dim-D}~.
\eea
The Bohr-Sommerfeld quantization condition imposes its dimension as a positive integer $-c+1/2=\ell \in \Z_{>0}$, or equivalently $t^2=q^{-(2\ell-1)}$, which is \eqref{case3}.

When $t=q^{-(2\ell-1)/2}$, the lowering operator annihilates the Macdonald polynomial
\be\label{null-2}
\pol(\sfL_{2\ell})\cdot P_{2\ell}(X;q,t)=0~.
\ee
Therefore, the quotient space
\begin{equation}\label{quotient-D}
  \scD_{2\ell}:=\PR/(P_{2\ell})
\end{equation}
by an ideal $(P_{2\ell})$ is a $2\ell$-dimensional representation of $\SH$. In fact, it is not irreducible, and decomposes into two irreducible representations
\be\label{finite-case2}
\scD_{2\ell}=\scD_\ell^{(1)}\oplus \scD_\ell^{(2)}~.
\ee
Because $P_j$ and $P_{2\ell-j-1}$ have the same eigenvalue of the Macdonald difference operator \eqref{Macdonald-eigen}
when $t=q^{-(2\ell-1)/2}$, their combinations indeed form bases of $\scD_\ell^{(1,2)}$
\be\label{D-basis}
\scD_\ell^{(1)}=\bigoplus_{j=0}^{\ell-1} \C_{q,t}\Bigl[\frac{P_j(X)}{P_j(t^{-1})}+ \frac{P_{2\ell-j-1}(X)}{P_{2\ell-j-1}(t^{-1})} \Bigr]~,\quad \scD_\ell^{(2)}=\bigoplus_{j=0}^{\ell-1} \C_{q,t}\Bigl[\frac{P_j(X)}{P_j(t^{-1})}- \frac{P_{2\ell-j-1}(X)}{P_{2\ell-j-1}(t^{-1})} \Bigr]~.
\ee
Consequently, they are related by the sign change $\scD_\ell^{(2)}=\xi_1(\scD_\ell^{(1)})$. In fact, the support \eqref{brane-poly} of the brane $\brP$ intersects with $\D_{1,2}$ at $t=q^{-(2\ell-1)/2}$ so that  $\scD_\ell^{(1)}\oplus \scD_\ell^{(2)}$ can be obtained as the quotient of $\PR$ as in \eqref{quotient-D}.

Even when  $t=-q^{-(2\ell-1)/2}$, the shortening condition \eqref{null-2} holds, but the eigenvalues \eqref{Macdonald-eigen} of the $y$-operator have the opposite sign as in \eqref{zeta2-support}. Therefore, the corresponding irreducible representations can be obtained by the sign change $\xi_2$ in \eqref{sign-changes} from $\scD_\ell^{(1,2)}$.

As a result, for $t^2= q^{-(2\ell-1)}$, there are four irreducible finite-dimensional representations \cite[Thm 2.8.1]{Cherednik-book} that are obtained from $\scD_\ell^{(1)}$ by the sign changes $\xi_{1,2}$.
This is analogous to the relationship among the exceptional divisors under the sign changes \eqref{signchange-Di}. Therefore, we identify these modules to the spaces of open $(\Bcc,\brDi)$-strings as
\bea \label{4-irrep}
\scD_\ell^{(1)}=\Hom(\Bcc,\brD1)~,& \qquad  \scD_\ell^{(2)}=\xi_1(\scD_\ell^{(1)})=\Hom(\Bcc,\brD2)~,& \cr
\scD_\ell^{(3)}:=\xi_2(\scD_\ell^{(1)} )=\Hom(\Bcc,\brD3)~,& \qquad  \scD_\ell^{(4)}:=\xi_2(\scD_\ell^{(2)})=\xi_3(\scD_\ell^{(1)} )=\Hom(\Bcc,\brD4)~.&
\eea
The modules $\scD^{(1,2)}_\ell$ can be obtained as the quotient of the polynomial representation because the support \eqref{brane-poly} of $\frakB_\bfP$ intersects with $\D_1$ and $\D_2$. On the other hand, its $\xi_2$-image \eqref{brane-poly-iota} intersects with  $\D_3$ and $\D_4$. (See also Figure \ref{fig:branes-reps}.)

Under the $\PSL(2,\Z)$ action, the four irreducible representations are transformed as in \eqref{PSL-Di}. Namely, the modular $T$-transformation $\tau_-$ exchanges $\scD_\ell^{(3)}$ and $\scD_\ell^{(4)}$ whereas $\scD_\ell^{(1)}$ and $\scD_\ell^{(2)}$ are invariant. Also, the modular $S$-transformation  $\sigma$ exchanges $\scD_\ell^{(2)}$ and $\scD_\ell^{(3)}$ whereas the modules $\scD_\ell^{(1)}$ and $\scD_\ell^{(4)}$ are invariant.

\bea\label{PSL-Di-2}
\tau_+ &: \scD_\ell^{(2)} \leftrightarrow \scD_\ell^{(4)} \quad \textrm{and} \quad  \scD_\ell^{(1)},~ \scD_\ell^{(3)}\quad \textrm{are invariant}~, \\
\tau_- &: \scD_\ell^{(3)} \leftrightarrow \scD_\ell^{(4)} \quad \textrm{and} \quad  \scD_\ell^{(1)},~  \scD_\ell^{(2)}\quad \textrm{are invariant}~, \\
\sigma &: \scD_\ell^{(2)} \leftrightarrow \scD_\ell^{(3)} \quad \textrm{and} \quad  \scD_\ell^{(1)},~ \scD_\ell^{(4)}\quad \textrm{are invariant}~.
\eea
Thus, only the module $\scD_\ell^{(1)}=\Hom(\Bcc,\brD1)$ among the four modules becomes a $\PSL(2,\Z)$ representation.

Let us find the modular $S$ and $T$ matrices for this $\PSL(2,\Z)$ representation.
As we have seen, the polynomial representation $\PR$ captures both $\scD_\ell^{(1)}$ and $\scD_\ell^{(2)}$ so that the $S$-matrix \eqref{S-poly} truncates a matrix of size $2\ell\times 2\ell$ under the shortening condition \eqref{case3}. However, the matrix has rank $\ell$ and it acts non-trivially only on $\scD_\ell^{(1)}$  under the change \eqref{D-basis} of basis
\be\label{base-change-S}
\wt S_{jj'}:= G^{-1} S_{jj'} G(q,t=q^{-(2\ell-1)/2})\big|_{\scD^{(1)}_{\ell}}~,\qquad 0 \le j,j' \le \ell-1
\ee
where $G$ is a matrix of size $2\ell\times 2\ell$ that changes the basis according to \eqref{D-basis}. This gives the geometric interpretation of the basis change in \cite[\S4.1]{Kozcaz:2018usv}.
As a result, we find the following explicit forms of the $S$ and $T$ matrices, and a 3d interpretation of our $A$-model setup in \S\ref{sec:SU2} will identify an intrinsic physical meaning of the $\PSL(2,\Z)$ representation:

\begin{conjecture}\label{conj:MTC-D1}
    The space $\scD^{(1)}_{\ell}$ is an $\ell$-dimensional $\mathrm{PSL}(2,\Z)$ representation, with modular $S$ and $T$ matrices given by
\begin{equation}
\begin{aligned}
T_{jj'} \big|_{\scD^{(1)}_{\ell}}&= e^{\frac{(\ell-1)\pi i}6}q^{-\frac{(2 \ell - 1) (\ell - 1)}{6}} ~q^{\frac{j(k-j)}2} \delta_{jj'}
\qquad 0 \le j,j' \le \ell-1 \label{STD1} \\
S_{jj'} \big|_{\scD^{(1)}_{\ell}}&=b_{\ell}^{-1} g_{j}(q, t=q^{-(2\ell-1)/2})^{-1}\wt S_{jj'}~.
\end{aligned}
\end{equation}
The $\mathrm{PSL}(2,\Z)$ representation comes from a modular tensor category associated to the Argyres-Douglas theory of type $(A_1,A_{2(\ell-1)})$. These matrices coincide with those of the $(2,2\ell+1)$ Virasoro minimal model at $q=e^{-2\pi i /(2\ell+1)}$.
\end{conjecture}
Here we normalize \eqref{base-change-S} by the Macdonald norm \eqref{Mac-norm-A1} and
$$
b_{\ell}=2\prod_{i=0}^{\ell-2}(q^{1/2+ i}-q^{-1/2-i})
$$
so that $S^2=1$. We also normalize \eqref{Mac-T} by $e^{(\ell-1)\pi i/6} q^{-(2 \ell - 1) (\ell - 1)/6}$ so that $(ST)^3=1$.

For instance, when $\ell=2$, these matrices become
$$T\big|_{\scD^{(1)}_{\ell=2}}=e^{\frac{\pi i}6}\begin{pmatrix}
q^{-\frac12} & 0 \\
0 & q^{\frac12}
\end{pmatrix}~,\qquad S\big|_{\scD^{(1)}_{\ell=2}}=\frac{i}{q^{\frac12}-q^{-\frac12}}
\begin{pmatrix}
    1 & -(q-1+q^{-1}) \\
    1 & -1
\end{pmatrix}~.
$$
When $q=e^{-2 \pi i /5}$, they coincide with the modular matrices of the $(2,5)$ Virasoro minimal model although an appropriate change of basis is required to bring the $S$-matrix into the standard form.

\begin{table}[ht]
  \begin{center}
    {\renewcommand{\arraystretch}{1.4}\begin{tabular}{c|c|c}
  \hline
     finite-dim irrep  & shortening condition & $\displaystyle A$-brane condition \\
  \hline
     $\displaystyle \mathscr{F}_{m}^{(x_m,y_m)}$ & $\displaystyle q^{m} =1$ & $ m=\frac{1}{\hbar }$ \\
    \hline
     $\displaystyle \mathscr{U}_{n}$ & $\displaystyle q^{2n} =1$ & $ n=\frac{1}{2\hbar }$ \\
    \hline
     $\displaystyle \mathscr{V}_{k+1}$ & $\displaystyle t^{2} =-q^{-k}$ & $k=\frac{1}{2\hbar } +\frac{\tgamma _{p} +i\talpha _{p}}{i\hbar }$ \\
    \hline
     $\displaystyle \mathscr{D}_{\ell }$ & $\displaystyle t =q^{-\ell +1/2}$ & $ \ell =\frac{\tgamma _{p} +i\talpha _{p}}{2i\hbar }$ \\
    \hline
  \end{tabular}}
    \end{center}
  \caption{A summary of finite-dimensional irreducible representations of $\protect\SH$ with corresponding shortening and $A$-brane conditions.}
  \label{fig:summary}
\end{table}

\subsection{Bound states of branes and short exact sequences:   morphism matching}\label{sec:bound-state}
We have hitherto studied generic conditions when an individual $A$-brane supported on a compact irreducible Lagrangian can exist. Next, we will figure out the situation in which two distinct $A$-branes are present at a singular fiber of the Hitchin fibration. When two distinct $A$-branes intersect at a singular fiber, they will form a bound state. In this section, we will study a bound state of compact $A$-branes and identify the corresponding $\SH$-module. This provide evidence of the equivalent morphism structure under the functor \eqref{eq:functor}, restricting to the subcategory of compact Lagrangian $A$-branes with that of finite-dimensional $\SH$-modules.

\subsubsection{\texorpdfstring{At singular fiber of type $I_2$}{At singular fiber of type I2}}
As seen in \S\ref{sec:generic-fiber} and \S\ref{sec:Ui}, the compact branes $\brF$ and $\brUi$ can exist when $q$ is a root of unity and $t$ is generic. As Figure \ref{fig:I_2} shows, the irreducible components $\bfU_1$ and $\bfU_2$ at  the singular fiber $\pi^{-1}(b_1)$ of type $I_2$ intersect at two points $p_1$ and $p_2$. Therefore, the Floer complex \cite{floer1988morse,floer1989witten} (or morphisms) of the two $A$-branes $\frakB_{\UU_1}$ and $\frakB_{\UU_2}$ is
\begin{equation}
\Hom^*(\frakB_{\UU_1},\frakB_{\UU_2}):=CF^*(\frakB_{\UU_1},\frakB_{\UU_2})\cong\bC\langle p_1\rangle\oplus \bC\langle p_2\rangle~.
\end{equation}
Note that the Floer complexes $CF^*(\frakB_{\UU_1},\frakB_{\UU_2})$ and $CF^*(\frakB_{\UU_2},\frakB_{\UU_1})$ and the differentials on them are Poincar\'e-dual to each other. Namely,
each intersection point $p_i$ defines generators of both complexes, whose degrees sum to $2$ (the complex dimension of the target).

\begin{figure}[ht]\centering
  \includegraphics{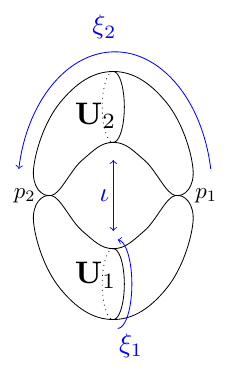}
  \caption{At the singular fiber $\pi^{-1}(b_1)$, $\xi_2$ exchanges the irreducible components, $\UU_1$ and $\UU_2$, by the $180^{\circ}$ rotation along the $(0,1)$-circle (longitude). Therefore, $\xi_2$ exchanges $p_1$ and $p_2$. On the other hand, $\iota$ exchanges $\UU_1$ and $\UU_2$ by fixing $p_1$ and $p_2$. Besides, $\xi_1$ maps each irreducible component to itself by the  $180^{\circ}$ rotation along the $(1,0)$-circle (meridian).}
  \label{fig:I_2}
\end{figure}

This implies that there are two bound states of $\frakB_{\UU_1}$ and $\frakB_{\UU_2}$ as $A$-branes. Let us consider one natural candidate for them: a brane $\brF^{\lambda}$ degenerating into the singular fiber $\pi^{-1}(b_1)$ of type $I_2$. First of all, the dimension $m$ of $\Hom(\Bcc,\brF^{\lambda})$ needs to be even $m=2n$ in order for the brane to be supported on a singular fiber because the evaluation of the integral cohomology class $[F'_\bfF/(2\pi)]$ over a singular fiber cannot be odd like \eqref{odd-integral}. There is also a topological constraint to be a bound state of the branes $\frakB_{\UU_1}$ and $\frakB_{\UU_2}$. As illustrated in Figure \ref{fig:I_2},
a one-cycle, say the $(1,0)$-cycle, of a torus is pinched to a double point at two locations so that the singular fiber $\pi^{-1}(b_1)$ topologically consists of two copies of $\bC\bfP^1$. Therefore, it has the unique bounding spin structure along the $(1,0)$-cycle, which is Neveu-Schwarz. Consequently, only a brane $\brF^{(-,+)}$ with trivial holonomy and the Neveu-Schwarz spin structure along the $(1,0)$-cycle of $\F$ can degenerate to a bound state of the branes $\frakB_{\UU_1}$ and $\frakB_{\UU_2}$ at the singular fiber $\pi^{-1}(b_1)$.

There is indeed a corresponding representation of $\SH$.
We see that the support $\UU_{1} \cup \UU_{2}$  is invariant under $\tau_-$ (as a set). Thus, a brane $\brF^{(x_{2n},+)}$ can enter the singular fiber when the corresponding module $\scF_{2n}^{(-,+)}$ is $\tau_-$-invariant, namely when  the two ideals
$$
(X^{2n}+X^{-2n}-x_{2n}-x_{2n}^{-1}), \qquad
(\tau_-(X^{2n}+X^{-2n}-x_{2n}-x_{2n}^{-1}))~
$$
coincide.
Under the condition \eqref{case1}, the $2n$-th Macdonald polynomial takes the form $P_{2n}=X^{2n}+X^{-2n}+2=(X^n+X^{-n})^2$, and \eqref{Mac-T} yields $\tau_-(1)=1$ and $\tau_-(P_{2n})=t^{-2n}P_{2n}$. For a generic value of $t$, only when $x_{2n}=-1$, we therefore have the $\tau_-$-invariant module $\scF_{2n}^{(-,+)}\cong \PR/(P_{2n})$.
Moreover, since $P_{2n}=(P_{n})^2$ under \eqref{case1}, the quotient of the polynomial representation $\PR$ yields a short exact sequence
\bea\label{SES4}
0\to \scU_n^{(2)}\to \scF_{2n}^{(-,+)}\to \scU_n^{(1)}\to0~.
\eea
The representation $\scF_{2n}^{(-,+)}$ corresponds to the bound state $\brF^{(-,+)}$.  As explained in \S\ref{sec:finite-rep}, the raising operator \eqref{raising} of $\PR$ does not become null because the prefactor $(1-q^{2j}t^2)$ cancels with the denominator of $P_{j+1}$. Consequently, this short exact sequence \eqref{SES4} does not split as a direct sum, but rather is a nontrivial extension of  $\scU_n^{(1)}$  by $\scU_n^{(2)}$. This is analogous to the fact that $\C[X]/(X^{2n})\to \C[X]/(X^{n})$ cannot split as a $\C[X]$-module. As such, when the gradings are chosen such that $\scU_n^{(1,2)}$ are in degree zero, the degree of the corresponding morphism between the $A$-branes is one, and corresponds to the class in~$\Hom^1(\frakB_{\UU_1},\frakB_{\UU_2})$ represented by~$\frakB_{\FF}^{(-,+)}$.\footnote{Often literature in mathematics uses the notation $\textrm{Ext}^1$ instead of $\Hom^1$. Here they have the same meaning.}
Although this paper does not determine the degree of the morphism in the $A$-brane category, the representation category of $\SH$ predicts one. Even in what follows, non-trivial extensions in the representation category give a description of degree-one morphisms (extensions or bound states) of various distinct compact $A$-branes. Determining the degree of the morphisms directly in the $A$-brane category is left for future work.

Since $\Hom^*(\frakB_{\UU_1},\frakB_{\UU_2})$ is two-dimensional, there must be another generator. To identify it, we consider the symmetries. As Figure \ref{fig:I_2} illustrates, $\xi_2$ and $\iota$ exchange the irreducible components $\UU_1$ and $\UU_2$ at the singular fiber. More precisely, $\xi_2$ acts on the singular fiber as the $180^{\circ}$ rotation along the $(0,1)$-circle (longitude) so that the intersection points $p_{1,2}$ are exchanged by $\xi_2$. On the other hand, $\iota$ exchanges $\UU_1$ and $\UU_2$ by fixing $p_{1,2}$. Consequently, the images of the brane $\brF^{(-,+)}$ under the symmetries $\xi_2$ and $\iota$ are non-isomorphic objects in the $A$-brane category. They indeed span the morphism space
\begin{equation}\label{Hom21}
  \Hom^1(\frakB_{\UU_2},\frakB_{\UU_1}) \cong \bC\langle  \xi_2(\brF^{(-,+)})\rangle  \oplus \bC \langle \iota(\brF^{(-,+)})\rangle ~.
\end{equation}
As a result, two irreducible branes can form bound states in more than one
way.
Similarly, the images of the brane $\scF_{2n}^{(-,+)}$ under the symmetries $\xi_2$ and $\iota$ are non-isomorphic in the representation category of $\SH$ for $n>1$.  The image of the short exact sequence \eqref{SES4} under $\xi_2$ becomes
\bea\label{SES1}
0\to  \scU_n^{(1)} \to \xi_2(\scF_{2n}^{(-,+)})\to \scU_n^{(2)}\to0~.
\eea
Likewise, The image of the short exact sequence \eqref{SES4} under $\iota$ becomes
\bea\label{SES1v}
0\to  \scU_n^{(1)} \to \iota(\scF_{2n}^{(-,+)})\to \scU_n^{(2)}\to0~.
\eea
By using the polynomial representation \eqref{poly-rep-sym}, one can read off the action of the generators $x$ and $y$ on these representations as
\begin{align}
 x\Big|_{\xi_2(\scF_{2n}^{(-,+)})}=  \raisebox{-2.2cm}{\includegraphics[width=5.5cm]{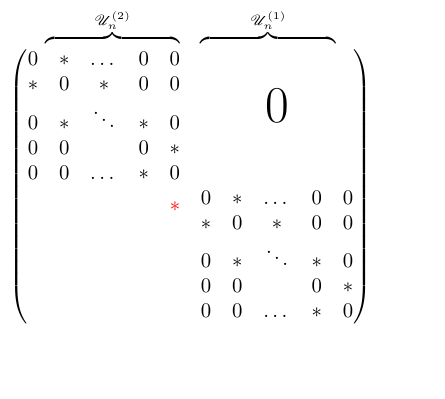}}~,\qquad   x\Big|_{\iota(\scF_{2n}^{(-,+)})}= \raisebox{-2.2cm}{\includegraphics[width=5.5cm]{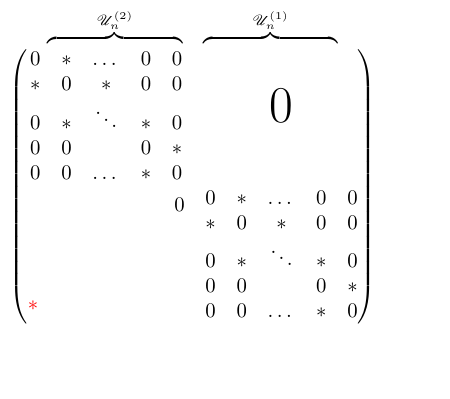}}~,
\end{align}
on the basis where $y$ acts diagonally as $\textrm{diag}(t+t^{-1},qt+q^{-1}t^{-1},\ldots, q^{2n-1}t+q^{1-2n}t^{-1})$. Note that the upper-left block and lower-right matrices of the $x$ actions are the same whereas the lower-left matrices are different. These matrices explicitly show that $\xi_2(\scF_{2n}^{(-,+)})$ and $ \iota(\scF_{2n}^{(-,+)})$ are not isomorphic.

In fact, the composition $\xi_2\circ\iota$ leaves $\frakB_{\UU_1}$ and $\frakB_{\UU_2}$ as they are, respectively. However, it maps $\brF^{(-,+)}$ to a different object. Correspondingly, we have a short exact sequence
\bea\label{SES4v}
0\to\scU_n^{(2)}\to \xi_2\circ\iota(\scF_{2n}^{(-,+)})\to \scU_n^{(1)}\to0~,
\eea
which is not isomorphic to \eqref{SES4}. Therefore, they span the morphism space of two dimensions
\begin{equation}\label{Hom12}
  \Hom^1(\frakB_{\UU_1},\frakB_{\UU_2}) \cong \bC\langle  \brF^{(-,+)}\rangle  \oplus \bC \langle \xi_2\circ\iota(\brF^{(-,+)})\rangle ~,
\end{equation}
which is Poincar\'e-dual to \eqref{Hom21}. In conclusion, when two irreducible branes intersect two points, they can form non-isomorphic bound states with the same support in the $A$-brane category, and these bound states are related to subtleties defining $A$-branes supported on singular submanifolds.

At the other singular fibers $\pi^{-1}(b_{2,3})$, there are similar bound states. As in \eqref{SL2Z-singular}, $\sigma\in\PSL(2,\bZ)$ maps \eqref{Hom21} to $\Hom^1(\frakB_{\UU_3},\frakB_{\UU_4})$. Also,  $\tau_+\in\PSL(2,\bZ)$ maps \eqref{Hom21} to $\Hom^1(\frakB_{\UU_5},\frakB_{\UU_6})$.

\subsubsection{\texorpdfstring{At global nilpotent cone of type $I_0^*$}{At global nilpotent cone of type I0*}}

Next, let us consider the case in which both the $A$-branes $\frakB_{\bfV}$ and $\frakB_{\bfD_i}$ exist. In order for both $\frakB_{\bfV}$ and $\frakB_{\bfD_i}$ to be Lagrangian,  \eqref{Lag-BunG} and  \eqref{Lag-D} need to be satisfied, which implies $\tgamma_p=0$ and $\hbar$ is real whereas $\talpha_p$ and $\tbeta_p$ can be arbitrary. Therefore, the symplectic form must be $\omega_\X=\omega_K/\hbar$.
 In this situation, $\bfF$ and $\bfU_i$ are also Lagrangian with respect to the symplectic form.  Moreover, the quantization conditions, \eqref{dim-D} and \eqref{dim-BunG}, for both $\brDi$ and $\brV$ are
\be \label{2quant-cond}
-c+\frac12=\ell~, \qquad \frac{1}{2\hbar}+2c-1=k+1~,
\ee
which implies that $1/2\hbar=2\ell +k+1$. In other words, the two shortening conditions lead to the other one
$$
\eqref{case3} \ \  \textrm{and}\ \ \iota\textrm{\eqref{case2}} \quad \longrightarrow \quad  \eqref{case1} \ \  \textrm{where} \ \ n=2\ell+k+1~.
$$
Under this condition, there are therefore finite-dimensional representations of three kinds, $\Hom(\Bcc,\frakB_{\bfU_i})$, $\Hom(\Bcc,\frakB_{\bfV})$ and $\Hom(\Bcc,\frakB_{\bfD_i})$. On the representation theory side, the quotient of the polynomial representation yields a short exact sequence
\be\label{SES2}
0\longrightarrow  \iota(\scV_{k+1}) \longrightarrow \scU_n^{(1)}\xrightarrow{f} \scD_\ell^{(1)}\oplus \scD_\ell^{(2)} \longrightarrow 0~.
\ee
We also note that there exist similar short exact sequences for the images of $\scU_n^{(1)}$ under the symmetry $\Xi \times \PSL(2,\Z)$ in \S\ref{sec:Ui} under the same shortening condition.

In a similar fashion, if the branes $\brDi$ and $\brUi$ exist simultaneously, their quantization conditions guarantee the existence of $\brV$. Also, if we assume the presence of the branes $\brV$ and $\brUi$, then the quantization condition for $\brDi$ follows. In fact,
it is straightforward to check that, under the relation $n=k+1+2\ell$, we have
\bea
\eqref{case3}\ \  \textrm{and}\ \ \eqref{case1}& \ \ \longrightarrow \ \ \iota\eqref{case2}~,\cr
\eqref{case1} \ \  \textrm{and}\ \ \iota\eqref{case2}& \ \ \longrightarrow\ \ \eqref{case3}~.
\eea
Subsequently, we have the short exact sequence \eqref{SES2}.

If $\tbeta_p\neq0$, the Hitchin fibration has the three singular fibers of type $I_2$ (Figure \ref{fig:3I_2}), and the Lagrangians $\bfV$ and $\brDi$ are not contained in a Hitchin fiber. Thus, the short exact sequence \eqref{SES2} implies that a Hamiltonian isotopy can deform the brane  $\frakB_{\bfU_1}$ in such a way that it contains $\brV$ as subbranes. The situation becomes much more lucid when $\tbeta_p=0$. As $\tbeta_p\to 0$, the three singular fibers meet simultaneously and transform into the singular fiber of type $I_0^*$, which is the global nilpotent cone. In this process, the $A$-brane $\frakB_{\bfU_1}$ becomes a bound state of $\frakB_{\bfD_1}$, $\frakB_{\bfD_2}$ and  $\frakB_{\bfV}$ because of \eqref{U-homology}. The short exact sequence \eqref{SES2} indeed corresponds to the bound state as illustrated in Figure \ref{fig:branes-reps}. A similar story holds for the other branes $\brUi$ and they become bound states of irreducible branes according to the relation \eqref{U-homology} of the second homology group.

\begin{figure}[ht]\centering
	\includegraphics[width=\textwidth]{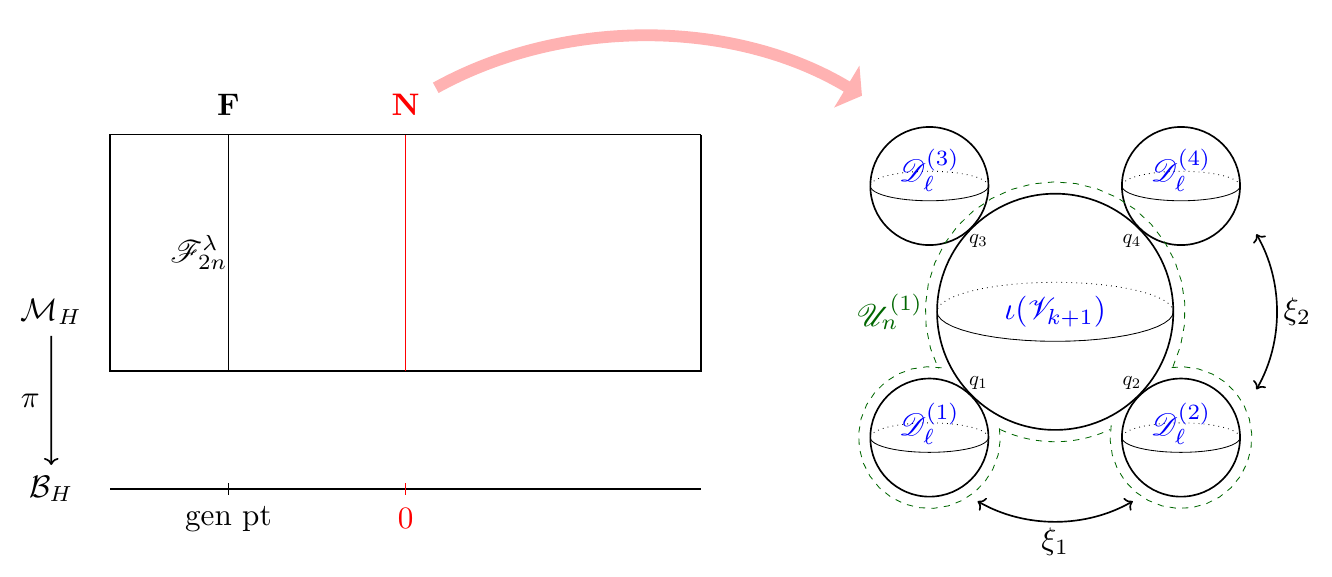}
	\caption{This figure depicts the correspondence between compact supports of $(B,A,A)$-branes and finite-dimensional modules of the spherical DAHA when $\hbar=1/2n$, $\talpha_p/2\hbar=\ell$ and $\tbeta_p=0=\tgamma_p$. Note that $n=2\ell+k+1$.}\label{fig:branes-reps}
\end{figure}

As explained above, the short exact sequence \eqref{SES2} does not split into the direct sum  because the raising operator \eqref{raising} of $\PR$ never becomes null. Geometrically, the choice of the direction of the arrows in  \eqref{SES2} comes from how the support \eqref{brane-poly} of $\brP$ intersects with the global nilpotent cone. As explained in \S\ref{sec:BunG} and \S\ref{sec:divisorTQFT}, the support \eqref{brane-poly} of $\brP$ cuts through real one-dimensional slices of the exceptional divisors $\D_{1,2}$, but it does not intersect with $\V$.  As a result, the brane $\frakB_{\bfV}$ becomes a subbrane of  $\frakB_{\U_1}$ whereas $\frakB_{\bfD_1}\oplus\frakB_{\bfD_2}$ becomes its quotient.

On the other hand, the $\iota$-image $\iota(\brP)$ intersects with $\V$ whereas it does not with exceptional divisors $\D_{i}$. Consequently, there is a short exact sequence
\be\label{SES}
0\longrightarrow  \scD^{(3)}\oplus \scD^{(4)} \longrightarrow \iota(\scU_n^{(1)})\longrightarrow \iota(\scV_{n-2\ell})\longrightarrow 0
~.
\ee
Once we take $\xi_2$-image of this short exact sequence, we have
\be\label{SESv}
0\longrightarrow  \scD_\ell^{(1)}\oplus \scD_\ell^{(2)}\xrightarrow{g}  \xi_2\circ \iota(\scU_n^{(1)})\longrightarrow \iota(\scV_{n-2\ell})\longrightarrow 0
~.
\ee
because $\iota(\scV_{n-2\ell})$ is $\xi_2$-invariant.

Now we are ready to compare the morphism structures of the two categories under the shortening condition $\hbar=1/2n$ and $\talpha_p/\hbar=\ell$. As Figure \ref{fig:branes-reps} illustrates, the supports of branes $\brV$ and $\brane_{\D_1}$ intersect at one point $q_1$ so that the morphism space between them is one-dimensional:
\begin{equation}
\Hom^1(\brane_{\D_1},\brV)\cong\bC\langle q_1 \rangle~.
\end{equation}
 This means that there is one bound state of $\brV$ and $\brane_{\D_1}$. Indeed, we find the corresponding representation from \eqref{SES2}:
\be\label{SES2-f}
0\longrightarrow  \iota(\scV_{k+1}) \longrightarrow f^{-1}(\scD_\ell^{(1)})\longrightarrow \scD_\ell^{(1)} \longrightarrow 0~.
\ee
Its Poincar\'e dual in the representation category can be obtained from \eqref{SESv}
\be\label{SESv-g}
0\longrightarrow  \scD_\ell^{(1)}\longrightarrow  \xi_2\circ \iota(\scU_n^{(1)})/g(\scD_\ell^{(2)})\longrightarrow \iota(\scV_{n-2\ell})\longrightarrow 0
~.
\ee
By using the sign change group $\Xi$, we obtain short exact sequences analogous to \eqref{SES2-f}, which changes from $\scD_\ell^{(1)}$ to $\scD_\ell^{(i)}$ ($i=2,3,4$).
We can further pursue the comparison of the morphism structure. In the $A$-brane category, the morphism space between  $\brV$ and $\brane_{\D_1}\oplus \brane_{\D_2}$ is two-dimensional:
\begin{equation}
\Hom^1(\brane_{\D_1}\oplus \brane_{\D_2},\brV)\cong\bC\langle q_1 \rangle\oplus \bC\langle q_2 \rangle~.
\end{equation}
It is easy to find the corresponding representations
\bea
0\longrightarrow  \iota(\scV_{k+1}) \longrightarrow f^{-1}(\scD_\ell^{(1)})\oplus \scD_\ell^{(2)}  \longrightarrow \scD_\ell^{(1)} \oplus \scD_\ell^{(2)} \longrightarrow 0~,\cr
0\longrightarrow  \iota(\scV_{k+1}) \longrightarrow f^{-1}(\scD_\ell^{(2)})\oplus \scD_\ell^{(1)}  \longrightarrow \scD_\ell^{(1)} \oplus \scD_\ell^{(2)} \longrightarrow 0~.
\eea
In fact, the short exact sequence \eqref{SES2} can be understood as the diagonal element corresponding to $q_1+q_2 \in \Hom^1(\brane_{\D_1}\oplus \brane_{\D_2},\brV)$. More generally, we have
\be
\Hom^1(\oplus_{i\in I}\brane_{\D_i},\brV)\cong \oplus_{i\in I}\bC\langle q_i \rangle~,
\ee
where $I$ is a subset of $\{1,2,3,4\}$. The diagonal element in the representation category is
\be
0\longrightarrow  \iota(\scV_{k+1}) \longrightarrow \scN_{|I|\ell+k+1}^I  \longrightarrow \oplus_{i\in I} \scD_\ell^{(i)} \longrightarrow 0~,
\ee
where $|I|$ is the cardinality of the set $I$. We write the corresponding $A$-brane
\begin{equation}
\brane_{\N^I}\in \Hom^1(\oplus_{i\in I}\brane_{\D_i},\brV)~,
\end{equation}
which is supported on $\N^I:=\cup_{i\in I} \D_i\cup \V$.

If the cardinality $|I|$ is three, the corresponding brane is supported on  $\V$ plus three exceptional divisors, and the representation $\scN^I$ is not obtained by a quotient of the polynomial representation.  Therefore, these are \emph{new finite-dimensional representations}, which do not appear in the theorems of Cherednik~\cite[\S2.8--9]{Cherednik-book}.

When $I=\{1,2,3,4\}$, the support of the corresponding brane is the entire global nilpotent cone $\N$ \eqref{nilconetoy} so that we simply write it as $\brN$. It turns out that this brane gives rise to another interesting bound state in the $A$-brane, which we will see below.

\paragraph{The global nilpotent cone.}
In fact, when $n-k-1$ is odd (or equivalently $c\in \Z_{\le0}$), there is a short exact sequence
\be\label{SES3}
0\longrightarrow  \iota (\scV_{k+1}) \longrightarrow \scF_{2n}^{(+,+)}\longrightarrow \scN_{2n-k-1}\longrightarrow 0~,
\ee
In fact, when both \eqref{case1} and $\iota$\eqref{case2} are satisfied, we have $\pol(\sfL_{2n-k-1})\cdot P_{2n-k-1}=0$. Furthermore, when $n-k-1$ is odd, then $\scN_{2n-k-1}:=\scP/(P_{2n-k-1})$ becomes an irreducible module of dimension $2n-k-1$. The short exact sequence \eqref{SES3}  illustrates that the module $\scN_{2n-k-1}$ can also be obtained by the quotient $\scF_{2n}^{(+,+)}/ \iota (\scV_{k+1}) $.

When $\tbeta_p=0=\tgamma_p$, the Hitchin fibration has one singular fiber of type $I_0^*$, and the entire global nilpotent cone $\bfN$ is Lagrangian with respect to $\omega_\X$.
The short exact sequence \eqref{SES3} indeed depicts the situation where the brane $\frakB_\bfF^{\lambda=(+,+)}$ with the Ramond spin structures enters the global nilpotent cone. Since it has a different spin structure, the brane is \emph{not} decomposed into each irreducible component. As a result, the brane $\frakB_\bfF^{\lambda=(+,+)}$ becomes the bound state of two branes; $\brV$ and $\brN$.
Actually, using the fiber class relation \eqref{fiber-class-rel} with \eqref{dim-generic-fiber} and \eqref{dim-BunG}, one can evaluate the dimension formula for an $A$-brane $\brN$
\bea
\dim \Hom(\Bcc,\frakB_\bfN) = \int_{\bfN}\frac{F+B}{2\pi}= \int_{\bfN} \frac{\omega_I}{2\pi \hbar}=\frac{1}{2\hbar}-2c+1~.
\eea
From \eqref{2quant-cond}, this is equal to $2n-k-1$, and the space of $(\Bcc,\brN)$-strings therefore corresponds to the module $\scN_{2n-k-1}$ in \eqref{SES3}.

One delicate point arises for constructing the Chan-Paton bundle for $\brN$ because $\bfN$ is \emph{not} a manifold. Since $\V$ is linked with the exceptional divisors $\D_i$ in $\brN$, the Chan-Paton bundle for $\brN$ is no longer well-defined at the four joining points of $\V$ and $\D_i$. The Chan-Paton bundle becomes a putative ``line bundle'' $\cL'$ over each exceptional divisor $\D_i$ and the curvature $F'_{\bfN}$ of its connection has a half-integral flux over it \cite{Freed:1999vc}
$$
\int_{\D_i}F'_{\bfN}=-\frac12~,
$$
while it cancels with the $B$-field due to \eqref{deformed-flat}
$$F'_{\bfN}+B\big|_{\bfN}=0~.$$
In other words, $\cL'$ restricted to an exceptional divisor $\D_i$ is a ``square root'' of the $\cO(-1)\to \CP^1$ bundle and the $B$-field flux over it is $1/2$.
As a result, we have
$$
\int_{\D_i}\frac{F+B}{2\pi}= \int_{\D_i} \frac{\omega_I}{2\pi \hbar}=\frac{n-k-1}2\in \frac12+\Z~,
$$
which gives the condition that $n-k-1$ is odd.

Under this circumstance, the line bundle $\cL\to \X$ \eqref{Bcc} for $\Bcc$ is actually the $2n$-th tensor product of the determinant line bundle \cite[\S8]{huybrechts2010geometry} of the Hitchin moduli space. As a result, the geometric quantization of $\bfV$ provides the quantum Hilbert space $\scV_{k+1}$ on a once-punctured torus in Chern-Simons theory \cite{Gukov:2010sw}. The additional series $\scV_{k+1}$ at a primitive $2n$-th root of unity $q=e^{\pi i/n}$ is called \emph{perfect representation} \cite[\S2.9.3]{Cherednik-book}.
Moreover, when $n=k+2$, the additional series  $\scV_{k+1}$ of dimension $k+1$ is isomorphic to the well-known Verlinde formula of $\wh{\fraksl}(2)_{k}$ with level $k$ for a torus (without puncture) \cite{Verlinde:1988sn}.

\bigskip

Let us end this section by commenting briefly on future directions.
There are an enormous number of non-compact Lagrangian submanifolds in the moduli space of Higgs bundles that have been studied in their own right: for example, the image of the Hitchin section, the brane of opers (see \cite{Nekrasov:2010ka,Mikhaylov:2017ngi,Balasubramanian:2017gxc,Gaiotto:2021tsq} in a similar context), or the $A$-polynomial of any knot \cite{Gukov:2003na}. Each of these geometric objects should naturally be associated with an \SH-module whose behavior precisely matches the geometric properties of the object, just as we demonstrate for compact Lagrangians and for the (generalized) polynomial representation. It would be of great interest to further pursue this correspondence for infinite-dimensional representations, even just in the rank-one case.

It would also be interesting to connect explicitly with other mathematical contexts in which algebraic approaches to the Fukaya category or equivalences between Fukaya categories and module categories appear. To give just one example, in \cite{LekiliWF}, Etg\"u and Lekili study the Chekanov-Eliashberg dg-algebra associated with a Legendrian link in a Weinstein four-manifold for a given graph. They show that this algebra is $A_\infty$-quasi-isomorphic to, roughly speaking, the endomorphism algebra of a collection of generating objects of the wrapped Fukaya category of the surface, and go on to recover the multiplicative preprojective algebra studied in~\cite{CBS} in the context of the Deligne--Simpson problem from the Legendrian link. When the graph in question is the affine $D_4$ Dynkin diagram, it is expected that the corresponding preprojective algebra is related to DAHA. (We thank A.~Oblomkov for private communication related to this point.) The computations of the (wrapped) Fukaya category of the above four-manifolds in \cite{LekiliWF} thus may provide an interesting perspective on our Claim~\ref{Th:Claim1} as well as its generalization to other algebras.

\section{3d theories and modularity}
\label{sec:3d}

In this section, we connect the brane setup of the above 2d $A$-model to 3d/3d correspondence and shed light on various modular representations coming from geometry. In particular, we explain the origin for the explicit form of the $S$ and $T$ matrices in Conjectures \ref{conj:rcs} and \ref{conj:MTC-D1}.
The modular action in Conjecture \ref{conj:rcs} turns out to be the one of refined Chern-Simons theory \cite{Aganagic:2011sg}.
On the other hand, the modular action in  Conjecture \ref{conj:MTC-D1} is a ``hidden'' (surprising) one;
it is realized on the vector space spanned by the set of connected components of fixed points under the Hitchin $\U(1)_{\beta}$ action on the moduli space of wild Higgs bundles associated to a certain Argyres-Douglas theory.
Furthermore, we propose how non-standard (\textit{e.g.}\ logarithmic) modular data of $\MTC [M_3]$ can be described in terms of the $A$-model on the Hitchin moduli space associated with the Heegaard decomposition of $M_3$ and discuss possible connections to skein modules of closed oriented 3-manifolds.

One advantage of connecting the 2d $A$-model to the three-dimensional perspective is that all of these modular actions admit a natural categorification. In other words, in all of these instances it makes sense to ask if the space of open strings in the Hitchin moduli space can be realized as the Grothendieck groups of a tensor category (possibly, non-unitary or non-semisimple):
$$
\SL(2,\Z) \ \rotatebox[origin=c]{-90}{$\circlearrowright$} \ K^0 (\MTC)~.
$$
Finally, we will see that, in the opposite direction, the relation to the 2d $A$-model offers a unifying home for the above-mentioned modular data.

\subsection{DAHA and modularity}\label{sec:modularity}

The fivebrane system in M-theory that provides geometric origins of the modular representations on DAHA modules is the following familiar setting for the 3d/3d correspondence
\be\label{3d3d}
\begin{matrix}
	{\mbox{\rm space-time:}} & \quad & S^1 & \times_{q,t} & \big(TN & \times  &T^* M_3\big) \\
	{\mbox{\rm $N$ M5-branes:}} & \quad & S^1 &\times_q& D^2 & \times &  M_3
\end{matrix}
\ee
where $M_3$ is a 3-manifold, $D^2$ is a two-dimensional disk (or a cigar), and $TN \cong \R^4$ is the Taub-NUT space.
Writing the local complex coordinates $(z_1,z_2)$ on $TN$, such that $z_1$ also parametrizes $D^2$, we turn on the $\Omega$-background, \textit{i.e.}\ a holonomy along $S^1$ that provides a twisting of $TN$ via an isometry
\be\label{omega}
(z_1, z_2)\to (qz_1, t^{-1}z_2)~.
\ee
In this setting, the symmetry group of the 6d $(2,0)$ theory on the M5-branes is reduced to
\be\label{M5-sym}
\SO(6)_E\times \SO(5)_R\to \SO(3)_1 \times \SO(3)_2\times \SO(3)_R\times \SO(2)_R~,
\ee
where $\SO(3)_1$ and $\SO(3)_2$ are the space-time symmetry of $S^1\times_q D^2$ and $M_3$, respectively, and $\SO(3)_R$ is the symmetry of a cotangent fiber of $T^*M_3$.  We perform a topological twist by taking the diagonal subgroup $\SO(3)_{\textrm{diag}}$ of $\SO(3)_2\times \SO(3)_R$ so that the resulting theory is partially topological (along $M_3$). After the partial topological twist, the effective theory on $S^1\times_q D^2$ only depends on topology (but not the metric) on $M_3$ and is described by 3d $\cN=2$ theory often denoted $\cT[M_3]$ \footnote{In this section, we restrict ourselves to $\SU(N)$ gauge group so that $\cT[M_3,\SU(N)]=\cT[M_3]$.}, with the $R$-symmetry given by $\SO(2)_R$ in \eqref{M5-sym}. When $M_3$ is a Seifert manifold, there is an extra $\U(1)_S$ symmetry associated with the two directions in the cotangent bundle normal to the Seifert fiber. As a result, the partition function, called the half-index, of the 3d $\cN=2$ theory $\cT[M_3]$ on $S^1\times_q D^2$ with a 2d $\cN=(0,2)$ boundary condition $\cB$ in this setting is defined as
\be\label{refined-index}
Z_{\cT[M_3]}(S^1\times_q D^2, \cB)=\Tr(-1)^F e^{-\beta(\Delta-R-J_3/2)} q^{J_3+S}t^{R-S}~,
\ee
where $S$ and $R$ are charges of $\U(1)_S$ and $\SO(2)_R$, respectively, $\Delta$ is the Hamiltonian, and $J_3$ is an eigenvalue of the Cartan subalgebra of $\SO(3)_1$. The difference between $\U(1)_S$ and $\SO(2)_R$ is customarily denoted $\U(1)_\beta$ in \cite{Gukov:2015sna,Gukov:2016gkn,Gukov:2017kmk}, and its fugacity is the variable $t$ in \eqref{refined-index}.

Notice that the system \eqref{3d3d} does not involve a once-punctured torus which was used to define the Hitchin moduli space and the parameter $t$ as in the previous section. However, for gauge groups of type $A$, the following two physical systems are expected to be closely related:
\begin{itemize}[nosep]
    \item 6d $(2,0)$ theory on $S^1\times C_p$ with $C_p$ being a once-punctured torus.
    \item 4d $\mathcal{N}=2^*$ theory  on $S^1$.
\end{itemize}
Although the two systems would have different spectra,\footnote{For example, many KK modes of the 6d theory on $T^2$ have no counterparts in the 4d theory. Even if one replaces the 4d $\mathcal{N}=2^*$ theory with 6d $(2,0)$ theory on a torus (with the mass parameter replaced by holonomies for a $\U(1)$ subgroup of the R-symmetry on $T^2$) the full spectrum is still different. One way to see this is that the latter theory depends on all three $\U(1)$ holonomies on $T^2\times S^1$ in a periodic way, and they are completely symmetric, while this is not the case for the former theory obtained from a punctured torus.} their BPS sectors are expected to be equivalent. In particular, at low energy both systems realize a 3d sigma-model onto the Hitchin moduli space. The deformation parameters can also be identified as follows.

On one side, the (classical) deformations are parametrized by the triplet $(\alpha_p,\beta_p,\gamma_p)$ of monodromy parameters around the puncture as introduced before.
On the other side, for the 4d $\mathcal{N}=2^*$ theory, the triplet of deformation parameters is given by the complex mass of the adjoint hyper-multiplet in 4d together with the holonomy of the $\U(1)$ flavor symmetry along the circle. In the system \eqref{3d3d}, the 4d $\mathcal{N}=2^*$ theory is obtained by the compactification of the 6d theory on $T^2\subset M_3$ with  holonomy for $\U(1)_\beta$ along $S^1$ \eqref{refined-index}. In particular, the parameter $t$ defined above is identified with the $t$ in DAHA.  In this section, we will be looking at questions whose answers depend holomorphically on $t$, as required by supersymmetry on $M_3$, and the other deformation parameter $\beta_p$ won't play a role. For example, what complex connections on $T^2\subset M_3$ can be extended to the entire $M_3$ is a question that is ``holomorphic in $J$'' (and given by intersections of $(A,B,A)$-branes in the Hitchin moduli space). Notice that this non-trivial relation only holds for a gauge group of type $A$, while for other types the class $\mathcal{S}$ construction of 4d $\mathcal{N}=2^*$ theory is generally unknown, and the once-punctured torus does not lead to either the 4d $\mathcal{N}=2^*$ theory or DAHA.

One statement of the 3d/3d correspondence is the duality between the \textit{non-perturbative} complex $\SL(N,\C)$ Chern-Simons theory on $M_3$ and the 3d $\cN=2$ theory $\cT[M_3]$ on $S^1\times_q D^2$, so that the partition functions of both sides are identified. As explained in \cite{Gukov:2016gkn,Gukov:2017kmk,Gukov:2019mnk}, for a particular class of boundary conditions $\cB_b$ labeled by $b\in (\Spin^c(M_3))^{N-1}$, the partition function of the 3d $\cN=2$ theory $\cT[M_3]$ on $S^1\times_q D^2$ counts BPS states and, therefore, has a $q$-expansions with integer coefficients and integer $q$-powers\footnote{Up to an overall factor $q^{\Delta_b}$ that plays an important role but is not relevant to the present discussion.}
\be\label{half-index}
\widehat { Z }_{\cT[M_3],b } ( q,t ) :=
Z_{\cT[M_3]}(S^1\times_q D^2, \cB_b)~.
\ee
The relation to Chern-Simons theory involves the same space of boundary conditions with a ``dual'' basis, related to $\cB_b$ via  the $\bS$-matrix
$$
\mathbb{S}_{ab}=\frac{\sum_{\sigma\in \frakS_N}e^{2\pi i \sum_{i=1}^{N-1}\lk(a_i,b_{\sigma(i)})}}{|\mathrm{Stab}_{\frakS_N}(a)|\cdot|\Tor H_1(M_3,\Z)|^{(N-1)/2}}~.
$$
In particular, the partition function of the non-perturbative $\SL(N,\C)$ Chern-Simons theory on $M_3$ is given by
\be
\Big( - \frac{\log q}{4\pi i} \Big)^{\frac{N-1}{4}}
\sum_{a,b\in (\Spin^c(M_3))^{N-1}} e^{2\pi i \bar k \cdot  \lk(a,a)}  \mathbb{S}_{a b } \widehat { Z }_{\cT[M_3],b } ( q,t )
\ee
with generic $|q|<1$, and specializes to that of $\SU(N)$ Chern-Simons theory when $q \to e^{2\pi i / \bar k}$ with integer (renormalized) level $\bar k = k + N$. The origin of $\log q$ factors is explained in \cite{Pei:2015jsa}.
Note that the linking pairing $\lk$ on $\Spin^c(M_3)$ is defined by the Pontryagin duality.
We will see shortly that it is the basis of BPS partition functions \eqref{half-index} and the corresponding boundary conditions $\cB_b$ that are most naturally related to DAHA.

Consider a simple example where $M_3 = L(p,1)$ is a Lens space. The lens space $L(p,1)$ can be constructed by gluing two solid tori with a homeomorphism between the boundary tori sending the meridian $(1,0)$ of one torus to a $(1,p)$ cycle of the other.
The corresponding 3d $\cN=2$ $\SU(N)$ gauge theory $\cT[L(p,1)]$ consists of one adjoint chiral multiplet $\Phi$ with $R$-charge 2 and $\cN=2$ Chern-Simons term with level $p$.
Consequently, the factor $\wh Z_{\cT[L(p,1)],b}$ labeled by $b\in (\Spin^c(M_3))^{N-1}$ is defined by
\be\label{Z-hat}
\wh Z_{\cT[L(p,1)],b} (q,t)={\frac{1}{N!}} \int_{|X|=1} \frac{d X}{2\pi i X} \;  \Upsilon(X;q,t) \Theta_{b}^{\bZ^{N-1};p}(X, q) ~,
\ee
where
$$
\Upsilon(X)= \prod_{\alpha\in \sfR} \frac{(X^\alpha; q^2)_\infty }{ (t^2X^\alpha; q^2)_\infty}~, \qquad \Theta_{b}^{\bZ^{N-1};p}(X, q)=\sum_{n\in p\bZ^{N-1}+b}q^{2\sum_{i=1}^{N-1}{n_i^2}/p}\prod_{i=1}^{N-1}X_i^{n_i}~.
$$
Here we impose the Neumann boundary condition at the boundary $\partial (S^1\times_q D^2)$ on the vector multiplet  and adjoint chiral multiplet, which give rise to the numerator and denominator of the Macdonald measure $\Upsilon$ by one-loop determinant \cite{Yoshida:2014ssa} (see also \eqref{Macdonald-product}). In addition,  the boundary partition function $\Theta_{b}^{\bZ^{N-1};p}$ encodes the information about the Chern-Simons term with level $p$, and 2d $\cN=(0,2)$ boundary condition at the boundary $\partial (S^1\times_q D^2)$ is labeled by $b\in (\Spin^c(M_3))^{N-1}$. In fact, $\wh Z_{\cT[L(p,1)],b}$ can be understood as the half-index of the 3d/2d coupled system.  For more detail, we refer to \cite{Gukov:2016gkn,Gukov:2017kmk}.

When the lens space $L(p,1)$ is constructed by gluing two solid tori, we can include a Wilson loop in each solid torus. The \emph{reduced} partition function with boundary condition specified by $\Spin^c$ structure $b$ results in
\be\label{half-index-Wilson}
\wh Z_{\cT[L(p,1)],b}(\lambda,\mu) =\frac{1}{ \wh Z_{\cT[L(p,1)],b} } {\frac{1}{N!}} \int_{|X|=1} \frac{d X}{2\pi i X} \;  \Upsilon(X)\;
\Theta_{b}^{\bZ^{N-1};p}(X) \; P_\lambda(X)\overline{P_\mu(X)}~,
\ee
where the conjugation  $f\mapsto \overline f$ is defined in \eqref{conjugate}.
In particular, when $p=0$, \ie $ L(0,1)\cong S^1\times S^2$, the partition function vanishes unless the total charge of two Wilson loops is zero. This defines the Macdonald inner product  \eqref{Macdonald-norm}
$$\langle P_\lambda,P_\mu\rangle= \wh Z_{\cT[L(0,1)],0}(\lambda,\mu) =\delta_{\lambda,\mu}~g_\lambda(q,t)~.$$
In the case of $M_3=S^3$, this defines the symmetric bilinear pairing \cite{cherednik1995macdonald,etingof1996representation,kirillov1996inner}
\be \label{SUN-S}
\textbf[P_\lambda,P_\mu\textbf]=\wh Z_{\cT[L(1,1)],0}(\lambda,\mu) =P_\lambda(q^{-2\mu}t^{-2\rho})P_\mu(t^{-2\rho})~,
\ee
where $\rho$ is the Weyl vector of $\fraksl(N)$. As in Appendix \ref{app:sym-bilinear}, this pairing $\C_{q,t}[X]^{\frakS_N}\times \C_{q,t}[X]^{\frakS_N}\to \C_{q,t}$ can be defined by transforming the holonomy  $\Tr\,(X)$ along the $(1,0)$-cycle in one solid torus to the holonomy  $\Tr\,(Y)$ along the $(0,1)$-cycle, and it acts on loop operators in the other solid torus via the polynomials representation when they link:
\be\label{sym-blinear-2}
\textbf[f(X),g(X)\textbf]=\mathrm{pol}(f(Y^{-1}))\cdot g(X)\Big|_{X\mapsto t^{-2\rho}}
\ee
for $f,g\in \C_{q,t}[X]^{\frakS_N}$.  In the case of $\SU(2)$, this is indeed \eqref{S-poly}. This can be viewed as a deformed version of the construction of the skein module of type $A_{N-1}$
\be \label{Skein-Heegaard}
\Sk(M_3,\SU(N))=\Sk(M_3^+,\SU(N))\mathop{\otimes}_{\Sk(C,\SU(N))}\Sk(M_3^-,\SU(N))
\ee
of a closed oriented 3-manifold $M_3$ by using a Heegaard splitting $M_3= M_3^{+}\cup_C M_3^{-}$.
 As seen in \S\ref{sec:poly-rep}, the polynomial representation $\repP$ of $\SH$ can be understood as a deformed Skein module of a solid torus $S^1\times D^2$. In \eqref{half-index-Wilson}, $P_\lambda(X)$ (resp. $\overline{P_\mu(X)}$) can be actually regarded as a basis element of the deformed skein module of one (resp. the other) solid torus, and the boundary partition function $\Theta$ glues the two solid tori by the $S$-transformation \eqref{brane-poly-SL2Z}. Thus, the spherical DAHA acts on the left-module via the polynomial representation whereas it acts on the right-module via its $S$-transformation. As a result, the $S$-transformation $\sigma(\repP)$ of the polynomial representation, called the functional representation, can be defined by the symmetric bilinear pairing, which is presented in Appendix~\ref{app:func-rep}.

$$
\begin{tikzpicture}
\node[rounded rectangle,fill=blue!10] at (0,0) {6d $(2,0)$ theory of type $A_{N-1}$ on  $ S^1\times_q D^2 \times S^1_\tau\times_\zeta C$};
\node[rounded rectangle,fill=blue!10] at (4,-2.5) {3d $\cN=2$ theory $\cT[S^1_\tau\times C]$  on $S^1 \times_q D^2$ };
\node[rounded rectangle,fill=blue!10] at (-4,-2.5) {2d $\sigma$-model $S^1_\tau\times \bfI\to\cM_{H} (C_p,\SU(N))$};
\node at(0,-2.5) {$\cong$};
\draw[->] (-2,-.5) to node [left=0.25] {on $S^1\times_q S^1\times C$} (-4,-2);
\draw[->] (2,-.5) to node [right=0.25] {on $ S^1_\tau\times_\zeta C$} (4,-2);
\end{tikzpicture}
$$

\begin{figure}[ht]
	\centering
	\includegraphics[width=10cm]{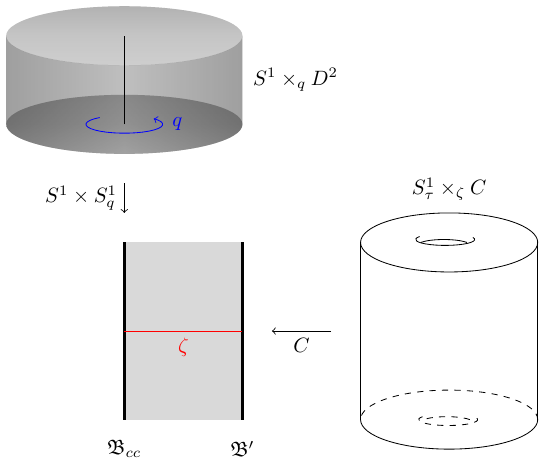}
	\caption{The relation between 3d $\cN=2$ theory $\cT[S^1\times_\zeta C]$ and 2d sigma-model. A mapping torus $S^1\times_\zeta C$ where the top and bottom tori are identified by $\zeta\in \SL(2,\Z)$ gives rise to an $\SL(2,\Z)$ duality wall on the worldsheet of $(\Bcc,\frakB')$-string.}
	\label{fig:6d-to-sigma}
\end{figure}

Moreover, the relation between 3d $\cN=2$ theory $\cT[M_3]$ to the 2d sigma-model explored in \S\ref{sec:2d} becomes manifest from the fivebrane system \eqref{3d3d}. For the sake of brevity, let $M_3= S^1_\tau\times C$ where $C\cong T^2$. As described above, the compactification of the 6d theory on $C$ with $\U(1)_\beta$ holonomy along $S^1$ leads to 4d $\cN=2^*$ theory, and the $t$ parameter in \eqref{3d3d} can be identified with the ramification parameters $(\talpha_p,\tgamma_p)$ via
\be\label{t-ramification}
tq^{-\frac12}=\exp(-\pi  ( \tgamma_p+i \talpha_p))~.
\ee
As in Figure \ref{fig:6d-to-sigma}, we further compactify the 4d $\cN=2^*$ theory on a two-torus $T^2=S^1\times S^1_q\subset S^1\times_q D^2$ to obtain the 2d sigma-model $ S^1_\tau\times \bfI\to \MH(C_p, \SU(N))$ where the interval $\bfI=[0,1]$ is obtained by reducing along $S^1_q \subset D^2$. The canonical coisotropic brane $\Bcc$ arises at the boundary of the strip $S^1_\tau\times \bfI$ corresponding to the center of $D^2$ \cite{Nekrasov:2010ka}. In addition, a boundary condition of 3d $\cN=2$ theory at $\partial(S^1\times_q D^2)$ gives rise to a brane $\frakB'$ at the other boundary of the strip $ S^1_\tau\times \bfI$ in the 2d sigma-model.

The theory $\cT[S^1_\tau\times C]$ consists of three $\cN=2$ adjoint chiral multiplets $Q,\wt Q$ and $\Phi$ where the Neumann boundary condition is imposed on the $\cN=2$ vector multiplet and chiral multiplets $\wt Q$ and $\Phi$, and the Dirichlet boundary condition is imposed on the $\cN=2$ chiral multiplet $Q$ at $\partial (S^1\times_q D^2)$.
Moreover, the form \eqref{refined-index} of the refined index tells us that fermions are periodic and a field $\Psi$ is identified along the time circle $S^1$
\be\label{identification}  q^{(J_3+S)}t^{(R-S)}\Psi(x^0+\beta,z_1)\sim  \Psi(x^0,z_1)~.\ee
The time derivative is replaced as $\partial_t \to \partial_t -R-J_3/2$ due to $e^{-\beta(\Delta-R-J_3/2)}$.
\begin{table}[ht]\centering
	\begin{tabular}{cccc}
		&$\U(1)_R$&$\U(1)_S$&bdry cond.\\
		$\Phi$&$2$&0&N\\
		$Q$&0&$-2$&D\\
		$\wt Q$&0&0&N
	\end{tabular}
\end{table}

One important lesson that we learn in this subsection is that the Hilbert space of a non-perturbative complex Chern-Simons TQFT on a 2-torus is the space of representations of the spherical DAHA at $t=1$. A categorified version of this statement would be a relation between the category of line operators in the $\wh Z$ TQFT and the category of modules of $\SH_{t=1}$,
\be
\MTC (\wh Z) \cong \Rep( \SH_{t=1} )~.
\ee
Again, we remind that here and in other places, $\MTC$ refers to a tensor category where some of the traditional conditions may need to be relaxed, \textit{e.g.}\ it may have an infinite number of simple objects, be non-unitary or non-semisimple. (The latter generalization typically appears when one tries to ``truncate'' a category with infinitely many simple objects to a finite-dimensional structure.) The modular representations that arise from such generalizations are, in general, more delicate and interesting than familiar vector-valued modular forms that describe the space of genus-1 conformal blocks in a rational VOA. Of course, in some special cases, these more interesting and exotic generalizations do not arise, and $\MTC$ is a genuine modular tensor category in its full mathematical sense (justifying the name for generalizations as well); this happens in some of the examples discussed in the following subsections and also in various examples considered in \cite{Gukov:2016gkn,Feigin:2018bkf,Dedushenko:2018bpp}.

\subsubsection{\texorpdfstring{$\SU(2)$: refined Chern-Simons and TQFT associated to Argyres-Douglas theory}{SU(2): refined Chern-Simons and TQFT associated to Argyres-Douglas theory}}\label{sec:SU2}

This connection of 3d theories to the 2d sigma-model clarifies the geometric origin of the modular action.
It was proposed in \cite{Aganagic:2011sg} that the fivebrane system \eqref{3d3d} with $N=2$ M5-branes gives rise to $\SU(2)$ refined Chern-Simons theory on $M_3$ when the parameters are subject to\footnote{The parameters $(q_{\textrm{ours}},t_{\textrm{ours}})$ in this paper are related to the parameters  $(q_{\textrm{AS}},t_{\textrm{AS}})$ in \cite{Aganagic:2011sg} via $q_{\textrm{our}}=q_{\textrm{AS}}^{1/2}$ and $t_{\textrm{our}}=t_{\textrm{AS}}^{1/2}$.}
\be \label{rCS-variables}
q=\exp\Bigl(\frac{\pi i}{k+2c} \Bigr)~,\qquad t=\exp\Bigl(\frac{c\, \pi i}{k+2c} \Bigr)~.
\ee
This condition is equivalent to the existence \eqref{case2} of the brane $\brV$ in the 2d sigma-model \S\ref{sec:BunG} so that the field identification \eqref{identification} under \eqref{rCS-variables} leads to the boundary condition $\frakB'=\brV$ upon the reduction as in Figure \ref{fig:6d-to-sigma}.
Therefore, the module $\Hom(\Bcc,\brV)$ of DAHA in the 2d sigma-model can be identified with the Hilbert space of $\SU(2)$ refined Chern-Simons theory on $T^2$ spanned by $\{| P_j \rangle \}$ $(j=0,\ldots,k)$.
The projective action of $\SL(2,\Z)$ on the Hilbert space is manifest in refined Chern-Simons theory, and the matrix elements can be obtained via the 3d/3d correspondence. In fact, the pairing \eqref{SUN-S} at $N=2$ (which is equal to \eqref{S-poly}) becomes of rank $(k+1)$ when \eqref{rCS-variables} holds; it gives the modular $S$-matrix  in Conjecture \ref{conj:rcs} up to a suitable normalization with the Macdonald norm \eqref{Mac-norm-A1}.
Upon reduction to the sigma-model, it can be interpreted as the $S$-duality wall in the worldsheet of the $(\Bcc,\brV)$-string.
Thus, the gluing of the two states $\lambda,\mu\in\Hom(\Bcc,\brV)$ by the $S$-duality wall in the $(\Bcc,\brV)$-string can be understood as the Hopf link configuration in refined Chern-Simons theory on $S^3$, illustrated in Figure \ref{fig:Hopf-link-brane}.

\begin{figure}[ht]
	\centering
	\includegraphics[width=9cm]{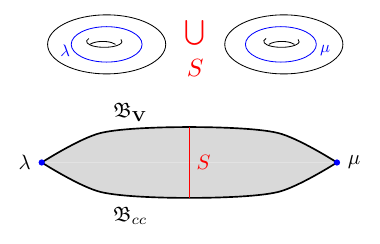}
	\caption{The $(\Bcc,\brV)$-string with the $S$-duality wall gives rise to refined Chern-Simons invariant of the Hopf link in $S^3$.}
	\label{fig:Hopf-link-brane}
\end{figure}

Although the parameters $q$ and $t$ are subject to $t^2q^{k}=-1$, there is one free parameter left.
If $c$ is generic, refined Chern-Simons theory cannot arise from a fusion category due to Ocneanu rigidity (for instance, see \cite{etingof2016tensor}) and, therefore, it is not a modular tensor category ($\MTC$).\footnote{If we further impose the condition that $q$ is a root of unity, the 3d theory on $M_3$ becomes an $\MTC$ \cite{kirillov1996inner}.} Nonetheless, it provides torus link invariants as we will briefly review below.
 In addition, the half index in \eqref{half-index-Wilson} provides the deformation of WRT-invariants of the lens space $L(p,1)$, and moreover $\wh Z_{\cT[L(p,1)],b}(q,t)$ in \eqref{Z-hat} exhibits positivity \cite{Gukov:2016gkn,Gukov:2017kmk}. Despite the failure to be a fusion category, the half indices shed new light on the topology of three-manifolds and link invariants via the 3d/3d correspondence.

\bigskip

The relation between a 3d theory and Conjecture \ref{conj:MTC-D1} is more interesting. It was argued in \cite{Kozcaz:2018usv} that the field identification \eqref{identification} under the condition $t^2q^{2\ell-1}=1$ is equivalent to the class $\cS$ construction for the Argyres-Douglas theory of type $(A_1,A_{2(\ell-1)})$ in \cite{Cecotti:2010fi,Xie:2012hs}, which we briefly review below.
The 4d $\cN=2$ Argyres-Douglas theory of type $(A_1,A_{2(\ell-1)})$ can be geometrically engineered by compactifying two M5 branes on a sphere $C\cong\bC\bfP^1$ with one \emph{wild (irregular) singularity} at infinity.
The theory is specified by the Hitchin system on $C_{\textrm{wild}}$ where the Higgs field has the asymptotic behavior at infinity described by
\be \label{wild-ramification}
\varphi(z_1)dz_1 \sim z_1^{\frac{2\ell-1}2} \sigma_3dz_1~,
\ee
where $z_1$ is the coordinate of $C\backslash \infty$ and $\sigma_3$ is the third Pauli matrix. Thus, we denote this Argyres-Douglas theory by $\cT[C_{\textrm{wild}},\SU(2)]$.
The Hitchin action on the moduli space of Higgs bundles can be identified with the $\U(1)_\beta$ symmetry defined below \eqref{refined-index}
$$
\U(1)_\beta:(A, \varphi) \rightarrow (A, e^{i \theta} \varphi)~.
$$
In the brane setting \eqref{3d3d}, we cannot consider the Hitchin system with \eqref{wild-ramification} on $D^2$ in general. However, when the $\Omega$-deformation parameters are subject to $t^2q^{2\ell-1}=1$, the field identification \eqref{identification} for the Higgs field is consistent along the time circle $S^1$
$$
tq^{\frac{2\ell-1}2}\varphi(x_0+\beta,z_1)=\varphi(x_0+\beta,z_1) \sim \varphi(x_0,z_1)~.
$$
Hence, under \eqref{case3}, it is effectively equivalent to the following brane setting:
\be\label{A1Aeven}
\begin{matrix}
	{\mbox{\rm space-time:}} & \quad & S^1 & \times & T^*C_{\textrm{wild}} & \times  &T^* M_3 \\
	{\mbox{\rm $2$ M5-branes:}} & \quad & S^1 &\times& C_{\textrm{wild}} & \times &  M_3
\end{matrix}
\ee
This system is investigated in detail (including Argyres-Douglas theories of other types) \cite{Fredrickson:2017yka,Fredrickson:2017jcf,Kozcaz:2018usv,Dedushenko:2018bpp}, and remarkably there turns out to be an $\SL(2,\Z)$ representation on the set of connected components of $\U(1)_\beta$ fixed points
\be
\SL(2,\mathbb{Z}) \ \rotatebox[origin=c]{-90}{$\circlearrowright$} \ \Bigl\langle \text{components of $\U(1)_{\beta}$ fixed points in } \MH(C_{\textrm{wild}},G) \Bigr\rangle~.
\ee
Moreover, considering the topologically twisted partition function $Z(S^{1} \times M_{3})$ of the Argyres-Douglas theory $\cT[C_{\textrm{wild}},\SU(2)]$, this $\SL(2,\Z)$ representation can be categorified. Namely, there is a modular tensor category $\MTC[A_1,A_{2(\ell-1)}]$ on $M_3$ whose simple objects are in one-to-one correspondence  with $\U(1)_{\beta}$ fixed points.
In fact, the Argyres-Douglas theory of type $(A_1,A_{2(\ell-1)})$ possesses the discrete global symmetry $\Z_{2\ell+1}$, and if we impose a holonomy $q=e^{-\frac{2\pi\gamma i}{2\ell+1}}$ ($\gamma\in\Z_{2\ell+1}^\times$) of this discrete global symmetry along $S^1$, then the modular matrices in Conjecture \ref{conj:MTC-D1} are those of the corresponding $\MTC[A_1,A_{2(\ell-1)}]$ on $M_3$.
Although the $S$ and $T$ matrices in Conjecture \ref{conj:MTC-D1} satisfy the $\PSL(2,\Z)$ relation even for a generic $q$, the Ocneanu rigidity again forbids them to be those of an $\MTC$. Rather, they connect $\MTC$'s for different values of a holonomy $q=e^{-\frac{2\pi\gamma i}{2\ell+1}}$ with $\gamma\in\Z_{2\ell+1}^\times$ by the one-parameter family with $q$.

When $\gamma=1$, the modular matrices coincide with those of the $(2,2\ell+1)$ Virasoro minimal model \cite{Kozcaz:2018usv,Dedushenko:2018bpp}.
Note that the $(2,2\ell+1)$ Virasoro minimal model is the chiral algebra of the Argyres-Douglas theory of type $(A_1,A_{2(\ell-1)})$ \cite{Cordova:2015nma}.
However, the topologically twisted partition function $Z(S^{1} \times M_{3})$ (therefore $\MTC[A_1,A_{2(\ell-1)}]$) receives the contribution from Coulomb branch operators whereas a vacuum character of the chiral algebra is given by Higgs branch operators \cite{Beem:2013sza}. It is worth noting that there are generally many chiral algebras with the same representation categories \cite{Feigin:2018bkf} so that this coincidence remains very mysterious.  (It is sometimes called ``4d symplectic duality''.)

\subsubsection{\texorpdfstring{$\SU(N)$: higher rank generalization}{SU(N): higher rank generalization}}\label{sec:SUN}

Let us briefly consider a higher rank generalization of the 3d modularity.
The moduli space of flat $G_\bC$-connections over a two-torus $C\cong T^2$ is the quotient space $(T_\bC\times T_\C)/W$ of the product of the two complex maximal tori by the Weyl group. In particular, when $G_\bC=\SL(N,\C)$, the fixed points under the action of the Weyl group $W=\frakS_N$ consist of the center $\Z_N\times\Z_N \subset T_\bC\times  T_\bC$ so that there are $N^2$ torsion points on the moduli space  $\V:=(T\times T)/\frakS_N$ of $\SU(N)$-bundles over a torus. For higher ranks, tame ramifications of Higgs bundles are classified by Levi subgroups of $\SU(N)$ or equivalently partitions of $N$ \cite{Gukov:2006jk}. To obtain the spherical DAHA $\SH(\frakS_N)$ of type $A_{N-1}$ as $\Hom(\Bcc,\Bcc)$, a simple puncture corresponding to the $[1,N-1]$ partition needs to be introduced on $C$. Although we have not understood topology and symplectic geometry of the Hitchin moduli space $\MH(C_p,\SU(N))$ over a torus with a simple puncture (for instance, the number of irreducible components of the global nilpotent cone), we can generalize Conjecture \ref{conj:rcs} and \ref{conj:MTC-D1} to the higher ranks. It is a very interesting problem to generalize the analysis in this paper to arbitrary semi-simple gauge groups.

In refined Chern-Simons theory with $\SU(N)$ gauge group \cite{Aganagic:2011sg}, the parameters $q$ and $t$ are usually expressed in terms of a positive integer $k\in \Z_{>0}$ and the continuous parameter $c$:
\be\label{rCS-variables-N}
q=\exp\Bigl(\frac{\pi i}{k+c\,N} \Bigr)~,\qquad t=\exp\Bigl(\frac{c\, \pi i}{k+c\,N} \Bigr)~,
\ee
so that they are subject to the relation $t^Nq^k=-1$.
Under this condition, the moduli space $\V$ of $\SU(N)$-bundles on $C_p$ is a Lagrangian submanifold in the symplectic manifold $(\MH(C_p,\SU(N)),\omega_\X)$.
As in the $A_1$ case, finite-dimensional representations in the higher rank spherical DAHA $\SH(\frakS_N)$ can be studied by using the raising and lowering operators \cite{KN:1998} in the polynomial representation $\scP$. The Hilbert space $\Hom(\Bcc,\brV)$ of $\SU(N)$ refined Chern-Simons theory is spanned by the basis $P_\lambda$ labeled by Young diagrams $\lambda\subset [k^{N-1}]$ inscribed in the $k\times (N-1)$ rectangle.
The modular action on the Hilbert space is described by $S$ and $T$ matrices of rank $\frac{(N+k-1)!}{(N-1)!k!}$,
\be \label{SUN-modular}
S_{\lambda\mu}=P_{\lambda}(q^{-2\mu}t^{-2\rho})P_{\mu}(t^{-2\rho})~,\qquad T_{\lambda \mu}= \delta_{\lambda \mu}\cdot
q^{\frac{1}{N}|\lambda|^2-||\lambda||^2} t^{||\lambda^{t}||^2-N|\lambda|}~,
\ee
where $\left\Vert \lambda \right\Vert^2 =\sum \lambda_{i}^{2}$, and  $\lambda^{t}$ denotes the transposition of the Young diagram $\lambda$. They indeed compute invariants of a Seifert manifold and a torus link \cite{Aganagic:2011sg,Cherednik:2011nr,cherednik2016daha}. Regarding $P_\lambda(X)$ as an element of $\SH(\frakS_N)$, one can define the invariant of a torus link $T_{m,n}$ by
\be\label{DAHA-Jones}
\theta(\zeta_{m,n}(P_\lambda(X))) ~,\qquad  \zeta_{m,n}=\begin{pmatrix}m&n\\\ast&\ast\end{pmatrix}\in \mathrm{SL}(2,\mathbb{Z})~.
\ee
where $\zeta_{m,n}$ acts projectively on $P_\lambda(X)\in\SH(\frakS_N)$, and $\theta:\SH(\frakS_N)\to \C_{q,t}$ is the evaluation map defined in \eqref{evaluation}. The large $N$ limit is conjectured to be equal to the Poincar\'e polynomial of the HOMFLY-PT homology of a torus link up to a change of variables when colors are labeled by a rectangular Young diagram.

\bigskip

After a simple puncture is added on a two-torus $T^2$, the moduli space becomes smooth and the $N^2$ torsion points turn into the corresponding $N^2$ exceptional divisors. Let us denote them by $\bfD_i^{(N)}$ ($i=1,\ldots,N^2$). They become Lagrangian submanifolds with respect to $\omega_\X$ when $t^N=q^{-M}$, or
\be
q=\exp\left(\frac{2\pi i}{M+cN} \right)~,\qquad t=\exp\left(\frac{2c\pi i}{M+cN} \right)~,
\ee
with coprime $(M,N)$.
In fact, under the shortening condition $t^N=q^{-M}$, there are $N^2$ irreducible $\SH(\frakS_N)$-modules of dimension $\frac{(N+M-1)!}{(N-1)!M!N}$, corresponding to the exceptional divisors.  Among them, only one irreducible component $\bfD_1^{(N)}$ is invariant under $\PSL(2,\Z)$, which is analogous to $\bfD_1$ in the $A_1$ case. We are interested in the modular matrices acting on the corresponding finite-dimensional representation of $\SH(\frakS_N)$.

With the shortening condition $t^N=q^{-M}$, a finite-dimensional module arises as a quotient of the polynomial representation whose basis is spanned by Macdonald polynomials $P_\lambda$ with $\lambda\subset [M^{N-1}]$ inscribed in the $M\times (N-1)$ rectangle. This decomposes into $N$ irreducible modules, and the other $N(N-1)$ irreducible modules can be obtained by their orbits under the symmetry $\Xi\times\PSL(2,\Z)=H^1(C,\Z_N)\times\PSL(2,\Z)$ of $\SH(\frakS_N)$. They correspond to $\Hom(\Bcc,\frakB_{\bfD_i^{(N)}})$. From the brane perspective, the support of the brane of the polynomial representation intersects with the corresponding $N$ exceptional divisors. When $t^N=q^{-M}$, $N$ Macdonald polynomials $P_{\lambda^{(i)}}$ $(i=1, \ldots, N)$ of type $A_{N-1}$, where $\lambda^{(i)}\subset [M^{N-1}]$, are degenerate at each eigenvalue of the Dunkl operator
$$
D(u)=\sum_{r=0}^{n}(-u)^{r} D^{(r)}~, \qquad D^{(r)}=\sum_{\substack{I \subset[1,\ldots ,N]\\ |I|=r}} \prod_{\substack{i \in I\\ j\not\in I}} \frac{t X_{i}-t^{-1}X_{j}}{X_{i}-X_{j}} \varpi_{i} \quad(r=0,1, \ldots, N)~.
$$
Here we write variables of the Macdonald polynomials defined in Appendix \ref{app:poly} as $X_i/X_j:=X^\alpha$ for a root $\alpha=e_i-e_j$ and the $q$-shift operators act as $\varpi_{i} X_j=q^{\delta_{ij}} X_j$. We also note that $D^{(0)}=1=D^{(N)}$.
Out of the $N$ irreducible finite-dimensional modules, only one irreducible representation becomes a $\PSL(2,\Z)$ representation, and its basis is spanned by
\be \label{SUN-base-change}
\Bigl\{\sum_{i=1}^N P_{\lambda^{(i)}}(X)/P_{\lambda^{(i)}}(t^{-\rho})\Bigl\}_{\lambda^{(i)}\subset [M^{N-1}]}~.
\ee
In fact, the modular $S$-matrix $S_{\lambda\mu}$ in \eqref{SUN-modular} becomes of rank $\frac{(N+M-1)!}{(N-1)!M!N}$ with the shortening condition $t^N=q^{-M}$.
As in the $A_1$ case \eqref{base-change-S}, we can make a change of basis to \eqref{SUN-base-change} to obtain a $\frac{(N+M-1)!}{(N-1)!M!N}$-dimensional $\PSL(2,\Z)$ representation on the irreducible $\SH(\frakS_N)$-module explicitly.

By a similar argument to the one above, the fivebrane system \eqref{3d3d} at $t^N=q^{-M}$ is equivalent to the Argyres-Douglas theory of type  $(A_{N-1},A_{M-1})$ \cite{Cecotti:2010fi,Xie:2012hs} on $S^1\times M_3$, which admits a class $\cS$ construction with an $\SU(N)$ Hitchin system on $\CP^1$ with a wild singularity at $z=\infty$ where the eigenvalues of the Higgs field grow as $|\varphi| \sim |z^{M / N} dz|$. Therefore, the modular matrices acting on the module $\Hom(\Bcc,\frakB_{\bfD_1^{(N)}})$ can be understood as those of an $\MTC[A_{N-1},A_{M-1}]$ associated to the $(A_{N-1},A_{M-1})$ Argyres-Douglas theory, which categorifies the $\SL(2,\Z)$ action on fixed points of the $\U(1)_\beta$ action on the corresponding wild Hitchin moduli space \cite{Fredrickson:2017jcf}.
As a higher rank generalization of Conjecture~\ref{conj:MTC-D1}, it is expected that they are related to the modular matrices in the $(N,M+N)$ minimal model of the $W_N$-algebra, which is the chiral algebra of the $(A_{N-1},A_{M-1})$ Argyres-Douglas theory \cite{Cordova:2015nma}. In fact, by normalizing them appropriately with the Macdonald norm  \eqref{Macdonald-norm} of type $A_{N-1}$,
the modular matrices at $q=e^{-2\pi i/(M+N)}$ coincide with those \eqref{WN-MM} of the $W_N(N,M+N)$ minimal model \cite{Beltaos:2010ka}, which are reviewed below.
However, we should keep in mind the same caution as the one given at the end of the previous subsection \S\ref{sec:SU2}.

Remarkably, the space $\Hom(\Bcc,\frakB_{\bfD_1^{(N)}})$ has another intriguing interpretation.
In the limit of the spherical rational Cherednik algebra $\SH^{\mathrm{rat}}_{\hbar,c}(\frakS_{N})$, the target space of the sigma-model becomes the Hilbert scheme of $(N-1)$-points on the affine plane $\C^2$, and the exceptional divisor $\bfD_1^{(N)}$ only remains to be a compact Lagrangian submanifold, called \emph{punctual Hilbert scheme}. (See also Appendix~\ref{sec:3d-N=4}.) It is known that its geometric quantization provides the unique finite-dimensional representation of  $\SH^{\mathrm{rat}}_{\hbar,c}(\frakS_{N})$ \cite{berest2003finite,gordon2005rational,gordon2006rational} and it is furthermore isomorphic to the lowest $a$-degree $\cH_{\textrm{bottom}}(T_{N,M})$ of HOMFLY-PT homology of the $(N,M)$ torus knot $T_{N,M}$ \cite{Gorsky:2012mk}. Thus, we have isomorphisms of the following vector spaces
\be
K^0(\MTC[A_{N-1},A_{M-1}])\cong \Hom(\Bcc,\frakB_{\bfD_1^{(N)}})\cong \cH_{\textrm{bottom}}(T_{N,M})~.
\ee

\bigskip

In what follows, we briefly review the modular matrices of the $W_N(N,M+N)$ minimal model \cite{Beltaos:2010ka}.
These minimal models admit a coset description:
\be\label{coset}
W_N(N,M+N)=\frac{\SU(N)_k\times \SU(N)_1}{\SU(N)_{k+1}}~, \qquad \textrm{with} \quad  k=\frac{N}{M}-N~.
\ee
Therefore, their modular matrices are constructed from those of $\SU(N)_k$ affine Lie algebra \cite{Beltaos:2010ka}. The primary fields in the $\SU(N)_k$ WZW model are classified by
$$
\Phi(N;n):=\Bigl\{\lambda=\left(\lambda_{1}, \ldots, \lambda_{N-1}\right) \in \mathbb{Z}_{>0}^{N-1} \ \vert \   \ \sum_{i=1}^{N-1}\lambda_{i}<n=k+N\Bigr\}
$$
where the vacuum corresponds to $\varrho=(1,\ldots,1)\in\Phi(N;n)$, and the $S$ matrix is given by
$$
S_{\lambda \mu}^{(N ; n)}= \frac{1}{in\sqrt{N}}\exp [2 \pi i\frac{ t(\lambda) t(\mu)}{N n}] \det\Bigl(\exp [-2 \pi i\frac{\lambda[\ell] \mu[m]} {n}]\Bigr)_{1 \leq \ell,m \leq N}
$$
with
$$
\lambda[i]=\sum_{i \leq \ell<N}\left(\lambda_{\ell}+1\right)~, \qquad t(\lambda):=\sum_{j=1}^{N-1} j \lambda_{j}~.
$$
The primary fields of the $W_N(N,M+N)$ minimal model are in one-to-one correspondence with the following set
$$
\Phi[W_N(N,M+N)]=\Bigl\{(\varrho,\lambda)  \ \vert \ \lambda  \in \Phi(N;N+M)~,\ t(\lambda)\equiv 0 \mod N  \Bigr\}
$$
The modular $S$ and $T$ matrices of the  $W_N(N,M+N)$ minimal model  are
\begin{align}\label{WN-MM}
S_{(\varrho,\lambda)(\varrho,\mu)}&=(N (N+M))^{\frac{3-N}2} \exp \left[-2 \pi i \frac{t(\varrho) (t(\mu)+t(\lambda))}{N}\right] S_{\varrho \varrho}^{(N ; N /(N+M))} S_{\lambda \mu}^{(N ;(N+M)/N)}~, \cr
T_{(\varrho,\lambda)(\varrho,\mu)}&=-i\delta_{\lambda\mu}\exp \left[\pi i\frac{(N+M) \varrho-N \lambda)\cdot ((N+M) \varrho-N \lambda)}{(N+M)N}\right]~,
\end{align}
where the inner product is defined by
$$
\lambda \cdot \mu:=\sum_{1 \leq i<N} \frac{i(N-i)}{N} \lambda_{i} \mu_{i}+\sum_{1 \leq i<j<N} \frac{i(N-j)}{N}\left(\lambda_{i} \mu_{j}+\lambda_{j} \mu_{i}\right)~.
$$

\subsection{\texorpdfstring{Relation to skein modules and MTC$[M_3]$}{Relation to skein modules and MTC[M3]}}\label{sec:SL2Z-skein}

In the above discussion, we already encountered the skein modules of 3-manifolds and the algebraic data of line operators $\MTC [M_3,G]$ in 3d $\cN=2$ theory $\cT [M_3,G]$,
$$
\MTC[M_3,G]:=\lat{Line}\bigl[\cT[M_3,G]\bigr]
$$
that also enters ``gluing'' of vertex algebras associated to 4-manifolds \cite{Feigin:2018bkf}, twisted indices of $\cT [M_3,G]$ on general 3-manifolds \cite{Gukov:2016gkn}, and modular properties of $q$-series invariants $\wh Z (M_3)$ \cite{Cheng:2018vpl}.

Since 3d theory $\cT [M_3,G]$ has only $\cN=2$ supersymmetry, it cannot be topologically twisted on a general 3-manifold and, therefore, does not lead to a full 3d TQFT that could have been associated to a tensor category (of its line operators) in a familiar way. Nevertheless, as was pointed out in \cite{Gukov:2016gkn}, the structure of line operators and partially twisted partition functions in $\cT [M_3,G]$ in many ways is close to (and, in some cases, is described by) that of a tensor category. Hence, the name $\MTC [M_3]$, or $\MTC [M_3, G]$. The simple objects of $\MTC [M_3, G]$ are complex $G_{\C}$ flat connections on $M_3$. For example, when $M_3$ is the Poincar\'e sphere and $G = \SU(2)$, there are three simple objects in $\MTC [M_3, G]$ and $K^0 (\MTC [M_3, G])$ has rank 3. In this example, and more generally, when all flat $G_{\C}$-connections on $M_3$ are isolated, they can be identified with the intersection points of two Heegaard branes $\frakB_{H^{\pm}}$ associated with the Heegaard decomposition of $M_3$, illustrated in Figure \ref{fig:Heegaard}.

\begin{figure}[ht]\centering
  \includegraphics[width=\textwidth]{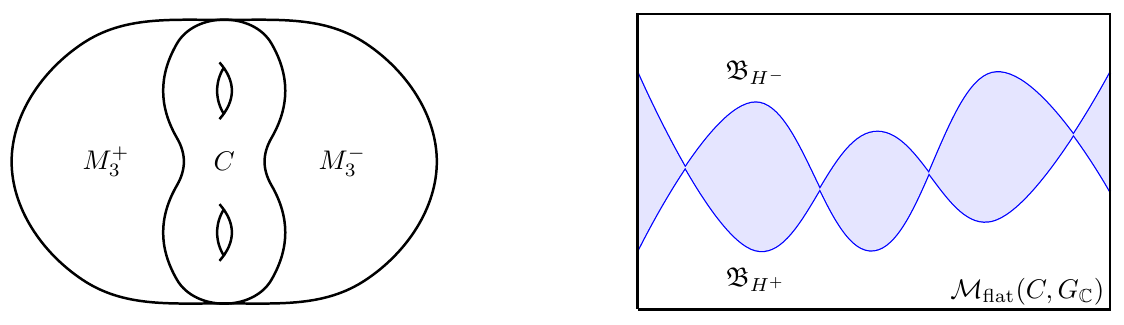}
  \caption{A Heegaard decomposition (left panel) of a closed oriented 3-manifold leads to an interpretation of $K^0 (\MTC [M_3])$ as the space of $(\frakB_{H^{+}},\frakB_{H^{-}})$-strings in $\MF(C,G_\C)$.}\label{fig:Heegaard}
\end{figure}

Specifically, let $M_3= M_3^{+}\cup_C M_3^{-}$ be a Heegaard splitting of a closed oriented 3-manifold $M_3$. As in \eqref{skein-module}, 3-manifolds with boundary $\partial M_3^\pm=C$ define the $(A,B,A)$-branes $\frakB_{H^{\pm}}$ supported on Lagrangian submanifolds $\MF(M_3^\pm,G_\C)$ in $\MF(C,G_\C)$. Hence, $K^0 (\MTC [M_3])$ can be interpreted as the space of open strings between two Heegaard branes $\frakB_{H^{\pm}}$ associated to $M_3^{\pm}$ and illustrated in Figure \ref{fig:Heegaard}.
Furthermore, via a complex analogue of the Atiyah-Floer conjecture (see \textit{e.g.}\ \cite{Gukov:2007ck}), this ring is expected to be isomorphic to the complex $G_\C$ Floer homology $HF^{\textrm{inst}}_{0}(M_3,G_\C)$ of $M_3$:
\be\label{sk-floer}
K^0 (\MTC [M_3]) \cong \Hom^0(\frakB_{H^{+}},\frakB_{H^{-}})\cong HF_0^{\textrm{symp}}(\MF(C,G_\C);H^+,H^-)\cong HF^{\textrm{inst}}_{0}(M_3,G_\C)~.
\ee
Here both symplectic and instanton Floer homology groups are $\Z$-graded, and we take the zeroth degree of the homology groups. Physically, this grading comes from non-anomalous $\U(1)$ R-symmetry.

Indeed, the relevant system here is a stack of M5-branes on $\bR\times T^2\times M_3$, and we are interested in the Hilbert space $\mathcal{H}_{\cT[M_{3} \times T^{2},G]}$. We can interpret this Hilbert space as that of 3d $\cN=2$ theory $\cT[M_{3},G]$ on $T^2$. The Hilbert space is $\Z$-graded by the $\U(1)$ $R$-symmetry of the 3d $\cN=2$ theory.
On the other hand, we can compactify the 6d $\cN=(2,0)$ theory on $T^2$, and perform the topological twist of the 4d $\cN=4$ theory considered in \cite{Yamron:1988qc}. The two types of topological twists of the 4d $\cN=4$ theory in \cite{Yamron:1988qc}, Vafa-Witten twist \cite{Vafa:1994tf} and Marcus/GL-twist \cite{Marcus:1995mq,Kapustin:2006pk}, are equivalent on $\bR\times M_3$, and the BPS equations on $M_3$ are satisfied by complex $G_\C$-flat connections. As a result, the Hilbert space can be understood as complex Floer homology of $M_3$. Consequently, the Hilbert space admits two interpretations \cite{Gukov:2016gkn}:
$$
\mathcal{H}_{\cT[T^2,G]}(M_{3})\cong\mathcal{H}_{\cT[M_{3} \times T^{2},G]}\cong\mathcal{H}_{\cT[M_{3},G]}\left(T^{2}\right)~.
$$
In general, complex Floer homology groups are infinite-dimensional due to the presence of reducible solutions and non-compactness of moduli spaces. Nonetheless, it is graded by the $R$-charges of the 3d $\cN=2$ supersymmetry, and we expect that the zeroth degree piece gives precisely \eqref{sk-floer}.

Note, that for some manifolds, like $M_3 = T^3$, all complex flat connections are reducible. (In this example, simply because $\pi_1 (M_3)$ is abelian.) Such examples illustrate especially well how the infinite-dimensional complex Floer homology of $M_3$ is re-packaged into its finite-dimensional version $K^{0}\left(\MTC[M_{3},G]\right)$. Moreover, half-BPS line operators in $\cT[M_{3},G]$ are in one-to-one correspondence with states of the Hilbert space of $\cT[M_{3},G]$ on $T^2$. The mapping class group of $T^2$ acts on this Hilbert space, justifying the name for $K^0 (\MTC [M_3])$. In practice, this can be a log-modular action, as in \cite{Cheng:2018vpl}.

A somewhat similar ``regularization'' of the complex Floer theory is provided by the skein module $\Sk(M_3,G)$, which was recently shown to be finite-dimensional \cite{Gunningham:2019kac} for any closed oriented 3-manifold $M_3$. Physically, the $\SU(N)$-skein module of $M_3$ is a set of all formal linear combinations of line operators in complex $\SL(N,\C)$ Chern-Simons theory, defined as  \cite{turaev1990conway,przytycki2006skein}:
$$\Sk(M_3,\SU(N)) = \bC[q^\pm] (\textrm{isotopy classes of framed oriented links in }  M_3) / \textrm{skein relations}~.$$
where the skein relations are given by
$$\centering
\includegraphics[width=11cm]{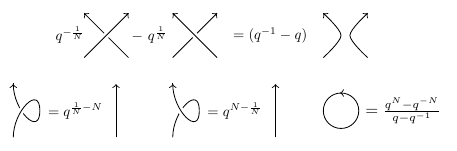}~.
$$
The analogue for Cartan types other than $A$ is not well explored, and would be an excellent direction for future work.

Focusing on $G=\SU(N)$ and $G_\C=\SL(N,\C)$, the above discussion suggests that there may be a relation between $K^{0}(\MTC[M_3,G])$ that describes line operators in $\cT[M_3]$ and the skein module $\Sk(M_3,G)$. This relation cannot be a simple isomorphism because, \textit{e.g.}\ for $M_3 = T^3$ and $G=\SU(2)$, $K^{0}(\MTC[M_3,G])$ has 10 simple objects whereas $\rank \Sk(M_3,G) = 9$ \cite{carrega2017nine,gilmer2018kauffman}. Relegating a better understanding of this relation to future work,\footnote{The above mentioned examples of the Poincar\'e sphere and $M_3 = T^3$ suggest that the general relation for $G=SU(2)$ might be $\rank \Sk(M_3,G) = \rank K^{0}(\MTC[M_3,G]) - 1$. Although we do not know any counterexample to this potential relation, we should stress that the role of ``$-1$'' is likely to be delicate and cannot be simply attributed to, say, reducible flat connections (as in the case of the Poincar\'e sphere). For example, in the case of $M_3 = T^3$, all complex flat connections are reducible, as was already pointed out in the main text.} here we merely conjecture that it commutes with the $\SL(2,\Z)$ action, so that $\Sk(M_3,G)$ also enjoys a (possibly, log-) modular action
$$
\SL(2,\Z) \ \rotatebox[origin=c]{-90}{$\circlearrowright$} \ \Sk(M_3,G)~.
$$

As a next natural step, we now turn our attention to a relation between the skein algebra $\Sk(C)$ of a Riemann surface $C$ and line operators of the 4d $\cN=2$ theory $\cT[C]$, in particular in the case when $C$ is a (punctured) torus.

\section{4d theories, fivebranes, and M-theory}\label{sec:4d}

In this section, we study line operators in the 4d $\cN=2^*$ theories. A 4d $\cN=2$ theory of class $\cS$ arises from a compactification of 6d $\cN=(2,0)$ theory on a once-punctured torus $C$. The spectrum of line operators in the theory depends on additional discrete data, a maximal isotropic lattice $\lat{L}\subset H^1(C,Z(G))$ where line operators must be invariant under the discrete group $\lat{L}$. Therefore, we will show that
a non-commutative algebra of line operators of a 4d $\cN=2^*$ theory on $S^1\times\R\times_q \R^2$ with the $\Omega$-background is the $\lat{L}$-invariant subalgebra of spherical DAHA. Also, we give an explicit geometric relation between Hitchin moduli spaces and an elliptic fibration of the Coulomb branch of a 4d $\cN=2^*$ theory in the rank-one case. Besides, we include a surface operator of Gukov-Witten type in the story, and consider an algebra of line operators on a surface operator to realize the full DAHA instead of the spherical DAHA.
An advantage of the fivebrane system of class $\cS$ is that we can relate line operators of a 4d theory to boundary conditions of a 2d sigma-model by a compactification. Taking this advantage, we propose a canonical coisotropic brane $\wh \frakB_{cc}$ of higher rank which realizes the full DAHA as $\Hom(\wh \frakB_{cc},\wh \frakB_{cc})$.

\subsection{\texorpdfstring{Coulomb branches of 4d $\cN=2^*$ theories of rank one}{Coulomb branches of 4d N=2* theories of rank one}}\label{sec:4dCoulomb}

In this subsection, we study a stack of M5-branes on $C\times S^1 \times \R^3$.
A 4d $\cN=2$ theory of class $\cS$ is constructed by a compactification of the 6d (2,0) theory of type $G$ ($G$ is of Cartan type $ADE$) on a Riemann surface $C$ \cite{Gaiotto:2009we,Gaiotto:2009hg} (generally with punctures) with additional discrete data $\lat{L}$ \cite{Tachikawa:2013hya} (see also \cite{Gukov:2020btk} where such choice is referred to as a ``polarization on $C$''), denoted by $\cT[C,G,\lat{L}]$. The basic information of a theory of class $\cS$ is encoded in a Hitchin system
\begin{equation}
\label{eq:spectral}
\begin{tikzcd}
\Sigma \arrow[r,hook]
& T^*C \arrow[d] \\
& C
\end{tikzcd}\qquad,
\end{equation}
where $\Sigma$ is a Seiberg-Witten curve. The Coulomb branch of the theory on $\R^4$, called the $u$-plane, is an affine space $\cB_u=\bigoplus_{k=1}^r H^0(C,K_C^{\otimes d_k})$ where the exponents $d_k$ depend on $G$. Given a point $u\in \cB_u$, the Seiberg-Witten curve $\Sigma$ is expressed as the characteristic polynomial $\det(xdz-\varphi)=f(x,u(z))$ where $x,z$ are local coordinates of the fiber and base of $T^*C$. To introduce the additional data, we pick a symplectic basis of $H_1(C)$ of $C$ of genus $g$ in terms of intersection numbers
$$(\alpha_1,\ldots,\alpha_g,\beta_1,\ldots,\beta_g)\in H_1(C)~, \qquad \alpha_i\cdot \alpha_j=0=\beta_i\cdot \beta_j~, \ \alpha_i\cdot\beta_j=\delta_{ij}=-\beta_j\cdot\alpha_i~,$$
which yields a symplectic lattice $(H^{1}(C,Z(G)),\omega)$ where $Z(G)$ is the center of $G$.
In fact, the additional data are given by a maximal isotropic sublattice $\lat{L}\subset(H^{1}(C, Z(G)),\omega)$, and they specify an allowed set of charges of line operators that are compatible with the Dirac quantization conditions \cite{Aharony:2013hda}. Given a maximal isotropic sublattice $\lat{L} \subset (H^{1}(C,Z(G)),\omega)$, the Coulomb branch $\cM_C(C,G,\lat{L})$ of the $\cT[C,G,\lat{L}]$ theory on $S^1\times \R^3$ admits an elliptic fibration over the $u$-plane $\cB_u$ \cite{Gaiotto:2010be,Tachikawa:2013hya}
$$
\pi: \cM_C(C,G,\lat{L}) \rightarrow \cB_u~.
$$
This is sometimes called the Donagi-Witten integrable system of class $\cS$. In fact, $H^{1}(C, Z(G))$ freely acts on the Hitchin fibration $\pi:\MH(C,G)\to \cB_H$ fiberwise, and the Coulomb branch can be obtained by the quotient of $\pi:\MH(C,G)\to \cB_H$ by a fiberwise action of $\lat{L}$  so that
$$\cM_C(C,G,\lat{L})=\MH(C,G)/\lat{L} ~.$$
Note that this action can be obtained by twists of a Higgs bundle by a flat $Z(G)$-bundle over $C$ associated to $\lat{L}$, and it acts freely on a generic fiber of the Hitchin fibration.
Therefore, the Coulomb branch $\cM_C(C,G,\lat{L})$ inherits a \HK structure from the Hitchin moduli space  $\MH(C,G)$.

Of our interest are certainly the class $\cS$ theories of type $A_1$ associated to the once-punctured torus $C_p$, namely 4d $\cN=2^*$ theories of rank one \cite{Gorsky:1993dq,Donagi:1995cf}. The lattice $H^{1}(C_p,\Z_2)=\Z_2\oplus\Z_2$ with the natural symplectic form can be identified with the electric and magnetic charges of line operators of the $\cN=2^*$ theory wrapping $S^1$. Line operators with charges $\lambda=(\lambda_{e}, \lambda_{m})$ and $\nu=(\nu_{e}, \nu_{m})$ must be subject to the Dirac quantization condition
$$
\omega(\lambda,\nu)=\lambda_e\nu_m-\lambda_m\nu_e \in 2 \mathbb{Z} ~.
$$
There are three ways to pick a maximal isotropic lattice, corresponding to (0,1) (1,0) and (1,1) $\in H^{1}(C_p,\Z_2)=\Z_2\oplus\Z_2$.
They are known as $\SU(2)$, $\SO(3)_{+}$ and $\SO(3)_{-}$ gauge theories, respectively, where the theta angles of $\SO(3)_\pm$ differ by $2\pi$.
Under the  $\SL(2,\Z)$ transformation on the complexified gauge coupling (electromagnetic duality), these theories are related to each other as follows:
\bea\label{N=4-SL2Z}
\begin{tikzpicture}
  \node (A)  at (0,1.6) {$\SU(2)$};
    \node (B) at (-1.5,0) {$\SO(3)_+$};
  \node (C)  at (1.5,0)  {$\SO(3)_-$};
  \draw[<->] (A)  to node [left,scale=0.8,blue] {$S$} (B);
  \draw[<->] (B)  to node [below,scale=0.8,blue] {$T$} (C);
    \draw[<->] (C)  to node [right,scale=0.8,blue] {$TST$} (A) ;
\path[->] (A) edge  [loop above] node [scale=0.8,blue]{$T$} (A);
\path[->] (B) edge  [out=175,in=185,loop] node [left,scale=0.8,blue] {$TST$} ();
\path[->] (C) edge  [out=5,in=-5,loop] node[right,scale=0.8,blue] {$S$} ();
\end{tikzpicture}
\eea

\begin{figure}\centering
    \includegraphics[width=\textwidth]{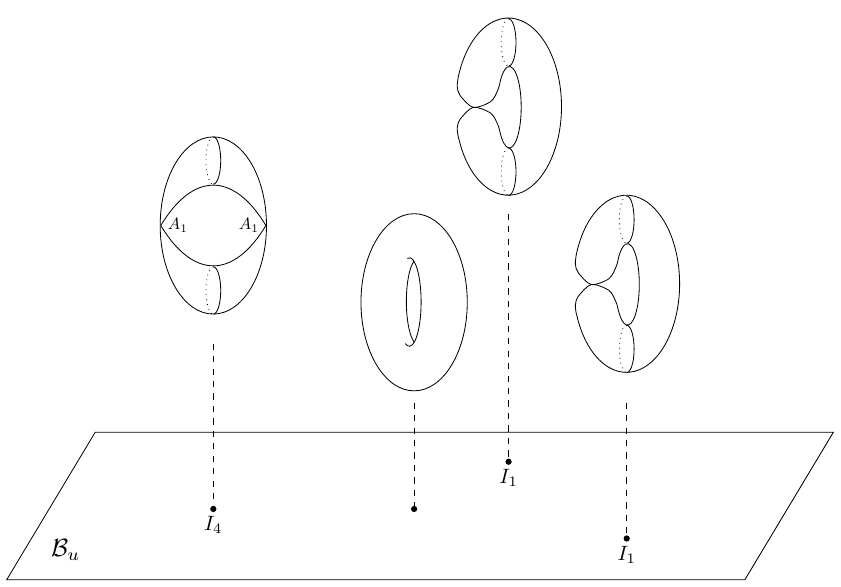}
    \caption{Schematic illustration of elliptic fibration of Coulomb branch $\cM_C(C_p,\SO(3)_+)\to \cB_u$}
    \label{fig:CB}
\end{figure}

\begin{figure}\centering
    \includegraphics[width=\textwidth]{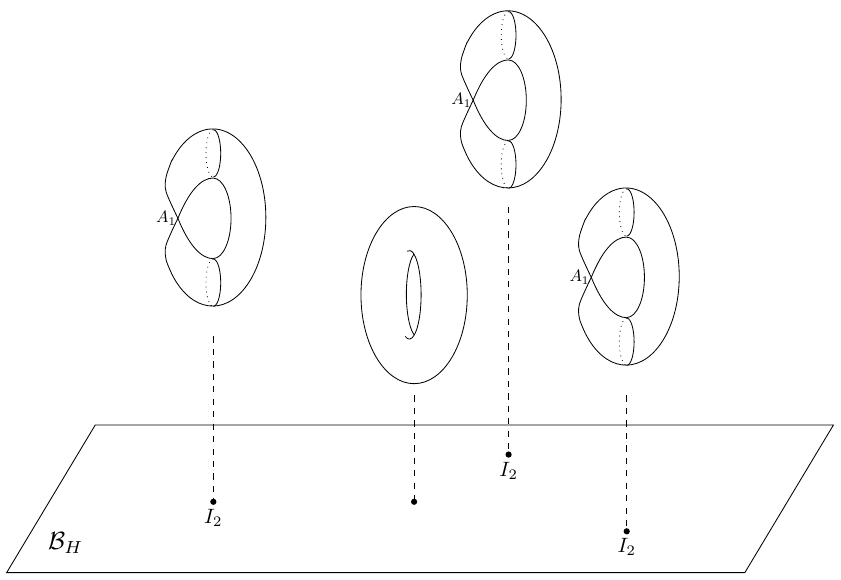}
    \caption{Schematic illustration of the Hitchin fibration of $\MH(C_p,\SO(3))\to \cB_H$}
    \label{fig:MHSO3}
\end{figure}

\bea\label{Z2-Z2-quotient}
\begin{tikzpicture}
  \node (A)  at (0,1.5) {$\MH(C_p,\SU(2))$};
    \node (C)  at (0,0) {$\cM_C(C_p,\SO(3)_+)$};
      \node (B)  at (-4,0) {$\cM_C(C_p,\SU(2))$};
        \node (D)  at (4,0) {$\cM_C(C_p,\SO(3)_-)$};
          \node (E)  at (0,-1.5) {$\MH(C_p,\SO(3))$};
          \draw [->] (A) -- node [inner sep=5mm,left,blue] {$\xi_2$} (B);
          \draw [->] (A) -- node [left,blue] {$\xi_1$} (C);
          \draw [->] (A) -- node [inner sep=5mm,right,blue] {$\xi_3$} (D);
          \draw [->] (B) -- node [inner sep=5mm,left,blue] {$\xi_1$} (E);
          \draw [->] (C) -- node [left,blue] {$\xi_2$} (E);
          \draw [->] (D) -- node [inner sep=5mm,right,blue] {$\xi_1$} (E);
\end{tikzpicture}
\eea

Next, we study the geometry of the Coulomb branches of the 4d $\cN=2^*$ theories of rank one on $S^1\times \R^3$.
The Coulomb branches can be obtained by $\Z_2$ quotients of the Hitchin moduli space $\cM_{H}(C_p,\SU(2))$ by $\xi_{i}\in\Xi=H^1(C_p,\Z_2)$ ($i=1,2,3$) \cite[\S8.4]{Gaiotto:2010be} as in \eqref{Z2-Z2-quotient}.
The ramification parameters ${\frac{1}{2}}(\tbeta_p+i\tgamma_p)$ at the Higgs field $\varphi$ is indeed equivalent to the complex mass of the adjoint hypermultiplet in the 4d $\cN=2^*$ theory. The ramification parameter  $\alpha_p$ is the holonomy along $S^1$ for the $\U(1)$ flavor symmetry.
Let us investigate the action of $\Xi$ on the Hitchin moduli space $\cM_{H}(C_p,\SU(2))$ at a generic ramification more in detail. As in Figure \ref{fig:3I_2}, $\cM_{H}(C_p,\SU(2))$ with a generic ramification has three singular fibers of Kodaira type $I_2$. As described in \S\ref{sec:target}, the action of $\Xi$ on each fiber in the Hitchin fibration $\MH(C_p,\SU(2))\to \cB_H$ is of order two, and the action is moreover free on a generic fiber. Hence, an interesting part is the action on the singular fibers.
Two irreducible components, $\bfU_{2i-1}$ and $\bfU_{2i}$, in the singular fiber $\pi^{-1}(b_i)$ can be understood as two $\CP^1$'s meeting at the north and south pole as double points.
As illustrated in Figure \ref{fig:I_2}, the element $\xi_{1}$ in \eqref{sign-changes} acts on each irreducible component of the singular fiber $\pi^{-1}(b_1)$ as the $180^{\circ}$ rotation around the polar axis of $\CP^1$. Likewise, $\xi_{1}$  acts on a generic fiber $\bfF$ nearby as the $180^{\circ}$ rotation along the $(1,0)$-cycle (meridian) of $\bfF$. As we have seen in \S\ref{sec:Ui}, the singular fiber $\pi^{-1}(b_1)$ is mapped to $\pi^{-1}(b_2)$ by the modular $S$-transformation $\sigma$. Therefore, in the neighborhood of the singular fiber $\pi^{-1}(b_2)$, $\xi_{1}$  acts on a generic fiber $\bfF$ as the $180^{\circ}$ rotation along the $(0,1)$-cycle (longitude) of $\bfF$. Consequently, $\xi_{1}$ exchanges the two irreducible components $\bfU_3$ and $\bfU_4$ by the corresponding rotation on the singular fiber $\pi^{-1}(b_2)$. In a similar fashion, $\tau_+\in \PSL(2,\Z)$ maps the singular fiber $\pi^{-1}(b_1)$ to the other fiber $\pi^{-1}(b_3)$. Therefore, $\xi_{1}$  acts on a generic fiber $\bfF$ as the $180^{\circ}$ rotation along the $(1,1)$-cycle of $\bfF$ around the singular fiber $\pi^{-1}(b_3)$. Moreover, it exchanges the two irreducible components $\bfU_5$ and $\bfU_6$ with additional rotation around the polar axis of $\CP^1$.
The actions of $\xi_{2}$ and $\xi_{3}$ are obtained by the cyclic permutations of $b_{i}$ ($i=1,2,3$).

Since $\xi_i$ acts freely on a generic Hitchin fiber with order two, the quotient of the Hitchin fibration $\MH(C_p,\SU(2))\to \cB_H$ by $\xi_i$ provides the structure of an elliptic fibration of the Coulomb branch. Namely, this double cover is obtained by an isogeny of each elliptic fiber of degree two \cite{Argyres:2015ffa,Argyres:2015gha}.
As illustrated in Figure \ref{fig:I_2}, $\xi_1$ acts on each irreducible component of the singular fiber $\pi^{-1}(b_1)$ by the $180^{\circ}$ rotation so that the quotient by its action turns the double points $p_{1,2}$ into the $A_1$ orbifold points. In fact, the quotient can be understood as a particular limit of the fiber of Kodaira type $I_4$. Generically, the fiber of type $I_4$ consists of four $\CP^1$'s joining like a necklace, or the affine $\widehat A_3$ Dynkin diagram.
The quotient is indeed the zero-volume limit of the two disjoint $\CP^1$'s as in Figure \ref{fig:blowup}.
On the other hand, the quotient of the other singular fibers $\pi^{-1}(b_{2,3})$ by $\xi_1$ identifies the two irreducible components and the two double points by the rotation, yielding the fiber of Kodaira type $I_1$.
Again, the quotients of $\xi_{2}$ and $\xi_{3}$ are obtained by the cyclic permutations of $b_{i}$ ($i=1,2,3$).
As a result, the quotient of the Hitchin moduli space  $\MH(C_p,\SU(2))$ by $\xi_i$ leads to an elliptic fibration $\cM_C\to \cB_u$ of the Coulomb branch with one singular fiber of type $I_4$ at $b_i\in \cB_u$ and two singular fibers of type $I_1$ at $b_{i+1},b_{i+2}\in \cB_u$ ~\cite{Argyres:2015ffa,Argyres:2015gha} as illustrated in Figure \ref{fig:CB}. Hence, an $\cN=2^*$ theory of rank one enjoys a subgroup of $\SL(2,\Z)$ that fixes the singular fiber of type $I_4$. One can easily read off such a subgroup from \eqref{PSL-Ui} that is consistent with a duality group  \eqref{N=4-SL2Z} of an $\cN=2^*$ theory. Note that $\tau_+$ and $\tau_-$ correspond to the $T$ and $TST$ elements, respectively, of the electromagnetic duality of the 4d $\cN=2^*$ theories of rank one, which is different from the matrix assignment in \eqref{tau-pm}.
\begin{figure}[ht]\centering
\includegraphics[width=\textwidth]{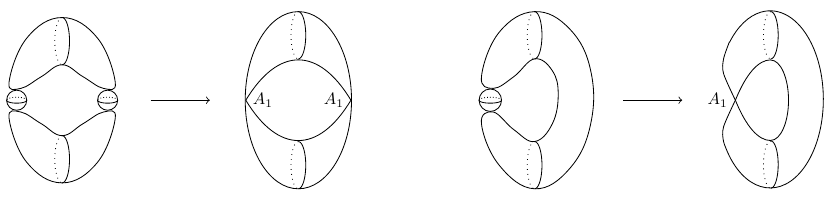}
\caption{(Left) The Coulomb branch $\cM_C$ contains a particular limit of the fiber of type $I_4$ so that there are two $A_1$ singularities. (Right) Each singular fiber of the Hitchin fibration $\MH(C_p,\SO(3))\to \cB_H$ with a generic ramification is a certain limit of the fibers of type $I_2$.}\label{fig:blowup}
\end{figure}

So far we have studied the Coulomb branches with generic ramification parameters $(\talpha_p,\tbeta_p,\tgamma_p)$. When $\tbeta_p=0=\tgamma_p$, the Hitchin fibration $\MH(C_p,\SU(2))\to \cB_H$ has one singular fiber of type $I_0^*$ at the global nilpotent cone, and it is easy from Figure \ref{fig:Hitchin-fibration} to see the quotient by $\xi_i$. For instance, at generic values of $\talpha_p$, the quotient of the global nilpotent cone by $\xi_1$ again leads to the singular fiber of type $I_0^*$, but the volumes of $\D_2$ and $\D_4$ shrink to zero in this case. Therefore, it has two $A_1$ singularities. Similarly, the quotient by another generator $\xi_i$ can be obtained by exchanging non-trivial exceptional divisor $\D_3$ to another one ($\D_2$ ($\xi_2$) or $\D_4$ ($\xi_3$)).

As in \eqref{Z2-Z2-quotient}, the moduli space $\MH(C_p,\SO(3))$ of $\SO(3)$-Higgs bundles over $C_p$ can be obtained by the further quotient of the Coulomb branch by the other generator of $\Xi$. By the further quotient, the two irreducible components and the two $A_1$ singular points are identified in the singular fiber of type $I_4$, and the quotient of each singular fiber of type $I_1$ by the $180^{\circ}$ rotation of around the polar axis turns the double point into the $A_1$ singularity. As a result, all the singular fibers of  $\MH(C_p,\SO(3))\to \cB_H$ can be understood as the limit of a singular fiber of type $I_2$ in which one of the irreducible components shrinks to zero as in Figure \ref{fig:MHSO3}.
When $\tbeta_p=0=\tgamma_p$, the global nilpotent cone is again the singular fiber of type $I_0^*$, but it has three $A_1$ singularities for generic values of $\talpha_p$ \cite[Figure 1]{Gukov:2010sw}.

\subsection{Algebra of line operators}\label{sec:line}

It is known \cite{Gaiotto:2009hg,Gaiotto:2010be} that loop operators along $S^1$ in a 4d $\cN=2$ theory $\cT(C,G,\lat{L})$ on $S^1\times \R^3$ form a commutative algebra that is the coordinate ring $\OO(\cM_C(C,G,\lat{L}))$ of the Coulomb branch holomorphic in complex structure $J$.
Once we introduce the $\Omega$-background $S^1\times \bR\times_q \R^2$, loop operators wrapped on $S^1$ are localized on the axis of the $\Omega$-deformation as depicted in Figure \ref{fig:omega-background}. Consequently, they are forced to come across each other as they exchange their positions, which yields non-commutative deformation of the algebra \cite{Nekrasov:2010ka,Ito:2011ea,Yagi:2014toa,Bullimore:2016nji,Dedushenko:2018icp,Okuda:2019emk}. Thus, an algebra of line operators of a 4d $\cN=2$ theory on the $\Omega$-background provides the deformation quantization $\OO^q(\cM_C)$ of the coordinate ring of its Coulomb branch, a.k.a. \emph{quantized Coulomb branch}.
\begin{figure}[ht]\centering
\includegraphics[width=14cm]{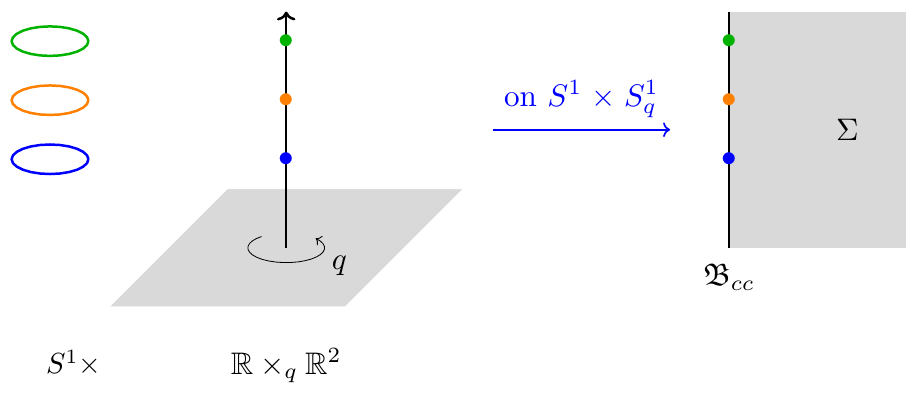}
\caption{An algebra of line operators (colored circles) in a 4d $\cN=2$ theory becomes non-commutative in the $\Omega$-background $S^1\times \bR\times_q \R^2$, which provides deformation quantization of holomorphic coordinate ring of the Coulomb branch. The 4d $\cN=2$ theory compactified on $S^1\times S^1_q$ is described by 2d $A$-model  $\Sigma\to \cM_C$ on the Coulomb branch where the boundary condition at $\partial \Sigma$ is given by $\frakB_{cc}$. Here $\R^2 \supset S^1_q$ is the circle generating the $\Omega$-deformation.}\label{fig:omega-background}
\end{figure}

Now we are ready to discuss quantized Coulomb branches of the 4d $\cN=2^*$ theories of rank one and their relation to spherical DAHA $\SH$. Once we specify a maximal isotropic lattice $\lat{L}\subset H^2(C,Z(G))$, we can read off charges of
line operators in a 4d $\cN=2$ theory subject to the Dirac quantization condition. The cases of the 4d $\cN=2^*$ theories of rank one have been studied in detail \cite{Gaiotto:2010be,Aharony:2013hda}. In fact, the generators $x,y,z$ of $\SH$ in \eqref{SH-gen} correspond to the minimal Wilson (1,0), 't Hooft (0,1) and dyonic  (1,1) line operator, respectively. Therefore, it is natural to expect that the relations of quantized Coulomb branches to $\SH$ are as follows:
\begin{itemize}
\item The $\SU(2)$ theory has line operators of charge $(\lambda_{e}, \lambda_{m})$ with $\lambda_{e} \in \mathbb{Z}, \lambda_{m} \in 2 \mathbb{Z}$, including a Wilson operator with the fundamental representation. Therefore, the quantized Coulomb branch is isomorphic to the $\xi_2$-invariant subalgebra of $\SH$
\be\OO^q(\cM_C(C_p,\SU(2)))\cong \SH^{\xi_2}\ee
generated by
\be\label{SU2-generator} x=(X+X^{-1})\mathbf{e}~, \qquad y^2-1=(Y^2+1+Y^{-2})\mathbf{e}~.\ee
\item The $\SO(3)_+$ theory has line operators of charge $(\lambda_{e}, \lambda_{m})$ with $\lambda_{e} \in 2 \mathbb{Z}, \lambda_{m} \in \mathbb{Z}$, including an 't Hooft operator with the fundamental representation. Therefore, the quantized Coulomb branch is isomorphic to the $\xi_1$-invariant subalgebra of $\SH$
\be\OO^q(\cM_C(C_p,\SO(3)_+))\cong \SH^{\xi_1}\ee
generated by
\be\label{SO3+generator} x^2-1=(X^2+1+X^{-2})\mathbf{e}~, \qquad y=(Y+Y^{-1})\mathbf{e}~.\ee
\item The $\SO(3)_-$ theory has line operators of charge $(\lambda_{e}, \lambda_{m})$ with $\lambda_{e}, \lambda_{m} \in \mathbb{Z}$ such that $\lambda_{e}+\lambda_{m} \in 2 Z$, including a minimal dyonic operator $(\lambda_{e}, \lambda_{m})=(1,1)$.  Therefore, the quantized Coulomb branch is isomorphic to the $\xi_3$-invariant subalgebra of $\SH$
\be\OO^q(\cM_C(C_p,\SO(3)_-))\cong \SH^{\xi_3}\ee
generated by
\be x^2-1=(X^2+1+X^{-2})\mathbf{e}~, \qquad z=(q^{-1/2} Y^{-1}X + q^{1/2} X^{-1}Y)\mathbf{e}~.\ee
\end{itemize}

To see the connection to a 2d sigma-model in \S\ref{sec:2d}, we can employ a trick similar to Figure \ref{fig:6d-to-sigma}. Namely, we can compactify a 4d $\cN=2$ theory on $T^2\cong S^1\times S^1_q$ as illustrated in Figure \ref{fig:omega-background}, which leads to 2d $A$-model $\bR\times \R_+\cong \Sigma \to \cM_{C}(C,G,\lat{L})$ on the Coulomb branch. Here $S^1_q\subset \R^2$ is the circle around the axis of the $\Omega$-background. As a result, the axis of the $\Omega$-background on which loop operators meet each other becomes the boundary $\partial \Sigma$. Therefore, the boundary $\partial \Sigma$ should give rise the quantized Coulomb branch $\OO^q(\cM_C)$ so that the canonical coistoropic boundary condition $\frakB_{cc}$ naturally shows up at $\partial \Sigma$ \cite{Nekrasov:2010ka}. By the state-operator correspondence, a loop operator in the 4d $\cN=2$ theory becomes a state in $\Hom(\Bcc,\Bcc)$ up on the compactification.
In this way, $\Bcc$ arises from ``the axis of the $\Omega$-deformation'' (or a tip of a cigar as in \cite{Nekrasov:2010ka}).

We have seen that the $\SU(2)$ and $\SO(3)_\pm$ theories are related by $\PSL(2,\Z)$ so that the quantized Coulomb branches are indeed isomorphic. We expect the conjectural functor \eqref{eq:functor} exists even when $\X=\cM_C(C_p,\SU(2),\lat{L})$ are the Coulomb branches of the 4d $\cN=2^*$ theories of rank one. Thus, we can compare the representation category $\Rep(\OO^q(\X))$ of the quantized Coulomb branch with the $A$-brane category $\ABrane(\X,\omega_\X)$ of the Coulomb branch as in \S\ref{sec:2d}. In fact, we can construct a polynomial representation of a quantized Coulomb branch by using the two generators, and finite-dimensional modules can be obtained by quotients of the polynomial representation under the corresponding shortening conditions. The geometry of the Coulomb branches is explored in the previous section, and we confirm that there is a one-to-one correspondence between finite-dimensional modules of the quantized Coulomb branch and compact $A$-branes in the Coulomb branch as in \S\ref{sec:2d}.

For illustrative purposes, let us briefly study representation theory of the quantized Coulomb branch $\OO^q(\cM_C(C_p,\SO(3)_+))\cong \SH^{\xi_1}$. As in \eqref{SO3+generator}, it is generated by $x^2$ and $y$. Consequently, the polynomial representation of $\SH$ splits into the $\pm$ eigenspaces of the $\Z_2$ action $\xi_1:X\to -X$ as the $\SH^{\xi_1}$-modules for generic $(q,t)$. Therefore, $\SH^{\xi_1}$ acts on $\wt\scP:=\C_{q,t}[X^{\pm2}]^{\Z_2}$ (resp. $\wt\scP:=(X+X^{-1})\C_{q,t}[X^{\pm2}]^{\Z_2}$) under the polynomial representation, which is spanned by Macdonald polynomials of even (resp.  odd) degrees. We can define a raising and lowering operator as in \eqref{RL} with these generators
\bea\label{RL-tilde}
\wt\sfR_j&:=q^{j-1}t\frac{q^{-1}(x^2-1)y-qy(x^2-1)}{q^2-q^{-2}}+ \frac{q^{2j}(q^2-t^2)(1-t^2)}{1-q^{2+2 j} t^{2}}~,\cr
\wt\sfL_j&:=q^{-1-j}t^{-1}\frac{q^{-1}y(x^2-1)-q(x^2-1)y}{q^2-q^{-2}}-\frac{(q^2-t^2)(1-t^2)}{t^2(q^2-q^{2j} t^2)}~.
\eea
They act on Macdonald polynomials as
\begin{align}
\pol(\wt\sfR_j)\cdot P_{j}(X;q,t)&= (1-q^{2 j-1} t^{2}) P_{j+2}(X;q,t)~,\label{raising2} \\
\pol(\wt\sfL_j)\cdot P_{j}(X;q,t)&= -\frac{q^{-2j}(1-q^{2j})(1-q^{2(j-1)})(1-q^{2(j-2)} t^{4})(1-q^{2(j-1)} t^{4})}{t^2(1-q^{2(j-1)} t^2)^{2}(1-q^{2(j-2)} t^2)}P_{j-2}(X;q,t)~.\label{lowering2}
\end{align}
Thus, using these operations, we can study finite-dimensional representations as quotients of the polynomial representation $\wt\scP$ of $\SH^{\xi_1}$ as in \S\ref{sec:finite-rep}. Focusing on the polynomial representation of even degrees, when $q$ is a $2n$-th root of unity (\textit{i.e.}\ $q^{2n}=1$), the ideal $(X^{2n}+X^{-2n}-x_{2n}-x_{2n})$ is invariant under the action of $\SH^{\xi_1}$ since the $q$-shift operator $\varpi$ acts trivially on $X^{\pm2n}$. Consequently, the quotient by this ideal yields an $n$-dimensional finite-dimensional representation
\be
\wt\scF^{x_{2n},+}_{n}:=\wt\scP/(X^{2n}+X^{-2n}-x_{2n}-x_{2n}) ~.
\ee
An analogous representation can be obtained from the polynomial representation of odd degrees.
They correspond to branes supported on a generic fiber with prescribed holonomy in $\cM_C(C_p,\SO(3)_+)$. Comparing \eqref{Fxm} at $m=2n$, the dimension is a half because it splits into the $\pm$ eigenspaces of the $\Z_2$ action $\xi_1$. Furthermore, when $n$ is even $n=2p$, we have another finite-dimensional representation because one lowering operator \eqref{lowering2} becomes null:
\be
\wt\scU_{p}:=\wt\scP/(P_{n}) ~.
\ee
This corresponds to a brane supported on an irreducible component in the singular fiber $\pi^{-1}(b_1)$ in Figure \ref{fig:CB}. The dimension is a half of \eqref{Un} due to the $\xi_1$ invariance. Under the condition, we have the short exact sequence analogous to \eqref{SES4}
\begin{equation}
0\to \wt\scU_{p} \to \wt\scF^{-,+}_{n} \to \xi_2(\wt\scU_{p}) \to 0~.
\end{equation}
This can be interpreted as a bound state formed by the branes supported on the two irreducible components at the singular fiber $\pi^{-1}(b_1)$ in Figure \ref{fig:CB}. Similarly, we can obtain finite-dimensional irreducible representations analogous to $\scV$ in \eqref{rCS-Hilbert-space} and $\scD$ in \eqref{D-basis} under the same shortening conditions where the dimensions are halved, respectively.

Note that the deformation quantization of the holomorphic coordinate ring of the Hitchin moduli space $\MH(C_p,\SO(3))$ can be understood as the spherical subalgebra of ``$\SO(3)$ DAHA''. It is isomorphic to the $\Xi$-invariant subalgebra of $\SH$
$$\OO^q(\MH(C_p,\SO(3)))\cong\SH^\Xi $$
generated by $x^2-1$ and $y^2-1$. Since the weight and root lattices of $\SU(2)$ and $\SO(3)$ are related by $\sfQ(\SU(2))=\sfP(\SO(3))\subset \sfP(\SU(2))=\sfQ(\SO(3))$, this is consistent with the construction of DAHA from the symplectic lattice $(\sfP\oplus \sfP^\vee,\omega)$ described at the beginning of \S\ref{sec:2d} and the deformation quantization of the moduli space $\MF(C_p,\PSL(2,\C))$ of flat $\PSL(2,\C)$-connections. Again, we can study representation theory of the $\SO(3)$ DAHA from the perspective of brane quantization though the detail is omitted here.

\bigskip

There is yet another way \cite{Kapustin:2006pk,Kapustin:2007wm} to connect the 4d $\cN=4$ theory to a 2d sigma-model and to see a category of line operators in the 4d $\cN=4$ theory. (See also \cite{Dimofte:2019zzj} for a similar analysis in 3d.) Let us consider a line operator supported on $\bR\times \textrm{pt} \subset \bR\times \R^3$ in the 4d $\cN=4$ theory with gauge group $G$. Then, the neighborhood around the line operator at $\textrm{pt}\in \R^3$ consists of two disks glued along with a punctured disk $\raisebox{-.1cm}{\includegraphics[width=1cm]{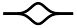}}=\bC\cup_{\C^\times} \bC$, called a ``raviolo'' \cite{Bullimore:2016hdc,Nakajima:2017bdt}. The ``effect'' of a line operator is measured by the modification of field configurations from one disk $\C$ to the other disk $\C$, namely a \emph{Hecke modification}. The compactification of the 4d $\cN=4$ theory on the raviolo leads to a 2d sigma-model on the Hitchin moduli space $\MH(\raisebox{-.1cm}{\includegraphics[width=1cm]{Raviolo}},G)$ of the raviolo, and a line operator gives rise to a boundary condition of the worldsheet as depicted in Figure \ref{local-line}. It was shown \cite{Kapustin:2006pk} that a Wilson operator provides a boundary condition of type $(B,B,B)$ whereas an 't Hooft operator gives a boundary condition of type $(B,A,A)$ in $\MH(\raisebox{-.1cm}{\includegraphics[width=1cm]{Raviolo}},G)$.  Since a boundary condition for a Wilson operator is holomorphic $(B,B,B)$ in every complex structure on $\MH(\raisebox{-.1cm}{\includegraphics[width=1cm]{Raviolo}},G)$, its fusion with another line operator preserves all supersymmetry. Taking into account the fact that the fusion of a Wilson and an 't Hooft operator leads to a dyonic operator, a dyonic operator hence gives rise to a brane of type $(B,A,A)$ in 2d sigma-model on $\MH(\raisebox{-.1cm}{\includegraphics[width=1cm]{Raviolo}},G)$. Thus, upon the compactification, line operators in the 4d $\cN=4$ theory all become $B$-branes of type $I$ on $\MH(\raisebox{-.1cm}{\includegraphics[width=1cm]{Raviolo}},G)$, and an algebra structure can be defined by the convolution product of $B$-branes.

\begin{figure}[ht]
  \centering\includegraphics[width=\textwidth]{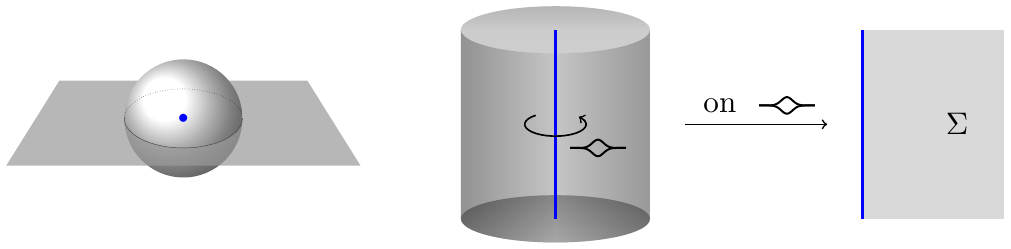}
  \caption{(Left) The  neighborhood around a line operator at $\textrm{pt}\in \R^3$ is a ``raviolo''. (Right) A line operator (blue) gives rise to a boundary condition in the 2d sigma-model upon the compactification of the 4d $\cN=4$ theory on the raviolo.}
  \label{local-line}
\end{figure}

To formulate this idea into a mathematical model \cite{garland1995affine,vasserot2005induced,bezrukavnikov2005equivariant}, let us first consider the moduli space of $G$-bundles over the raviolo. Since the coordinate ring of $\bC$ is  the formal power series ring  $\OO:=\bC\llbracket z\rrbracket$ and that of $\C^\times$ is its
 field $\scK:=\bC(\!( z )\!)$ of fractions, the moduli space of $G$-bundles over $\raisebox{-.1cm}{\includegraphics[width=1cm]{Raviolo}}$ can be expressed as a double coset model, namely the space $G_\C^\scK :=G_\bC(\!( z )\!)$ of transition functions over the punctured disk $\C^\times$ modulo the spaces of gauge transformations $G_\C^\OO :=G_\bC\llbracket z\rrbracket$ over each $\bC$:
\be\label{BunGP1}
\BunG(\raisebox{-.1cm}{\includegraphics[width=1cm]{Raviolo}})=G_\C^\OO\backslash\, G_\C^\scK \, /G_\C^\OO~.
\ee
In fact, if we take only the right quotient by the gauge transformation, the resulting space $\Gr(G_\C):= G_\C^\scK/G_\C^\OO$ is called the \emph{affine Grassmannian}.

To consider the Hitchin moduli space, we need to introduce the Higgs field. This can be achieved by considering the \emph{affine Grassmannian Steinberg variety}
\begin{equation}\label{Grassmannian-Steinberg}
    \cR = \{ (x,[g])\in \frakg_\C^\OO\times \Gr(G_\C) \mid
    \operatorname{Ad}_{g^{-1}}(x) \in \frakg^\OO_\bC\}~.
\end{equation}
The quotient $G_{\mathbb{C}}^{\OO}\backslash\cR$ is the moduli space of a pair of $G$-bundles and sections of its adjoint associated bundle over the raviolo, which can be regarded as the mathematical model of $\MH(\raisebox{-.1cm}{\includegraphics[width=1cm]{Raviolo}},G)$.
Taking the $B$-model viewpoint in complex structure $I$, the category $\lat{Line}$ of line operators in the 4d $\cN=4$ theory with gauge group $G$ and zero theta angle is equivalent to the derived category of $G_{\mathbb{C}}^{\OO}$-equivariant coherent sheaves on $\cR$
$$
\lat{Line}\bigl[\cT[C,G,\lat{L}]\bigr]\cong D^b\mathsf{Coh}^{G_\C^\OO}(\cR)~,
$$
where a maximal isotropic lattice $\lat{L}$ is chosen in such a way that the theta angle is zero.
For instance, it is easy to see that it automatically incorporates the category of Wilson operators as
$$
\mathsf{Coh}^{G_\C^\OO}(\textrm{pt})\cong\mathsf{Coh}^{G}(\textrm{pt})\cong\mathsf{Rep}(G)~.
$$
By taking the Grothendieck ring of this category, we obtain the algebra of line operators in the 4d $\cN=4$ theory \cite{bezrukavnikov2005equivariant}
\be\label{N=4-K}
K^{G_\C^\OO}(\cR)\cong \bC[T_\bC\times T_\C^\vee]^W~,
\ee
which is indeed isomorphic to the coordinate ring of the Coulomb branch
\begin{equation}\label{4d-N=4-MC}
\cM_C(C,G,\lat{L})=\frac{T_\bC\times T_\C^\vee}{W}
\end{equation}
of the 4d $\cN=4$ theory on $S^1\times \R^3$ holomorphic in complex structure $J$ \cite{garland1995affine,vasserot2005induced,bezrukavnikov2005equivariant}. The coordinates $\C[T_\bC]^W$ and $\C[T_\C^\vee]^W$ are spanned by Wilson and 't Hooft line operators, respectively. It is important to note that the Coulomb branch is \emph{not} isomorphic to the moduli space $\cM_{\mathrm{flat}}(C,G_\C)$ of flat $G_\C$-connections on a two-torus $C\cong T^2$ in \eqref{Mflat-T2} as a holomorphic symplectic manifold. It is rather a quotient of $\cM_{\mathrm{flat}}(C,G_\C)$ by $\lat{L}$.

To obtain the algebra of line operators in the 4d $\cN=2^*$ theory, we turn on the equivariant action $\C^\times_t$ on the cotangent fiber of the affine Grassmannian $\Gr(G_\C)$, which is equivalent to switching on the ramification parameters
 \eqref{t-ramification}.
Moreover, its quantization can be further achieved by introducing the equivariant action $\C^\times_q$of the loop rotation $z\mapsto qz$. In this way, we obtain the quantized Coulomb branch of the 4d $\cN=2^*$ theory on $S^1\times \R^3$
\be\label{Gr-SH}
K^{(G_\C^\OO\times \C^\times_t)\rtimes \C_q^\times}(\cR)\cong \SH(W)^{\lat{L}}~.
\ee
As we have seen in the examples of rank one, it is \emph{not} isomorphic to the spherical subalgebra $\SH(W)$ of DAHA associated to $W$. It is rather the $\lat{L}$-invariant subalgebra of $\SH(W)$.

Even with the same gauge group $G$, discrete theta angles provide different theories as in the examples of rank one. Generally, they are distinguished by characteristic classes of Higgs bundles such as the Stiefel-Whitney classes $w_2$ and $w_4$ \cite{Frenkel:2007tx,Aharony:2013hda}. Above we consider only the cases in which the theta angle is zero, but we can generalize it to a theory with non-trivial discrete theta angle by constructing the moduli space of Higgs bundles with non-trivial topological classes over the raviolo.

\subsection{Including surface operator}\label{sec:surface}

So far, we focus on physical realizations of the spherical DAHA $\SH(W)$ and its subalgebras, and it is natural to ask whether we can realize DAHA $\HH(W)$ itself. To see that, let us consider an algebra of line operators on a surface operator of Gukov-Witten type \cite{Gukov:2006jk}.
A surface operator of Gukov-Witten type arises as an intersection of M5-branes at codimension two locus:
$$\begin{matrix}
&{\mbox{\rm space-time:}} & \quad &\R^4 & \times   &T^*C & \times&\R^3 \\
&{\mbox{\rm $N$ M5-branes:}} & \quad  &\R^4 & \times & C & \times  & \textrm{pt}  \\
{\mbox{\rm (surface operator)}} & {\mbox{\rm M5'-brane:}} & \quad&  \R^2 & \times &C & \times &\R^2 \\
\end{matrix}$$
where $C\cong  T^2$ is a two-torus.
Here a surface operator is supported on $\R^2 \times \textrm{pt} \subset \R^4$ in the $\cN=4$ $\SU(N)$ theory $\cT[C,\SU(N),\lat{L}]$.
A half-BPS surface operator breaks the gauge group down to a \emph{Levi subgroup} $L \subset G$, and as in \eqref{tame} the singular behavior of the gauge field around the surface operator is
\bea
A & =  \alpha d \vartheta + \ldots \,, \label{asing}
\eea
where $z=r e^{i \vartheta}$ is a local coordinate of the plane normal to the surface operator. The singular behavior for one $\Phi$ of the adjoint chiral scalars is described by
$$D_{\bar z}\Phi=(\beta+i \gamma) \delta^{(2)}(z,\bar z)~.$$
A surface operator can also be microscopically realized as a 2d $\cN=(4,4)$ theory coupled to the 4d $\cN=4$ theory where the triple $(\alpha,\beta,\gamma)$ can be understood as the $\cN=(4,4)$ Fayet-Iliopoulos parameters.
In addition to the triple $(\alpha,\beta,\gamma)$, they are also labeled by the theta angles $\eta$ of the 2d theory.
The quadruple $(\alpha,\beta,\gamma,\eta)$ takes a value in the $W_L$-invariant part of $T \times \frakt \times \frakt \times T^{\vee}$ where we write $W_L$ for the Weyl group of the Levi subgroup $L$.
We remark that surface operators exist in $\cN=2$ supersymmetric gauge theories where the parameters $\beta$ and $\gamma$ are absent due to the number of supersymmetry.

\begin{figure}[ht]
\centering
\includegraphics[width=\textwidth]{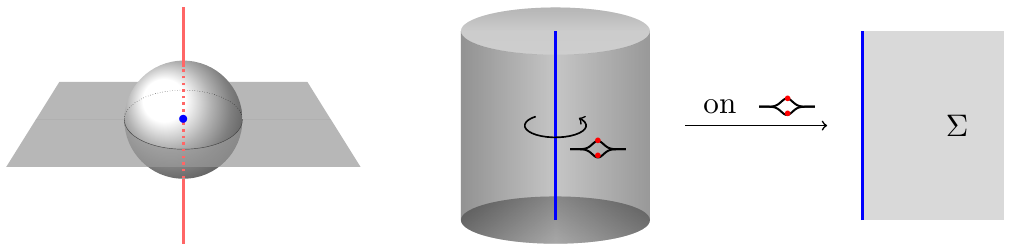}
\caption{The raviolo around a line operator (blue) on the surface operator (red) has tame ramifications at the centers of the two disks. A line operator (blue) gives rise to a boundary condition in the 2d sigma-model upon the compactification of the 4d $\cN=4$ theory on the raviolo.}
\label{fig:Iwahori}
\end{figure}

In the following, we consider a category and algebra of line operators on a surface operator that breaks the gauge group $G$ to its maximal torus $T$, which is often called the \emph{full} surface operator. In this case, the corresponding Weyl group is that of the gauge group $W_T=W$. Since we will eventually consider the 4d $\cN=2^*$ theory, we set $\b=0=\g$ and the surface operator is parametrized by the pair $(\a,\eta)$.
As in the previous subsection, we can study this by compactifying the 4d theory on the ``raviolo''. However, due to the presence of the surface operator, there are ramifications at the centers of the two disks of the raviolo $\raisebox{-.2cm}{\includegraphics[width=1cm]{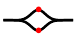}}$ around a line operator (Figure \ref{fig:Iwahori}). We write the resulting Hitchin moduli space by $\MH(\raisebox{-.2cm}{\includegraphics[width=1cm]{Raviolo-ramification}},G)$, and we are interested in a category of $B$-branes of type $I$ on $\MH(\raisebox{-.2cm}{\includegraphics[width=1cm]{Raviolo-ramification}},G)$.

Again, we first consider the moduli space of $G$-bundles over the ramified raviolo $\raisebox{-.2cm}{\includegraphics[width=1cm]{Raviolo-ramification}}$ to formulate this into a mathematical model \cite{vasserot2005induced,varagnolo2009finite,varagnolo2010double}. The full surface operator breaks the space of gauge transformations on a disk from the loop group $G_\C^\OO$ to the \emph{Iwahori} subgroup  \be\label{Iwahori}\scI:=\{ a_0 +a_1 z +a_2 z^2 +\cdots \in G_\C^\OO | a_0\in B  \}\ee  that is the preimage of a Borel subgroup $B$ under the projection $G_\C^\OO\to G_\bC$. Hence, the moduli space of $G$-bundles over $\raisebox{-.2cm}{\includegraphics[width=1cm]{Raviolo-ramification}}$ can be expressed as the double coset space
$$
\BunG(\raisebox{-.2cm}{\includegraphics[width=1cm]{Raviolo-ramification}})=\scI\backslash G_\C^\scK/\scI~.
$$
Actually, $\Fl(G_\C):=G_\C^\scK/\scI$ is called the \emph{affine flag variety},
which is the fiber bundle over the affine Grassmannian with the full flag variety $G_\C/B$ a fiber
$$
\begin{tikzcd}
G_\C/B \arrow[r]
& \Fl(G_\C) \arrow[d]\\
& \Gr(G_\C)
\end{tikzcd}~.
$$
As a mathematical model of the Hitchin moduli space $\MH(\raisebox{-.2cm}{\includegraphics[width=1cm]{Raviolo-ramification}},G)$, we can consider the moduli space $\scI\backslash \cZ$ of a pair of $G$-bundles and sections of its adjoint associated bundles on $\raisebox{-.2cm}{\includegraphics[width=1cm]{Raviolo-ramification}}$ where $\cZ$ is the \emph{affine flag Steinberg variety} defined as
\be\label{flag-Steinberg}
\mathcal{Z} = \{ (x,[g])\in \textrm{Lie}(\scI)\times \Fl(G_\C) \mid
    \operatorname{Ad}_{g^{-1}}(x) \in  \textrm{Lie}(\scI)\}~.\ee
Consequently, the category of line operators on the full surface operator in the 4d $\cN=4$ theory with gauge group $G$ and zero theta angle is equivalent to the derived category of ${\scI}$-equivariant coherent sheaves on $\cZ$
\be\label{line-surface-category}
\mathsf{Line}[\cT[C,G,\lat{L}],T]\cong D^b\mathsf{Coh}^{\scI}(\cZ)~,
\ee
where a maximal isotropic lattice $\lat{L}$ is chosen in such a way that the theta angle is zero.
This includes the category of Wilson operators on the full surface operator
$$
\mathsf{Coh}^{\scI}(\textrm{pt})\cong\mathsf{Coh}^{T}(\textrm{pt})\cong\mathsf{Rep}(T)~,
$$
which sees that the gauge group $G$ is broken to the maximal torus $T$ due to the surface operator.
Clearly, the Grothendieck ring of the category \eqref{line-surface-category} is the algebra of line operators on the full surface operator in the 4d $\cN=4$ theory \cite{vasserot2005induced,varagnolo2009finite,varagnolo2010double}
$$
 K^{\scI}(\mathcal{Z})=\mathbb{C}[T_\bC\times T_\C^\vee]\rtimes \mathbb{C}[W]~.
$$
Unlike \eqref{N=4-K}, this ring is non-commutative because line operators on the surface operators know the order of multiplications even without quantization ($\Omega$-deformation).

By introducing the equivariant actions as in \eqref{Gr-SH}, we obtain an algebra of line operators on the full surface operator in the 4d $\cN=2^*$ theory with gauge group $G$ and zero theta angle   on the $\Omega$-background
\be \label{linealge-surface}
 K^{(\scI \times \C^\times_t)\rtimes \C_q^\times}(\mathcal{Z})\cong \HH(W)^{\lat{L}}~.
\ee
This is isomorphic to the $\lat{L}$-invariant subalgebra of DAHA $\HH(W)$. In the case of $G=\SU(2)$, this is isomorphic to the $\xi_2$-invariant subalgebra of DAHA $\HH$ generated by $X,Y^2,T$. For $G=\SO(3)$, it is isomorphic to the $\xi_1$-invariant subalgebra of DAHA  $\HH$ generated by $X^2,Y,T$. (See \S\ref{sec:DAHA-main}.)

Although we consider the full surface operator that breaks a gauge group $\SU(N)$ to $\mathrm{S}[\U(1)^N]$ here, we can instead include a surface operator of another type associated to a partition of $N$. For this, we replace the Borel subgroup $B$ in \eqref{Iwahori} by a parabolic subgroup $P$ associated to a partition of $N$. In this way, we can obtain a variant of DAHA as an algebra of line operators on a surface operator.

 \subsection*{Canonical coisotropic brane of higher ranks}
In the previous subsection, the canonical coisotropic brane $\frakB_{cc}$ emerges as the boundary condition at the axis of the $\Omega$-deformation by compactifying the 4d $\cN=2^*$ theory on $T^2\cong S^1\times S^1_q$. Moreover, an algebra of line operators can be understood as the algebra of $(\frakB_{cc},\frakB_{cc})$-strings in 2d $A$-model on the Coulomb branch. It is natural to ask how to describe the boundary condition at the axis of the $\Omega$-deformation in the presence of the full surface operator up on the same compactification (Figure \ref{fig:omega-surface-Bcc}).

\begin{figure}[ht]
\centering
\includegraphics[width=14cm]{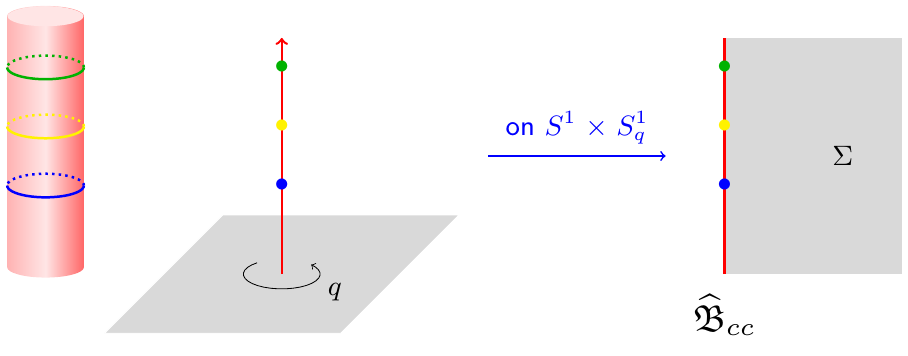}
\caption{The 4d $\cN=2^*$ theory with the Gukov-Witten surface operator on $S^1\times S^1_q$ is described by 2d $A$-model  $\Sigma\to \cM_C(C_p,G,\lat{L})$ where the boundary condition  at $\partial \Sigma$ is described by the canonical coisotropic brane $\widehat{\frakB}_{cc}$ of higher rank.}
\label{fig:omega-surface-Bcc}
\end{figure}

For this purpose, we refer to the idea \cite{haiman2002vanishing,bezrukavnikov2006cherednik} employed in the geometric construction of rational Cherednik algebra. Roughly speaking, spherical DAHA $\SH(W)$ can be interpreted as the subalgebra averaged over the action of the Weyl group $W$. We want to construct the part of the Weyl group by constructing a brane of higher rank since the algebra of line operators on the full surface operator realizes \eqref{linealge-surface}. In fact, there is the natural construction of the Weyl group just by taking projection $T_\bC\times T_\C^\vee\to(T_\bC\times T_\C^\vee)/W$. If we regard the Coulomb branch of the 4d $\cN=2^*$ theory as the resolution  $\eta:\cM_C(C_p,G,\lat{L})\to(T_\bC\times T_\C^\vee)/W$, we can define their fiber-product $\frakY$ via
$$
\begin{tikzcd}
\frakY \arrow[r] \arrow[d, "\rho"]
& T_\bC\times T_\C^\vee \arrow[d] \\
\cM_C(C_p,G,\lat{L}) \arrow[r,  "\eta" ]
& \frac{T_\bC\times T_\C^\vee}{W}
\end{tikzcd}~.
$$
Therefore, $\frakY$ can be understood as the \emph{universal family} of $\cM_C(C_p,G,\lat{L})$.  Following \cite{haiman2002vanishing}, we define the ``unusual'' tautological bundle $\cP := \rho_*\OO(\frakY)$ on $\cM_C(C_p,G,\lat{L})$, which is called \emph{Procesi bundle}, by the push-forward of the sheaf $\OO(\frakY)$ of regular functions (or the trivial bundle)  on  $\frakY$. By construction, the Procesi bundle $\cP$ has rank $|W|$, with  the regular representation of the Weyl group $W$ on every fiber. Then, the canonical coisotropic brane $\wh \frakB_{cc}:=\cP\otimes \frakB_{cc}$ in the presence of the full surface operator is indeed the tensor product of the original line bundle $\cL$ for $\frakB_{cc}$ \eqref{Bcc} with the Procesi bundle $\cP$. Consequently, the algebra of $(\wh\frakB_{cc},\wh\frakB_{cc})$-strings realizes an algebra of line operators on the full surface operator
$$
\Hom(\wh\frakB_{cc},\wh\frakB_{cc})\cong \HH(W)^{\lat{L}}~.
$$

If we replace $\cM_C(C_p,G,\lat{L})$ and $T_\C^\vee$ by the Hitchin moduli space $\MH(C_p,G)$ and $T_\bC$, respectively, in the construction above, we obtain the full DAHA as $\End (\wh\frakB_{cc})\cong \HH(W)$.
In particular, when $q=1$, the bundle $\cP\otimes \cL$ is equivalent to  the vector bundle $E$ constructed in \cite[Corollary 6.1 and 6.2]{Oblomkov:aa}. In fact, the space of $(\wh\frakB_{cc},\frakB_{cc})$-strings can be understood as $\protect\HH(W)$-left and $\protect\SH(W)$-right module $\HH(W)\mathbf{e}$, to which the Procesi bundle $\cP$ is associated (Figure \ref{fig:He}).

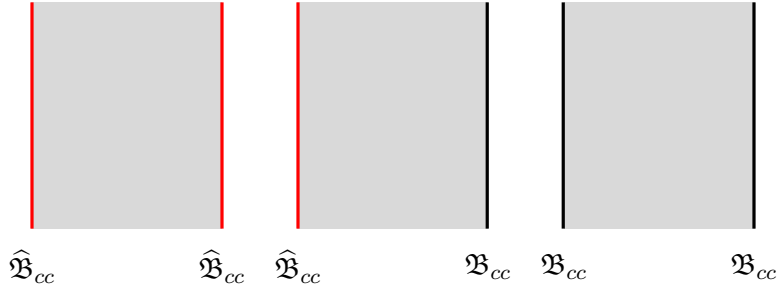
\begin{figure}[ht]
\centering
\begin{tikzpicture}
\fill [gray!30] (-6,0) rectangle (-3.5,3);
\fill [gray!30] (-2.5,0) rectangle (0,3);
\fill [gray!30] (1,0) rectangle (3.5,3);
\draw[very thick] (-0.01,0)--(-0.01,3);
\draw[very thick] (3.5,0)--(3.5,3);
\draw[very thick] (.99,0)--(.99,3);
\draw[very thick,red] (-6,0)--(-6,3);
\draw[very thick,red] (-2.5,0)--(-2.5,3);
\draw[very thick,red] (-3.5,0)--(-3.5,3);
\node at (-6,-.5) {$\wh \frakB_{cc}$};
\node at (-0,-.5) {$\frakB_{cc}$};
\node at (-3.5,-.5) {$\wh \frakB_{cc}$};
\node at (-2.5,-.5) {$\wh\frakB_{cc}$};
\node at (1,-.5) {$\frakB_{cc}$};
\node at (3.5,-.5) {$\frakB_{cc}$};
\end{tikzpicture}
\caption{$(\wh\frakB_{cc},\wh\frakB_{cc})$-strings and $(\frakB_{cc},\frakB_{cc})$-strings form DAHA $\protect\HH$ and spherical DAHA $\protect\SH$, respectively. Hence, a $(\wh\frakB_{cc},\frakB_{cc})$-string leads to $\protect\HH$-left and $\protect\SH$-right module $\protect\HH \mathbf{e}$.}
\label{fig:He}
\end{figure}

This has the following remarkable consequence.
In \S\ref{sec:2d}, we have seen that given an $A$-brane $\frakB'$, the space of $(\frakB_{cc},\frakB')$-strings can be understood as a representation of spherical DAHA $\SH(W)$. In fact, given a $(\frakB_{cc},\frakB')$-string, its joining with a $(\wh\frakB_{cc},\frakB_{cc})$-string always yields a $(\wh\frakB_{cc},\frakB')$-string (Figure \ref{fig:Morita-equiv}), which receives the action of the full DAHA $\HH(W)$. In a similar fashion, by reversing a $(\frakB_{cc},\frakB')$-string, one can obtain a $(\frakB_{cc},\frakB')$-string from a $(\wh\frakB_{cc},\frakB')$-string.  This leads to the Morita equivalence of the two representation categories
\be\label{Morita}
\Hom(\wh\frakB_{cc},\frakB_{cc}) :\mathsf{Rep}(\SH(W)) \xrightarrow{\sim} \mathsf{Rep}(\HH(W))~.
\ee
Of course, not every object produces an equivalence of this type. We expect that both $\frakB_{cc}$ and~$\wh\frakB_{cc}$ can be understood as generating objects of the category of $A$-branes.
In general, generating objects are far from unique, and their non-uniqueness is one way that Morita equivalences arise.
For example, any free $R$-module is a generating object, giving rise to the usual Morita equivalences between matrix algebras. Since $\hat\frakB_{cc}$ is, in a sense, analogous to a higher-rank module over $\frakB_{cc}$, we expect a similar story here, but do not pursue this in this paper.

In particular, since the space of $(\wh\frakB_{cc},\frakB_{cc})$-strings is associated to the Procesi bundle $\cP$, the dimension formula for the representation corresponding to the space of $(\wh\frakB_{cc},\brL)$-strings for a compact Lagrangian submanifold $\L$ is obtained by just tensoring $\cP$ in \eqref{dimension2}
\bea\nonumber
\dim \Hom (\wh\frakB_{cc},\brL)&=\dim H^0(\L,\cP\otimes \frakB_{cc}\otimes \brL^{-1})~,\cr
&=|W|  \dim \Hom (\frakB_{cc},\brL)~.
\eea

\begin{figure}[ht]
\centering
\includegraphics[width=5cm]{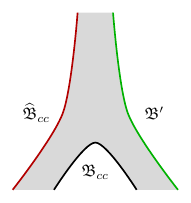}
\caption{Joining of a $(\wh\frakB_{cc},\frakB_{cc})$-string and a $(\frakB_{cc},\frakB')$-string leads to  $(\wh\frakB_{cc},\frakB')$-string.}
\label{fig:Morita-equiv}
\end{figure}

\acknowledgments
We would like to thank
\begin{center}
  D.~Ben-Zvi, M.~Bullimore, I.~Cherednik, M.~Dedushenko, D.~E.~Diaconescu, T.~Dimofte, P.~Etingof, M.~Fluder, D.~Jordan, T.~Kimura, A.~Kirillov Jr., M.~Kontsevich, T.~Q.~T.~Le, Y.~Lekili, M.~Mazzocco, V.~Mikhaylov, A.~Mellit, G.~W.~Moore, H.~Nakajima, A.~Neitzke, A.~Oblomkov, T.~Pantev, P.~Putrov, M.~Romo, P.~Samuelson, A.~Sanders, O.~Schiffmann, Peng~Shan, E.~Sharpe, V.~Shende, Y. ~Soibelman, M.~Sperling, Y.~Tachikawa, A.~Tripathy, E.~Vasserot, B.~Webster, P.~Wedrich, B.~Williams, Wenbin~Yan, Ke~Ye, Y.~Yoshida, R.D.~Zhu
\end{center}
  for valuable discussion and correspondence. A portion of this work was performed
\begin{itemize}[nosep]
  \item at the American Institute of Mathematics supported by a SQuaRE grant ``New connections between link homologies and physics'',
\item at the Aspen Center for Physics supported by National Science Foundation grant PHY-1066293 during the program, ``Boundaries and Defects in Quantum Field Theories'' and PHY-1607611 during the workshop, ``Quantum Knot Homology and Supersymmetric Gauge Theories'',
\item at the Kavli Institute for Theoretical Physics in Santa Barbara supported by the National Science Foundation under Grant No. NSF PHY-1748958 during the workshop ``Quantum Knot Invariants and Supersymmetric Gauge Theories'', and
\item at the International Centre for Theoretical Sciences (ICTS) during the program, ``Quantum Fields, Geometry and Representation Theory'' (Code: ICTS/qftgrt2018/07)~.
\end{itemize}
S.N. thanks  the Center for Quantum Geometry of Moduli Spaces, Max Planck Institute for Mathematics Bonn, and IHES for hospitality.
I.S. thanks Uppsala University, the Center for Quantum Geometry of Moduli Spaces, the Mathematisches Forschungsinstitut Oberwolfach, and Fudan University for hospitality. We are grateful to organizers of seminars, workshops and conferences which have provided us chances to present this work.

The work of S.G., S.N., and D.P. was supported by the Walter Burke Institute for Theoretical Physics and the U.S. Department of Energy, Office of Science, Office of High Energy Physics, under Award No.\ DE{-}SC0011632.
The work of S.G. was also supported by the National Science Foundation under Grant No.~NSF DMS 1664227.
The work of S.N. was also supported by the center of excellence grant ``Center for Quantum Geometry of Moduli Space'' from the Danish National Research Foundation (DNRF95) and by NSFC Grant No.11850410428 and Fudan University Original Project (No. IDH1512092/002).
The work of D.P. was partly supported by the ERC-SyG project, Recursive and Exact New Quantum Theory (ReNewQuantum) which received funding from the European Research Council (ERC) under the European Union's Horizon 2020 research and innovation programme under grant agreement No 810573.
The work of I.S. was supported by the Deutsche Forschungsgemeinschaft (DFG, German Research Foundation) under Germany's Excellence Strategy EXC 2181/1\,--\,390900948 (the Heidelberg STRUCTURES Excellence Cluster), and by the Free State of Bavaria.
P.K. is partially supported by the AMS Simons Travel Grant.

\newpage
\appendix

\section{Glossary of symbols}\label{app:notation}

As a general rule, single symbols in sans-serif type are used to denote lattices or other free $\Z$-modules. Words set in sans-serif type (e.g. \ABrane) refer to categories. Italic symbols may be used for groups, algebras, or other classes of objects.
Calligraphic letters, such as~$\mathcal{M}$ or $\mathcal{B}$, denote objects which are moduli spaces or closely related to moduli spaces.
We reserve bold-face type for distinguished Lagrangian submanifolds of such moduli spaces, in particular for the support of branes (such as $\F$ for a generic fiber of the Hitchin fibration).

Capital gothic letters (e.g.~\X{} for the target space) are used for objects equipped with the structure required by the topological $A$-model. As such, $\brane$ denotes the $A$-brane associated to particular data; for instance, $\brane_\F$ denotes a brane supported on the generic fiber of the Hitchin fibration. Note, though, that extra data on which the brane depends may be left implicit.

Script letters will denote modules of the algebra $\OO^q(\X)$ (precisely which algebra is  intended will be clear from the context). We always suggestively use the same letter for a brane and its corresponding representation, so that (for example) \repF{} will be identified with~\brF{} under the correspondence~\eqref{eq:functor}.

Occasionally, we cannot help using the same symbol for different notions due to the lack of letters. For instance, the symbol $T$ is used for a maximal torus of a compact Lie group, a generator of DAHA, a two-torus and an element of $\SL(2,\Z)$. However, we believe that the context is a sufficient guide.

\bigskip
In the following, we list notations of the paper.
\begin{itemize}[nosep]
  \item $\MH$ is a Hitchin moduli space (a moduli space of Higgs bundles), $\cM_C$ is a Coulomb branch, $\MF$ is a moduli space of flat connections
\item \MS{} is an affine variety over~$\C$ which admits the structure of a non-compact \HK manifold. Depending on context, it may be $\C^\times \times \C^\times$, or the moduli space of flat connections with coefficients in a complex Lie group.
\item $G$ is a compact gauge group, $T$ is a Cartan subgroup of $G$, and $L$ is a Levi subgroup of $G$.
\item $G_\C$ is the complexification of $G$, $T_\C$ is a Cartan subgroup of $G_\C$, $B$ is a Borel subgroup of $G_\C$, and $P$ is a parabolic subgroup of $G_\C$.
\item $\sfR$ denotes a finite root system, and $\dt \sfR:=\sfR\oplus \bZ\delta$ denotes the corresponding affine root system. $\sfR_+$ denotes a set of positive roots.
\item $\sfQ$ and $\sfQ^\vee$ denote the root and coroot lattice of $\sfR$, respectively. Similarly, $\sfP$ and $\sfP^\vee$ denote the weight and coweight lattice of $\sfR$, respectively. $\sfP_+$ denotes the set of dominant weights.
\item $\mathsf{\Lambda}$ and $\mathsf{\Lambda}^\vee$ denote the character and cocharacter lattice, respectively. The center $Z(G)$ and the fundamental group $\pi_1(G)$ of $G$ are given by quotients of lattices $Z(G)=\sfP^\vee/\mathsf{\Lambda}^\vee$ and  $\pi_1(G)=\mathsf{\Lambda}^\vee/\sfQ^\vee$, respectively. Their duals are $Z(G)^\vee=\mathsf{\Lambda}/\sfQ$ and $\pi_1(G)^\vee= \sfP/\mathsf{\Lambda}$.
\item $\{\a_1,\ldots,\a_n\}$ denotes the set of simple roots and  $\{\a_1^\vee,\ldots,\a_n^\vee\}$ denotes the set of simple coroots
\item $\Weyl:=( s_1,\ldots, s_n) $  is the Weyl group of $\sfR$ where $s_i$ denotes the reflection with respect to $\a_i$.
\item $\a_0:=\delta-\theta $ with $\theta$ being the highest short root of $\sfR$ so that $\{\a_0,\a_1,\ldots,\a_n\}$  is the set of simple roots of the affine root system $\dt \sfR$.
\item $\AWeyl:=( s_0,s_1,\ldots, s_n) $ is  the affine Weyl group for $\dt \sfR$. Note that we have an isomorphism
$$
\AWeyl=\Weyl \ltimes \mathsf{t}(\sfQ^\vee)
$$
where $\sfQ^\vee$ acts by affine translation $\mathsf{t}$ $$\mathsf{t}(\a^\vee)(\lambda)=\lambda-(\lambda,\a^\vee)\delta~.$$
\item $\AWeyl^e:=\Weyl\ltimes \mathsf{t}(\sfP^\vee)$ denotes the \emph{extended} affine Weyl group, which can be expressed as
$$
\AWeyl^e=\Pi  \ltimes \AWeyl
$$
where $\Pi := \sfP^\vee/\sfQ^\vee$ acts faithfully on the set of simple roots $\{\a_0,\a_1,\ldots,\a_n\}$. Hence, if $\pi_r(\a_i)=\a_j$ for $\pi_r\in \Pi$, then $\pi_r s_i\pi_r^{-1}=s_j$. (See Figure \ref{fig:extended-affine-Weyl}.)
\item $\BB(W)$,  $\HH(W)$, and $\WW$  denote the double affine braid group, the double affine Hecke algebra and the double affine Weyl group associated to a Weyl group $W$ of a root system.
\end{itemize}

\begin{figure}[ht]
\centering
\includegraphics[width=14cm]{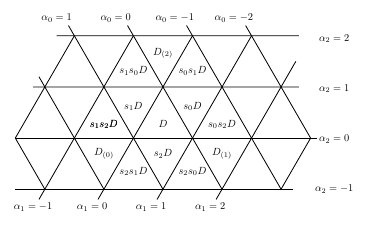}
\caption{The action of the extended affine Weyl group $\protect \AWeyl^e$ of type $\protect\dt{A}_2$ with generators $s_0,s_1,s_2,\pi$. The reflection with respect to the axis $\a_i=0$ is generated by $s_i$ so that $s_i^2=1$. The braid relation can be seen as $D_{(i)}=s_{i+2}s_{i+1}s_{i+2}D=s_{i+1}s_{i+2}s_{i+1}D$, which is equivalent to the reflection with respect to the axis $\a_i=1$. The rotation of the triangle domain $D$ around its center by 120${}^\circ$ is generated by $\pi$ that satisfies $\pi^3=1$. One can also convince oneself $\pi s_i=s_{i+1}\pi$. Here the indices $i=0,1,2$ are taken to be cyclic $i\cong i+3$. If we view the extended affine Weyl group as $\protect \AWeyl^e=\Weyl\ltimes \mathsf{t}(\sfP^\vee)$, $\mathsf{t}(\omega_1)$ and  $\mathsf{t}(\omega_2)$ translate  the domain  $D$ to $s_0s_2D$ and $s_0s_1D$, respectively, where $\omega_1,\omega_2$ are the fundamental weights of $A_2$. }
\label{fig:extended-affine-Weyl}
\end{figure}

\section{Basics of DAHA}\label{app:DAHA}
In this appendix, we will summarize the basics of double affine Hecke algebras. This appendix can be regarded as a concise review of \cite{kirillov1997lectures,macdonald2003affine,Cherednik-book,haiman2006cherednik}.

\subsection{DAHA}\label{app:DAHA-gen}

Let $t := \{t_0,\ldots , t_n\}$ be a collection of formal variables associated to an affine Weyl group $\dt W$ where we identify $t_i = t_j$ if the reflections $s_i$ and $s_j$ are conjugate in $\dt W$. Note that we have at most two distinct $t$'s and, in particular, all the variables $t_i$ are identical for simply-laced types ($A,D,E$). Let $q$ be a formal variable and let $\bC[q^{\pm\frac1m},t_0^{\pm},\ldots , t_n^{\pm}]$ be the ring of Laurent polynomials where $m$ is the minimum positive integer satisfying $( \sfP, \sfP^\vee) = \frac1m \bZ$. Namely, $m=2$ for type $D_{\textrm{even}}$; $m=1$ for type $B_{\textrm{even}}$ and type $C$; otherwise $m=|\Pi|$.
We consider a multiplicative system $M$ in $\bC[q^{\pm\frac1m},t_0^{\pm},\ldots , t_n^{\pm}]$ generated by elements of the form $(q^{\ell} t_i-q^{-\ell} t_i^{-1})$ for any non-negative integer $\ell \in \Z_{\ge0}$ with $i=0,\ldots,n$. We define the coefficient ring \CR{} to be the localization (or formal ``fraction'') of the ring $\bC[q^{\pm\frac1m},t_0^{\pm},\ldots , t_n^{\pm}]$ at~$M$:
$$
\CR =M^{-1}\bC[q^{\pm\frac1m},t_0^{\pm},\ldots , t_n^{\pm}]~.
$$
Moreover, we define $\C_t:=\CR /(q^{\frac1m}-1)$ and $\C_q:=\CR /(t_0-1, \ldots,t_n-1)$.

The DAHA $\HH(W)$ associated to a Lie group $G$ is the $\C_{q,t}$-algebra generated by elements $T_0,\ldots,T_n, \Pi, X^{\sfP}$ with relations
\begin{enumerate}\setlength{\parskip}{-0.1cm}
\item (braid relation) $T_iT_jT_i\cdots = T_jT_iT_j \cdots$ with $m_{ij}$ terms on each side.  ($m_{ij}$ is defined via the Coxeter relations $(s_is_j)^{m_{ij}} = 1$ in the Weyl group $W$.)
\item (quadratic relation) \be\label{quadratic-relation}(T_i-t_i)(T_i+t_i^{-1})=0~.\ee
\item (first affinization) $\pi_r T_i\pi_r^{-1}=T_j$ if $\pi_r(\a_i)=\a_j$.
\item (second affinization) For $i=0,\ldots,n$,
\be\label{XT-relation}\begin{array}{ll}
T_i X^\mu=X^\mu T_i    &\qquad \textrm{if} \quad  (\mu,\a_i^\vee)=0\\
 T_i X^\mu=X^{s_i(\mu)} T_i^{-1}    &\qquad \textrm{if} \quad  (\mu,\a_i^\vee)=1
\end{array}\ee
where $\a_0^\vee:=-\theta^\vee$.
\item $\pi_r X^\mu \pi_r^{-1}=X^{\pi_r(\mu)}$.
\end{enumerate}

The deformation parameters $t_i$ manifestly appear in the quadratic relation \eqref{quadratic-relation} whereas
the other deformation parameter $q$ implicitly shows up in the second relation of \eqref{XT-relation} via $q=X^\delta$. The most important part is $i = 0$ because  $s_0(\mu)=\mu-\a_0 =\mu+\theta-\delta$ for $(\mu,\a^\vee_0)=1$. In this case,  the relation becomes
$$T_0 X^\mu = X^{s_0(\mu)}T_0^{-1}= X^{\mu-\a_0}T_0^{-1} = q^{-1}X^{\mu+\theta} T_0^{-1}~.$$

As the name suggests, DAHA $\HH(W)$ contains two affine Hecke algebras as subalgebras:
\begin{itemize}\setlength{\parskip}{-0.1cm}
\item (affine Hecke algebra for the root system $\sfR^\vee$) $\dt H^X(W):=\C_t ( T_1,\ldots,T_n,X^ \sfP) $
\item  (affine Hecke algebra for the root system $\sfR$) $\dt H^Y(W):= \C_t( T_0,T_1,\ldots,T_n,\Pi) $
\end{itemize}
where both contain finite Hecke algebra $H(W) := \C_t ( T_1,\ldots,T_n)$.
Indeed, the first representation  $\dt H^X(W)$ is called the  \emph{Bernstein presentation} whereas the second one is called the \emph{Coxeter presentation}. The second one $\dt H^Y(W)$ can also be expressed in the Bernstein presentation by using an isomorphism as vector spaces
$$H(W) \otimes \mathbb{C}[Y^{\sfP^\vee}] \xrightarrow{\sim} H^Y(W)~,
$$
where the affine translation in $\dt H^Y(W)$ is generated by $Y$:
\be\begin{array}{ll}
Y^\lambda:=T_{\mathsf{t}(\lambda)}\qquad & \textrm{if} \quad \lambda\in\sfP^\vee_+\cr
Y^\lambda := Y^\mu (Y^\nu)^{-1} \qquad & \textrm{if}\quad  \lambda=\mu-\nu \quad  \mu,\nu\in\sfP^\vee_+~.
\end{array}
\ee

\subsubsection{Double affine braid group and double affine Weyl group}
The double affine braid group $\BB(W)$ is a group generated by elements $T_0,\ldots,T_n, \Pi, X^{\sfP}$ only with the relations 1,3,4,5 above (without the relation 2).  In other words, DAHA can be understood as the $\C_{q,t}$-group algebra of $\BB(W)$ with the quadratic relation \eqref{quadratic-relation}
$$
\HH(W)=\C_{q,t}[\BB(W)]/((T_i-t_i)(T_i+t_i^{-1}))~.
$$
The double affine Weyl group $\WW$ is $\BB(W)$ with the quadratic relations $T_i^2=1$. Therefore, its $\C_{q}$-group algebra is the $t=1$ limit of DAHA $\HH(W)$, which is also introduced at the beginning of \S\ref{sec:2d}
$$
\C_{q}[\WW]=\HH_{t=1}(W)~.
$$
In the Bernstein presentation, the generators $X$ and $Y$ form the Heisenberg group as a subgroup of $\WW$ with the relation
$$
X^{\mu } Y^{\lambda  } = q^{( \mu,\lambda   ) } Y^{\lambda  } X^{\mu } , \qquad \textrm{for} \quad \mu \in \sfP ,\  \lambda \in \sfP^\vee~.
$$

\subsubsection{PBW theorem for DAHA}
First of all, the following Poincar\'e-Birkhoff-Witt theorem for DAHA plays a very important role in the representation theory.
Every element $h \in \HH(W)$ can be \emph{uniquely} written in the form
$$
h=\sum_{r\in\Pi,w \in \dt W,\mu\in \sfP} a_{\mu,w,r}(h)~ X^\mu  T_w \pi_r \qquad a_{\mu,w,r}(h)\in\C_{q,t}
$$
in the Coxeter presentation for $\dt H^Y(W)$, or
\be\label{PBW}
h=\sum_{\lambda\in\sfP^\vee,w\in W,\mu\in\sfP}b_{\mu,w,\lambda}(h) ~  X^\mu T_w Y^\lambda\qquad b_{\mu,w,\lambda}(h)\in\C_{q,t}
\ee
in the Bernstein presentation  for $\dt H^Y(W)$.

\subsubsection{Spherical subalgebra}

For $w \in W$, let $t_w := t_{i_1} \cdots t_{i_k}$, where $w = s_{i_1} \cdots s_{i_k}$ is a reduced decomposition. Then, it is well-defined.
We define an element in the group algebra $\C_t[W]$
\begin{equation}\label{ideomp-gen}
\mathbf{e} :=  \frac{\sum_{w\in W}t_w T_w}{\sum_{w\in W}(t_w)^2}~.
\end{equation}
Then, we have $T_i \mathbf{e}=t_i \mathbf{e}$, and it is moreover an idempotent  $\mathbf{e}^2=\mathbf{e}$. Subsequently,
we can define the spherical subalgebra  as
$\SH(W):=\mathbf{e}\HH(W) \mathbf{e}\subset\HH(W)$, called  \emph{spherical DAHA}. Roughly speaking, the spherical DAHA can be understood as the subalgebra ``averaged over'' the Weyl group symmetry.

At $t=1=q$, the spherical DAHA becomes commutative
$$
\SH(W)\Big|_{t=1=q}=\bC[T_\bC\times T_\bC]^W~.
$$
Indeed, it is isomorphic to the coordinate ring of the moduli space
\be\label{Mflat-T2}\MF(T^2,G_\C)=\frac{T_\bC\times T_\bC}W\ee
of flat $G_\bC$-connections on a torus \cite{Oblomkov:aa}.

\subsubsection{\texorpdfstring{Braid group and $\mathrm{SL}(2,\mathbb{Z})$ action}{Brad group and SL(2,Z) action}}\label{app:SL2-action}

The braid group on  three strands is given by  $B_3= (\tau_\pm\text{ : } \tau_+\tau^{-1}_-\tau_+ = \tau^{-1}_-\tau_+\tau^{-1}_-)$, and the relation to $\SL(2,\Z)$ is given by the short exact sequence
\be
0\to \Z\to B_3\to \SL(2,\Z)\to1
\ee
where the kernel of the projection $B_3\to \SL(2,\Z)$ is $\Z$ generated by $\left(\tau_{+} \tau_{-}^{-1} \tau_{+}\right)^{4}$.

DAHA $\HH(W)$ receives an action of the braid group on three strands $B_3= (\tau_\pm\text{ : } \tau_+\tau^{-1}_-\tau_+ = \tau^{-1}_-\tau_+\tau^{-1}_-)$.
To see the action explicitly, let $\omega_i$, $\omega_i^\vee$ ($i=1,\dots,n$) denote the fundamental weights and coweights respectively, and we define $X_i:=X^{\omega_i}$ and $Y_i=Y^{\omega_i^\vee}$ as the corresponding generators of $\HH(W)$ in the Bernstein presentation. Then, the action reads
\begin{equation}\label{DAHA-T}
\tau_+: \begin{cases}
      X^\mu\mapsto X^\mu\\
      T_j\mapsto T_j\\
      Y_i\mapsto X_iY_i q^{-(\omega_i,\omega_i)}
   \end{cases}~,\qquad\qquad \tau_-: \begin{cases}
      X_i\mapsto Y_iX_i q^{(\omega_i,\omega_i)}\\
      T_j\mapsto T_j\\
      Y^\lambda\mapsto Y^\lambda
   \end{cases}~.
\end{equation}
The element $\sigma=\t_+\t_-^{-1}\t_+=\t_-^{-1}\t_+\t_-^{-1}$ maps the generators as
\be\label{sigma}
\sigma:  \begin{cases}  X^\mu\mapsto Y^{-\mu^\vee}\\
      T_j\mapsto T_j\\
 Y^{\lambda^\vee} \mapsto T_{w_\circ}^{-1}X^{w_\circ(\lambda)}T_{w_\circ}
   \end{cases}~,\qquad\qquad
   \sigma^2:  \begin{cases}  X^\mu\mapsto T_{w_\circ}^{-1}X^{-w_\circ(\mu)}T_{w_\circ}\\
      T_j\mapsto T_j\\
 Y^\lambda\mapsto T_{w_\circ}^{-1}Y^{-w_\circ(\lambda)}T_{w_\circ}
   \end{cases}~,
\ee
where $w_\circ$ is the longest element of the Weyl group $W$.
Moreover, $\sigma^4$ acts as the conjugation by $T_{w_\circ}^2$, namely $\sigma^4(x ) =T_{w_\circ}^{-2}(x)T_{w_\circ}^2$ for any $ x\in \HH(W)$.

Since $T_w$ acts on a generator of $\SH(W)$ as the multiplication by $t_w$, the action of $\sigma^4$ is trivial on the spherical subalgebra $\SH(W)$ in general: $$\sigma^4\Big|_{\SH(W)}\cong \textrm{id}_{\SH(W)}~.$$
Therefore, the $B_3$ action factors through $\SL(2,\Z)$ on $\SH(W)$, and each element corresponds to the following matrix element on $\SH(W)$ \footnote{Although we follow the notation of \cite{Cherednik-book} for the transformations $\tau_\pm$ on the generators of DAHA here, we change matrix assignments \eqref{tau-pm2} to $\tau_\pm$ from \cite{Cherednik-book}.}:
\be\label{tau-pm2} \begin{pmatrix}1 & 0 \\ 1 & 1\end{pmatrix} \leftrightarrow \tau_{+}~, \quad\begin{pmatrix}1 & 1 \\ 0 & 1\end{pmatrix} \leftrightarrow \tau_{-}~, \quad \begin{pmatrix}0 & -1 \\ 1 & 0\end{pmatrix} \leftrightarrow \sigma~.\ee

\subsubsection{Polynomial representation of DAHA}\label{app:poly}
The most basic representation of DAHA, first studied by Cherednik, is called \emph{the polynomial representation}
\bea\label{poly-rep-full}
\pol:\HH(W) &\to \End(\C_{q,t}[X^\sfP]); \\
 T_i&\mapsto t_is_i+(t_i-t_i^{-1})\frac{s_i-1}{X^{\a_i}-1} \qquad \mathrm{for} \ i=0,\ldots, n\\
 \pi_r&\mapsto \pi_r \qquad \mathrm{for} \  \pi_r\in \Pi~.
\eea
Here we denote the group algebra of the weight lattice $\sfP$ by $\C_{q,t} [X^\sfP]$, which contains the group algebra $\C_t[X^{\dt\sfP} ]$ of the affine weight lattice $\dt\sfP  := \sfP\oplus \bZ \delta$ as a subalgebra by setting $X^{\lambda+r\delta} := q^rX^\lambda$. Hence, $X^\mu$ act on $\C_{q,t}[X^\sfP]$ by multiplication, and the action of the extended affine Weyl group $\dt W^e$ on $\C_{q,t}[X^\sfP]$, which appears as $s_i$ in \eqref{poly-rep-full}, is given by
 $$
 w(X^\mu) := X^{w(\mu)} = q^{-(v(\mu), \lambda)}X^{v(\mu)}~,
 $$
 where we write an element $w\in \dt W^e$ as $w=\mathsf{t}(\lambda)v$ with $\lambda\in \sfP^\vee$, $v\in W$. Note that  the polynomial representation is faithful.

 \subsubsection*{Non-symmetric Macdonald polynomials}

We define an involution $f\mapsto f^*$ for $f=\sum_\lambda f_\lambda(q,t) X^\lambda\in \C_{q,t}[X]$ as
$$
f^*=\sum_\lambda f_\lambda(q^{-1},t^{-1}) X^{-\lambda}~.
$$
We also define the weight function
$$
\Xi:=\prod_{\a\in\sfR_+} \frac{(X^\alpha;q^2)_\infty(q^2X^{-\alpha};q^2)_\infty}{(q^{2c_\a}X^\alpha;q^2)_\infty(q^{2c_\a+2}X^{-\alpha};q^2)_\infty}~.
$$
Note that $c_\a$ relates $q$ and $t=\{t_0,\dots, t_n\}$ via $q^{c_\a}=t_\a$ where
$c_\a=c_\b$ if $\a$ and $\b$ are in the same $W$-orbit, \textit{i.e.} there exists $w\in W$ such that $w(\a)=\b$. In particular, there is only one parameter $c$ for simply-laced ADE types.
Then, the product $\textbf(\ ,\ \textbf)$ on $\C_{q,t}[X]$ is defined as the \emph{constant term} (C.T.) of $fg^*\Xi$
$$ \textbf(f,g\textbf)=\textrm{C.T.}(fg^*\Xi)~,$$
which can be taken by the integral
$$
{\rm C.T.} \ h(X) := \oint_{|X_1|=1} \frac{dX_1}{X_1} \ldots \oint_{|X_{\rk(G)}|=1} \frac{dX_{\rk(G)}}{X_{\rk(G)}}~ h(X)~.
$$
Using the product $\textbf(\ ,\ \textbf)$, one can show the existence and uniqueness of \emph{non-symmetric Macdonald polynomials} $E_\lambda(X;q,t)$ as follows.

 For each $\lambda\in \sfP$, there is a unique polynomial $E_\lambda\in \C_{q,t}[X]$ such that
\begin{itemize}\setlength\itemsep{-1pt}
\item $E_\lambda = X^\lambda + \textrm{lower terms}~,$
 \item  $\textbf(E_\lambda,X^\mu\textbf)=0 \quad  \textrm{for} \ {}^\forall \mu<\lambda~.$
\end{itemize}
Under the polynomial representation,
 these are eigenfunctions of the $Y$-operators consequently as
\bea\label{non-sym-Macdonald}
\pol(f(Y))\cdot E_\lambda &= f(q^{-2\lambda-2\rho_c(\lambda)}) E_\lambda
\eea
where
$$
\rho_c(\lambda)=\frac12\sum_{\a\in \sfR_+} \eta((\lambda,\a) ) c_\a\a\qquad \text{with} \qquad \eta(x)=\begin{cases}1&\textrm{if} \ x>0\\ -1&\textrm{if} \ x\le0
\end{cases}~.
$$

 \subsubsection*{Symmetric Macdonald polynomials}

We define another involution $f\mapsto \overline f$ for $f=\sum_\lambda f_\lambda(q,t) X^\lambda\in \C_{q,t}[X]$ as
\be\label{conjugate}
\overline f=\sum_\lambda f_\lambda(q,t) X^{-\lambda}~.
\ee
We also define another product \cite[\S VI.9]{macdonald1998symmetric} on $\C_{q,t}[X]$ as
\be\label{Macdonald-product}\langle f,g\rangle=\frac{1}{|W|}\textrm{C.T.}(f \overline g\, \Upsilon)\qquad
\textrm{where}
\qquad \Upsilon= \prod_{\alpha\in \sfR} \frac{(X^\alpha; q^2)_\infty }{ (t_\a^2X^\alpha; q^2)_\infty}~. \ee
Since  $ \Upsilon=\overline \Upsilon$, it is symmetric $\langle f,g\rangle=\langle g,f\rangle$.
Furthermore, for a dominant weight $\lambda\in \sfP_+$, we define \emph{monomial symmetric functions}
 $$
 m_\lambda=\sum_{\mu\in W(\lambda)} X^\mu~.
 $$
Using these data, one can show the existence and uniqueness of \emph{symmetric Macdonald polynomials} $P_\lambda(X;q,t)$ as follows.

  For each dominant weight $\lambda\in \sfP_+$, there is a unique function $P_\lambda\in \C_{q,t}[X]^W$ such that
 \begin{itemize}\setlength\itemsep{-1pt}
 \item $  P_\lambda = m_\lambda + \textrm{lower terms}~,$
  \item $ \langle P_\lambda,m_\mu\rangle=0 \quad  \textrm{for} \ {}^\forall \mu \in \sfP_+$ such that  $\mu<\lambda~,$
  \end{itemize}
Under the product, the symmetric Macdonald polynomials are orthogonal. In the case of type $A_{N-1}$, the norm is explicitly given by
\be \label{Macdonald-norm}
 \langle P_\lambda,P_{\lambda'}\rangle=g_{\lambda}\delta_{\lambda\lambda'}~, \qquad g_{\lambda}(q,t)=\prod_{(i,j)\in\lambda}\frac{1-q^{2(\lambda_i-j+1)}t^{2(\lambda^T_j-i)}}{1-q^{2(\lambda_i-j)}t^{2(\lambda^T_j-i+1)}}\frac{1-t^{2(N-i+1)} q^{2(j-1)}}{1-t^{2(N-i)} q^{2j}}~,
\ee
where we express a Young diagram $\lambda=\lambda_1\ge\cdots \ge \lambda_{N-1}$.

Let $f$ be a symmetric polynomial $f\in \C_{q,t}[Y^{\sfP^\vee}]^W$. Then, under the polynomial representation,  they are eigenfunctions
\be\label{general-Dunkl}
\pol(f(Y_1,\dots,Y_n))\cdot P_\lambda=f(q^{-2\lambda-2\rho_c})P_\lambda~,
\ee
where  $\rho_c$ is the formal expression
$$\rho_c:=\frac12\sum_{\a\in \sfR_+} c_\a\a~.$$
Note that, for type ADE, it simplifies to $q^{-2\lambda-2\rho_c}=q^{-2\lambda}t^{-2\rho}$  where $\rho$ is the Weyl vector.
The non-symmetric and symmetric Macdonald polynomials are related by the idempotent \eqref{ideomp-gen}
$$
P_\lambda=\Big(\sum_{w\in W}(t_w)^2\Big)\,\mathbf{e}E_\lambda~.
$$

 \subsubsection{Symmetric bilinear form}\label{app:sym-bilinear}

There exists an anti-involution $\phi : \HH(W) \to \HH(W)$ fixing $q, t$ such that
$$
\phi:  \begin{cases}  X^\mu\mapsto Y^{-\mu^\vee}\\
      T_w\mapsto T_{w^{-1}}\\
 Y^{\lambda^\vee}\mapsto X^{-\lambda}
   \end{cases}$$
   It is easy to see $\phi^2=\id$.
If we restrict the action of $\phi$ to the spherical subalgebra $\SH(W)$, then it is the same as $\sigma$ in \eqref{sigma}
\be\label{phi-sigma}
\phi\Big|_{\SH(W)}\cong\sigma\Big|_{\SH(W)}~.
\ee

We define the evaluation map $\theta:\HH(W)\to \C_{q,t}$ as
\be\label{evaluation}
\theta(h):=\pol(h) \cdot 1\Big|_{X\mapsto q^{-2\rho_c}}
\ee
Then, one can define a symmetric $\C_{q,t}$-bilinear form $\textbf[ \ , \ \textbf]\,. : \HH(W)\times \HH(W)\to \C_{q,t}$ as
\be\label{sym-form}
\textbf[h,h'\textbf]:=\theta(\phi(h)h')~.
\ee
In particular, for $f,g\in \C_{q,t}[X]$, we have
\be\label{sym-bilinear}
\textbf[f,g\textbf]=\pol(f(Y^{-1}))\cdot g\Big|_{X\mapsto q^{-2\rho_c}}~.
\ee
For instance, when both are symmetric Macdonald polynomials, $f=P_\lambda(X)$, $g=P_\mu(X)$, \eqref{general-Dunkl} tells us
\be \label{Macdonald-pairing2}
\textbf[P_\lambda(X),P_\mu(X)\textbf]=P_\lambda(q^{2\mu+2\rho_c})P_\mu(q^{-2\rho_c})~,
\ee
which is the refined Chern-Simons invariants of the Hopf link \cite{Aganagic:2011sg}.
In fact, using the PBW theorem, one can bring $\phi(h)h'$ into the form  \eqref{PBW}
$$
\phi(h)h'=\sum_{\lambda\in\sfP^\vee,w\in W,\mu\in\sfP}b_{\mu,w,\lambda}(\phi(h)h') ~  X^\mu T_w Y^\lambda~.
$$
 Then, the bilinear from $\textbf[ \ , \ \textbf]$ is indeed expressed by
 $$
\textbf[h,h'\textbf]=\sum_{\lambda\in\sfP^\vee,w\in W,\mu\in\sfP}b_{\mu,w,\lambda}(\phi(h)h') ~  q^{-2(\mu,\rho_c)}\, t_w \,q^{2(\lambda,\rho_c)}~.
 $$
Since $\phi$ and the generators $\tau_+,\tau_-\in \mathrm{SL}(2,\Z)$ obey the relation $\phi\tau_-\phi=\tau_+$, we have
\be\label{tau-bracket}
 \textbf[\tau_+(h),h'\textbf]=\textbf [h,\tau_-(h')\textbf]~.
\ee

 \subsubsection{Degenerations} \label{app:degeneration}

 There is a two-step degeneration of DAHA:
\bea\nonumber\begin{array}{ccccc}
\textrm{DAHA}&&  \textrm{trigonometric}&&\textrm{rational}\\
\HH(W)  &\rightsquigarrow &\HH^{\textrm{tri}} (W) &\rightsquigarrow &
\HH^{\textrm{rat}} (W)
\end{array}\eea

\subsubsection*{Trigonometric degeneration}

 Let $c_i$, $i = 0,\dots , n$ be formal variables such that $c_i = c_j$ whenever $s_i$ and $s_j$ are conjugate. We will also take commuting variables $\hat x_1, \dots , \hat x_n$ and, for $\mu\in \sfP$, we will denote
$$\hat x_\mu :=\sum( \mu,\a_j^\vee) \hat x_j~.$$
The extended affine Weyl group $\dt W^e$ acts on the space $\bC[c,\hbar][\hat x_1,\dots,\hat x_n]$ where  $s_1, \dots , s_n$ acts in the standard way $s_i(\hat x_\mu) = \hat x_{s_i(\mu)}$ for $i = 1,\dots,n$.  On the other hand, the affine translation involves $\hbar$ as $\mathsf{t}(\lambda)(\hat x_\mu) = \hat x_\mu -(\mu,\lambda)\hbar$ for $\lambda\in \sfP^\vee$.

 The \emph{trigonometric Cherednik algebra} (a.k.a. \emph{graded Cherednik algebra} \cite{oblomkov2016geometric} or \emph{degenerate double affine Hecke algebra} \cite[\S1.6]{Cherednik-book} also \cite{Schiffmann:2012aa}), $\HH^{\textrm{tri}}(W)$ is the $\bC[c, \hbar]$-algebra generated by the extended affine Weyl group $\dt W^e$ and pairwise commuting variables $\hat x_1, \dots , \hat x_n$, subject to the following relations:
 \bea
 s_i\hat x_\mu-\hat x_{s_i(\mu)}s_i &= -c_i(\mu,\a_i^\vee), \qquad \textrm{for} \ i=1,\ldots, n \cr
 s_0\hat x_\mu -s_0(\hat x_{\mu})s_0 &=c_0(\mu,\theta^\vee), \cr
 \pi_r\hat x_\mu &= \hat x_{\pi_r(\mu)}\pi_r.
 \eea
This algebra can be obtained by taking the leading relation in the $\beta$ expansion from DAHA $\HH(W)$ as in Appendix~\ref{app:Trigonometric-A1}.

Since the variables $Y, \hat x$ are not symmetric, the algebra $\HH^{\textrm{tri}}(W)$ admits two polynomial representations.
One is called the \emph{differential polynomial representation} on the group algebra $\bC[c,\hbar][Y^{\sfP^\vee}]$. The generators  $Y$ and $w\in W$ act naturally, and $\hat x_\mu$ acts via the
\emph{trigonometric differential Dunkl operator}
$$ D^{\textrm{tri}}_\mu:=\hbar \partial_\mu+\sum_{\a\in \sfR_+} \frac{c_\a(\mu,\a^\vee)}{1-Y^{-\a^\vee}} (\textrm{id}-s_\a)-(\mu,\rho_c^\vee)~.$$
Here the derivative $ \partial_\mu$ acts on $\bC[c,\hbar][Y]$ of the weight lattice $\sfP^\vee$as
$$ \partial_\mu(Y^\lambda) = (\mu,\lambda)Y^\lambda ~,$$
and $\rho_c^\vee$ is the formal expression
$$\rho_c^\vee:=\frac12\sum_{\a\in \sfR_+} c_\a\a^\vee~.$$

The other is called the \emph{difference-rational polynomial representation} on the group algebra $\bC[c, \hbar ][\hat x]$. The generator $\hat x$ acts by multiplication.
By defining the \emph{Demazure-Lusztig operators}
$$S_i:=s_i+\frac{c_i}{\hat x_{\a_i}}(s_i-\textrm{id})~$$
for $i=0,1,\dots,n$, the action of the extended affine Weyl group $\dt W^e$ is given by $S_w:=\pi_rS_{i_1}\cdots S_{i_j}$ for $\dt W^e\ni w=\pi_rs_{i_1}\cdots s_{i_j}$ a reduced expression.

\subsubsection*{Rational degeneration}
The \emph{rational Cherednik algebra} $\HH^{\textrm{rat}}(W)$ \cite{berest2003finite} is the $\bC[c,\hbar]$-algebra generated by $\bC[\sfP]$, $\bC[\sfP^\vee]$ and $W$ subject to the relations
$$w\hat x=w(\hat x)w~,\quad w\hat y=w(\hat y)w~,\quad  [\hat y,\hat x]=\hbar(\hat x,\hat y) - \sum_{\a\in \sfR_+} c_\a (\a,\hat y) (\hat x,\a^\vee) s_\a$$
 for $w\in W$, $\hat x\in\bC[\sfP]$ and $\hat y\in \bC[\sfP^\vee]$. This algebra can be obtained from trigonometric Cherednik algebra $\HH^{\textrm{tri}}(W)$ by taking the leading relation in the $\beta$ expansion as in Appendix~\ref{app:Rational-A1}.

The rational Cherednik algebra $\HH^{\textrm{rat}}(W)$ admits a polynomial representation on the space $\bC[c,\hbar][\sfP]$ where $\hat x$ and $w$ act in the standard way, and $\hat y$ is assigned to the \emph{rational Dunkl operator}
\be\label{rational-Dunkl} D^{\textrm{rat}}_{\hat y} :=\hbar \partial_{\hat y}- \sum_{\a\in \sfR_+} c_\a\frac{(\a,y)}{\a}(\textrm{id}-s_\a )~,\ee
 where $\partial_{\hat y}(\hat x) = (\hat x,\hat y)$, $\hat x\in\sfP$.
 \bigskip

 \subsubsection*{Spherical subalgebras}

We can also define spherical subalgebras of both trigonometric and rational Cherednik algebras by
$$\SH^{\textrm{tri}}(W):=\mathbf{e} \HH^{\textrm{tri}}(W)\mathbf{e}~, \qquad \SH^{\textrm{rat}}(W):=\mathbf{e} \HH^{\textrm{rat}}(W)\mathbf{e}~,$$
where  $\mathbf{e} $ now is the trivial idempotent of the group $W$
 $$\mathbf{e} = \frac{1}{|W|} \sum_{w\in W}w~.$$

At $c=0=\hbar$, the spherical subalgebras of trigonometric and rational Cherednik algebras become the commutative rings
\be\label{trig-rat-SH}
\SH^{\textrm{tri}}(W)\Big|_{c=0=\hbar}=\bC[\frakt_\bC\times T_\bC]^W~,\qquad \SH^{\textrm{rat}}(W)\Big|_{c=0=\hbar}=\bC[\frakt_\bC\times \frakt_\bC]^W~.
\ee
Therefore, they are the coordinate rings of the \HK manifolds $$\frac{\frakt_\bC\times T_\bC}{W}~, \quad\textrm{and}\quad \frac{\frakt_\bC\times \frakt_\bC}{W}~,$$ respectively \cite{bezrukavnikov2005equivariant} where their complex structures inherit from their elliptic origin \eqref{Mflat-T2}.

\subsection{\texorpdfstring{DAHA of type $A_1$}{DAHA of type A1}}\label{app:DAHA-A1}
Now we will study DAHA $\HH:=\HH(\Z_2)$ of type $A_1$ more in detail. In the case of type $A_1$, the relations $1,\dots,5$ in Appendix~\ref{app:DAHA-gen} amount to
$$\HH=\C_{q,t}\bigl[ T^{\pm 1}, X^{\pm 1}, \pi^{\pm 1} \bigr] \bigg/
\left\{ \begin{array}{cc}TXT=X^{-1}~,
& \pi X\pi^{-1}=qX^{-1}~,\\
\pi^2=1~, & (T-t)(T+t^{-1})=0\end{array}\right\} .$$
Here $T$ and $X$ generate affine Hecke algebra $\dt H^X$ in the Bernstein presentation whereas $T$ and $\pi$ generate the other  affine Hecke algebra $\dt H^Y$ in the Coxeter presentation.  One can write $\dt H^Y$  in the Bernstein presentation by introducing $Y= \pi T$, yielding another representation of  $\HH$
\begin{equation}\label{DAHA-A1-2}
\HH=\C_{q,t}\bigl[ T^{\pm 1}, X^{\pm 1},Y^{\pm 1} \bigr] \bigg/
\left\{ \begin{array}{cc}TXT=X^{-1}~,
 & Y^{-1}X^{-1}YXT^2=q^{-1}~,\\
TY^{-1}T=Y~,& (T-t)(T+t^{-1})=0\end{array}\right\} .
\end{equation}
In this representation, the topological interpretation using the punctured torus becomes apparent as we have seen in Appendix~\ref{sec:DAHA-main}.

\subsubsection{Polynomial representation}\label{app:DAHA-poly}
 The \emph{polynomial representation} of DAHA of type $A_1$
\be\label{poly-rep-nonsym}
\pol:\HH \rightarrow \textrm{End} (\C_{q,t}[X^\pm])
\ee
is explicitly given by
$$T\mapsto ts+\frac{t-t^{-1}}{X^2-1}(s-1)~,\quad \pi \mapsto s\varpi~,\quad
X\mapsto X~,\quad Y\mapsto s\varpi T~,$$
where $s$ is the reflection $s(X)=X^{-1}$, and $\varpi$ is the shift operator $\varpi(X)=qX$.

To study the polynomial representation, we introduce a basis of $\C_{q,t}[X^\pm]$ spanned by \emph{non-symmetric Macdonald polynomials} of type $A_1$
$$
E_{j}(X;q,t)=\begin{cases}  X^{j}\, {}_2\phi_1(q^{-2j+2}, t^2;q^{-2j+2} t^{-2};q^2  ;q^2t^{-2}X^{-2})  & \textrm{if}    \ j>0 \\  X^{-j}\, {}_2\phi_1(q^{2j}, t^2;q^{2j} t^{-2};q^2  ;t^{-2}X^{2}) & \textrm{if}    \ j\le 0 \end{cases}~,
$$
where ${}_2\phi_1(a_1,a_2;b_1;q;X)$ is the basic hypergeometric series.
As we have seen in \eqref{non-sym-Macdonald}, they are eigenfunctions of the $Y$-operator under the polynomial representation
\be\label{Y-action}
\pol(Y)E_j=\begin{cases}q^{-j} t^{-1} E_j&\textrm{if}\ j>0  \\ q^j t E_j&\textrm{if}\ j\le0 \end{cases}~.
\ee
If we introduce a total ordering on the monomials
$$
1\prec X\prec X^{-1} \prec X^2 \prec X^{-2} \prec X^{3}\prec X^{-3}\prec \cdots~,
$$
then they are subject to the condition
$$
E_j(X;q,t)=X^j+ \textrm{lower order terms}~.
$$
The coefficients of the polynomials $E_j$ are rational functions of $q$ and $t$ in general.

For generic $q$ and $t$, the non-symmetric Macdonald polynomials $E_j$ are recursively determined by using the \emph{intertwining} (raising) operators \cite[\S2.6.2]{Cherednik-book}.
Explicitly, they are defined as
$$
\sfA:=XT~,\qquad \sfB:=t\Bigl(T+\frac{t-t^{-1}}{Y^{-2}-1}\Bigr)~,
$$
and they act on non-symmetric Macdonald polynomials as
\bea\label{AB-action}
\pol(\sfA)\cdot E_j&=t^{-\textrm{sign}(j)}  E_{1-j}~,\\
\pol(\sfB)\cdot E_j&=\begin{cases}  E_{-j}& \textrm{for} \ j>0 \\  \frac{t (1-q^{-2j}) (1-q^{-2j} t^4)}{(1-q^{-2j} t^2)^2} E_{-j} & \textrm{for} \ j\le 0 \end{cases}~.
\eea
In fact, these intertwining operators $\sfA,\sfB$ provide an inductive procedure for calculating
the non-symmetric polynomials
\be\label{total-ordering}1=E_0\stackrel{\sfA}{\longleftrightarrow}E_1\stackrel{\sfB}
{\longleftrightarrow}E_{-1}\stackrel{\sfA}{\longleftrightarrow}
E_2\stackrel{\sfB}{\longleftrightarrow}E_{-2}\stackrel{\sfA}
{\longleftrightarrow}\cdots.
\ee
Because of $t\in \C^\times$, the intertwining operator  $\sfA$ is always invertible. Hence, in order to classify finite-dimensional representations, one has to understand when $\sfB$ is not invertible.
Hence, as we have seen in \eqref{lowering} and \eqref{factor-vanish}, finite-dimensional modules arise when the factor
\be\label{factor-vanish2}
 \frac{t (1-q^{-2j}) (1-q^{-j} t^2)(1+q^{-j} t^2)}{(1-q^{-2j} t^2)^2}
\ee
vanishes where $j\le0$. Then, it is easy to see that there is a one-to-one correspondence between finite-dimensional modules of $\HH$ and $\SH$ because the shortening conditions for \eqref{factor-vanish} and \eqref{factor-vanish2} are equivalent. In other words, for finite-dimensional representations, one can easily verify the consequence of the Morita equivalence \eqref{Morita-equiv} explicitly by the action of the intertwining operators. Therefore, the classification of the finite-dimensional representations of $\HH$ \cite[\S2.8]{Cherednik-book} is in one-to-one correspondence with that of $\SH$ in \S\ref{sec:finite-rep}.

The symmetric Macdonald polynomials $P_j$ \eqref{sym-Macdonald-A1} for $j\ge0$ can be obtained by symmetrizing $E_{\pm j}$ with the idempotent $\mathbf{e}=(T+t^{-1})/(t+t^{-1})$
\bea
P_{j}(X;q,t)=
\begin{cases} (1+t^2)\,\mathbf{e} E_j (X;q,t)~, \\ \frac{(1+t^2)(1-q^{2j}t^2)}{(1-q^{2j}t^4)}\, \mathbf{e} E_{-j} (X;q,t)~, \end{cases}
\eea
or by the combination of $E_j$ and $E_{-j}$
\be\label{PEE}
P_{j}(X;q,t)=E_{-j}(X;q,t)+\frac{t^2(1-q^{2j})}{1-q^{2j}t^2}E_{j}(X;q,t)~.
\ee
On the other hand, non-symmetric Macdonald polynomials $E_j$ can also be obtained from symmetric ones $P_j$ as follows. If one expresses
$$
P_j=\sum_{i=0}^j c_{j,i}(q,t) (X^i+X^{-i})~,
$$
then
$$
E_j=\sum_{i=-j}^j d_{j,i}(q,t)X^i   \qquad \textrm{where} \qquad d_{j,i}(q,t):=\frac{1-q^{i+j}}{q^{i-j}-q^{i+j}}c_{j,i}(q,t)~.
$$
Furthermore, one can easily determine $E_{-j}$ for $j>0$ from $P_j$ by \eqref{PEE}. Hence, the polynomial representations for DAHA $\HH$ \eqref{poly-rep-nonsym} and spherical DAHA $\SH$ \eqref{poly-rep-sym} are equivalent, which follows from the Morita equivalence \eqref{Morita-equiv}.

To conclude this subsection, let us analyze how Macdonald polynomials are transformed under $\mathrm{PSL}(2,\Z)$ transformation. In fact, under the $\tau_-$ transformation \eqref{MCGloop}, non-symmetric Macdonald polynomials are subject to
\be\label{tau-e}
\tau_{- } ( E_{j} ) = q^{-\frac{j^{2 } }2}t^{- | j |} E_{j }~,
\ee
for  $j\in \bZ$. This can be proved by the induction in the sequence \eqref{total-ordering}. For $j=0$, it is easy to see that $\tau_{- } ( 1)=1$. Suppose that  $q,t\in \C^\times$ are generic so that all non-symmetric Macdonald polynomials are well-defined, and \eqref{tau-e} holds for $E_j$. Since it follows from \eqref{MCGloop} that the operator $\sfB$ commutes with $\tau_-$, $E_{-j}$ transforms in the same way as $E_j$ under  $\tau_-$. For $j<0$, we have
\bea
\tau_-(E_{1-j})=&t^{-1} q^{\frac12} \pol(YXT) \cdot\tau_-(E_{j}) \\[1.5ex]
\stackrel{\textrm{induction}}{=}&t^{-1} q^{\frac12} q^{- \frac{j^{2 }} 2}t^{j }  \pol(YXT) \cdot E_{j}\\[1.5ex]
\stackrel{\eqref{AB-action}}{=}&q^{\frac12} q^{-\frac{j^{2 }} 2}t^{j }  \pol(Y) \cdot E_{1-j}\\[1.5ex]
\stackrel{\eqref{Y-action}}{=}&q^{- \frac{(1-j)^{2 } }2}t^{-1+ j }  E_{1-j}~.
\eea
This completes the proof.

Furthermore, one can show from \eqref{MCG-quantum} that the raising \eqref{raising} and lowering \eqref{lowering} operators behave under the $\tau_-$ transformation as
$$ \tau_-(\sfR_j)=q^{-j-\frac12}t^{-1}\sfR_j~, \qquad  \tau_-(\sfL_j)=q^{j-\frac12}t\sfL_j~.$$
A similar induction shows that the symmetric Macdonald polynomials are transformed  as
\be\label{tau-p}
\tau_{- } ( P_{j} ) = q^{-\frac{j^{2 } }2}t^{-  j } P_{j }~,
\ee
which is consistent with  \eqref{PEE}.
This formula indeed gives the modular $T$-transformation for representations of $\SH$ in \S\ref{sec:2d}.

\subsubsection{Functional representation}\label{app:func-rep}
Here we shall give an explicit definition of the functional representation of DAHA of rank one \cite[\S2.7.1]{Cherednik-book}. Restricted to the spherical DAHA $\SH$, it can be understood as the $S$-transformation $\sigma(\repP)$ of the polynomial representation $\repP$ of $\SH$ as in \S\ref{sec:poly-rep} and \S\ref{sec:modularity}.

The assignment
$$
f(X)~\mapsto~\textbf[f(X),E_n(X)\textbf]/E_n(t^{-1})=\begin{cases} f(q^{n}t )& \textrm{for} \ n>0\\f(q^{n}t^{-1} )& \textrm{for} \  n\le0 \end{cases}
$$
provides the function defined only over the set $\mathsf{Fun}$ of $X=q^{n}t^{\textrm{sgn}(n-\frac12)}$ for $n\in\bZ$. Thus,  one can obtain the discretization of the polynomial representation \eqref{poly-rep-nonsym}  on this set of functions defined over $X\in \mathsf{Fun}$ for $h\in \HH$
\be\nonumber
\mathrm{disc}:\textbf[f(X),E_n(X)\textbf]/E_n(t^{-1}) ~\mapsto~ \textbf[\pol(h)\cdot f(X)  ,E_n(X)\textbf]/E_n(t^{-1}) ~.
\ee
More explicitly, one can read off the action as
\bea\nonumber
(\mathrm{disc}(X)\cdot f)(q^{n}t^{\textrm{sgn}(n-\frac12)})&=q^{n}t^{\textrm{sgn}(n-\frac12)}f(q^{n}t^{\textrm{sgn}(n-\frac12)})~,\cr
(\mathrm{disc}(\pi)\cdot f)(q^{n}t^{\textrm{sgn}(n-\frac12)})&=f(q^{1-n}t^{\textrm{sgn}(\frac12-n)})~,
\eea
and
\bea\nonumber
&(\mathrm{disc}(T)\cdot f)(q^{n}t^{\textrm{sgn}(n-\frac12)})\cr&=\frac{tq^{2n}t^{2\textrm{sgn}(n-\frac12)} -t^{-1}}{q^{2n}t^{2\textrm{sgn}(n-\frac12)}-1}f(q^{-n}t^{\textrm{sgn}(\frac12-n)})-\frac{t-t^{-1}}{q^{2n}t^{2\textrm{sgn}(n-\frac12)}-1}f(q^{n}t^{\textrm{sgn}(n-\frac12)})~.
\eea
Though $q^{-n}t^{\textrm{sgn}(\frac12-n)}$ is not in the set $\mathsf{Fun}$ when $n=0$, $tq^{2n}t^{2\textrm{sgn}(n-\frac12)} -t^{-1}=0$ in this case. Therefore, it is well-defined.

The functional representation is defined by the symmetric bilinear form \eqref{sym-form}. If we restrict the functional representation to the spherical subalgebra $\SH$, it is the modular $S$-transformation $\sigma$ of the polynomial representation due to \eqref{phi-sigma}.

\subsubsection{\texorpdfstring{Trigonometric Cherednik algebra of type $A_1$}{Trigonometric Cherednik algebra of type A1}}\label{app:Trigonometric-A1}

To find the trigonometric degeneration, we set $X =e^{\beta \hat x}, q=e^{\beta \hbar}, t=e^{\beta \hbar c}$ and $T =s e^{\beta cs}$, where $s\in \Z_2$ is the reflection. Then,  taking the leading order in $\beta$ from \eqref{DAHA-A1-2}, the generators $s,\hat x,Y$ satisfy the following relations:
\be \label{tri-SH-A1}s^2 =1~,\quad sY^{-1}s =Y~,\quad s\hat x+\hat xs=-2c~,\quad Y^{-1}\hat xY-\hat x=\hbar+2cs~.\ee
The algebra generated by $s,\hat x,Y$ with these relations is called the trigonometric Cherednik algebra $\HH^{\textrm{tri}}_{\hbar,c}$ of type $A_1$. It enjoys a $\Z_2$-symmetry
\be \label{tri-xi2}
\xi_2:s \mapsto s~,\ \hat x \mapsto \hat x ~, \  Y\mapsto -Y~, \ \hbar\mapsto \hbar~, \ c\mapsto c~,
\ee
which is reminiscent of \eqref{automorphisms1}.
The idempotent is defined by
\be\label{idem-deg}
\mathbf{e}:=\frac{1+s}{2}~,
\ee
and the spherical subalgebra is defined as $\SH^{\textrm{tri}}_{\hbar,c}:=\mathbf{e}\HH^{\textrm{tri}}_{\hbar,c}\mathbf{e}$, which is generated by
\be v:=\mathbf{e}\hat x^2\mathbf{e}~, \quad \textrm{and} \quad y=\mathbf{e}Y\mathbf{e}\ee

\subsubsection{\texorpdfstring{Rational Cherednik algebra of type $A_1$}{Rational Cherednik algebra of type A1}}\label{app:Rational-A1}

In the trigonometric Cherednik algebra $\HH^{\textrm{tri}}_{\hbar,c}$, we write $Y = e^{\beta \hat y}$ and we rescale $\hat x \mapsto \hat x/\beta$. Then, the generators $s, \hat x, \hat y$ satisfy the following relations
$$s^2 =1~,\quad s\hat x=-\hat xs~,\quad s  \hat y=- \hat y s~,\quad [\hat x,\hat y]=\hbar+2cs~,$$
in the leading order of $\beta$.
The algebra generated by $s, \hat x, \hat y$ with these relations is the rational Cherednik algebra
$\HH^{\textrm{rat}}_{\hbar,c}$.

Let us study the spherical subalgebra $\SH^{\textrm{rat}}_{\hbar,c}:=\mathbf{e}\HH^{\textrm{rat}}_{\hbar,c}\mathbf{e}$ where the idempotent $\mathbf{e}$ is the same as  \eqref{idem-deg}.
The spherical subalgebra $\SH^{\textrm{rat}}_{\hbar,c}$ is a $\bC[c,\hbar]$-algebra generated by
$$E = \frac12\mathbf{e}\hat x^2\mathbf{e}~, \qquad F = -\frac12\mathbf{e}\hat y^2\mathbf{e}~, \qquad   H= -\mathbf{e}\hat x\hat y \mathbf{e}~.$$
with relations
\bea\label{sl2}
[{E},{F}]=\hbar\Bigl({H}+\frac\hbar2+c\Bigr)~,&\qquad [{H},{E}]=2\hbar {E} ~,\qquad [{H},{F}]=-2\hbar {F} ~,\cr
&4{E}{F}={H}(-{H}+\hbar -2c)~.
\eea
Thus, $E,F,H$ can be interpreted as the $\fraksl(2)$-triple.
Writing the Casimir element of $\fraksl(2)$  as
$$\Omega = EF+FE +\frac12 \Bigl(H+\frac\hbar2+c\Bigr)^2~,$$
it is easy to verify that there is an isomorphism \cite[Proposition 8.2]{etingof2002symplectic}
\be\label{SH-rat-A1}
U(\fraksl(2))/\langle \Omega-\tfrac12(c+\tfrac12\hbar)(c-\tfrac32\hbar) \rangle \cong \SH_{\hbar,c}^{\textrm{rat}}~.
\ee

At $c=0$, the rational Cherednik algebra of type $A_1$ is isomorphic to
$$
\HH^{\textrm{rat}}_{\hbar,c=1}=\OO^q(\bC^2)\rtimes \bC[\bZ_2]~,
$$
where $\hat x, \hat y$ are the generators of $\OO^q(\bC^2)$ and $s$ is the generator of $\bZ_2$. If we consider the open set $\bC^\times \times \bC\subset \bC^2$,
the deformation quantization $\OO^q(\bC^\times \times \bC)$ of its coordinate ring is generated by $\hat x^{\pm}, \hat y$. Then, we can define a subalgebra generated by
\bea\label{sl2-2}
E:=\frac{\hat x^2}{2}~, \qquad F:= -\frac{\hat y^2}2 - c \hat x^{-1} \hat y ~,\qquad
H :=  - \hat x\hat y ~,
\eea
which is isomorphic to the spherical subalgebra $\SH^{\textrm{rat}}_{\hbar,c}$ with generic $c$ given in \eqref{SH-rat-A1} with the relations \eqref{sl2}. This construction indeed provides the polynomial representation \eqref{rational-Dunkl} of the spherical subalgebra of $\SH^{\textrm{rat}}_{\hbar,c}$ acting on $\bC[x^{\pm2}]$ by substituting $-\hbar\partial_x$ for $\hat y$.

\section{Quantum torus algebra}
\label{app:qt}

In this appendix, we study representations of quantum torus algebras and their symmetrization. We aim at understanding representations of the simplest quantum torus algebra in terms of 2d $A$-model on the affine variety $\C^\times\times \C^\times$. Therefore, we first review representation theory of the quantum torus algebra, and we subsequently find $A$-branes for two important families of representations of the quantum torus algebra: \emph{cyclic representations} and \emph{polynomial representations}.  The quantum torus algebra and DAHA are closely related, and the relationship becomes clearer when we symmetrize the algebra and representations by the distinguished outer automorphism $\Z_2$. Although representations of the quantum torus algebra and symmetrized quantum torus are rather simple, they play a helpful guide for analogous analysis for DAHA in \S\ref{sec:2d}.

\subsection{Representations of quantum torus algebra}\label{app:rep-QT}

It is well-known that the so-called quantum torus algebra is given by
\begin{equation}
\QT = \C[ q^\pm ]\langle X^\pm, Y^\pm \rangle / ( Y^{-1}X^{-1}YX = q^{-1})~.
\label{eq:Rdef}
\end{equation}
Here the generator $q$ is a formal parameter that commutes with both $X$ and~$Y$. This algebra occurs throughout physics, and it is indeed a close cousin of the most basic ingredient in the whole theory of quantum mechanics. If we allowed arbitrary (rather than integral) powers of $X$ and $Y$, and set $q=\exp(2\pi i \hbar)$, it would be the group algebra of the exponentiated form of the canonical commutation relations.
Here, since only integral powers are allowed, we are dealing with a discrete-valued analogue of canonically conjugate variables.

In fact, $\QT$ is the group algebra of a specific Heisenberg group, associated to a discrete analogue of a symplectic pairing.
As is well-known, the linear representations of a discrete group correspond precisely to modules over its group algebra, so the presence of the scalars is irrelevant: studying $\QT$-modules is the same as studying representations of this Heisenberg group.

To construct it, we introduce Heisenberg groups in somewhat greater generality. Let $S$ be any locally compact topological ring with characteristic not equal to two, and $\lat{V}$ a free, finitely generated $S$-module (so isomorphic to~$S^n$). A symplectic pairing on~$\lat{V}$ is defined to be a function
\begin{equation}
\omega: \lat{V} \times \lat{V} \rightarrow S~,
\end{equation}
which is skew-symmetric and $S$-bilinear. (Note that non-degeneracy condition on~$\omega$ is not necessary here so that we use a pairing rather than a symplectic form. However, in all the examples we consider below, $\omega$ is a non-degenerate symplectic form.)
Then, the following is true:

\begin{theorem}[{\cite[Theorem~7.1]{Kleppner}}]
The group $H^2(\lat{V},S)$, which classifies group extensions of the additive group~$\lat{V}$, is isomorphic to the group of symplectic pairings on~$\lat{V}$.
\end{theorem}
(It is a pleasant exercise to check that any bilinear function defines an inhomogeneous group cochain. To obtain the full result, one shows that the antisymmetrization of any cochain is necessarily bilinear, and actually defines a set of unique representatives for the cohomology.)

There is therefore a unique central extension of groups defined by the sequence
\begin{equation}
0\rightarrow S \rightarrow \Heis(\lat{V},\omega)\rightarrow \lat{V} \rightarrow 0~,
\end{equation}
which is called the Heisenberg group associated to the data $(\lat{V},\omega)$. In our example, we take $S$ to be~$\Z$, $\lat{V}$ to be~$\Z^2$, and the pairing to be defined by sending a fixed ordered basis of~$\lat{V}$ to $+1$. It is then easy to see that $\QT$ is the group algebra of~$\Heis(\Z^2,\omega)$, and that we can similarly define the quantum torus algebra
\deq{
\QT(\lat{V},\omega) = \C[\Heis(\lat{V},\omega)]
}
associated to a general finitely generated free $\Z$-module with symplectic pairing. In the sequel, \QT{} without decoration will always refer to the standard module $\Z^2$ with its standard pairing.

Our algebra $\QT$ has an interesting collection of outer automorphisms: in fact,
\begin{equation}\label{Out-QT}
\Out(\QT)  = \Sp(2,\Z) = \SL(2,\Z).
\end{equation}
The inner automorphisms just consist of shifting $X$ and~$Y$ by powers of~$q$, and so make up the group~$\Z^2$. The full automorphism group is a central extension of $\SL(2,\Z)$ by this module:
\begin{equation}
\begin{tikzcd}[column sep = 1.2 em]
0  \ar[r] & \Inn(\QT) \ar[r] & \Aut(\QT) \ar[r] & \Out(\QT) \ar[r] &  0 \\
0 \ar[r] & \Z^2 \ar[r] \arrow{u}{\cong} & \Z^2 \rtimes \SL(2,\Z) \ar[r] \arrow{u}{\cong}& \SL(2,\Z) \ar[r] \arrow{u}{\cong}&  0
\end{tikzcd}~.
\end{equation}
More generally, for the Heisenberg group on a higher-dimensional standard symplectic lattice $\Z^{2n}$, the group of automorphisms is the semidirect product $\Z^{2n} \rtimes \Sp(2n,\Z)$.

Heisenberg groups over more general rings now make starring appearances in numerous branches of mathematics, including even number theory~\cite[for example]{Weil}. For that reason (and to highlight their connection to DAHA), we have chosen to emphasize the generality of the concept here even though the cases of most interest to everyday physics are $S=\R$ or~$\Z$. For some physical applications of the latter algebra in the context of Chern--Simons theory, see for example~\cite{Gukov:2003na} and related work on  the $AJ$-conjecture.

\subsubsection{Unitary representations}\label{app:unitary}

The unitary representation theory of Heisenberg groups is well-understood, thanks to the Stone--von~Neumann theorem and its generalizations due, in particular, to Mackey~\cite{MackeySvN,MackeyBook}. We will just quickly recall enough context to state the main theorem. (This subsection is intended to mathematically classify unitary representations of $\QT$ so that the reader could skip to the three cases at the end for the first reading.)  Let us start quoting the following definition:
\begin{definition}
Let $G$ be a locally compact topological group and $\Gamma$ a Borel space with $G$-action. A \emph{system of imprimitivity based on~$(G,\Gamma)$} consists of a separable Hilbert space with a unitary, strongly continuous representation of~$G$, together with
 a  $G$-equivariant, projection-valued measure, assigning projection operators on the Hilbert space to the Borel sets of~$\Gamma$.
\end{definition}
The basic example of such data is the space of $L^2$ functions with respect to a $G$-invariant measure on~$\Gamma$; $G$ is then represented by pull-back, and the projections are just characteristic functions of subsets. Actually, one can show that essentially all systems of imprimitivity arise in this fashion.

The most interesting examples for physics arise when a group $\hat{G}$ is the semidirect product of an abelian normal subgroup $A$ with some other group $G$ that acts on~$A$: that is, there is a  split short exact sequence
\deq[eq:compseq]{
0 \to A \to \hat{G} \to G \to 0~.
}
Then one chooses $\Gamma$ to be the character space of the abelian normal subgroup $A$, with its obvious $G$-action.
\begin{theorem}
The set of isomorphism classes of unitary irreducible representations of~$\hat{G}$ is in bijection with the  set of systems of imprimitivity based on~$(G, \Gamma)$, up to unitary equivalence.
\end{theorem}
It is further true~\cite{Varadarajan} that when $G$ acts transitively on~$\Gamma$, the inequivalent systems of imprimitivity based on~$(G,\Gamma)$ are in one-to-one correspondence with unitary representations of the point stabilizer subgroup
\begin{equation}
G_x = \{ g \in G | g\cdot x = x \} \quad (x \in \Gamma)~.
\end{equation}
This means that unitary representations of~$\hat{G}$ are classified by the $G$-orbits of~$\Gamma$, together with unitary representations of the corresponding stabilizer subgroups $G_x$.

For the physically minded reader, this can be understood as generalizations of Wigner's technique of the little group~\cite{Wigner}.  There, the decomposition~\eqref{eq:compseq} for the Poincar\'e group is the obvious (and canonical) one,
\begin{equation}
0 \to \R^d \to \R^d \rtimes \SO(d) \to \SO(d) \to 0~.
\end{equation}
The theorem of Stone and von~Neumann is indeed an example of precisely the same very powerful general logic.

To apply this technique to representations of~$\Heis(\lat{V},\omega)$, we must choose a maximal isotropic subspace $\lat{L}$ of~$(\lat{V},\omega)$.
There is then a unique abelian normal subgroup $A_\lat{L}$ of the Heisenberg group determined by~$\lat{L}$, which can be non-canonically identified with $\lat{L} \oplus S$.  In fact, this choice defines a (non-canonical) semidirect product structure
\begin{equation}
0 \rightarrow A_\lat{L} \rightarrow \Heis(\lat{V},\omega) \rightarrow \lat{V}/\lat{L} \rightarrow 0.
\end{equation}
Understanding representations thus reduces to studying the $(\lat{V}/\lat{L})$-orbits in~$\Gamma$. For simplicity, let us reduce to the one-dimensional Heisenberg group, where $\lat{V} = S^2$ and $\omega$ is the standard Darboux pairing. Then $\Gamma = S^\vee \times S^\vee$ is the product of two copies of the Pontryagin dual of~$S$.  If we abuse our exponential notation from~\eqref{eq:Rdef} above, we can write an element of $\Heis(\lat{V},\omega)$ as $q^a X^b Y^c$, where $a,b,c\in  S$. Then, $A_\lat{L}$ is the set of elements of the form $q^a X^b$, and a splitting of the short exact sequence is provided by the subgroup of elements $Y^c$. The action of $\lat{V}/\lat{L}$ is then determined by the formula
\begin{equation}
Y^{-c} \Bigl( q^a X^b \Bigr) Y^c = q^{a-bc} X^b~,
\end{equation}
which also specifies the action on~$\Gamma$ by pullback.

In our example of quantum torus algebra, $S=\Z$, so that $\Gamma = \U(1) \times \U(1)$. If we abuse notation, we can denote the image of $q$ (the choice of central character) simply by $q$. Of course, Schur's lemma ensures that the center must act by a chosen central character in any irreducible representation. This central character is nothing other than the numerical value of~$\hbar$, and plays the same role in this more general context. As such, no confusion should arise, and we will frequently confuse formal and numerical $q$ in what follows. We will denote the image of~$X$ by~$x_1$. Then the group action on $\Gamma$ is just
\begin{equation}
Y : (q,x_1) \mapsto (q,x_1/q)~,
\end{equation}
and of course fixes the central character.

It is straightforward to enumerate the orbits:
\begin{enumerate}[nosep]
\item $q = 1$. An orbit is any single point $x_1 \in \U(1)$. The stabilizer is~$\Z$. A representation is thus determined by $x_1$ and another phase $y_1$. These are the abelian representations of~$\Z^2$.
\item $q\in \U(1)$ is \emph{not} a root of unity. There is a single orbit. The stabilizer is trivial. There is thus a unique irreducible representation up to isomorphism. (This is the Stone--von~Neumann theorem.)
\item $q$ is a primitive $m$-th root of unity $\mu_m$. The space of orbits is $\U(1)/\mu_m$, each orbit consisting of $m$ points. The stabilizer is $m\Z\subset \Z$. A representation is again determined by two phases, one being the $m$-th power of~$x_1$, the other being (in a precise sense) the ``$m$-th power of $y_1$.''
\end{enumerate}
In fact, we can construct the corresponding representations explicitly; we do this (in greater generality) in what follows.

\subsubsection{Non-unitary representations}

Let us first fix $q$ to be a primitive $m$-th root of unity, corresponding to the third case above. There is then an obvious unitary finite-dimensional representation of the quantum torus algebra given by the well-known ``clock'' and ``shift'' matrices, acting on a standard vector space $\C^m$ with basis $e_i$:
\begin{equation}
\xi e_i = e_{i+1}, \quad \varpi^- e_i = q^{-i} e_i.
\end{equation}
It  is clear that mapping $X\mapsto  \xi$ and  $Y \mapsto \varpi^-$ defines the structure of an $m$-dimensional $\QT$-module with central character $q$.  What is  perhaps slightly less obvious is that every such representation is closely related to this one, differing only by a rescaling.

\begin{definition}[cyclic representation]
Fix a pair $\lambda = (x_m,y_m) \in \C^\times \times \C^\times$. Let $x_1$ and~$y_1$ be any $m$-th roots of~$x_m$ and~$y_m$ respectively. Then the \emph{cyclic representation of weight~$\lambda$}, denoted $\repF^\lambda$, is defined on an $m$-dimensional complex vector space by the map
\begin{equation}\label{cyclic-rep}
\QT \rightarrow \End(\repF^\lambda):  X\mapsto x_1 \xi, ~ Y \mapsto y_1 \varpi^-.
\end{equation}
(Note that the isomorphism type does not depend on the choices of~$x_1$ and~$y_1$.) We may sometimes write $\repF^\lambda_m$ to specify the dimension of the representation.
\end{definition}

For an $m$-th primitive root of unity $q$, it is known \cite{Kashaev:1994pj,Bai,Kim} that every representation of~$\QT$ in which $X$ and $Y$ are invertible is a direct sum of irreducible representations of the form~$\repF^\lambda$ for some choices of weights. Two representations $\repF^\lambda$ and~$\repF^{\lambda'}$ are isomorphic precisely when $\lambda=\lambda'$.
Note that the weight $\lambda$ and $q$ together define a central character in this instance. Any representation must factor through  the quotient $\QT/( q^m-1 )$, in which the elements $X^m$ and $Y^m$ are also central.

\bigskip

If $q$ is generic, we have no hope of finding a finite-dimensional representation by the results of the previous subsection, but the construction above generalizes to produce an infinite-dimensional module. Let $\PR = \C[q^\pm,X^\pm]$ be the space of Laurent polynomials in the central generator $q$ and the variable $X$.
We can define an analogue of the operator $\varpi^-$ above by the rule
\begin{equation}\label{shift-operator}
\varpi^- X^i = q^{-i} X^i.
\end{equation}
The role of the matrix $\xi$ will be played simply by multiplication by~$X$.

\begin{definition}[polynomial representation]
Choose a weight $y_1 \in \C^\times$. The \emph{polynomial representation of weight $y_1$}, denoted~$\PR^{y_1}$, is defined by the map
\begin{equation}\label{poly-rep-QT}
\QT \to \End(\PR^{y_1}): X \mapsto X,  ~ Y \mapsto y_1 \varpi^-.
\end{equation}
\end{definition}
It is clear that the shift automorphism (multiplication by~$X$) intertwines the representations $\PR^{y_1}$ and~$\PR^{y_1/q}$. However, this automorphism is outer (it is not \emph{conjugation} by $X$, but rather a generator of $\SL(2,\Z)$), so that these representations are not actually isomorphic.

\subsubsection{Geometric viewpoint}

To obtain a geometric perspective on its representation theory, we will take a slightly different way of looking at the algebra $\QT$. If we set $q=1$ (or, equivalently, taking the quotient of~$\QT$ by~$\C[S]$), we obtain a commutative group algebra. We are free to think of any commutative algebra as the coordinate ring of a certain affine space. For our example, it is indeed the product of two punctured affine lines:
\begin{equation}
\QT \xrightarrow[q\to 1]{} \OO(\C^\times \times \C^\times).
\end{equation}
This space will play an important role in the sequel, and so we will denote it by $\MS$. In this paper, we consider the representation theory of more complicated algebras, and $\MS$ will therefore denote more complicated spaces, depending on the context. (As our main interest, \MS{} in \S\ref{sec:2d} is the moduli space of flat $\SL(2,\C)$-connections on the once-punctured torus.)

What is common between all of these examples are certain key properties of~\MS: First of all, it will always be a non-compact complex manifold, so that it has a large and interesting algebra $\OO(\MS)$ of holomorphic functions with polynomial growth at infinity. (In fact, in this paper, \MS{} will always be an affine variety over~$\C$.)
It will also be a holomorphic symplectic manifold; in our current example, there is a natural holomorphic symplectic form on $\C^\times \times \C^\times$, which we take to be
\begin{equation}
\Omega_J = \frac{1}{2\pi i} \frac{dX}{X} \wedge \frac{dY}{Y}.
\end{equation}
In our examples, \MS{} will be even \HK; for reasons that will become clear later, we refer to the complex structure whose holomorphic symplectic structure is of interest to us now as $J$, motivating the choice of notation above. In the example at hand, we can see the \HK structure explicitly by using logarithmic coordinates,
\begin{equation}
X = \exp(r + i \varphi)~, \qquad Y = \exp(\rho + i \phi)~,
\end{equation}
and observing that both the real and imaginary parts of
\begin{equation}
\Omega_J = \omega_K + i \omega_I = \frac{1}{2\pi}\Bigl( dr \wedge d\phi  + d\varphi \wedge d\rho \Bigr) +  \frac{i}{2\pi} \Bigl(  d\varphi \wedge d\phi -dr\wedge d\rho \Bigr)
\end{equation}
are real symplectic forms. The third such symplectic structure is
\begin{equation}
\omega_J = \frac{1}{2\pi}\Bigl(dr\wedge d\varphi + d\rho\wedge d\phi\Bigr);
\end{equation}
it arises by identifying our space with the cotangent bundle of the two-torus.

Of course,  it is also trivial to see that this space admits an elliptic fibration, which can be written very simply in coordinates as
\deq[eq:easy-fib]{
\begin{aligned}
\pi: \MS &\to \C, \\
(X,Y) &\mapsto (r,\rho)~.
\end{aligned}
}
It is apparent that this map is holomorphic in complex structure $I$, and that its fibers are tori which are Lagrangian with respect to $\omega_J$ and~$\omega_K$. As such, in the language of mirror symmetry for \HK manifolds \cite{Kapustin:2006pk,Gukov:2006jk}, they are appropriate submanifolds to support branes of type $(B,A,A)$, which are the only topologically interesting compact Lagrangian branes in the $A$-model of $(\X,\omega_K)$.

In addition, $\QT$ is actually the deformation quantization $\OO^q(\C^\times\times\C^\times)$ of the coordinate ring with respect to the Poisson bracket defined by~$\Omega_J$. We can thus begin to think about the algebra $\QT$ within the context of brane quantization \cite{Gukov:2008ve} in 2d sigma-model, which is reviewed in \S\ref{sec:Bcc} and \S\ref{sec:Lagrangian}. It is the central idea of this paper, and this rather simple setting in this Appendix is instructive guidance for a similar geometric angle on DAHA in \S\ref{sec:2d}. (See also \cite{soibelman2001quantum} for the geometric approach to the quantum tori.)

\subsection{Branes for quantum torus algebra}
\label{app:branesQT}

We have understood the representation theory of~\QT{}. Also, we recall that the curvature $F$ of the $\Bcc$ line bundle in 2d sigma-model on the symplectic manifold $(\X=\C^\times \times \C^\times, \omega_\X)$ is subject to
$$F+B+i \omega_{\mathfrak{X}}=\frac{\Omega_{J}}{i \hbar}~,$$
to obtain $\QT$ as $\Hom(\Bcc,\Bcc)$. Now, we set up an equivalence between the categories $\Rep(\QT)$ and~$\ABrane(\X, \omega_\X)$  by finding $A$-branes for the cyclic representations \eqref{cyclic-rep} and the polynomial representations \ref{poly-rep-QT}.
It is the essential motivation for the study of our main examples, DAHA, but this builds intuition for the sorts of matching in \S\ref{sec:poly-rep} and \S\ref{sec:generic-fiber}.

\subsubsection{Cyclic representations}\label{app:cyclic-rep-QT}

Our first attempt will be to find $A$-branes for the finite-dimensional representations \eqref{cyclic-rep}. On general grounds, these should arise from branes supported on compact Lagrangians.
These branes are special because, from the geometry viewpoint, compact Lagrangian submanifolds are interesting and somewhat rare. Similarly, finite-dimensional representations are rare and distinctive from the standpoint of representation theory.

Due to the simple topology of $\X\simeq T^*T^2$ in this example, there is essentially one interesting class of compact branes: Lagrangian tori wrapping the unique generator of~$H_2(\MS,\Z)$. The condition that Lagrangians have vanishing Maslov class requires that no cycle in $H_1(\L,\Z)$ bound a disc in~\X; as such, it must map injectively to~$H_1(\X,\Z)\cong \Z^2$ under the obvious map.

These tori are holomorphic in complex structure $I$ and holomorphic Lagrangians with respect to $\Omega_I$ so that they can be Lagrangian submanifolds in a symplectic manifold $(\X,\omega_\X)$ only when $\omega_\X$ is proportional to $\omega_K$. According to \eqref{generic-Bcc}, we need to set $\theta=0$, and the parameter $\hbar$ is real.
It is easy to see that each equivalence class of such tori under Hamiltonian isotopies with respect to~$\omega_K$  contains precisely one torus of the form $\pi^{-1}(r,\rho)$, \textit{i.e.}\ a fiber of the map~\eqref{eq:easy-fib}. It is furthermore clear that no two such tori are equivalent. Each $A$-brane thus corresponds to exactly one special Lagrangian brane of type $(B,A,A)$, and branes are indexed by two real parameters $(r,\rho)$ recording their position. (See Figure~\ref{fig:brane-QT}.)

 Now, as we recalled above, an $A$-brane is a Lagrangian equipped with (among other data) a flat unitary line bundle. To choose such a flat line bundle on the torus, we need to additionally choose two $\U(1)$ holonomies, or equivalently pick a point in $\Jac(T^2) \cong (T^2)^\vee$. Thus, we write the two $\U(1)$ holonomies by $(e^{i\varphi^\vee},e^{i\phi^\vee})$.
 The set of $A$-branes is thus labeled  by two real for the position and two circle-valued parameters for the holonomies. This agrees precisely with the set $\lambda=(x_m,y_m)$ of the cyclic representations $\repF^\lambda$ of $\QT$ in \eqref{cyclic-rep} where the identification is
 \be
(x_m,y_m)=(e^{r+i\varphi^\vee},e^{\rho+i\phi^\vee})~.
 \ee

The last thing to check is the dimension of the space of open strings $\Hom(\Bcc,\brL)$. Since $\theta=0$ in \eqref{generic-Bcc}, the data of $\Bcc$ consist of $F+B=\omega_I/|\hbar|$ and $\omega_\X=-\omega_K/|\hbar|$. Also, the Chan-Paton bundle for the Lagrangian brane is endowed with flat connection \eqref{deformed-flat}. As a result, the dimension formula \eqref{dimension} becomes
\be\label{dimension-torus}
\dim \Hom(\Bcc,\brL)=\int_{T^2} \frac{\omega_I}{2\pi \hbar}=\frac1\hbar~.
\ee
Since this is the dimension of holomorphic sections, $\hbar = 1/m$ for some positive integer $m$. We can interpret this as the Bohr--Sommerfeld quantizability condition on the compact Lagrangian branes, and we therefore recover the condition that $q$ is a primitive $m$-th root of unity for the cyclic representation.

The involution
$$\zeta:\X\to \X;~(r,\varphi,\rho,\phi)\mapsto (-r,\varphi,-\rho,\phi)$$
is holomorphic in complex structure $I$ and antiholomorphic in complex structures $J$ and $K$. The central Lagrangian torus $\pi^{-1}(0,0)$ is exactly the set of fixed points under $\zeta$, and the $A$-brane on $\pi^{-1}(0,0)$ with the trivial holonomies gives rise to the unitary finite-dimensional representation, which is the third case at the end of Appendix \ref{app:unitary}.

\begin{figure}[ht]\centering
    \includegraphics[width=0.9\textwidth]{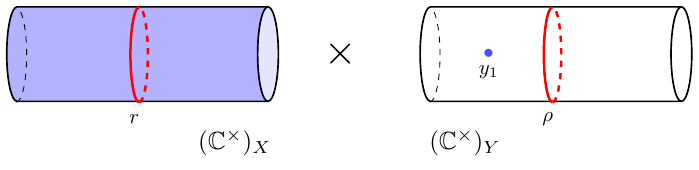}
    \caption{A brane (red) supported on a torus fiber gives the cyclic representation and the deformation parameter consists of its position $(r,\rho)$ and the $\U(1)^2$ holonomy. A brane (blue) supported on $(\C^\times)_X$ provides the polynomial representation and the deformation parameter consists of its position $|y_1|$ and the $\U(1)$ holonomy.}
    \label{fig:brane-QT}
\end{figure}

\subsubsection{Polynomial representations}\label{app:poly-rep-QT}

Next, we will find an $A$-brane corresponding to the polynomial representation $\PR^{y_1}$ in \eqref{poly-rep-QT}. Since it is an infinite-dimensional representation, we expect this brane to be non-compact.

Let us begin by just considering Lagrangian subspaces which are isomorphic to $T^* S^1$, embedded such that the generator of $H_1$ is mapped to the generator of $H^1((\C^\times)_X)$ under the embedding map.
While this choice is not canonical, other classes in $H_1(\C^\times \times \C^\times)$ could be obtained by acting with the outer automorphism group $\SL(2,\Z)$.
We are thus interested in Lagrangians that are graphs $\Gamma_Y$ of topologically trivial maps $Y=(\rho,\phi):(\C^\times)_X\to(\C^\times)_Y$,  which can be represented in the form
\deq[eq:Lag-prep]{
\Gamma_Y:=\{(X,Y(X))=(r,\varphi,\rho(r,\varphi),\phi(r,\varphi)) \subset (\C^\times)_X \times (\C^\times)_Y\}~~.
}
Recall that the symplectic form~\eqref{generic-Bcc} we are using is
\begin{equation}
\omega_\X = -\frac{\sin\theta}{2\pi|\hbar|} (d \varphi \wedge d \phi-d r \wedge d \rho) - \frac{\cos\theta}{2\pi|\hbar|} (d r \wedge d \phi+d \varphi \wedge d \rho)~,
\end{equation}
The condition for~\eqref{eq:Lag-prep} to define a Lagrangian submanifold is therefore
\deq[eq:Lagr-cond]{
 \sin\theta \Bigl( \pdv{\phi}{r}+ \pdv{\rho}{\varphi}  \Bigr)+\cos\theta \Bigl( \pdv{\rho}{r} - \pdv{\phi}{\varphi} \Bigr)  = 0~.
}
The two quantities in parentheses are just the two Cauchy--Riemann equations for the function $Y$. Therefore, the graph of a holomorphic map $(\C^\times)_X$ to~$(\C^\times)_Y$ defines an $A$-brane for any choice of~$\omega_\X$---in fact, an $(A,B,A)$-brane. But the only such holomorphic function that is isotopic to the constant map is itself the constant map!

Let us recall that we consider $A$-branes up to Hamiltonian isotopy. The Hamiltonian vector field associated to a generating function $f(r,\varphi)$ is
\begin{equation}
X_f =2\pi |\hbar|\Bigl[\cos\theta \Bigl( \pdv{f}{r} \pdv{}{\phi} + \pdv{f}{\varphi} \pdv{}{\rho} \Bigr)
- \sin\theta  \Bigl( \pdv{f}{r} \pdv{}{\rho} - \pdv{f}{\varphi} \pdv{}{\phi} \Bigr)  \Bigr],
\end{equation}
which generate the Hamiltonian flow of the form
\deq{
\delta \phi = \cos\theta \pdv{f}{r} + \sin\theta \pdv{f}{\varphi}, \qquad
\delta \rho = \cos\theta \pdv{f}{\varphi} - \sin\theta \pdv{f}{r}~.
}

It is easiest to understand the action of these vector fields if $\theta = 0$. Choosing $f(r,\varphi)=r$ then generates a global rotation of~$\phi$ in $(\C^\times)_Y$, constant over the Lagrangian, and any radially-dependent rotation can also be obtained by a suitable choice of~$f$. As such, the angle $\phi$ can be set to any chosen value anywhere on the Lagrangian; angular displacements do not lead to different objects in the $A$-model of $(\X,\omega_K)$! The other Hamiltonian vector field ensures that we can deform our Lagrangian by a Hamiltonian isotopy to set $\rho(r,\varphi)$ to its average value over $\varphi$, depending (in principle) on~$r$. But the Lagrangian condition~\eqref{eq:Lagr-cond} then ensures that $\rho$ is just constant.

As such, every $A$-brane in $(\X,\omega_K)$ of topological class $\C^\times$ can be brought by Hamiltonian isotopy to a unique $(A,B,A)$-brane which is a fiber of the projection on the second factor of $\C^\times \times \C^\times$.
This makes it easy to see that the Maslov class is zero; we can consider the symplectic universal cover of $\C^\times \times \C^\times$ by~$\C^2$, and such Lagrangians descend from affine subspaces there (which always have zero Maslov class).

There are exactly two real parameters defining an $(A,B,A)$-brane of this form in the category $\ABrane(\X,\omega_\X)$: the modulus $|Y|$ (or $\rho$) and the $\U(1)$ holonomy of the flat bundle over the Lagrangian due to $H_1(\L,\Z)=\Z$.
This matches precisely with the $\C^\times$ parameter $y_1$ for the polynomial representation $\PR^{y_1}$. Indeed, the representation space was just the algebra $\C[q^\pm,X^{\pm}]$ of Laurent polynomials, which is exactly the collection of holomorphic functions (in complex structure $J$) on this brane.

The story for generic values of~$\theta$ is similar. However, the integral curves of the Hamiltonian vector field generated by $\partial f/\partial r$ are no longer closed circles,  but rather spirals in~$\C^\times$ (obtained as the images of lines of fixed slope under the exponential map).

\subsection{Symmetrized quantum torus}
\label{app:SQT}

Now let us move one step forward to DAHA.
Recall from \eqref{Out-QT} that the outer automorphism group of~$\QT$ is~$\SL(2,\Z)$. How much of this symmetry is visible in the classical limit?
Since the classical limit corresponds to a trivial choice of the central character (or, equivalently, to the quotient by the ideal $(q-1)$), all inner automorphisms must act trivially, so that the residual symmetry is by the group $\Out(\QT) = \SL(2,\Z)$. Indeed, this is just the natural action of the group of canonical transformations on symplectic $\Z^2$.
Of  course, these automorphisms act on~\MS{} by  the exponentiated form of this action:
\begin{equation}\label{SL2Z-coord}
\SL(2,\Z)\ni\begin{pmatrix} a & b \\  c &  d  \end{pmatrix} : (X,Y) \mapsto \Bigl(X^a Y^b, X^c Y^d\Bigr)~.
\end{equation}

We will be particularly interested in one distinguished outer automorphism of order two, namely the central element $\kappa$ of~$\SL(2,\Z)$.
Let us study the extension of~$\QT$ by this outer automorphism. In other words, we are interested in the algebra~$\TQT$ defined by the short exact sequence
\begin{equation}
0 \to \QT \to \TQT \to \C[\Z_2] \to 0~,
\end{equation}
with respect to the outer automorphism $\kappa$. (Of course, this is just obtained  by applying the group algebra functor to  a corresponding semidirect product of groups.) We can think of this as adjoining a new generator $T$ of $\Z_2$ to~$\QT$, subject to the relations
\begin{equation}
T^2 = 1~, \quad  TXT = X^{-1}~,  \quad  TY^{-1}T = Y~.
\end{equation}

In the classical limit, $\TQT$ is not commutative. We therefore have no hope of telling a story analogous to the one we have been building for~$\QT$. However, rather than studying $\TQT$---which is like a $\kappa$-\emph{equivariant} version of~$\QT$---directly, we can imagine trying to replace it by a $\kappa$-\emph{invariant} version, whose classical limit should be (the functions on) the quotient of~\MS{} by~$\kappa$.

This is concretely accomplished as follows. The algebra $\TQT$  contains an idempotent element $\mathbf{e} = (1+T)/2$, which one can think of as implementing projection onto $\kappa$-invariants. We can therefore define the subalgebra
\begin{equation}
\SQT = \mathbf{e}\TQT \mathbf{e}~.
\end{equation}
In fact, the algebras $\SQT$ and~$\TQT$ are Morita equivalent; this equivalence is witnessed by the bimodule $\mathbf{e}\TQT$ \cite{Oblomkov:aa,oblomkov2004double}, and it is also discussed in \S\ref{sec:surface}.
We are thus free to study the representation category of either, and in the classical limit, we have that
\begin{equation}\label{flat-moduli-torus}
\SQT  \xrightarrow[q\to1]{} \OO(\MS)^{\Z_2} = \OO[(\C^\times \times \C^\times)/\Z_2]~.
\end{equation}
However, there is one subtlety. Since the geometric ramification parameter $\tilde{t}$ and DAHA parameter $t$ are related in \eqref{t-wtt}, $t=1$ means $\tilde{t}=q^{-\frac12}$. Therefore, $\SH_{t=1}$ can be understood as the deformation quantization of the coordinate ring of $\MF(C_p,\SL(2,\C))$ with ramification \be\label{1/2-shift}\tgamma_p+i\talpha_p=i\hbar~.\ee
Or, we believe that this $\frac12$ shift is related to the \emph{orbifold (Gepner) point} $(\talpha_p, \tbeta_p, \tgamma_p, \teta_p)=(0,0,0,\frac12)$ \cite{Aspinwall:1995zi,Douglas:1996xg,Blum:1997fw,Wendland:2000ry} in the sigma-model on $(\C^\times\times\C^\times)/\Z_2$. In fact, the evaluation of the left-hand side of \eqref{c-ramification} at the orbifold point yields $\frac12$. If $(\omega_\X^{-1}(F+B))^2=-1$ is satisfied, the canonical coisotropic brane $\Bcc$ can exist at the orbifold point. However, $\Bcc$ at the orbifold point is obscure because of the stringy nature so that we leave it for future work. Also, it is worth noting that the skein algebra $\Sk(T^2)$ of a torus discussed in \S\ref{sec:poly-rep} is isomorphic to $\SH_{t=1}$. Hence, $\Sk(T^2)$ is \emph{not} simply the deformation quantization of the coordinate ring of $\MF(T^2,\SL(2,\C))$. This stems from the fact that a Higgs bundle over a genus one curve $C\cong T^2$ is \emph{not} stable. Thus, we need to take into account the quantum correction described above.

Now let us consider representation theory of $\SQT$.
How can we study representations of $\SQT$ using the information we have from the non-equivariant case? In particular, can we make use of our understanding of $\QT$-modules?
Our situation is summarized by the following diagram:
\begin{equation}
\begin{tikzcd}
\SQT
	\arrow[hook, shift left]{r}{} & \TQT \arrow[two heads, shift left]{l}{} \arrow[dashed]{dr}{} & \\
& \QT \arrow[hook]{u}{} \arrow{r}{} & \End(U)\qquad .
\end{tikzcd}
\label{eq:liftdiag}
\end{equation}

The astute reader will notice that there seems to be an obvious arrow missing in~\eqref{eq:liftdiag} above.
Isn't $\SQT$ simply the invariant subset of $\QT$ with respect to the action of~$\kappa$? The answer is yes, but \emph{not canonically}. Indeed, an outer automorphism is an \emph{equivalence class} of automorphisms, and we must choose a particular
 lift to a specific non-inner automorphism in order to even talk about the action of~$\kappa$ on~$\QT$. For example, we could take
\begin{equation}
\hat{\kappa}: X \mapsto X^{-1}, ~ Y\mapsto Y^{-1}~,
\label{eq:klift}
\end{equation}
but any other choice that shifts $X$ and $Y$ by powers of~$q$ is equally valid and does not change the structure of the extension~$\TQT$.

Let us fix the lift~\eqref{eq:klift}.
It is then clear that the data of a lift of a $\QT$-module $U$ to an $\TQT$-module consists of an order-two endomorphism $\tau$ of~$U$
such that conjugation by~$\tau$
and intertwining with any lift of the outer automorphism $\kappa$ to an automorphism of~$\QT$. That is, the diagram
\begin{equation}
\begin{tikzcd}
\QT  \arrow{r}{} \arrow{d}{\hat{\kappa}} & \End(U) \arrow{d}{\tau} \\
\QT  \arrow{r}{} & \End(U)~,
\end{tikzcd}
\end{equation}
should commute, where the action of~$\tau$ on~$\End(U)$ is by conjugation.  We can  therefore lift the classification  of $\QT$-modules to~$\TQT$ by classifying such lifts, and we will do this in the next subsection.

\subsubsection{Representation theory}\label{app:pre-SQT}

In this subsection, we consider each $\QT$-module in turn. A \emph{lifting datum} for a module $U$ consists, as above, of an order-two endomorphism $\tau$ of~$U$ that intertwines with~$\hat{\kappa}$.

\begin{proposition}
\label{prop:liftPR}
Lifting data for the polynomial representation $\PR^{y_1}$ exist precisely when $y_1^2 = q^\ell$ for some integer $\ell$. When this condition holds, there is a unique lift up to a sign.
\end{proposition}
\begin{proof}
Let us consider the matrix elements of~$\tau$ in the polynomial basis:
\deq{
\tau (X^i) = \tau_{ij} X^j~.
}
We first consider the requirement imposed by conjugating with~$X$:
\deq[eq:conjX]{
\tau X = X^{-1} \tau  ~ \Rightarrow \tau_{(i+1)j} = \tau_{i(j+1)} ~.
}
This means that we can write our map in the form
\deq{
\tau : X^i \mapsto \sum_\ell a_\ell X^{\ell-i}
}
for some constants $a_\ell$. It remains to consider the action on~$Y$:
\begin{equation}
\begin{aligned}
\tau Y = Y^{-1} \tau  ~ &\Rightarrow
y_1 q^{-i} a_\ell X^{\ell - i} = a_\ell y_1^{-1} q^{\ell - i} X^{\ell - i} \\
&\Rightarrow a_\ell(y_1^2-q^{\ell}) = 0~.
\end{aligned}
\end{equation}
From here it follows that $a_\ell = 0$ unless $\ell$ is chosen such that $y_1^2= q^{\ell}$. In order to get an involution, we must choose $a_\ell = \pm 1$. Since $q$ is not a root of unity, this completes the proof.
\end{proof}

Note that the freedom in the choice of~$\ell$ is not very interesting: it is always possible to shift $y_1 \mapsto y_1/q$ by a (non-inner) automorphism of the algebra. Hence, setting $\ell=0$, we may consider $\tau(X)=X^{-1}$, and the polynomial representation can be lifted when $y_1=\pm1$.

\begin{proposition}
\label{prop:liftFD}
Let $q$ be a primitive $m$-th root of unity.
Lifting data for the cyclic representation $\repF^\lambda_m$ exist precisely when both weights are equal to plus or minus one. As above, there is then a unique lift up to a sign for each integer $\ell\pmod m$.
\end{proposition}
\begin{proof}
The calculation proceeds along similar lines. By conjugating with $X$, we obtain the requirement
\deq{
x_1 \tau_{(i+1)j} = x_1^{-1} \tau_{i(j+1)}~.
}
Since the indices are cyclic, this requirement can only be satisfied if $x_1^{2m} = 1$, \textit{i.e.}\ if the weight $x_m = \pm 1$. The calculation for~$Y$ is identical, and produces the restriction
\deq{
a_\ell(y_1^2-q^{\ell}) = 0~.
}
Since $q$ is an $m$-th root of unity, this implies that $y_m = \pm 1$.
\end{proof}

This completes the story as to lifting of irreducible representations of~$\QT$. However, it is important to note that \emph{reducible} representations can admit interesting new choices of lifting datum!

\begin{proposition}\label{prop:liftPR2}
Lifting data for the representation $\PR^{y_1} \oplus \PR^{y_1'}$ exist whenever $y_1 = q^\ell/y_1'$ for some~$\ell$. When this is true, there is a unique lift up to sign. If either $y_1$ or~$y_1'$ satisfies the conditions of Proposition~\ref{prop:liftPR}, additional lifts correspondingly exist.
\end{proposition}
\begin{proof}
This proceeds by computation, as above. $\tau$ can be decomposed into four blocks, each of which must satisfy~\eqref{eq:conjX} independently; correspondingly, there are four undetermined constants, which we will call $a_\ell$, $b_\ell$, $c_\ell$, and $d_\ell$. The conjugation by $Y$ then produces the conditions
\deq{
a_\ell(y_1^2-q^{\ell}) = d_\ell({y_1'}^2-q^{\ell})  = 0~, \quad
b_\ell(y_1y_1'-q^{\ell}) = c_\ell(y_1y_1'-q^{\ell}) = 0~.
}
The conditions are decoupled and those on $a$ and~$d$ just recover the conditions of Proposition~\ref{prop:liftPR}. The new conditions will also be satisfied for at most one choice of $\ell$, and we must choose $b_\ell c_\ell = 1$ to obtain an involution.
\end{proof}

\begin{proposition}\label{prop:liftFD2}
    Let $q$ be a primitive $m$-th root of unity. Lifting data for the representation $\repF^\lambda_m \oplus \repF_m^{\lambda'}$ consist whenever $x_m x_m'  = 1$ and $y_m y_m'=1$. When this is true, there is a unique lift up to sign. If either $y_1$ or~$y_1'$ satisfies the conditions of Proposition~\ref{prop:liftPR}, additional lifts correspondingly exist.
\end{proposition}
\begin{proof}
As above, the diagonal blocks reduce independently to the previous case (Proposition~\ref{prop:liftFD}). The conjugation by $X$ leads to the condition
\deq{
x_1 \tau_{(i+1)\bar{j}} = x_1^{-1} \tau_{i(\bar{j}+1)}
}
on the off-diagonal block; from cyclicity it then follows that $x_1x_1'$ is an $m$-th root of unity, and therefore that $x_m x_m'  = 1$. The calculation for $Y$ is identical, but reduces as in the proof of Proposition~\ref{prop:liftFD} since $q$ is an $m$-th root of unity.
\end{proof}

After choosing $\hat\kappa$, the algebra $\SQT$ maps injectively into~$\QT$, and we can therefore pull back representations. Consequently, the polynomial representation, the lift of $\PR^{y_1} \oplus \PR^{y_1'}$ with $y_1y_1'=1$, generically becomes an \emph{irreducible} representation of~\SQT{} on Laurent polynomials which are \emph{not} symmetrized. However, when $y_1 = \pm1$, (or, to be precise, $y_1^2=q^\ell$ for $\ell\in\bZ$), the spectrum of the $Y$-operator on $\PR^{\pm1}$ becomes two-fold degenerate under $Y\mapsto Y^{-1}$. Thus, $\PR^{\pm1}$ is compatible with the outer automorphism $\hat\kappa$.  Correspondingly, the representation $\PR^{\pm1}$ splits into two isomorphic irreducible representations on symmetric and antisymmetric Laurent polynomials.

For finite-dimensional representations, the story is identical. The pullback of the lift of $\repF^\lambda_m \oplus \repF_m^{\lambda'}$ to~\SQT{} is generically irreducible, but becomes reducible when $m = 2n$ is even and $y_m = \pm 1$. Precisely in this case, it becomes an $n$-dimensional representation of~\SQT{} on \emph{symmetric} Laurent polynomials, modulo the ideal $(X^n + X^{-n})$.

\subsubsection{Corresponding branes}\label{app:branes-SQT}

While we need to take into account the quantum correction \eqref{1/2-shift}, $\SQT$  is related to the coordinate ring on $(\C^\times\times \C^\times)/\Z_2$ as in \eqref{flat-moduli-torus}. Therefore, let us briefly consider the interpretation in terms of the $A$-brane category on $(\C^\times\times \C^\times)/\Z_2$.

First, we consider $A$-branes corresponding to the polynomial representations. As in Appendix \ref{app:poly-rep-QT}, an $A$-brane for the polynomial representation of $\SQT$ is supported on the constant locus of the symmetrized function $Y+Y^{-1}$, and it stays an $A$-brane at any value of $\hbar$ since it is of type $(A,B,A)$. Its preimages to $\C^\times\times\C^\times$ are generically two distinct branes supported on the $(\C^\times)_X$ planes with the values of $Y$ related by $y_1 y_1'=1$. This corresponds to the representation of $\SQT$ on \emph{non-symmetric} Laurent polynomials corresponding to Proposition \ref{prop:liftPR}, and it has the deformation parameter valued in $\C^\times$ for positions and the $\U(1)$ holonomy. However, at a fixed point of the $\Z_2$ action, only one brane for $\PR^{\pm1}$ can be symmetrized into a brane on $\C^\times/\Z_2$, and it therefore gives rise to the representation on \emph{symmetric} Laurent polynomials as in Proposition \ref{prop:liftPR}.
The extensions of these representations to DAHA are discussed in \S\ref{sec:poly-rep}.

Since the target $(\C^\times\times \C^\times)/\Z_2$ is the moduli space of flat $\SL(2,\C)$-connections over a torus $T^2$, it admits an elliptic fibration called the \emph{Hitchin fibration}, which was discussed in \S\ref{sec:target}. Roughly speaking, the elliptic fibration can be obtained by identifying the elliptic fibers in \eqref{eq:easy-fib} at $(r,\rho)$ and $(-r,-\rho)$ so that there is a singular fiber $T^2/\Z_2$ at $(r,\rho)=(0,0)$. Analogous to Appendix \ref{app:cyclic-rep-QT}, an $A$-brane supported on a generic fiber in the Hitchin fibration gives rise to an $m$-dimensional representation when $\hbar=1/m$ for $m\in \Z_{>0}$. Since its preimage in~$\C^\times\times\C^\times$ under the $\Z_2$ quotient consists of two distinct tori at different positions, it is precisely the pullback of the lift of $\repF_m^\lambda\oplus\repF_m^{\lambda^{-1}}$ in Proposition \ref{prop:liftFD2}. Consequently, there is the deformation parameter $\lambda$ for a position at the base and the $\U(1)\times \U(1)$ holonomy of the $A$-brane on a generic fiber.
The fiber at the center $(r,\rho)=(0,0)$ in \eqref{eq:easy-fib} is fixed as a set under the $\Z_2$ action so that the singular fiber in the Hitchin fibration is doubly covered by the torus. Thus, if the representation $\repF^{(\pm1,\pm1)}_{m}$ is even-dimensional $(m=2n)$, the corresponding $A$-brane degenerates two copies on the singular fiber. The $A$-brane with support on the singular fiber brings about the $n$-dimensional irreducible representation of $\SQT$ on the symmetric Laurent polynomials, obtained by the pullback of the lifting $\repF_{2n}^{(\pm1,\pm1)}$.
 There are no longer any deformation parameters, as this brane does not belong to any continuous family of $A$-branes. As explained in \S\ref{sec:Lagrangian}, a Chan-Paton bundle is generally endowed with a flat $\Spin^c$ structure and the choice of signs corresponds to that of a flat $\Spin^c$ structure with no holonomy (equivalently Spin structure in this case) of $K^{-1/2}_\F$. Indeed, the plus sign $+$ is the Ramond spin structure on a circle (a trivial real line bundle over a circle), and the minus sign $-$ is the Neveu-Schwarz spin structure (see \S\ref{sec:Ui} and \S\ref{sec:bound-state}).

 Although we see similar phenomena in branes for DAHA, finite-dimensional representations of DAHA are much richer because compact Lagrangians in the target for DAHA are more intricate, which is the main subject in \S\ref{sec:2d}.

\section{\texorpdfstring{3d $\cN=4$ theories and Cherednik algebras}{3d N=4 theories and Cherednik algebras}}\label{app:3dN=4}

In this appendix, we shall briefly discuss the relationship between Coulomb branches of 3d $\cN=4$ theories and the trigonometric and rational degeneration of DAHA of type $A_1$. The connection of the deformation quantization of Coulomb branches of 3d $\cN=4$ theories to variants of Cherednik algebras has been studied in \cite{bezrukavnikov2005equivariant,oblomkov2016geometric,Braverman:2016wma,Kodera:2016faj}. Here we shall provide a brief review as well as some implications to representation theory from brane quantization.

\subsection{\texorpdfstring{Coulomb branches of 3d $\cN=4$ theories}{Coulomb branches of 3d N=4 theories}}

If a 3d $\cN=4$ theory admits Lagrangian description, a Lagrangian is described by a 3d $\cN=4$ vector multiplet $\cV$ and a  3d $\cN=4$ hypermultiplet $\cH$ of gauge group $G$ in a representation $R$ of gauge group $G$. Moreover, they can be constructed a 3d $\cN=2$ vector multiplet  $(A_\mu,\sigma,\lambda_\a,d)$ and a  3d $\cN=2$ chiral multiplet $(\varphi,\psi_\a,F)$ where $\phi$ is a complex scalar, and $\sigma$ is a real scalar which can be regarded as the reminiscent of the $A_3$ component of 4d gauge field. In fact, the $\cN=4$ vector multiplet consists of a 3d $\cN=2$ vector multiplet  $(A_\mu,\sigma,\lambda_\a,d)$ and a  3d $\cN=2$ chiral multiplet $(\varphi,\psi_\a,F)$ in the adjoint representation, forming
$$
\cV = (A_{\mu}, \Lambda_{\alpha a\dot a}, \Phi_{ab}, D_{\dot{a}\dot b})~.
$$
The gauge field $A_{\mu}$, gaugino $\Lambda_{\alpha a\dot a}$, scalar $\Phi_{ab}$ and auxiliary field $D_{\dot{a}\dot b}$ transform in the trivial, $(\mathbf{2},\mathbf{2})$, $(\mathbf{3},\mathbf{1})$, and $(\mathbf{1},\mathbf{3})$, respectively, under the $\SU(2)_C\times\SU(2)_H$ $R$-symmetry of a 3d $\cN=4$ theory.
The $\cN=4$ hypermultiplet consists of $\cN=2$ chiral multiplets in the representation $R$ and its conjugate representation $\overline R$, forming
\bea
\cH = (q_{\dot a}, \widetilde q_{\dot a}, \psi_{\alpha a}, \widetilde \psi_{\alpha a}) \label{Hmul}
\eea
The scalars $q_{\dot a}, \widetilde q_{\dot a}$ and their fermionic superpartners $\psi_{\alpha a}, \widetilde \psi_{\alpha a}$ transform as $(\mathbf{1},\mathbf{2})$  and $(\mathbf{2},\mathbf{1})$, respectively, under the $R$-symmetry.

An $\cN=4$ gauge theory generically has a union of moduli spaces of vacua, called Coulomb branch  and  Higgs branch. The Higgs branch is a \HK  manifold, known as a Nakajima quiver variety \cite{Nakajima:1994nid}, parameterized by the expectation values of gauge-invariant operators in hypermultiplets.
 The Coulomb branch $\cM_C^{\mathrm{3d}}$ is a \HK manifold parameterized by the expectation values of gauge-invariant combinations of scalars in vector multiplets and monopole operators.

 In an abelian gauge theory, the classical Coulomb branch is $\R^3\times S^1$ where $\R^2$ is spanned by the expectation values of the complex scalar $\varphi$ and $\bR \times S^1$ is by that of the monopole operators $v^\pm = e^{\pm\frac1{g^2}(\sigma+i\gamma)}$. Note that $\gamma$ is a periodic scalar $\gamma \sim \gamma + 2\pi g^2$ called ``dual photon'' subject to $d\gamma = \ast dA$.
 However, it is well-known that a 3d $\cN=4$ Coulomb branch receives quantum corrections, which deform the classical moduli space. For instance, quantum corrected Coulomb branches in abelian gauge theories have been investigated in \cite[\S3]{Bullimore:2015lsa}.

In a non-abelian gauge $G$ theory, the scalars in the $\cN=4$ vector multiplet take expectation values in the Cartan subalgebra $\frakt \subset \frakg$.  For generic  expectation values, the gauge group is broken to a Cartan subgroup $T\cong \U(1)^{\rk(G)} \subset G$. Therefore, around a generic point, the Coulomb branch is locally $(\R^3\times S^1)^{\rk(G)}$. However, the Weyl group $W$ acts on the scalars as a residual gauge symmetry. In addition, it receives both perturbative and non-perturbative quantum corrections. In the end, the Coulomb branch is birationally equivalent \cite{Braverman:2016wma} to
\be\label{3d-Coulomb-birational}
\cM_C^{\mathrm{3d}}\approx \frac{(\R^3\times S^1)^{\rk(G)}}{W}~.
\ee
Moreover, it admits an interpretation as an integrable system
\be\label{projection}
\pi:\cM_C^{\mathrm{3d}} \to \C^{\rk(G)}
\ee
where a generic fiber is $(\C^\times)^{\rk(G)} $. This projection can be obtained by forgetting about monopole operators, and the base $\C^{\rk(G)}$ is parametrized by the gauge-invariant operators $\Tr(\varphi^n)$ of the complex scalar.

If a theory is the reduction of a 4d $\cN=2$ theory of class $\cS$ on $S^1$, the projection \eqref{projection} can be interpreted as a partial decompactification of the corresponding Coulomb branch $\cM_C(C,G,\lat{L})$. As explained in \S\ref{sec:4dCoulomb}, the Coulomb branch $\cM_C(C,G,\lat{L})$ of a 4d $\cN=2$ theory $\cT[C,G,\lat{L}]$ of class $\cS$ on $S^1\times \R^3$ where $S^1$ is a circle of radius $R$ is a \HK manifold (say, $\dim_\bR \MH(C,G)=4r$). It is moreover a completely integrable system so that a generic Hitchin fiber is a complex tori of volume $1/R^r$ which is homeomorphic to $T^{2r}$. As $S^1$ shrinks $R\to 0$, a generic fiber is decompactified to $(\bR\times S^1)^{r}$, and its volume diverges to infinity. Hence, a generic fiber and the base in \eqref{projection} locally parametrized by the monopole operators  and the complex scalars $\varphi$, respectively, are holomorphic in complex structure $I$.

\subsection{\texorpdfstring{3d $\cN=4$ Coulomb branches and Cherednik algebras}{3d N=4 Coulomb branches and Cherednik algebras}}\label{sec:3d-N=4}

\subsubsection*{Trigonometric Cherednik algebra}

Let us consider the 3d $\cN=8$ theory, which consists of a vector multiplet and  a hypermultiplet in the adjoint representation. It is a reduction of the 4d $\cN=4$ theory on a circle $S^1$. The Coulomb branch of the 3d $\cN=8$ theory with gauge group $G$ is actually
\be\label{3d-N=8-MC}
\cM_C^{\mathrm{3d}}\Big[\!\raisebox{-.35cm}{\begin{tikzpicture}
\node[draw,circle, inner sep=2pt] (a) at (0,0) {$G$};
    \draw  (a)   to [distance=3ex, in=150, out=210, loop] ();
\end{tikzpicture}}~
\Big] =\frac{\frakt_\bC\times T_\C^\vee}{W}~,
\ee
and it is indeed free from quantum corrections. Although the $\SU(2)_C$ symmetry treats all the \HK complex structures $(I,J,K)$ equally, it is  natural to take a viewpoint holomorphic in complex structure $J$ to see the connection to Cherednik algebra. A supersymmetric Wilson loop $\Tr\,\exp\oint_{S^1}(A+\varphi)$ in 4d reduces to the complex scalar $\varphi$ in 3d whereas an 't Hooft operator in 4d becomes a monopole operator in 3d.  As a result, $\cM_C^{\mathrm{3d}}$  is the trigonometric degeneration of \eqref{4d-N=4-MC} by taking the Cartan subgroup $T_\bC$ to the Cartan subalgebra $\frakt_\bC$.\footnote{In \cite[\S 1.6]{Cherednik-book}, the trigonometric limit is taken in the $S$-dual side, $T_\C^\vee \to \frakt_\C^\vee$, instead. However, the trigonometric limit $T_\bC \to \frakt_\bC$ of the Wilson loop is more natural from the physical viewpoint. Therefore, among the two polynomial representations in Appendix~\ref{app:degeneration}, the difference-rational polynomial representation is more natural than the differential polynomial representation for the reduction from 4d to 3d.}

It is clear from \eqref{trig-rat-SH} that the quantized Coulomb branch of the 3d $\cN=8$ theory and its mass deformation are related to the spherical subalgebra of the trigonometric Cherednik algebra $\SH^{\textrm{tri}}(W)$. To see this, let us consider the case of rank one. If the gauge group is either $\SU(2)$ or $\SO(3)$, the Coulomb branch \eqref{3d-N=8-MC} is $(\bC\times \C^\times)/\Z_2$, which has two $A_1$ singularities. However, the dual maximal torus is $T_\C^\vee=\frakt_\bC/\sfQ$ where $\sfQ(\SU(2))=\sfP(\SO(3))\subset \sfP(\SU(2))=\sfQ(\SO(3))$ so that the $\SO(3)$ Coulomb branch is a double cover of the $\SU(2)$ Coulomb branch. Therefore, the quantized Coulomb branch of the $\SO(3)$ theory is isomorphic to the spherical trigonometric Cherednik algebra $\SH^{\textrm{tri}}_{\hbar,c=0}$ at $c=0$. On the other hand, the quantized Coulomb branch of the $\SU(2)$ theory is the $\xi_2$-invariant subalgebra of $\SH^{\textrm{tri}}_{\hbar,c=0}$ generated by
$$v=\mathbf{e}\hat x^2\mathbf{e}~, \quad \textrm{and} \quad y^2-1= (Y^2+1+Y^{-2})\mathbf{e}~,$$
where the notation is the same as Appendix \ref{app:Trigonometric-A1}. In other words, the $\SO(3)$ theory is endowed with the minimal monopole operator whereas the $\SU(2)$ theory is not. This is consistent with what we have seen for 't Hooft operators in the $\SO(3)_+$ and $\SU(2)$ 4d $\cN=4$ theory, respectively, in \S\ref{sec:4dCoulomb}.

Turning on mass parameters breaks a half of supersymmetries and it is called 3d $\cN=4^*$ theory. Correspondingly, the two $A_1$ singularities turn into two exceptional divisors in the $\SO(3)$ Coulomb branch. (Figure \ref{fig:3d-nilcone}.) The three complex structures $(I,J,K)$ of the 3d $\cN=4^*$ $\SO(3)$ Coulomb branch $\cM_C^{\mathrm{3d}}$ inherits from those of the Coulomb branch $\cM_C(C_p,\SO(3)_+)$ in \S\ref{sec:4dCoulomb}. At the classical level, the period integrals of $(\omega_I,\omega_J,\omega_K)$ over an exceptional divisor provide the parameters $(\talpha_p,\tbeta_p,\tgamma_p)$ as in \eqref{integral-D}. Indeed, the complex mass parameter of the adjoint hypermultiplet is $\tbeta_p+i\tgamma_p$ whereas $\talpha_p$ corresponds to the real mass parameter. The triple $(\talpha_p,\tbeta_p,\tgamma_p)$ is transformed as the scalars $\Phi=(\sigma,\varphi)$ in the $\cN=4$ vector multiplet under $\SU(2)_C$. Consequently, the quantized Coulomb branch of the $\SO(3)$ 3d $\cN=4^*$ theory is isomorphic to the spherical trigonometric Cherednik algebra  $\SH^{\textrm{tri}}_{\hbar,c}$.

Therefore, the representation theory of the algebra can be studied in the context of brane quantization analogous to \S\ref{sec:2d}. However,
we should note that there is one difference between 3d $\cN=4$ theories and 4d $\cN=2$ theories. In trigonometric Cherednik algebra \eqref{tri-SH-A1}, one can set $\hbar=1$ by redefinition $\hat x\to \hbar \hat x$, $c\to \hbar c$. Hence, we can consider the 2d $A$-model on the 3d $\cN=4^*$ $\SO(3)$ Coulomb branch $\X$ with a fixed symplectic form $\omega_\X=-\omega_K$ where the parameter of $\SH^{\textrm{tri}}_{1,c}$ is identified by $c=-\talpha_p/2$.
More generally, once we fix a complex structure of any 3d $\cN=4$ Coulomb branch in which deformation quantization is performed, a symplectic form can be set to a particular \K form thanks to the $\SU(2)_C$ symmetry. On the other hand, a 4d $\cN=2$ superconformal theory is endowed with the $\U(1)_r\times \SU(2)_H$ $R$-symmetry, and $\U(1)_r$ rotates only $\omega_J$ and $\omega_K$. As a result, $\hbar$ cannot be fixed to a particular value, and the story becomes more delicate as seen in \S\ref{sec:2d}.

Without loss of generality, we can hence consider that the target $\X$ of the 2d $A$-model is the 3d $\cN=4^*$ $\SO(3)$ Coulomb branch with
\be \label{onlyalpha}
\talpha_p\neq0~,\quad \tbeta_p=0=\tgamma_p
\ee
and the symplectic form is $\omega_\X=-\omega_K$. The target can be understood as a partial decompactification of the Coulomb branch $\cM_C(C_p,\SO(3)_+)$ in \S\ref{sec:4dCoulomb} with the same ramification parameters \eqref{onlyalpha} along the $x$-direction. Upon 3d reduction, the singular fiber of type $I_0^*$ in $\cM_C(C_p,\SO(3)_+)\to \cB_u$ is decompactified to the left of Figure \ref{fig:3d-nilcone}, which consists of a cigar and the two exceptional divisors, denoted by $\D_1$ and $\D_3$.
Consequently, finite-dimensional representations come from branes $\frakB_{\bfD_1}$ and $\frakB_{\bfD_3}$ supported on the two exceptional divisors in Figure \ref{fig:3d-nilcone} that can exist only when $c=-(2\ell-1)/2$ as in \S\ref{sec:divisorTQFT}. Moreover, they are related by the $\xi_2$-reflection $\scD_\ell^{(1)}\to \scD_\ell^{(3)}$.
Indeed, the finite-dimensional representations of $\SH_{1,c}^{\textrm{tri}}$ (and $\SH_{1,c}^{\textrm{rat}}$) are classified by \cite{berest2003finite,etingof2007reducibility,varagnolo2009finite}, and this conclusion is consistent with \cite[Proposition 7.1]{berest2003finite}.

Let us remark the indications from the brane quantization on distinguished infinite-dimensional representations of $\SH^{\textrm{tri}}_{1,c}$ that stem from branes supported on a generic fiber $\C^\times$ of \eqref{projection} and the cigar in $\pi^{-1}(0)$. A brane supported on a generic fiber is labeled by a holonomy $y_0$ of the $\Spin^c$-bundle around $S^1\subset \C^\times$ and a position $v$ of the base in \eqref{projection}. Hence, we denote the corresponding infinite-dimensional representation by $V_{\C^\times}^{y_0,v}$. The $y$-weights of this representation are unbounded below and above. On the other hand, for a brane supported on the cigar, the $y$-weights of the corresponding representation, say $V_{\textrm{cigar}}$, are bounded below. The distinguished infinite-dimensional representations $V_{\C^\times}$ and $V_{\textrm{cigar}}$ are analogous to the principal and discrete series of $\SL(2,\R)$, respectively.

As in \S\ref{sec:bound-state}, we can consider bound states of branes. For instance, when $c=-(2\ell-1)/2$, the bound state of $\frakB_{\bfD_i}$ and the brane supported on the cigar gives rise to the short exact sequence
\be
0\longrightarrow  \iota(V_{\textrm{cigar}}) \longrightarrow V^{(i)}\longrightarrow \scD_\ell^{(i)} \longrightarrow 0~,\qquad i=1,3~.
\ee
This is analogous to \eqref{SES2-f}.
On the other hand, in the case of no holonomy with Ramond spin structure $y_0=+$, it can enter the singular fiber $\pi^{-1}(0)$ when $c\in \Z$, which yields
\be
0\longrightarrow  \iota (V_{\textrm{cigar}}) \longrightarrow  V_{\C^\times}^{+,0}\longrightarrow   V_{\pi^{-1}(0)} \longrightarrow 0~.
\ee
This is analogous to \eqref{SES3}.

If $G=\SU(2)$, the mass deformation develops only one exceptional divisor $\D_1$ (no $\D_3$) and the singular fiber $\pi^{-1}(0)$ has one $A_1$ singularity.
The quantized Coulomb branch of the 3d $\cN=4^*$ $\SU(2)$ theory is the $\xi_2$-invariant subalgebra of $\SH^{\textrm{tri}}_{1,c}$.

\bigskip

We have seen that the Coulomb branch of the 4d $\cN=4$ theory on $S^1\times \R^3$ can be constructed from affine Grassmannian Steinberg variety $\cR$ in \eqref{Grassmannian-Steinberg}.
Indeed, the Coulomb branch of the 3d $\cN=8$ theory can be obtained by taking the spectrum of the $G_\C^\OO$-equivariant Borel-Moore homology of $\cR$ \cite{bezrukavnikov2005equivariant,Braverman:2016wma}
\be
\Spec H^{G_\C^\OO}_*(\cR) = \frac{\frakt_\bC\times T_\C^\vee}W=\cM_C^{\mathrm{3d}}\Big[\!\raisebox{-.35cm}{\begin{tikzpicture}
\node[draw,circle, inner sep=2pt] (a) at (0,0) {$G$};
    \draw  (a)   to [distance=3ex, in=150, out=210, loop] ();
\end{tikzpicture}}~
\Big]~.
\ee
By introducing the same equivariant action as in \eqref{Gr-SH}, we obtain the quantized Coulomb branch $H^{(G_\C^\OO\times\C_t^\times) \rtimes\C^\times_\hbar}_*(\cR)$ of the 3d $\cN=4^*$ theory.
If $G=\SU(N)/\Z_N$, then it is isomorphic to the spherical subalgebra $\SH_{\hbar,c}^{\textrm{tri}}(\frakS_N)$ of the trigonometric Cherednik algebra of type $A_{N-1}$.

Motivated by this construction, the mathematical definition of Coulomb branches of general 3d $\cN=4$ quiver theories has been given in \cite{Braverman:2016wma,Braverman:2016pwk,Braverman:2017ofm}.
In addition, representation theory of a 3d $\cN=4$ quantized Coulomb branch has been studied in the context of boundary conditions of 3d $\cN=4$ gauge theory in physics literature \cite{Bullimore:2016nji,Bullimore:2016hdc,Dedushenko:2018icp}. Their $K$-theoretic version will potentially lead to a vast generalization of DAHA as in \cite{braverman2016cyclotomic} and its representation theory.

\begin{figure}[ht]
\includegraphics[width=\textwidth]{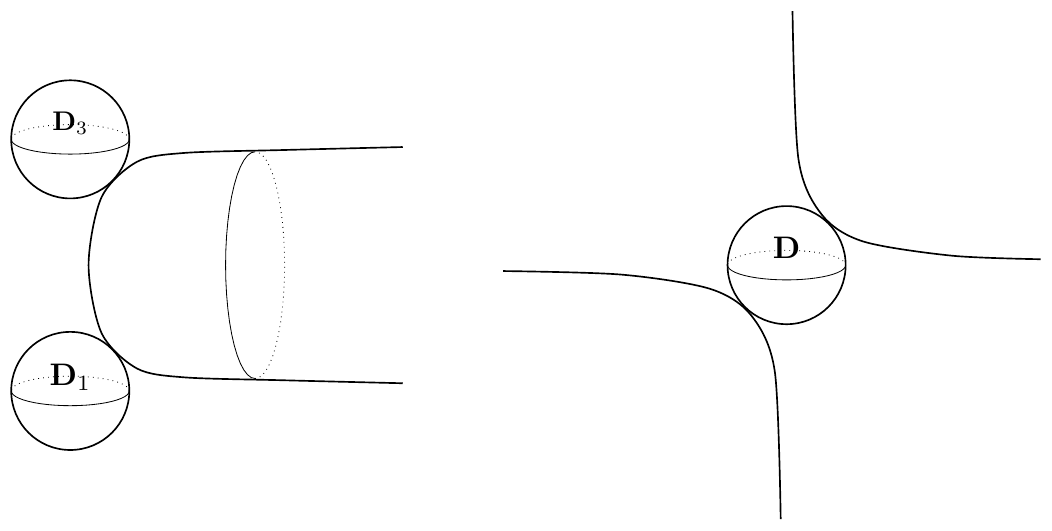}
\caption{(Left) Real locus of the Coulomb branch of 3d $\cN=4^*$ theory with $\SO(3)$ gauge group. \newline (Right) Real locus of the Coulomb branch of 3d $\cN=4$ SU(2) theory with one adjoint and one fundamental hypermultiplet.} \label{fig:3d-nilcone}
\end{figure}

\subsubsection*{Rational Cherednik algebra}

The Coulomb (and Higgs) branch of the 3d $\cN=4$ SQED with massless 2 flavors is $\C^2/\Z_2$. Once the mass parameters are turned on, the $A_1$ singularity of $\C^2/\Z_2$ is resolved, and the resulting manifold is called the Eguchi-Hanson space $T^*\bC\bfP^1$. From its complete \HK metric, it can be regarded as $S^1$ fibration over $\R^3$ where the fiber shrinks at two points in $\R^3$. For instance, the detail can be found in \cite[\S3.7]{Gukov:2008ve}. The deformation quantization of the Eguchi-Hanson space leads to the spherical subalgebra $\SH^{\textrm{rat}}_{\hbar,c}$ of the rational Cherednik algebra that is isomorphic to the universal enveloping algebra of the $\fraksl(2)$ Lie algebra \eqref{sl2}. The representation theory of this algebra has been exclusively  investigated from the viewpoint of brane quantization in \cite{Gukov:2008ve}.\footnote{Precisely speaking, we need the $1/2$ shift in $c$ as in \S\ref{sec:Bcc-SH} to connect to \cite{Gukov:2008ve}.} In particular, a brane supported on $\bC\bfP^1\subset T^* \bC\bfP^1 $ can exist only when $c=-(2\ell-1)/2$ like $\frakB_{\bfD_i}$ in \S\ref{sec:divisorTQFT}, and it gives rise to a finite-dimensional representation. On the other hand, branes supported on a cigar in Figure \ref{fig:3d-nilcone} or $T^*S^1\subset T^* \bC\bfP^1 $ bring about a discrete or principal series of $\SL(2,\R)$. We refer the reader to \cite{Gukov:2008ve} for more details.

Now let us consider the 3d $\cN=4$ $\U(N)$ gauge theory with one adjoint and one fundamental hypermultiplet. If the theory is massless, the Coulomb branch \cite{deBoer:1996ck} is the $N$-th symmetric product of $\bC^2$
$$
\cM_C^{\mathrm{3d}}\Big[\!\raisebox{-.35cm}{\begin{tikzpicture}
\node[draw,circle, inner sep=.22pt,scale=.8] (a) at (0,0){$\U(N)$};
\node[draw,rectangle, inner sep=2pt] (b) at (.8,0) {$1$};
    \draw  (a)   to [distance=3ex, in=150, out=210, loop] ();
    \draw  (a)   to  (b);
\end{tikzpicture}}~
\Big] =\mathrm{Sym}^N\bC^2~.
$$
Consequently, its coordinate ring is indeed isomorphic to the spherical subalgebra of the $\mathfrak{gl}(N)$ rational Cherednik algebra $\SH^{\textrm{rat}}(\mathfrak{gl}(N))$ at $c=0=\hbar$. In this sense, the spherical rational Cherednik algebra $\SH^{\textrm{rat}}(\mathfrak{gl}(N))$ can be interpreted as a ``Lie algebra'' of the spherical  DAHA $\SH(\mathfrak{gl}(N))$. Once mass parameters are turned on, the Coulomb branch is the Hilbert scheme of $N$-points on the affine plane $\C^2$, and the resolution of singularities can be understood as the Hilbert-Chow map
\be\label{Hilbert-Chow}\pi: \mathrm{Hilb}^N\bC^2 \to \mathrm{Sym}^N\bC^2~.\ee
If we remove the center-of-mass coordinate of $\mathrm{Hilb}^N$, we have  the spherical subalgebra $\SH^{\textrm{rat}}(\frakS_N)$ of the rational Cherednik algebra  of type $A_{N-1}$.

In the case of type $A_{N-1}$, it is known that $\textrm{Spec}(\SH^{\textrm{rat}}_{\hbar=0}(\frakS_N))$ has a unique compact holomorphic Lagrangian submanifold, called the punctual Hilbert scheme, which is the preimage $\pi^{-1}(0)$ of zero under the Hilbert-Chow map \eqref{Hilbert-Chow}.
As briefly discussed in \S\ref{sec:SUN}, this will give rise to the unique finite-dimensional representation of $\SH^{\textrm{rat}}(\frakS_N)$ when $c=-M/N$ with coprime $(M,N)$. This is consistent with the classification of finite-dimensional representations of the spherical rational Cherednik algebra \cite[Theorem 1.2 and 1.10]{berest2003finite}. (See also \cite{gordon2005rational,gordon2006rational} for a realization by geometric quantization.) Moreover, the finite-dimensional module is isomorphic to the lowest $a$-degree of the HOMFLY-PT homology of the $(M,N)$ torus knot \cite{Gorsky:2012mk}.

\bigskip

As a remark, we note that the Coulomb branch of 3d $\cN=4$ $\U(N)$ gauge theory with $\ell$ hypermultiplets in the fundamental representation \cite{deBoer:1996ck} is the $N$-th symmetric power of  the hypersurface $\cS_\ell$  determined by the equation $xy=z^{\ell}$ in $\C^3$, which can be expressed as
$$
\cM_C^{\mathrm{3d}}\Big[\!\raisebox{-.35cm}{\begin{tikzpicture}
\node[draw,circle, inner sep=.22pt,scale=.8] (a) at (0,0){$\U(N)$};
\node[draw,rectangle, inner sep=2pt] (b) at (.8,0) {$\ell$};
    \draw  (a)   to [distance=3ex, in=150, out=210, loop] ();
    \draw  (a)   to  (b);
\end{tikzpicture}}~
\Big] =\textrm{Sym}^N(\cS_\ell)~.
$$
The deformation quantization of the coordinate ring of the Coulomb branch leads to the \emph{cyclotomic rational Cherednik algebras} \cite{Kodera:2016faj,Okuda:2019emk}. The $\ell=0$ and $\ell=1$ specializations are the spherical subalgebras of the trigonometric and rational Cherednik algebras of type $\mathrm{GL}(N,\C)$, respectively. Hence, their representation theory can be similarly investigated in the context of brane quantization, which deserves future study.

\newpage

\bibliography{references}
\bibliographystyle{hyperamsalpha}

\end{document}